\documentclass[a4paper,12pt,cite,times,numbered,print,index]{Classes/PhDThesisPSnPDF}

\input{Preamble/preamble}

\title{Black hole microstates and supersymmetric localization}

\author{Seyed Morteza Hosseini}

\dept{Department of Physics G. Occhialini}

\university{Universit\`a degli Studi di Milano-Bicocca}
\crest{\includegraphics[width=0.2\textwidth]{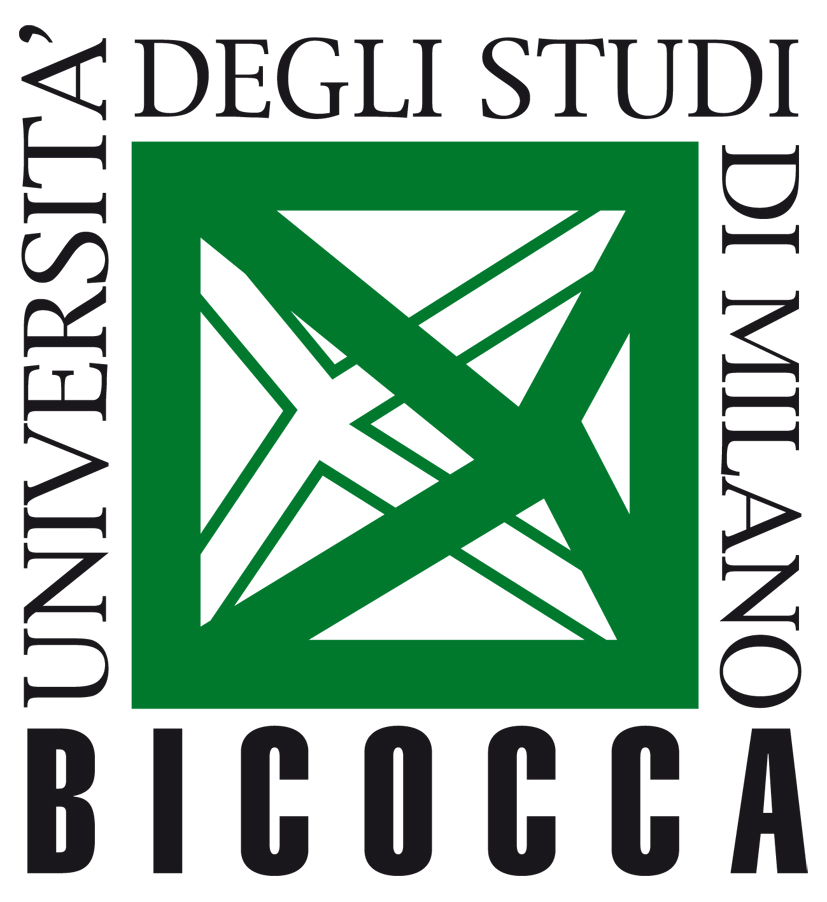}}

\advisor{Alberto Zaffaroni}
     
\advisorlinewidth{0.29\textwidth}

\degreetitle{Doctor of Philosophy}

\college{Curriculum in Theoretical Physics
Cycle XXX \par
Matricola: 798931 \par
SSD: FIS/02
}

\degreedate{February 2018} 

\subject{} \keywords{{} {PhD Thesis} {Physics} {University of Milano-Bicocca}}

\ifdefineAbstract
\fi

\ifdefineChapter
\fi

\begin{document}

\frontmatter

\maketitle


\begin{thesisright}


\textcircled{c} Copyright by Seyed Morteza Hosseini, 2018.\\
All Rights Reserved.

\end{thesisright}

\begin{dedication} 

Missione compiuta!

\end{dedication}


\begin{declaration}

This dissertation is a result of my own efforts. The work to which it refers is based on my PhD research projects:

\begin{enumerate}[label=\Roman*]
 \item
  ``Large $N$ matrix models for 3d ${\cal N}=2$ theories: twisted index, free energy and black holes,''
  with A.~Zaffaroni,
  JHEP {\bf 1608}, 064 (2016)
  \href{http://arxiv.org/abs/1604.03122}{[arXiv:1604.03122 [hep-th]]}.

 \item
  ``Large $N$ topologically twisted index: necklace quivers, dualities, and Sasaki-Einstein spaces,''
  with N.~Mekareeya,
  JHEP {\bf 1608}, 089 (2016)
  \href{http://arxiv.org/abs/1604.03397}{[arXiv:1604.03397 [hep-th]]}.

 \item
  ``The Cardy limit of the topologically twisted index and black strings in AdS$_{5}$,''
  with A.~Nedelin and A.~Zaffaroni,
  JHEP {\bf 1704}, 014 (2017)
  \href{https://arxiv.org/abs/1611.09374}{[arXiv:1611.09374 [hep-th]]}.
 \item
  ``An extremization principle for the entropy of rotating BPS black holes in AdS$_{5}$,''
  with K.~Hristov and A.~Zaffaroni,
  JHEP {\bf 1707}, 106 (2017)
  \href{https://arxiv.org/abs/1705.05383}{[arXiv:1705.05383 [hep-th]]}.
 \item
  ``Holographic microstate counting for AdS$_4$ black holes in massive IIA supergravity,''
  with K.~Hristov and A.~Passias,
  JHEP {\bf 1710}, 190 (2017)
  \href{https://arxiv.org/abs/1707.06884}{[arXiv:1707.06884 [hep-th]]}.
\end{enumerate}
I hereby declare that except where specific reference is made to the work of 
others, the contents of this dissertation are original and have not been 
submitted in whole or in part for consideration for any other degree or 
qualification in this, or any other university. 


\end{declaration}


\begin{acknowledgements}      

I would like to acknowledge all the people who have supported and helped me along the way to complete this thesis.

During my time in Milano I had the privilege of interacting with many talented individuals.
I would foremost like to thank my advisor Alberto Zaffaroni for his trust (although he will deny it), continuous support,
encouragement, inspiration and for sharing his ideas with me.
He has been a great mentor and the miglior amico for me.
A very BIG thank you to him for his patience not to kill me over the years!
I am also very grateful to Alessandro Tomasiello for his constant support
and for helping me to start my carrier as a PhD student in Milano-Bicocca.
Our discussions have always been a pleasure and I wish to thank him
for sharing his thoughts with me.
I am thankful to all of my wonderful colleagues who contributed to the results presented in this thesis:
Kiril Hristov, Noppadol Mekareeya, Anton Nedelin and Achilleas Passias.
I also wish to extend my gratitude towards Francesco Benini, Nikolay Bobev,
Gianguido Dall'Agata, Sameer Murthy and Alessandro Tomasiello 
for agreeing to be my thesis committee and for their valuable comments on the draft of this thesis.

Thanks to all group members I always enjoyed working here.
I would like to thank Antonio Amariti, Tom\'a\v{s} Je\v{z}o,
Sara Pasquetti, Silvia Penati, Valentin Reys, Paul Richmond, Gautier Solard and \'{A}lvaro V\'{e}liz-Osorio
for many useful discussions (both about physics and life).
I would like to thank my Iranian friends in Milano:
Mani Azadmand, Maryam Barzegar, Iman Mehrabinezhad, Ziba Najmi, Hedieh Rashidi and Somayeh Razzaghi
for supporting me throughout the years.
Special thanks go to Federica Crespi for her fantastic work over the years.
I would like to thank my fellow PhD students and young colleagues for a wonderful time together:
Francesco Azzurli, Stefano Baiguera, Luca Cassia, Riccardo Colabraro, Erika Colombo, Niccol\`o Cribiori,
G. Bruno De Luca, Silvia Ferrario Ravasio, Andrea Galliani, Carolina G\'omez, Federico Granata, Andrea Legramandi,
Gabriele Lo Monaco, Andrea Rota and Michele Turelli.

On a personal level I should admit that I am in love with Italy, Italians and EATalian food!
Last but definitely not least, I would like to express my deepest gratitude towards my family for their love, support and
for always believing in me.
I would not have made it to this point without them.

\end{acknowledgements}

\begin{abstract}

This thesis focuses mainly on understanding the origin of the Bekenstein-Hawking entropy
for a class of four- and five-dimensional BPS black holes in string/M-theory.
To this aim, important ingredients are holography and supersymmetric localization.

For supersymmetric field theories with at least four real supercharges
the Euclidean path integrals on $\Sigma_\fg \times T^{n}$ $(n=1,2)$ can be calculated exactly
using the method of supersymmetric localization.
The path integral reduces to a matrix integral that depends on background magnetic fluxes
and chemical potentials for the global symmetries of the theory.
This defines the topologically twisted index
which, upon extremization with respect to the chemical potentials,
is conjectured to reproduce the entropy of magnetically charged static BPS AdS$_{4/5}$ black holes/strings.

We solve a number of such matrix models both in three and four dimensions and provide general formulae in the large $N$ limit by which one can construct
the large $N$ matrix model associated with a particular quiver.
This is found by rewriting the matrix integral of quiver theories,
such that the poles in the Cartan contour integral are described
by the so-called Bethe ansatz equations and then by an effective twisted superpotential.
One of the main results of this thesis is a \emph{universal} formula -- named the index theorem --
for extracting the index from the twisted superpotential, leading to the conjecture
that the field theory extremization principle equals the attractor mechanism in 4D $\mathcal{N}=2$ gauged supergravity.

We then use these results to provide the microscopic realization of
the entropy of a class of BPS black holes in $\mathcal{N}=2$ gauged supergravity.
In particular, for the near-horizon geometries constructed in the four-dimensional
dyonic $\mathcal{N}=2$ gauged supergravity, that arises as a consistent truncation of
massive type IIA supergravity on $S^6$, we derive the Bekenstein-Hawking area law.
Finally, inspired by our previous results, we put forward an extremization principle for reproducing the Bekenstein-Hawking entropy
of a class of BPS electrically charged rotating black holes in AdS$_5\times S^5$.

\end{abstract}
%
%


\tableofcontents




\printnomenclature

\mainmatter


\chapter{Introduction and summary}  
\label{ch:1}
\ifpdf
    \graphicspath{{Chapter1/Figs/Raster/}{Chapter1/Figs/PDF/}{Chapter1/Figs/}}
\else
    \graphicspath{{Chapter1/Figs/Vector/}{Chapter1/Figs/}}
\fi

Black holes have more lessons in store for us.
We can assign a macroscopic entropy to a black hole
equal to one quarter of the horizon area measured in Planck
units \cite{Bekenstein:1972tm,Bekenstein:1973ur,Bekenstein:1974ax,Hawking:1974rv,Hawking:1974sw}.
More specifically, we find that
\be
 \label{intro:area:law}
 S_{\text{BH}} = \frac{\text{Area}}{4 G_{\text{N}}} \, ,
\ee
where $G_{\text{N}}$ is the Newton's gravitational constant.
The number of black hole microstates $d_{\rm micro}$ should then be given by
\bea
 \label{d:micro}
 d_{\rm micro} = \e^{S_{\rm BH}} \, .
\eea
But where are the microstates accounting for the black hole entropy?
A consistent theory of quantum gravity seems to be required in order to answer this question.
String theory, the prime candidate for such a theory,
provides a precise statistical mechanical interpretation of the Bekenstein-Hawking entropy \eqref{intro:area:law}
by representing the black holes as bound states of D-branes \cite{Dai:1989ua,Horava:1989ga,Polchinski:1995mt} and strings,
which allow us to construct and study supersymmetric gauge theories.
In the seminal paper \cite{Strominger:1996sh} the statistical entropy of asymptotically flat BPS black holes
(Reissner-Nordstr\"{o}m black holes) in type IIB string theory compactified to
five dimensions on $S^1 \times K3$ has been successfully identified
with the logarithm of the bound states degeneracy.

The situation for asymptotically anti de Sitter (AdS) black holes in ${\rm D} \geq 4$ is rather different
since we do not know the D-branes description of this class of black holes.
The question is then, what is the microscopic origin of the Bekenstein-Hawking entropy
for AdS black holes?
The gauge/gravity duality (holography) provides a nonperturbative definition of quantum gravity.
More concretely, one can formulate questions regarding quantum gravity in bulk spacetimes
as problems in lower-dimensional gauge theories living on their boundaries.
The gauge/gravity duality therefore renders a natural way of understanding the thermodynamic
black hole entropy \eqref{intro:area:law} in asymptotically AdS spacetimes
in terms of states in a dual conformal field theory (CFT).

The counting of microscopic BPS states has been recently achieved \cite{Benini:2015eyy,Benini:2016rke}
for a class of AdS black holes in four dimensions coming from compactification of M-theory on $S^7$,
thanks to the \emph{supersymmetric localization} \cite{Pestun:2007rz}.
See \cite{Cabo-Bizet:2017jsl} for a generalization of this setup to AdS$_4$ black holes with hyperbolic horizon.
Recent attempts to compute the logarithmic corrections to the entropy of this class of black holes can be found in \cite{Liu:2017vll,Jeon:2017aif}.

The localization principle allows one to reduce the path integral of the theory into
a finite-dimensional integral, \ie\;{\it matrix integral}, and compute some \emph{exact}
results for supersymmetric observables in strongly coupled quantum field theories (QFTs).
It thus gives a very precise predictions for the gauge/gravity duality.
In this thesis we report on the progresses in this direction.
Of particular importance is the \emph{topologically twisted index} of three-dimensional $\cN = 2$
and four-dimensional $\cN = 1$ gauge theories --- with an R-symmetry and $(\fg - 1) R_I \in \bZ$
($R_I$ being the R-charges of the matter fields) ---
on $\Sigma_\fg \times T^n$ with a partial topological A-twist \cite{Witten:1988ze,Witten:1991zz}, along
the genus $\fg$ Riemann surface $\Sigma_\fg$ \cite{Benini:2015noa}.
Here $T^n$ is a torus with $n= 1, 2$.

Before we move on, let us review some of the basic building blocks of this dissertation.

\section{Topological twist}
\label{ch:1:susy:twist}

If one attempts to put a supersymmetric theory on a curved spacetime without touching its Lagrangian,
supersymmetry will usually be broken by the curvature terms.
Thus, one has to be careful about how to define the theory on a curved spacetime.
In this section we give a brief overview of placing a QFT on a compact curved manifold $\cM$ while preserving some supersymmetry.
We refer the reader to \cite{Festuccia:2011ws} for a detailed analysis of rigid supersymmetric field theories on curved manifolds.

A uniform approach to this problem is to couple the flat space theory to off-shell supergravity.%
\footnote{The same supergravity theory can have different off-shell formulations,
and depending on which supercurrent multiplets exist in the matter theory,
one can couple the theory to different off-shell formulations.
This can lead to different classes of supersymmetric backgrounds for the same theory.
For instance in four dimensions, the Ferrara-Zumino multiplet \cite{Ferrara:1974pz}
can be coupled to ``old minimal supergravity'' \cite{Stelle:1978ye,Ferrara:1978em,Fradkin:1978jq}
while the R-multiplet (which contains the conserved R-current) to the ``new minimal supergravity'' \cite{Akulov:1976ck,Sohnius:1981tp}.}
Next, we require that the fluctuations in the gravitational field being decoupled, such that it remains a classical background.
This can be done by taking the rigid limit: sending the Newton's constant to zero while keeping fixed the background for
the metric and the bosonic auxiliary fields.
In this limit, we do not solve the equations of motion and we only impose the supersymmetry condition
(the vanishing of the gravitino variation):%
\footnote{Here, for simplicity, we only include the metric and the background gauge field $V_\mu$ that couples to the R-symmetry.}
\be
 \label{ch:1:GKS}
 D_\mu \epsilon
 = \left( \partial_\mu + \frac14 \omega_\mu^{a b} \gamma_{a b} + i V_{\mu} \right) \epsilon = 0 \, ,
\ee
where $V_\mu$ is a background gauge field for the R-symmetry.
We call this the \emph{generalized Killing spinor} (GKS) equation, and it should be solved for the spinors $\epsilon$.
The number of solutions for $\epsilon$ is the number of preserved supercharges.

If the theory at hand has a continuous non-anomalous R-symmetry one may perform a \emph{topological A-twist}
\cite{Witten:1988ze,Witten:1991zz} to solve the GKS equation.
The topological twist amounts to an identification of the spin connection with the R-symmetry
\be
 V_{\mu} \epsilon = \frac{i}{4} \omega_\mu^{a b} \gamma_{a b} \epsilon \, .
\ee
This corresponds to a flux $\frac{1}{2\pi} \int_{C_2} W$ for the R-symmetry curvature $W=\rd V$, where $C_2$ is any
compact two-cycle in $\cM$. Due to the R-symmetry background magnetic flux we will restrict to theories with integer R-charges.
Hence, a simple solution to \eqref{ch:1:GKS} is a constant spinor $\epsilon$.
In this background the spinors behave as scalars since the R-symmetry background has twisted their spin.
Let us stress that, the topological twist works only if the spacetime holonomy group can be embedded
into the R-symmetry group of the QFT in flat space.

The topological twist is precisely the way that branes, wrapping on nontrivial cycles in string/M-theory compactifications,
preserve supersymmetry \cite{Bershadsky:1995qy}.

\section[Magnetic black holes/strings in gauged supergravity]{Magnetic black holes/strings in gauged supergravity}
\label{ch:1:intro:magnetic:BH}

In this section we discuss the main features of the black holes/strings we consider.
We look for four-dimensional BPS black hole solutions, preserving at least two real supercharges,
which interpolate between an AdS$_2 \times \Sigma_\fg$ near-horizon region and an asymptotic AdS$_4$ vacuum.
They are the near-horizon geometry of $N$ D2$_k$/M2-branes wrapping $\Sigma_\fg$ \cite{Gauntlett:2001qs}.
The first examples of such analytic solutions, with $\fg > 1$ and constant scalar fields, found in \cite{Romans:1991nq}
and later studied further in \cite{Caldarelli:1998hg}.
The numeric evidence for black holes whose event horizons are Riemann surfaces of arbitrary genus ($\fg \neq 1$)
and have nontrivial scalar fields appeared first in \cite{Cucu:2003yk} but their analytic
construction was discovered in \cite{Cacciatori:2009iz}
(see also \cite{DallAgata:2010ejj, Hristov:2010ri,Halmagyi:2013sla,Katmadas:2014faa,Halmagyi:2014qza}).
They are solutions of ${\cal N} = 2$ supergravity with a gauged $\U(1)$ R-symmetry group.
The aforementioned topological twist consists of the cancellation of the spin connection on $\Sigma_\fg$
by the R-symmetry gauge vector field, and requires that the black hole solutions carry nontrivial Abelian magnetic charges $\fn^\Lambda$.
In general, the black holes can also support electric charges $q_\Lambda$.
In the case of the maximally supersymmetric $\SO(8)$ gauged supergravity (arising from Kaluza-Klein reduction of eleven-dimensional
supergravity on $S^7$) this allows one to consider only an Abelian $\U(1)^4$ truncation
$(\Lambda=1,2,3,4)$, often called ``STU model'',
as was the case in \cite{Benini:2015eyy,Benini:2016rke}.

The dual boundary theory is a relevant deformation of some 3D supersymmetric Chern-Simons-matter theories,
semi-topologically twisted by the presence of the magnetic charges.
Reducing down to $S^1$, the theory gives rise to a supersymmetric quantum mechanics and
the partition function on $\Sigma_\fg \times S^1$ computes the Witten index of the $\cN=2$
quantum mechanical sigma model \cite{Hori:2014tda,Hwang:2014uwa,Cordova:2014oxa}.
This naturally leads to a renormalization group (RG) flow across dimensions,
connecting the CFT$_3$ dual to asymptotic AdS$_4$
vacuum in the ultraviolet (UV) and the CFT$_1$ dual to the near-horizon
AdS$_2 \times \Sigma_\fg$ geometry in the infrared (IR).
Along the RG flow, the UV superconformal R-symmetry of the three-dimensional theory
generically mixes with the flavor symmetries
and, at the one-dimensional fixed point, it becomes a linear combination of the reference
R-symmetry and a subgroup of the flavor symmetries.
The R-symmetry that sits in the $\mathfrak{su}(1,1|1)$ superconformal algebra in the IR
is determined by extremizing the topologically twisted index,
whose value at the extremum is the regularized number of ground states.%
\footnote{We are evaluating an equivariant holomorphic index that provides
a regularization for the continuum of the ground states of the corresponding IR quantum mechanics.}
This is the so-called $\cI$-\emph{extremization principle} proposed in \cite{Benini:2015eyy,Benini:2016rke}
(see section \ref{sec:intro:I-extremization}).

We will also consider \emph{magnetically} charged BPS black strings in five-dimensional $\cN = 2$ Abelian gauged supergravity.
Black string solutions corresponding to D3-branes at a Calabi-Yau singularity 
have been recently studied in details in \cite{Benini:2012cz,Benini:2013cda}
(see also \cite{Bobev:2014jva,Benini:2015bwz,Klemm:2016kxw,Amariti:2016mnz}).
They can be viewed as domain-walls interpolating between maximally supersymmetric AdS$_5$ vacuum at infinity
and the near-horizon AdS$_3 \times \Sigma_\fg$ geometry.
This can be interpreted as an RG flow from an UV four-dimensional $\cN=1$ superconformal field theory (SCFT)
and an IR two-dimensional $\cN=(0,2)$ one.
The two-dimensional CFT is obtained by compactifying the four-dimensional theory on $\Sigma_\fg$ 
with a topological twist parameterized by a set of background magnetic charges.
The right-moving central charge of the two-dimensional CFT has been computed in
\cite{Benini:2012cz,Benini:2013cda,Bobev:2014jva,Benini:2015bwz}, and successfully compared with the
supergravity result for a variety of models.


\section{Supersymmetric localization}
\label{ch:1:sec:susy:localization}

The partition function of a local QFT is given by the Euclidean Feynman path integral
\be
 \label{ch:1:partition function}
 Z = \int \cD \phi \, \e^{- S[\phi]} \, ,
\ee
where $\phi$ denotes the set of fields in the theory and $S$ is the action functional.
In order to evaluate the partition function of a theory one needs to integrate over all possible classical field configurations.
Thus, besides Gaussian free theories, the partition function is too hard to compute.
Most of the time one has to work with the so called saddle-point approximation:
expanding the action around free fields and using perturbation theory.
However, if a theory is supersymmetric one may use the technique of supersymmetric localization
and compute the partition function exactly (in the sense of reducing them to a matrix model).

The argument for localization proceeds as follows.
Consider path integrals of supersymmetric gauge theories on compact manifolds $\cM$:
\be
 \label{ch:1:susy:path integral}
Z_{\cM} = \int \cD \phi \, \e^{- S[\phi]} \, ,
\ee
and let $\delta$ be a Grassmann-odd symmetry of these theories $(\delta S = 0)$.
We assume that $\delta$ is not anomalous and thus the measure of the path integral is invariant under $\delta$.
Consider now a deformation of the theories by a $\delta$-exact term
\be
 \label{ch:1:deformed path integral}
 Z_\cM (t) = \int \cD \phi \, \e^{- S[\phi] - t \delta V} \, ,
\ee
with $t \in \bR_{>0}$.
It is easy to see that the value of the partition function is independent of $t$:
\be
 \label{ch:1:localization:argument}
 \frac{\partial}{\partial t} Z_\cM (t)
 = - \int \cD \phi \, \e^{- S[\phi] - t \delta V} \delta V
 = - \int \cD \phi \delta \left( \e^{- S[\phi] - t \delta V} V \right) = 0 \, ,
\ee
and hence one can evaluate it as $t \to \infty$. Here we assumed that the integral decays sufficiently fast in field space
so there are no boundary terms at infinity. In this limit, if $\delta V$ has a positive definite bosonic part $(\delta V)_{\text{B}}$,
the integral \emph{localizes} to a submanifold of field space where
\be
 \label{ch:1:localizing locus}
 (\delta V)_{\text{B}} (\phi_0) = 0 \, .
\ee
Obviously, the above argument still holds true if we insert $\delta$-exact operators (observables), \ie\,$\cO=\delta \cX$.
One may insert both local operators (located at a point in spacetime)
and nonlocal operators (located along a submanifold) in the path integral.
They can be defined in different ways: order operators such as Wilson lines;
disorder operators such as monopole and 't Hooft line operators in three and four dimensions, respectively;
defect operators with their own actions and coupled to the bulk, \eg
$\,S_{\text{D}} = \int_\gamma \rd t \bar \psi \left( \partial_t - i A_t \right) \psi \, .$
These operators have vanishing vacuum expectation values since
\be
 \langle \cO \rangle_\cM = \int \cD \phi \, \e^{- S[\phi]} \delta \cX
 = \int \cD \phi \delta \left( \e^{-S[\phi]} \cX \right) = 0 \, .
\ee

Let us parameterize the fields around the localizing locus \eqref{ch:1:localizing locus}:
\be
 \phi = \phi_0 + t^{-1/2} \hat \phi\, ,
\ee
where the factor $t^{-1/2}$ is chosen because when the deformation term dominates at large $t$,
the kinetic term should be canonically normalized with no powers of $t$.
For large $t$, we can Taylor expand the action around $\phi_0$ as
\be
 S + t \delta V = S[\phi_0] + (\delta V)^{(2)} [\hat \phi] + \cO(t^{-1/2}) \, .
\ee
Only the value of the classical action on the saddle-point configuration $S[\phi_0]$
and the quadratic expansion of $\delta V$ around the fixed points matters.
We thus obtain the localization formula:
\be
 Z_\cM = \int_{\delta V_{\text{B}} (\phi_0) =0} \cD\phi_0 \, \e^{-S[\phi_o]} \,Z_{1\text{-loop}} [\phi_0] \, ,
\ee
by Gaussian integration. $Z_{1\text{-loop}}$ is the one-loop determinant (the ratio of fermionic and bosonic determinants)
of the deformation term $\delta V$.

\section{Supersymmetric Chern-Simons-matter theories in three dimensions}

In this section we review the construction of 3D supersymmetric Chern-Simons-matter
theories on $S^2 \times S^1$ \cite{Benini:2015noa}.%
\footnote{The result of localization for 4D $\cN=1$ field theories on $S^2 \times T^2$
is simply the elliptic generalization of the result in three dimensions \cite{Closset:2013sxa,Benini:2015noa}.
We do not consider the detailed construction of this class of theories here; rather,
we present the final formula for the matrix model.}
In order to preserve supersymmetry on this background one must perform a topological twist on $S^2$.
Thus, one of the important assumptions is that the theory should have a continuous $\U(1)_R$ symmetry.

After summarizing our conventions for spinors we describe the $S^2 \times S^1$ background of interest
and the background fields that we need to turn on in order to preserve some supersymmetry.
Then we concentrate on the supersymmetry variations corresponding to the topologically twisted theory.
We introduce the anticommuting supercharges by trading the anticommuting Killing spinors
in $\delta_{\epsilon}$ and $\delta_{\bar{\epsilon}}$ for their commuting counterparts.
Finally, we write the supersymmetric Lagrangians for gauge and matter fields.
Lagrangians invariant under the supersymmetry transformations were studied in \cite{Hama:2010av, Hama:2011ea}.
We refer the reader to \cite{Klare:2012gn, Closset:2012ru} for a more systematic analysis of supersymmetry on three-manifolds.
In this section we will closely follow the presentation of appendix B of \cite{Benini:2013yva}.

\subsection{Spinor conventions}

We will work with Euclidean space.
In Euclidean signature all fields get complexified and we will consider $\;\bar{}\;$-ed fields as independent fields.
The Dirac spinors are in the $\bf{2}$ of $\SU(2)$ with the index structure: $\psi^\alpha$, $\bar \psi^\alpha$.
We use the standard Pauli’s matrices for the Dirac matrices in vielbein space:
$\gamma^a = \smat{0 & 1 \\ 1 & 0}$, $\smat{0 & -i \\ i & 0}$, $\smat{1 & 0 \\ 0 & -1}$,
and also $\gamma^{a b} = \frac12 (\gamma^a \gamma^b - \gamma^b \gamma^a)$.
We take the charge conjugation matrix $C = -i \epsilon_{\alpha \beta} = \gamma_2$
(where $\epsilon^{1 2} = \epsilon_{1 2} = 1$) so that $C = C^{-1} = C^\dag = - C^\trans = - C^*$.
The charge conjugate spinors are $\epsilon^c = C \epsilon^*$ and  $\epsilon^{c\dag} = \epsilon^\trans C$.
Note that $\epsilon^{cc} = - \epsilon$.
Writing the spinor indices explicitly, the bilinear products are constructed as
\be
 \bar \epsilon \lambda \equiv \bar \epsilon^\alpha C_{\alpha \beta} \lambda^\beta \, , \qquad
 \bar \epsilon \gamma^\mu \lambda \equiv \bar \epsilon^\alpha \left( C \gamma^{\mu} \right)_{\alpha \beta} \lambda^\beta \, ,
 \qquad \etc. \, .
\ee
Noticing that the charge conjugation matrix $C$ is antisymmetric and $C \gamma^\mu$ are symmetric, it is easy to check that
\be
 \bar \epsilon \lambda = \lambda \bar \epsilon \, , \qquad \bar \epsilon \gamma^\mu \lambda = - \lambda \gamma^\mu \bar \epsilon \, .
\ee
$C \gamma^{\mu \nu}$ are also symmetric since $\gamma^{\mu \nu} = i \epsilon^{\mu \nu \rho} \gamma_\rho/ \sqrt{g}$
(where $\epsilon^{\mu \nu \rho}=1$).
We also have the following Fierz identity for anticommuting 3D Dirac fermions
\bea
 \left( \bar \lambda_1 \lambda_2 \right) \lambda_3 =
 - \frac12 \left( \bar \lambda_1 \lambda_3 \right) \lambda_2
 - \frac12 \left( \bar \lambda_1 \gamma^\rho \lambda_3 \right) \gamma_\rho \lambda_2 \, .
\eea

\subsection{Background geometry}
\label{ch1:susy background}

We will consider three-dimensional $\cN=2$ field theories on $S^2 \times S^1$ with the round metric
\be
 \rd s^2 = R^2 \big( \rd\theta^2 + \sin^2\theta\, \rd\varphi^2 \big) + \beta^2 \rd t^2 \, .
\ee
In order to preserve supersymmetry we perform a partial topological A-twist on $S^2$.
The vielbein one-forms are $e^1 = R\, \rd\theta$, $e^2 = R\sin\theta\, \rd\varphi$ on $S^2$
and $e^3 = \beta\, \rd t$ with $t \sim t+1$ on $S^1$.
As we already discussed in the previous section, to perform the topological twist
we turn on a background gauge field that couples to the $\U(1)_R$ symmetry
such that it cancels the spin connection for half of the supercharges:
\be
 \label{ch1:twist:R-charge}
 V = \frac12 \cos\theta\, d\varphi = - \frac12 \omega^{12} \, .
\ee
This corresponds to a magnetic flux $\frac1{2\pi} \int_{S^2} W = - 1$ for the R-symmetry curvature $W=\rd V$.
Here $\omega_\mu^{a b}$ is the spin connection.
In our notation the supersymmetry spinor $\epsilon = \smat{\epsilon_+ \\ \epsilon_-}$
has R-charge $-1$ so that the GKS equations is solved by
\be
 \label{ch1:sol:GKS}
 \epsilon = \Big( \begin{array}{c} \epsilon_+ \\ 0 \end{array} \Big)
 \qquad\quad \text{with} \qquad\quad \epsilon_+ = \text{const.} \, .
\ee
Due to the R-symmetry background flux, we will limit to theories with integer R-charges.

\paragraph*{$\fakebold{\Omega}$-background on $\fakebold{S^2\times S^1}$.}%
If the metric on $S^2$ has a rotational symmetry around an axis,
we may introduce a one-parameter family of deformations called the \emph{$\Omega$-background} in \cite{Closset:2014pda},
\be
 \label{ch1:Omega:bg}
 \rd s^2 = R^2 \big( \rd\theta^2 + \sin^2\theta (\rd\varphi - \varsigma\, \rd t)^2 \big) + \beta^2 \rd t^2 \, .
\ee
We refer to the parameter $\varsigma$ as an angular momentum parameter and note that $\varsigma=0$ is the round $S^2\times S^1$.
Here we take vielbein
\be
 e\ud{a}{\mu} = \mat{ R & 0 & 0 \\ 0 & R\sin\theta & -R\varsigma \sin\theta \\ 0 & 0 & \beta} \, ,
\ee
and the coordinates have the same periodicity as before, \ie\;$t \sim t+1$, $\varphi \sim \varphi + 2\pi$.
We can still perform the topological twist by turning on the background connection $V = - \frac12 \omega^{12}$
coupled to the R-symmetry current, and the covariantly constant spinor \eqref{ch1:sol:GKS}.
We call this the ``refined'' case.

\paragraph*{$\fakebold{\Sigma_\fg \times S^1}$ background.}%
We can preserve supersymmetry on $\Sigma_\fg \times S^1$ where $\Sigma_\fg$ is a Riemann surface of
arbitrary genus $\fg$, with the same choice of $V=- \frac12 \omega^{12}$ and the same covariantly constant spinor \eqref{ch1:sol:GKS}.
In general the R-symmetry field strength is given by
\be
 W_{12} = \frac12 \varepsilon^{\mu\nu} W_{\mu\nu} =
 - \frac14 R_s \qquad\quad\text{and}\quad\qquad \frac1{2\pi} \int_{\Sigma_\fg} W = \fg - 1 \, ,
\ee
where $R_s$ is the scalar curvature on $\Sigma_\fg$.

\subsection{Supersymmetry transformations}

Let us define the gauge field strength $F_{\mu \nu} = \partial_\mu A_\nu - \partial_\nu A_\mu - i \left[ A_\mu , A_\nu \right]$
and the gauge and metric covariant derivative $D_\mu = \nabla_\mu - i A_\mu$,
where $\nabla_\mu = \partial_\mu + \frac14 \omega_\mu^{ab} \gamma_{ab}$ is the metric covariant derivative.
We will also turn on a background connection coupled to the $\U(1)_R$ symmetry current and thus
$
 D_\mu = \nabla_\mu - i A_\mu - i V_\mu \, .
$

The basic multiplets of three-dimensional $\cN=2$ supersymmetry are the gauge (vector) multiplet, the chiral multiplet and
the anti-chiral multiplet, arising by dimensional reduction to three dimensions of the four-dimensional $\cN=1$ supersymmetry multiplets.
The R-charge and the scaling weight of the various fields are:
\begin{equation*}
\label{table of charges used}
\begin{array}{c|cc|ccccc|cccccc}
 & \epsilon & \bar \epsilon & A_\mu & \sigma & \lambda & \bar \lambda & D & \phi & \bar \phi & \psi & \bar \psi & F & \bar F\\
 \hline
 R & -1 & 1 & 0 & 0 & -1 & 1 & 0 & r & -r & r-1 & 1-r & r-2 & 2-r\\
 \Delta & 1/2 & 1/2 & 1 & 1 & 3/2 & 3/2 & 2 & r & r & r+1/2 & r+1/2 & r+1 & r+1 \\
\end{array}
\end{equation*}
%
%
%
The gauge multiplet $\cV$, in Lorentzian signature, includes a vector $A_\mu$, one real scalar $\sigma$,
a Dirac spinor $\lambda$ and a real auxiliary scalar $D$, all in the adjoint representation of the gauge group $G$. 
They transform under supersymmetry as
\bea
 \label{gauge:susy:variations}
 \delta A_\mu &= - \frac i2 (\bar\epsilon \gamma_\mu \lambda - \bar\lambda \gamma_\mu \epsilon) \, ,
 \hspace{5cm} \delta \sigma = \frac12 (\bar\epsilon \lambda - \bar\lambda \epsilon) \, , \\
 \delta\lambda &= \frac12 \gamma^{\mu\nu} \epsilon F_{\mu\nu} - D \epsilon + i \gamma^\mu \epsilon D_\mu \sigma
 + \frac{2i}3 \sigma \gamma^\mu D_\mu \epsilon \, , \\
 \delta\bar\lambda &= \frac12 \gamma^{\mu\nu} \bar\epsilon F_{\mu\nu} + D\bar\epsilon
 - i \gamma^\mu \bar\epsilon D_\mu \sigma - \frac{2i}3 \sigma \gamma^\mu D_\mu \bar\epsilon \, , \\
 \delta D &= - \frac i2 \bar\epsilon \gamma^\mu D_\mu \lambda - \frac i2 D_\mu \bar\lambda \gamma^\mu \epsilon
 + \frac i2 [\bar\epsilon \lambda, \sigma] + \frac i2 [\bar\lambda \epsilon,\sigma]
 - \frac i6 ( D_\mu \bar\epsilon \gamma^\mu \lambda + \bar\lambda \gamma^\mu D_\mu \epsilon) \, ,
\eea
The chiral multiplet $\Phi$ consists of a complex scalar $\phi$, a Dirac spinor $\psi$ and a complex auxiliary scalar $F$,
all in a representation $\fR$ of the gauge group.
The anti-chiral multiplets $\wb \Phi = (\bar \phi, \bar \psi, \bar F)$ has the same components as a chiral multiplet,
all in the conjugate representation $\wb\fR$.
The supersymmetry transformations of a chiral multiplet are given by
\bea
  \label{matter:susy:variations}
 \delta\phi &= \bar\epsilon \psi \, , \qquad\qquad
 \delta\psi = i \gamma^\mu \epsilon\, D_\mu\phi + i \epsilon \sigma \phi
 + \frac{2 i r}3 \gamma^\mu D_\mu \epsilon \, \phi + \bar\epsilon F \, , \\
 \delta\bar\phi & = \bar\psi \epsilon \, , \qquad\qquad
 \delta\bar\psi = i \gamma^\mu \bar\epsilon \, D_\mu\bar\phi + i \bar\epsilon \bar\phi \sigma
 + \frac{2 i r}3 \gamma^\mu D_\mu \bar\epsilon \, \bar\phi + \epsilon \bar F \, , \\
 \delta F &= \epsilon \big( i \gamma^\mu D_\mu \psi - i \sigma \psi - i \lambda \phi \big)
 + \frac{i (2 r - 1)}{3} D_\mu\epsilon \, \gamma^\mu \psi \, , \\
 \delta \bar F &= \bar\epsilon \big( i \gamma^\mu D_\mu \bar\psi - i \bar\psi \sigma + i \bar\phi \bar\lambda \big)
 + \frac{i (2 r - 1)}{3} D_\mu \bar\epsilon \, \gamma^\mu \bar\psi \, .
\eea
Here $\epsilon$ and $\bar{\epsilon}$ are independent spinors fulfilling the Killing spinor equations
\be
 \label{ch1:Killing:spinor:eq}
 D_\mu \epsilon = \gamma_\mu \hat\epsilon \, ,
 \qquad\qquad D_\mu \bar\epsilon = \gamma_\mu \hat{\bar\epsilon} \, ,
\ee
for some other spinors $\hat\epsilon, \hat{\bar\epsilon}$.
Splitting $\delta = \delta_\epsilon + \delta_{\bar\epsilon}$, these transformations realize the $\mathfrak{su}(1|1)$ superalgebra off-shell.
In order for the algebra to close, the Killing spinors need to satisfy the additional constraints
\be
 \label{ch1:additional:constraint}
 \gamma^\mu \gamma^\nu D_\mu D_\nu \epsilon
 = -\frac38 \big(R_s - 2 i W_{\mu\nu} \gamma^{\mu\nu}\big) \, \epsilon \, ,
 \quad\quad \gamma^\mu \gamma^\nu D_\mu D_\nu \bar\epsilon
 = -\frac38 \big(R_s + 2 i W_{\mu\nu} \gamma^{\mu\nu}\big) \, \bar\epsilon\, ,
\ee
with the same functions $R_s$ and $W_{\mu\nu}$ \cite{Hama:2010av,Hama:2011ea}.
Consistency requires that $R_s$ is the scalar curvature of the three-manifold and
$W_{\mu\nu} = \partial_\mu V_\nu - \partial_\nu V_\mu$ is the background gauge field strength.
Then, the resulting algebra reads
\be
 \label{ch:1:susy:commutators}
 \left[\delta_\epsilon, \delta_{\bar\epsilon}\right] = \cL^A_\xi  + i\Lambda +  \rho \Delta + i\alpha R \, ,
 \qquad\qquad \left[\delta_\epsilon, \delta_{\epsilon} \right]=0, \qquad\qquad \left[\delta_{\bar\epsilon}, \delta_{\bar\epsilon}\right]=0 \;,
\ee
where the parameters are defined as
\bea
 \xi^\mu & = i \bar\epsilon \gamma^\mu \epsilon \, , \qquad\qquad &
 \rho & = \frac i3 ( D_\mu \bar\epsilon \gamma^\mu \epsilon + \bar\epsilon \gamma^\mu D_\mu \epsilon) = \frac13 D_\mu \xi^\mu \,  , \\
 \Lambda & = \bar\epsilon \epsilon \sigma \, , \qquad\qquad &
 \alpha & = - \frac 13 ( D_\mu \bar\epsilon \gamma^\mu \epsilon - \bar\epsilon \gamma^\mu D_\mu \epsilon) - \xi^\mu V_\mu \, .
\eea
The algebra \eqref{ch:1:susy:commutators} implies that the commutator $\left[\delta_\epsilon, \delta_{\bar\epsilon}\right]$
is a sum of translation by $\xi_\mu$, a gauge transformation by $\Lambda$, a dilation by $\rho$ and a vector-like R-rotation by $\alpha.$
The Lie derivative $\cL_\xi$ along the Killing vector field $\xi$ is a \emph{metric independent} derivation.
On a $p$-form it is defined as $\cL_\xi = \{\rd, \iota_\xi\}$ in terms of the contraction $\iota_X$;
using the normalization $\alpha = \frac1{p!} \alpha_{\mu_1 \cdots \mu_p} \rd x^{\mu_1 \cdots \mu_p}$,
in components we have
\be
 [\cL_\xi \alpha]_{\mu_1 \cdots \mu_p} =
 \xi^\mu \partial_\mu \alpha_{\mu_1 \cdots \mu_p}
 + p \, (\partial_{[\mu_1} \xi^\mu) \, \alpha_{\mu | \mu_2 \cdots \mu_p]} \, .
\ee
The Lie derivative of spinors \cite{Kosmann:1972} reads (see \cite{Godina:2003tc} for a more thorough discussion.)
\be
\cL_\xi \psi = \xi^\mu \nabla_\mu \psi + \frac14 \nabla_\mu \xi_\nu \, \gamma^{\mu\nu} \psi \, .
\ee
We denoted by $\cL^A_\xi$ the ``gauge covariant'' version of the Lie derivative along $\xi$.
It acts on sections of some gauge bundle. On tensors it is obtained by replacing
$\partial_\mu \to \partial^A_\mu = \partial_\mu - i A_\mu$,
while on spinors it is obtained by replacing $\nabla_\mu \to \nabla_\mu^A$ in the first term.
The gauge covariant Lie derivative of the connection, which does not transform as a section of the adjoint bundle, is defined as
\be
 \cL_\xi^A A = \cL_\xi A - \rd^A(\iota_\xi A) \, ,\quad \quad
 (\cL_\xi^A A)_\mu = \xi^\rho F_{\rho\mu} = \xi^\rho \big( 2\partial_{[\rho} A_{\mu]} - i[A_\rho, A_\mu] \big) \, .
\ee

\subsection{Localizing supercharges}

We now proceed with the construction of two complex supercharges $Q,\wt Q$ in terms of \emph{commuting} covariantly constant spinors
$\epsilon$ and $\tilde\epsilon = - C \bar \epsilon^*$.%
\footnote{$D_\mu \epsilon = 0$, $\gamma_3 \epsilon = \epsilon$ and similarly for $\tilde \epsilon$, with the same R-charge $-1$.}
They are built as follows
\be
 \delta = \delta_\epsilon + \delta_{\bar\epsilon} = \epsilon^\alpha Q_\alpha + \bar\epsilon^\alpha \tilde Q_\alpha \, ,\qquad\qquad
 Q = \epsilon^\alpha Q_\alpha \;,\qquad\qquad \tilde Q = \tilde\epsilon^{c\, \alpha} \tilde Q_\alpha = - (\tilde\epsilon^\dag C)^\alpha \tilde Q_\alpha \;.
\ee
For convenience, we also rewrite $\;\bar{}\;$-ed spinors as $\bar\lambda = C (\lambda^\dag)^\trans$.
We obtain for the vector multiplet:
\bea
\label{ch1:susy:transformations:gauge:multiplet}
 Q A_\mu &= \frac i2 \lambda^\dag \gamma_\mu \epsilon \, , \hspace{2cm}
 Q\lambda = \frac12 \gamma^{\mu\nu} \epsilon F_{\mu\nu} - D\epsilon
 + i \gamma^\mu \epsilon \, D_\mu\sigma + \frac{2i}3 \sigma \gamma^\mu D_\mu\epsilon \, , \hspace{-4.3cm} \\
 \wt Q A_\mu &= \frac i2 \tilde\epsilon^\dag \gamma_\mu \lambda  \, , \hspace{2cm}
 \wt Q \lambda^\dag = - \frac12 \tilde\epsilon^\dag \gamma^{\mu\nu} F_{\mu\nu}
 + \tilde\epsilon^\dag D + i \tilde\epsilon^\dag \gamma^\mu D_\mu\sigma
 + \frac{2i}3 D_\mu \tilde\epsilon^\dag \gamma^\mu \sigma \, , \hspace{-6cm} \\
 QD &= - \frac i2 D_\mu \lambda^\dag \gamma^\mu \epsilon
 + \frac i2 [\lambda^\dag \epsilon, \sigma ] - \frac i6 \lambda^\dag \gamma^\mu D_\mu\epsilon \, , &
 Q \lambda^\dag &= 0 \, , \quad \quad &
 Q\sigma &= - \frac12 \lambda^\dag \epsilon \, , \\
 \wt Q D &= \frac i2 \tilde\epsilon^\dag \gamma^\mu D_\mu \lambda
 + \frac i2 [\sigma, \tilde\epsilon^\dag\lambda] + \frac i6 D_\mu \tilde\epsilon^\dag \gamma^\mu \lambda \, , &
 \wt Q \lambda &= 0 \, , &
 \wt Q \sigma &= - \frac12 \tilde\epsilon^\dag \lambda \, .
\eea
and for the chiral multiplet:
\bea
\label{ch1:susy:transformations:chiral:multiplet}
 Q\phi &= 0 \, , &
 \wt Q \phi &= - \tilde\epsilon^\dag \psi \, , \\
 Q\phi^\dag &= \psi^\dag \epsilon \, , &
 \wt Q \phi^\dag &= 0 \, , \\
 Q\psi &= \big( i \gamma^\mu D_\mu \phi + i\sigma \phi) \epsilon + \frac{2ir}3 \phi \, \gamma^\mu D_\mu \epsilon \, , &
 \wt Q \psi &= \tilde\epsilon^\dag F \, , \\
 \wt Q \psi^\dag &= \tilde\epsilon^\dag \big( -i\gamma^\mu D_\mu \phi^\dag
 + i \phi^\dag \sigma \big) - \frac{2ir}3 D_\mu \tilde\epsilon^\dag \gamma^\mu \phi^\dag \, , &
 Q \psi^\dag &= - \epsilon^{c \dag} F^\dag \, , \\
 QF &= \epsilon^{c \dag} \big( i \gamma^\mu D_\mu\psi - i \sigma\psi - i \lambda\phi \big)
 + \frac{i(2r-1)}3 D_\mu \epsilon^{c \dag} \gamma^\mu \psi \, , &
 \wt Q F &= 0 \, , \\
 \wt Q F^\dag &= \big( - i D_\mu \psi^\dag \gamma^\mu - i \psi^\dag \sigma
 + i \phi^\dag \lambda^\dag \big) \tilde\epsilon^c
 - \frac{i(2r-1)}3 \psi^\dag \gamma^\mu D_\mu \tilde\epsilon^\dag \, , \quad \; &
 Q F^\dag &= 0 \, .
\eea
In the localization computation we will use the supercharge $\cQ \equiv Q + \wt Q$.

\subsection{Supersymmetric Lagrangians}
\label{ch:1:ssec:susy:Lag}

One can easily construct supersymmetric actions on $S^2 \times S^1$,
\be
 S = \int \rd^3 x \sqrt{g} \, \cL \, ,
\ee
with $Q \cL = \wt{Q} \cL =0$ up to a total derivative.
In the following, we will consider the Yang-Mills Lagrangian, the various Chern-Simons terms,
the matter kinetic Lagrangian and superpotential interactions.
When the theory has some continuous flavor symmetry $J^f$, we may turn on supersymmetric backgrounds for the bosonic fields
in the corresponding flavor vector multiplet $\cV^f = (A_\mu^f , \sigma^f, \lambda^f, \lambda^{f\dag} , D^f)$.
This accounts for turning on magnetic flavor fluxes on $S^2$, flat flavor connections on $S^1$, and real masses.
We remark that whenever the gauge group has an Abelian factor, the flavor group includes a ``topological'' $\U(1)$ subgroup.

We work in Euclidean signature and this requires to double the number of degrees of freedom in a given multiplet.
This can be realized formally by considering each field and its Hermitian conjugate as transforming independently under supersymmetry.
When performing the path integral over the fields of a multiplet,
we will have to choose a middle-dimensional contour reducing the number of real independent fields to its canonical value.
We pick the ``natural'' one, in which ``real'' fields are real while $^\dag$ is identified with the adjoint operation.
We call such a contour the \emph{real contour}. In our conventions all Lagrangian terms have a positive definite real bosonic part,
which guarantee the convergence of the path integral.

The supersymmetric Yang-Mills (YM) Lagrangian is
\be
 \label{ch1:YM:Lag}
 \cL_\text{YM} = \Tr\bigg[ \frac14 F_{\mu\nu} F^{\mu\nu} + \frac12 D_\mu\sigma D^\mu\sigma
 + \frac12 D^2 - \frac i2 \lambda^\dag \gamma^\mu D_\mu\lambda - \frac i2 \lambda^\dag [\sigma,\lambda] \bigg] \, .
\ee
The Lagrangian $\cL_\text{YM}$ can be written as a $\cQ$-exact term, up to total derivatives:
\be
 Q \wt Q \Tr \Big( \tfrac12 \lambda^\dag\lambda + 2D\sigma \Big) \,\cong\,
 \tilde\epsilon^\dag\epsilon\, \cL_\text{YM} \, .
\ee
When performing localization, the path integral is only sensitive to the cohomology of $\cQ$ meaning that
$\cQ$-exact operators do not affect the integral. Thus, this Lagrangian can be used in the localization procedure
as the deformation term (see section \ref{ch:1:sec:susy:localization}).

The supersymmetric completion of a bosonic Chern-Simons (CS) action, for each simple or Abelian factor, reads
\be
 \label{ch1:CS:Lag}
 \cL_{\text{CS}} = - \frac{ik}{4\pi} \Tr \bigg[ \epsilon^{\mu\nu\rho} \Big( A_\mu \partial_\nu A_\rho
 - \frac{2i}3 A_\mu A_\nu A_\rho \Big) + \lambda^\dag \lambda + 2D\sigma \bigg] \, .
\ee
In general one can have a different CS level $k \in \bZ$ for each simple or Abelian factor in the gauge group;
however, we will be schematic with our notation and use the simple expression above.
The CS action is not $\cQ$-exact. If there are various Abelian factors in the gauge group,
we can consider adding to the action mixed CS terms between them:
\be
 \label{ch1:L mixed CS}
 \cL_\text{mCS} = - \frac{ik_{12}}{2\pi} \bigg[ \epsilon^{\mu\nu\rho} A_\mu^{(1)} \partial_\nu A_\rho^{(2)}
 + \frac12 \lambda^{(1)\dag} \lambda^{(2)} + \frac12 \lambda^{(2)\dag}\lambda^{(1)} + D^{(1)}\sigma^{(2)}
 + D^{(2)}\sigma^{(1)} \bigg] \, .
\ee
The mixed CS terms play a crucial r\^ole in turning on background fluxes or holonomies for the topological symmetries.
Remember that in three dimensions, any $\U(1)$ gauge symmetry yields a global symmetry
associated to the current $J_T^\mu = (\star F)^\mu = \frac12 \epsilon^{\mu\nu\rho} F_{\nu\rho}$,
being automatically conserved by the Bianchi identity: $\rd \star J_T=\rd F=0$.
The corresponding global symmetry $\U(1)_T$ is called \emph{topological symmetry}.
We may turn on a background gauge field $A^{(T)}$ for the $\U(1)_T$ symmetry by coupling it through
\be
 \label{ch1:top:Lag}
 \int A^{(T)} \wedge * J_T = \int \rd^3 x\, \sqrt g\, \epsilon^{\mu\nu\rho} A_\mu^{(T)} \partial_\nu A_\rho \, ,
\ee
where $A^{(T)}$ belongs to an external vector multiplet
$\cV^{(T)} = (A_\mu^{(T)} , \sigma^{(T)} , \lambda^{(T)} , \lambda^{(T)\dag} , D^{(T)})$.
On the supersymmetric background we have to set to zero the variation of the fermions in the external multiplet.
From \eqref{ch1:susy:transformations:gauge:multiplet} we get the conditions $D^{(T)} = i F_{12}^{(T)}$ and $\sigma^{(T)} = \text{const.}$.
The full Lagrangian is the supersymmetric completion of \eqref{ch1:top:Lag},
being obtained from \eqref{ch1:L mixed CS} by setting $k_{12}=1$ and thinking of $(1)$
as the background topological symmetry and $(2)$ as the gauge symmetry:
\be
 \label{ch1:L topological}
 \cL_\text{T} = - i \frac{A_3^{(T)}}{2\pi} \Tr F_{12} - i \frac{F_{12}^{(T)}}{2\pi} \Tr (A_3 + i\sigma)
 - i \frac{\sigma^{(T)}}{2\pi} \Tr D \, .
\ee
Notice that the three terms are separately supersymmetric and $\sigma^{(T)}$, a real mass for the topological symmetry,
is indeed a Fayet-Iliopoulos (FI) term. 

We can also add a mixed CS term between the R-symmetry and an Abelian flavor (or gauge) symmetry:
\be
 \label{ch1:RCS}
 \cL_\text{RCS} = -\frac{ik_R}{2\pi} \Big( \epsilon^{\mu\nu\rho} A_\mu \partial_\nu V_\rho + i \sigma W_{12} \Big) \, .
\ee

The standard kinetic matter Lagrangian for the chiral and antichiral multiplets $\Phi$, $\Phi^\dag$ coupled to $\cV$ is given by:
\be
 \label{ch1:matter:Lag}
 \cL_\text{mat} = D_\mu \phi^\dag D^\mu\phi + \phi^\dag \big( \sigma^2 + iD - r W_{12} \big) \phi
 + F^\dag F + i \psi^\dag ( \gamma^\mu D_\mu -\sigma ) \psi - i \psi^\dag \lambda\phi + i \phi^\dag \lambda^\dag \psi \, ,
\ee
where $r$ is the R-charge of $\phi$.
The covariant derivative $D_\mu$ in \eqref{ch1:matter:Lag} comprise the gauge fields,
the background field $V$ for the R-symmetry and background fields for the flavor symmetries of the theory.
A background flavor vector multiplet consists of the bosonic components $F_{12}^f$, $A_3^f$, $\sigma^f$
and $D^f$ which need to fulfill $D^f = i F_{12}^f$ in order to preserve supersymmetry.
We see that $F_{12}^f$ represents a Cartan-valued background magnetic flux for the flavor symmetry,
\be
 \frac1{2\pi} \int_{S^2} F^f = \fn \, ,
\ee
$A_3^f$ is a flat connection (or Wilson line) along $S^1$, and $\sigma^f$ is a real mass associated with the flavor symmetry.
The flavor magnetic flux $\fn$ will join the magnetic flux for the R-symmetry, providing a family of topological twists. 

The Lagrangian $\cL_\text{mat}$ is $\cQ$-exact, up to total derivatives:
\be
 Q \wt Q \big( \psi^\dag \psi + 2i \phi^\dag \sigma \phi \big) \,\cong\, \tilde\epsilon^\dag\epsilon\, \cL_\text{mat} \, ,
\ee
and will be used in the localization procedure.

Given a gauge-invariant, holomorphic function $W(\Phi)$, and of R-charge $r=2$,
one can write the superpotential Lagrangians:%
\footnote{The two Lagrangians are seperately supersymmetric.}
\be
 \cL_W = i F_W \, ,\qquad\qquad \cL_{\wb W} = i F_W^\dag \, ,
\ee
where
\be
 F_W = \parfrac{W}{\Phi_i} F_i - \frac12\, \parfrac{^2W}{\Phi_i \partial\Phi_j} \psi_j^{c\dag} \psi_i\pdag \, ,
 \qquad\qquad F_W^\dag = \parfrac{\wb W}{\Phi_i^\dag} F_i^\dag - \frac12\,
 \parfrac{^2 \wb W}{\Phi_i^\dag \partial\Phi_j^\dag} \psi_j^\dag \psi_i^c \, ,
\ee
are the F-terms of the chiral multiplet $W(\Phi)$ and its antichiral sister.
The two Lagrangians are $\cQ$-exact, up to total derivatives, due to
$\cQ \big( i \epsilon^{c\dag} \psi_W \big) \cong \tilde\epsilon^\dag\epsilon\, \cL_W$
and $\cQ\big( {-i \psi_W^\dag \tilde\epsilon^c} \big) \cong \tilde\epsilon^\dag\epsilon\, \cL_{\wb W}$.
Since we are working with the Wick rotation of real Lorentzian Lagrangians,
we take the two functions $W$ and $\wb W$ complex conjugate.

Finally, we can include the following supersymmetric Wilson loop in a representation $R$:
\be
 \label{ch1:def Wilson loop}
 W = \Tr_R \Pexp \oint \rd\tau \big( iA_\mu \dot x^\mu - \sigma \, | \dot x| \big) \, ,
\ee
as in \cite{Kapustin:2009kz}.
Here $\mathcal{P}$ denotes the usual path-ordering operator, $x^\mu(\tau)$ is the closed world-line of the Wilson loop,
$\tau$ is a parameter on it, $\dot x^\mu \equiv \rd x / \rd \tau$ and $| \dot x|$ is the length of $\dot x^\mu$.
Its supersymmetry variation reads
\be
 Q W \;\propto\; - \frac12 \lambda^\dag \gamma_\mu \epsilon \, \dot x^\mu + \frac12 \lambda^\dag \epsilon \, |\dot x| \, .
\ee
For $QW=0$ (and $\wt QW = 0$) we need $\dot x^1 = \dot x^2 = 0$, \ie\;the loop should be along the vector field $e_3$.
In the unrefined case $e_3 = \beta^{-1} \partial_t$ so that we can place the loop at an arbitrary point on $S^2$ and along $t$.
In the refined case $e_3 = \beta^{-1}(\partial_t + \varsigma \partial_\varphi)$ and
$x^\mu(\tau) = (\theta_0, \varsigma\tau, \tau)$. Thus, for irrational values of $\varsigma$ the loop does not close leaving us
with two choices: tune $\varsigma$ to rational values or place the loop at one of the two poles of $S^2$.

\section{ABJM theory}
\label{ch1:ABJM theory}

The ABJM theory \cite{Aharony:2008ug} describes the low-energy dynamics of $N$ M2-branes on $\bC^4 / \bZ_k$.
It is a three-dimensional supersymmetric Chern-Simons-matter theory with gauge group $\U(N)_k \times \U(N)_{-k}$,
four matter supermultiplets $(A_i, B_j)$, $i,j=1,2$, in bi-fundamental representations
and a manifest $\cN=6$ superconformal symmetry.
For $k=1,2$, the ABJM theory has monopole operators being conformal primaries of dimension 2
and transforming as vectors under Lorenz transformations. Such operators must be conserved currents,
which enables one to show that the superconformal symmetry is enhanced to $\cN=8$ \cite{Aharony:2008ug,Benna:2009xd,Bashkirov:2010kz}.
Using standard $\cN=2$ notation, this theory can be described by the quiver diagram
\bea
\begin{tikzpicture}[baseline, font=\footnotesize, scale=0.8]
\begin{scope}[auto,%
  every node/.style={draw, minimum size=0.5cm}, node distance=2cm];
\node[circle] (USp2k) at (-0.1, 0) {$N_k$};
\node[circle, right=of USp2k] (BN)  {$N_{-k}$};
\end{scope}
\draw[decoration={markings, mark=at position 0.9 with {\arrow[scale=1.5]{>}}, mark=at position 0.95 with {\arrow[scale=1.5]{>}}}, postaction={decorate}, shorten >=0.7pt]  (USp2k) to[bend left=40] node[midway,above] {$A_{i}$} node[midway,above] {} (BN) ; 
\draw[decoration={markings, mark=at position 0.1 with {\arrow[scale=1.5]{<}}, mark=at position 0.15 with {\arrow[scale=1.5]{<}}}, postaction={decorate}, shorten >=0.7pt]  (USp2k) to[bend right=40] node[midway,above] {$B_{j}$}node[midway,above] {}  (BN) ;  
\end{tikzpicture}
\eea
with two nodes representing the Chern–Simons theories, and four arrows  between the nodes representing
the bi-fundamental chiral multiplets. In addition, there is a quartic superpotential
\be
 W = \Tr ( A_1 B_1 A_2 B_2 - A_1 B_2 A_2 B_1 ) \, .
\ee
The ABJM theory has a manifest $\SU(2) \times \SU(2) \times \U(1)_T \times \U(1)_R$ global symmetry:
under the first $\SU(2)$ factor the $A_i$ transform as a doublet, and under the second $\SU(2)$ factor the $B_j$ transform as a doublet;
$\U(1)_T$ is the topological symmetry associated to the topological current $J_T = \star \Tr (F - \wt F)$
where $F$, $\wt F$ are the two field strengths;
$\U(1)_R$ is the R-symmetry under which $A_i$ and $B_j$ get multiplied by the same phase.

\section{The topologically twisted index}
\label{ch:1:twisted index:solving}

The topologically twisted index is the partition function for three- and four-dimensional gauge theories
with at least four supercharges on $\Sigma_\fg \times T^n$ $(n=1,2)$.
When it is refined with chemical potentials and background magnetic charges for the flavor symmetries,
it becomes an efficient tool for studying the nonperturbative properties of
supersymmetric gauge theories \cite{Closset:2013sxa,Benini:2015noa,Honda:2015yha,Closset:2015rna,Benini:2016hjo,Closset:2016arn}.
The large $N$ limit of the index contains interesting information about theories with a holographic dual.
In particular, the large $N$ limit of the index for the three-dimensional ABJM theory
was successfully used in \cite{Benini:2015eyy,Benini:2016rke} to provide
the first microscopic counting of the microstates of an AdS$_4$ black hole.
In this dissertation we extend the analyses of \cite{Benini:2015eyy,Benini:2016rke} to a wider class of field theories.

The index can be evaluated using supersymmetric localization and it reduces to a matrix model.
It can be written as the contour integral 
\begin{equation}
 \label{ch:1:index:generic:C}
 Z (\fn, y) = \frac1{|\mathfrak{W}|} \; \sum_{\fm \,\in\, \Gamma_\fh} \; \oint_\cC Z_{\text{int}} (\fm, x;  \fn , y) \, ,
\end{equation}
of a meromorphic differential form in variables $x$ parameterizing the Cartan subgroup
and subalgebra of the gauge group $G$, summed over the lattice of magnetic charges $\fm$ of the gauge group.
The index depends on complex fugacities $y_I = \e^{i \Delta_I}$ and magnetic charges $\fn_I$ for the flavor symmetries.
Here, $|\mathfrak{W}|$ is the order of the Weyl group of $G$.
Supersymmetric localization chooses a particular contour of integration $\cC$
and the final result can be recast in terms of the Jeffrey-Kirwan (JK) residue.
We refer the reader to \cite{Benini:2015noa,Benini:2016hjo} for a thorough analysis of the contour of integration.
As a difference with other well known matrix models arising from supersymmetric localization,
like the partition function on $S^3$ \cite{Kapustin:2009kz,Jafferis:2010un,Hama:2010av} or
the superconformal index \cite{Kim:2009wb}, in the large $N$ limit all the gauge magnetic fluxes
contribute to the integral making difficult its evaluation.
Here we use the strategy employed in \cite{Benini:2015eyy} to explicitly resum
the integrand and consider the contour integral of the sum%
\footnote{Here the sum is over a wedge $\Gamma_{\fh}^{\text{JK}}$ inside the magnetic lattice,
for which $Z_{\text{int}} (\fm, x; \fn , y)$ has poles inside the JK contour.}
\begin{equation}
 Z_{\text{resummed}} ( x;  \fn , y)
 = \frac1{|\mathfrak{W}|} \; \sum_{\fm \,\in\, \Gamma_{\fh}^{\text{JK}}}  Z_{\text{int}} (\fm, x; \fn , y) \, ,
\end{equation}
which is a complicated rational (in three dimensions) or elliptic (in four dimensions) function of $x$.
One can write a set of algebraic equations for the position of the poles, which we call \emph{Bethe ansatz equations} (BAEs)
(they actually are the BAEs of the dimensionally reduced theory on $\Sigma_\fg$ in the formalism of \cite{Nekrasov:2014xaa}),
and an \emph{effective twisted superpotential} $\wt\cW$ (or Yang-Yang functional \cite{Yang1969}) whose derivatives reproduce the BAEs.
The topologically twisted index is then given by the sum of the residues of $Z_{\text{resummed}}$ at the solutions to the BAEs. 

Let us note that in a three-dimensional $\cN = 2$ (four-dimensional $\cN = 1$) theory,
the R-symmetry can mix with the global symmetries and we can also write
\be
\fn_I = r_I + \fp_I \, ,
\ee
where $r_I$ is a reference R-symmetry and $\fp_I$ magnetic charges under the flavor symmetries of the theory.
Both $y_I$ and $\fn_I$ are thus parameterized by the global symmetries of the theory.
The invariance of each monomial term $W_a$ in the superpotential
under the symmetries of the theory imposes the following constraints
\be\label{supconstraints:intro}
 \prod_{I \in W_a} y_I = 1 \, , \qquad \qquad \sum_{I \in W_a} \fn_I = 2 (1 - \fg) \, ,
\ee
where the latter comes from supersymmetry, and, as a consequence,
\be\label{supconstraints2:intro}
 \sum_{I \in W_a} \Delta_I \in 2 \pi \mathbb{Z}  \, .
\ee
Here, the product and the sum are restricted to the fields entering in the monomial $W_a$.

Finally, the number of supersymmetric ground states $d_{\text{micro}}(\fn_I ,q_I)$ in the \emph{microcanonical} ensemble is given by
the Fourier transform of \eqref{ch:1:index:generic:C} with respect to the $\Delta_I$'s:
\be
 \label{ch:1:micro:density}
 d_{\text{micro}} (\fn_I ,q_I) = \int_{0}^{2\pi} \int_{0}^{2\pi} \dots \int_{0}^{2\pi}
 \bigg( \prod_{I}\frac{{\rm d} \Delta_I}{2 \pi} \bigg)
 \prod_a \delta\Big( \e^{i \sum_{I \in W_a} \Delta_I} - 1 \Big) \, Z (\fn_I , \Delta_I) \, \e^{- i \sum_{I} \Delta_I q_{I}} \, .
\ee

For simplicity, we restrict our discussion to the case of $\Sigma_\fg=S^2$ in the rest of the dissertation
since the generalization to an arbitrary Riemann surface is straightforward. Indeed, in the large $N$ limit,
the higher genus partition function receives a simple modification, as discussed in \cite{Benini:2016hjo}, as follows,
\begin{align}
 \label{ch:1:higher genus index}
 \log Z_{\fg \neq 1} (\fn_I) = (1 - \fg) \log Z_{\fg=0} (\fn_I / (1-\fg)) \, .
\end{align}

\subsection[\texorpdfstring{$\mathcal{N}=2$}{N=2} Chern-Simons-matter theories on \texorpdfstring{$S^2 \times S^1$}{S**2 x S**1}]{$\fakebold{\cN=2}$ Chern-Simons-matter theories on $\fakebold{S^2 \times S^1}$}
\label{index}

The topologically twisted index of a three-dimensional $\cN = 2$ supersymmetric Chern-Simons theory
with gauge group $G$ of rank $r$ and a set of chiral multiplets transforming in representations $\fR_I$ of $G$
is given by \cite{Benini:2015noa}:%
\footnote{For further developments see \cite{Benini:2016hjo,Cabo-Bizet:2016ars,Closset:2016arn,Closset:2017zgf}.}
\be
 \label{path}
 Z (\fn, y) = \frac1{|\mathfrak{W}|} \; \sum_{\fm \,\in\, \Gamma_\fh} \; \oint_\cC \;   \prod_{\text{Cartan}}
 \left (\frac{\rd x}{2\pi i x}  x^{k \fm} \right ) \prod_{\alpha \in G} (1-x^\alpha) \,  \prod_I \prod_{\rho_I \in \fR_I}
 \bigg( \frac{x^{\rho_I/2} \, y^{\rho^f_I / 2}}{1-x^{\rho_I} \, y^{\rho^f_I}} \bigg)^{\rho_I(\fm)- \rho^f_I (\fn)  +1}  \, ,
\ee
where the index $I$ runs over all matter fields in the theory.
Given a weight $\rho_I$ of the representation $\fR_I$, we use the notation $x^{\rho_I} = \e^{i \rho_I(u)}$.
$\alpha$ are the roots of $G$ and $\rho^f_I$ is the weight of the chiral multiplet under the flavor symmetry group.
In this formula,%
\footnote{$\beta$ is the radius of $S^1$.}
$x=\e^{i(A_t + i\beta \sigma)}$ parameterizes the gauge zero modes,
where $A_t$ is a Wilson line on $S^1$ and runs over the maximal torus of $G$ while $\sigma$ is the real scalar in  the vector
multiplet and runs over the corresponding Cartan subalgebra. $\fm$ are  gauge magnetic fluxes living in the co-root
lattice $\Gamma_\fh$ of $G$ (up to gauge transformations).
$k$ denotes the Chern-Simons coupling for the group $G$,
and there can be a different one for each Abelian and simple factor in $G$.

As already discussed, each Abelian gauge group in three dimensions is associated with a topological $\U(1)$ symmetry.
The contribution of a  topological symmetry with fugacity $\xi =\e^{i \Delta_m}$ and magnetic flux $\ft$ to the index is given by
\be
\label{topological}
Z^\text{top} = x^\ft \, \xi^\fm \, ,
\ee
where $x$ is the gauge variable of the corresponding $\U(1)$ gauge field.

Th index can be interpreted as a trace over a Hilbert space of states $\cH$ on $S^2$,
\begin{equation}
 \label{intro:3d:trace}
 Z(\fn, v) = \Tr\nolimits_{\cH} (-1)^{F} \e^{-\beta H} \e^{i \sum_I \Delta_I J_I} \, ,
\end{equation}
where $J_I$ are the generators of the flavor symmetries.
The Hamiltonian $H$ on $\Sigma_\fg$ explicitly depends on the flavor magnetic fluxes $\fn_I$ and the real masses $\sigma_I$.
Due to the supersymmetry algebra $Q^2 = H - \sigma_I J_I$ only states with $H=\sigma_I J_I$ contribute.
The index is a holomorphic function of $v_I$ with $v_I = \Delta_I + i \beta \sigma_I$.
We also identify $\Delta_I$ with flavor flat connections.

The partition function for theories in the $\Omega$-background (see section \ref{ch1:susy background}) is
\bea
 \label{conclusions:path:refined}
 Z (\fn, y , \zeta) = \frac1{|\mathfrak{W}|} \; & \sum_{\fm \,\in\, \Gamma_\fh} \; \oint_\cC \;
 \prod_{\text{Cartan}} \left (\frac{\rd x}{2\pi i x}  x^{k \fm} \right )
 \zeta^{-\sum_{\alpha>0} |\alpha(\fm)|} \prod_{\alpha \in G} \big(1-x^\alpha \zeta^{|\alpha(\fm)|}\big) \\
 & \times \prod_I \prod_{\rho_I \in \fR_I} \big( x^{\rho_I} y^{\rho^f_I} \big)^{B / 2} \;
 \frac{\big(x^{\rho_I} y^{\rho^f_I} \zeta^{1+B}; \zeta^2 \big)_\infty}{\big(x^{\rho_I} y^{\rho^f_I} \zeta^{1-B}; \zeta^2\big)_\infty} \, ,
 \qquad B = \rho_I(\fm) - \rho^f_I(\fn) + 1 \, ,
\eea
where $\zeta = \e^{i \varsigma / 2}$ is the fugacity for the angular momentum $L_\varphi$ of rotations along $\varphi$,
and the $q$-Pochhammer symbol is defined as
\be
 ( x; q )_{\infty} = \prod_{n = 0}^{\infty} ( 1 - x \, q^{n} ) \, ,  \qquad \text{ for } \; 0 \leq q < 1 \, .
\ee
The contribution of a $\U(1)$ topological symmetry is the same as before.

Reducing down to $S^1$, the $\Omega$-deformed partition function computes the quantum mechanical index
\be
 \log Z (\fn, v , \varsigma) = \Tr_\cH (-1)^F \e^{- \beta H} \e^{i \sum_{I} J_I \Delta_I} \e^{i \varsigma L_\varphi} \, .
\ee

In this dissertation we shall evaluate the \emph{unrefined} index at large $N$ for real $v_I$,
setting all real masses $\sigma_I$ to zero.
One can easily extend it to the complex plane employing holomorphy.

\subsection[\texorpdfstring{$\mathcal{N}=1$}{N=1} gauge theories on \texorpdfstring{$S^2 \times T^2$}{S**2 x T**2}]{$\fakebold{\cN=1}$ gauge theories on $\fakebold{S^2 \times T^2}$}
\label{The topologically twisted index}

The topologically twisted index of an $\cN = 1$ gauge theory with vector and chiral multiplets
and a non-anomalous $\U(1)_R$ symmetry in four dimensions is given by \cite{Benini:2015noa}%
\footnote{One can pull out a gauge independent factor --- the supersymmetric Casimir energy ---
from \eqref{path integral index} \cite{Closset:2017bse}.}
\be
\begin{aligned}
 \label{path integral index}
 Z (\fn, y , q) = \frac1{|\mathfrak{W}|} \; \sum_{\fm \,\in\, \Gamma_\fh} \; \oint_\cC \;  & \prod_{\text{Cartan}} \left (\frac{\rd x}{2\pi i x}  \eta(q)^{2} \right )
 (-1)^{\sum_{\alpha > 0}\alpha (\fm)} \prod_{\alpha \in G} \left[ \frac{\theta_1(x^\alpha ; q)}{i \eta(q)} \right] \\
 \times & \prod_I \prod_{\rho_I \in \fR_I} \bigg[ \frac{i \eta(q)}{\theta_1(x^{\rho_I} y^{\rho^f_I} ; q)} \bigg]^{\rho_I(\fm)- \rho^f_I (\fn) + 1} \, ,
\end{aligned}
\ee
where $q=\e^{2\pi i\tau}$ and $\tau$ is the complex modulus of the torus.
Here, the zero mode gauge variables $x = \e^{i u}$ parameterize the Wilson lines on the two directions of the torus 
\be
u = 2 \pi \oint_{\textmd{A-cycle}} A - 2 \pi \tau \oint_{\textmd{B-cycle}} A \, , 
\ee
and are defined modulo 
\be
 u_i \sim u_i + 2 \pi n + 2 \pi m \tau\, ,\qquad\qquad  n\, ,m \in \mathbb{Z} \, .
\ee
In this formula, $\theta_1(x; q)$ is a Jacobi theta function and $\eta(q)$ is the Dedekind eta function
(see appendix \ref{Elliptic functions}).
Let us remark that, there exist particular choices of background magnetic fluxes $\fn$ for which the $\sum_\fm$
truncates to a single set of gauge fluxes $\fm$ \cite{Gadde:2015wta}.
However, for generic background fluxes this does not happen and we need to sum an infinite number of contributions.

In order for the integrand in \eqref{path integral index} be a well-defined meromorphic function on the torus
the one-loop contributions must be invariant under the transformation $x^{\rho} \to q^{\rho(\gamma)}\, x^{\rho}$,
where $\gamma \,\in\, \Gamma_\fh$. 
Applying $x^{\rho} \to q^{\rho(\gamma)}\, x^{\rho}$ and using \eqref{theta:function:shift} we find
\bea
 \label{elliptic:transform:vec}
 Z_{1-\text{loop}}^{\text{gauge, off}} & \to Z_{1-\text{loop}}^{\text{gauge, off}}
 \prod_{\alpha \in G} (-1)^{-\alpha(\gamma)} \, \e^{- i \pi \tau \alpha(\gamma)^2} \, \e^{-i \alpha(u) \alpha(\gamma)} \, , \\
 Z_{1-\text{loop}}^{\text{chiral}} & \to Z_{1-\text{loop}}^{\text{chiral}}
 \prod_{\rho_I \in \fR_{I}} (-1)^{\rho_I(\gamma) B} \, \e^{ i \pi \tau \rho_I(\gamma)^2 B} \,
 \e^{i \rho_I(u) \rho_I(\gamma) B} \, \e^{ i \rho_I(\gamma) \rho^f_I(\Delta) B} \, .
\eea
Putting everything together, the total prefactor in the integrand vanishes if we demand the following anomaly cancellation conditions:
\bea
 &\sum_{\alpha\in G} \alpha(\gamma)^2 + \sum_{I} \sum_{\rho_I\in\fR_{I}} \left(\fn_I - 1\right) \rho_I(\gamma)^2 = 0 \, , && \mbox{$\U(1)_R$-gauge-gauge anomaly}\, ,\\
 &\sum_{\alpha\in G} \alpha(\gamma) \alpha(u) + \sum_{I} \sum_{\rho_I\in\fR_{I}} \left(\fn_I - 1\right) \rho_I(\gamma) \rho(u) = 0 \, , && \mbox{$\U(1)_R$-gauge-gauge anomaly}\, ,\\
 &\sum_I\sum_{\rho_I \in \fR_{I}} \rho_I(\gamma)^2 \, \rho_I(\fm) = 0 \, , && \mbox{gauge$^3$ anomaly}\, ,\\
 &\sum_I\sum_{\rho_I \in \fR_{I}} \rho_I(\gamma) \, \rho(u) \, \rho_I(\fm) = 0 \, , && \mbox{gauge$^3$ anomaly}\, ,\\
 &\sum_I\sum_{\rho_I \in \fR_{I}} \rho_I(\gamma) \, \rho_I(\fm) \, \rho^f_I(\Delta) = 0 \, , && \mbox{gauge-gauge-flavor anomaly}\, ,\\
 &\sum_I\sum_{\rho_I \in \fR_{I}} \left(\fn_I - 1\right) \rho_I(\gamma) \, \rho^f_I(\Delta) = 0 \, , && \mbox{$\U(1)_R$-gauge-flavor anomaly} \, .
\eea
The signs cancel out automatically for all D3-brane quivers since the number of arrows entering a node equals the number of arrows leaving it.

Notice that the index can be interpreted as a trace over a Hilbert space of states on $S^2 \times S^1$
\be
 \label{trace}
 Z (\fn, y, q) =  {\rm Tr}_{\cH}  (-1)^F q^{H_L} \prod_I y_I^{J_I} \, ,
\ee
where the Hamiltonian $H_L$ on $S^2 \times S^1$ explicitly depends on the magnetic fluxes $\fn_I$.

\section[The \texorpdfstring{$\cI$}{I}-extremization principle]{The $\fakebold{\cI}$-extremization principle}
\label{sec:intro:I-extremization}

The Bekenstein-Hawking entropy of a dyonic BPS black hole in AdS$_4$ with a charge vector $(\fn_I, q_I)$
can be obtained by extremizing
$\cI ( \Delta_I ) \equiv \log Z_{\text{SCFT}} \left( \Delta_I , \fn_I \right) - i \sum_{I} \Delta_I q_I$,
at large $N$,
\bea
 \label{SCFT:extremization:intro}
 \frac{\partial \cI \left( \Delta_I \right) }{\partial \Delta_{1,\ldots}}
 \bigg|_{\sum_{I \in W_a} \Delta_I \in 2 \pi \bZ} (\bar \Delta_I) = 0 \, ,
\eea
and evaluating it at its extremum $\bar \Delta_I$:
\bea
 \label{I-extremization:Legendre:intro}
 \cI \big|_{\text{crit}} ( \fn_I , q_I ) =
 S_{\rm BH} ( \fn_I , q_I) \, .
\eea
In the purely magnetic case $(q_I = 0)$ the extremization \eqref{I-extremization:Legendre:intro}
leads to real values for the critical points $\bar\Delta_I$ and the index $\cI (\bar \Delta_I)$.
However, in the dyonic case the saddle-point is complex and one has to
impose a constraint on the charges that the index $\cI (\bar \Delta_I)$ is a real positive quantity
[see the discussion around \eqref{main_result}].
This procedure, dubbed $\cI$-\emph{extremization} in \cite{Benini:2015eyy,Benini:2016rke}, comprised two conjectures:
\begin{enumerate}
 \item Extremizing the index unambiguously determines the exact R-symmetry in the unitary $\cN=2$ superconformal quantum mechanics in the infrared (IR).
 \item The value of the index at its extremum is the regularized number of ground states.
\end{enumerate}


\section{Attractor mechanism}
\label{ch:1:attractor:mechanism}

In this section we give a short introduction to an important feature of
static BPS black holes in four-dimensional $\cN=2$ gauged supergravity.
More details can be found in the main text.

The four-dimensional supergravity theory has $n_{\rm V}$ Abelian vector multiplets,
parameterizing a special K\"ahler manifold $\cM$ with metric $g_{i \bar{j}}$,
in addition to the gravity multiplet (thus a total of $n_{\rm V}+1$ gauge fields and $n_{\rm V}$ complex scalars).
The presence of hypermultiplets just add algebraic constraints \cite{Halmagyi:2013sla,Klemm:2016wng}.
So we only concern ourselves with vector multiplets.
The scalar manifold is defined by the prepotential $\cF \left( X^\Lambda \right)$,
which is a homogeneous holomorphic function of sections $X^\Lambda$.

Let us define the central charge of the black hole $\cZ$ and the superpotential $\cL$,
\be
\label{intro:central charge:superpotential}
\begin{aligned}
 \cZ = \e^{\cK / 2} \left( q_\Lambda X^\Lambda - \fn^\Lambda F_\Lambda \right) \, , \qquad
 \cL = \e^{\cK / 2} \left( g_\Lambda X^\Lambda - g^\Lambda F_\Lambda \right) \, .
\end{aligned} \ee
where $(g^\Lambda , g_\Lambda)$ are the magnetic and electric gaugings of the theory, $\cK$ is the K\"ahler potential and
\be
 F_\Lambda \equiv \frac{\partial \cF(X^\Lambda)}{\partial X^\Lambda} \, .
\ee
Here $\Lambda=0,\ldots,n_{\text{v}}$. Then, the Bekenstein-Hawking entropy of the black hole with magnetic and
electric charges $(\fn^\Lambda, q_\Lambda)$ can be obtained by extremizing the functional \cite{DallAgata:2010ejj}
\begin{equation}
 \label{attractor:intro}
 \cI_{\text{sugra}}(X^\Lambda) = - i \frac{\cZ}{\cL} \, ,
\end{equation}
with respect to $X^\Lambda$. This is the so-called \emph{attractor mechanism} \cite{Ferrara:1996dd}:
the area of the black hole horizon is given in terms of conserved charges
and is independent of the asymptotic moduli.

In the rest of the thesis we work in the gauge
\be
 g_\Lambda X^\Lambda= 1 \, .
\ee
In consistent models we can always apply an electric-magnetic duality transformation
so that the corresponding gauging becomes purely electric, \ie, $g^\Lambda = 0$.

\section{Main results}

In the course of our analysis, we find a number of interesting general results which we shall review in this section.

\subsection[\texorpdfstring{$\mathcal{N}=2$}{N=2} field theories on \texorpdfstring{$S^2 \times S^1$}{S**2 x S**1}]{$\fakebold{\cN=2}$ field theories on $\fakebold{S^2 \times S^1}$}
\label{ch:1:sec:N=2:3D}

In the large $N$ limit, we find a simple universal formula for computing
the index from the twisted superpotential, $\wt\cW(\Delta_I)$, as a function of the chemical potentials,
\begin{equation}\label{Z large N conjecture0:3d}
 \log Z  (\Delta_I,\fn_I) = - \frac{2}{\pi} \, \wt\cW(\Delta_I) \,
 - \sum_{I}\, \left[ \left(\fn_I - \frac{\Delta_I}{\pi}\right) \frac{\partial \wt\cW(\Delta_I)}{\partial \Delta_I}
  \right] \, .
\end{equation}
We call this the {\it index theorem}.
It allows to avoid the many technicalities involved in taking the residues and including exponentially small corrections to the index. 
By comparing the index theorem with the attractor formula for the entropy of asymptotically AdS$_4$ black holes,  we are also led to conjecture a relation between
the twisted superpotential and the prepotential of the dimensionally truncated  gauged supergravity describing the compactification on AdS$_4\times Y_7$, with $Y_7$  a Sasaki-Einstein manifold. 
This relation is discussed in section \ref{discussion}.
 
Furthermore, we find an explicit relation between the twisted superpotential and the $S^3$ free energy of the same $\cN \geq 2$ gauge theory.
Although the two matrix models are quite different at finite $N$,  the BAEs and the functional form of the twisted superpotential
in the large $N$ limit are {\it identical} to the matrix model equations of motion and  free energy functional for the path integral on $S^3$ found in \cite{Jafferis:2011zi}. 
This result implies that the index can be extracted from the free energy on $S^3$ and its derivatives in the large $N$ limit. It also implies a relation with the volume
functional of (Sasakian deformations of) the internal manifold $Y_7$.  
These relations deserve a better understanding.

\subsection[\texorpdfstring{$\mathcal{N}=1$}{N=1} field theories on \texorpdfstring{$S^2 \times T^2$}{S**2 x T**2}]{$\fakebold{\cN=1}$ field theories on $\fakebold{S^2 \times T^2}$}
\label{ch:1:4d:N=1:intro}

The explicit evaluation of the topologically twisted index in four dimensions is a strenuous task, even in the large $N$ limit.
However, the index greatly simplifies if we identify the modulus $\tau=i \beta/2\pi$ of the torus $T^2$
with a \emph{fictitious} inverse temperature $\beta$,%
\footnote{The torus partition function at a given $\tau$ would correspond to a thermal ensemble
while the elliptic genus is only counting extremal states.
Thus, the temperature represented by $\im \tau$ is fictitious.}
and take the \emph{high-temperature} limit $(\beta \to 0)$.
In this limit, we can use the modular properties of the elliptic functions under the $\SL(2,\bZ)$ action to simplify the result.

In the high-temperature limit, we find a number of intriguing results, valid to leading order in $1 / \beta$.

First, we obtain an explicit relation between the twisted superpotential
and the R-symmetry 't\,Hooft anomalies of the ultraviolet (UV) four-dimensional $\cN=1$ theory
\bea\label{tHoof:anomaly:0}
\wt\cW ( \Delta_I)= \frac{\pi^3}{6 \beta} \left[ \Tr R^3 ( \Delta_I) - \Tr R ( \Delta_I) \right]
= \frac{16 \pi^3}{27 \beta} \left[ 3c \left( \Delta_I \right) - 2 a ( \Delta_I) \right] \, ,
\eea
where $R$ is a choice of $\U(1)_R$ symmetry and the trace is over all fermions in the theory.
Here, we use the chemical potentials $\Delta_I/\pi$ to parameterize a trial R-symmetry of the $\cN=1$ theory.
Details about this identification are given in the main text.
In writing the second equality in \eqref{tHoof:anomaly:0} we used the relation between
conformal and R-symmetry 't\,Hooft anomalies in $\cN = 1$ SCFTs \cite{Anselmi:1997am},
\be\label{generalac}
a = \frac{9}{32} \Tr R^3 - \frac{3}{32} \Tr R \, ,
\qquad \qquad c = \frac{9}{32} \Tr R^3 - \frac{5}{32} \Tr R \, .
\ee

Secondly, the value of the index as a function of the chemical potentials $\Delta_I$ and
the set of magnetic fluxes $\fn_I$, parameterizing the twist,
can be expressed in terms of the trial left-moving central charge of the two-dimensional $\cN = (0,2)$ SCFT  as
\be
 \label{index theorem:2d central charge0}
 \log Z (\Delta_I,\fn_I) = \frac{\pi^2}{6 \beta} c_{l} \left( \Delta_I , \fn_I \right) \, .
\ee
This is related to the trial right-moving central charge $c_r$ by the gravitational anomaly
$k$ \cite{Benini:2012cz,Benini:2013cda},
\be
\label{cl:cr:k}
 c_r - c_l = k \, , \qquad \qquad k = - \Tr \gamma_3 \, .
\ee
Here, $\gamma_3$ is the chirality operator in two dimensions.%
\footnote{With our choice of chirality operator the gaugino zero modes have $\gamma_3 = 1$.}
We should emphasize that \eqref{index theorem:2d central charge0} does not receive logarithmic or polynomial corrections
in powers of $c_l / \beta$, it is \emph{exact} up to exponentially suppressed contributions.
The density of states $d_{\text{micro}}$ is then given by \eqref{ch:1:micro:density} and it reads
\be
 d_{\text{micro}} (\fn_I , q_0) \propto \cN(\fn_I)
 \left( S_{\text{Cardy}} \right)^{- (\cR+1)/2}
 I_{(\cR + 1)/2} (S_{\text{Cardy}})
 \, ,
\ee
where $I_{\nu}$ is the standard modified Bessel function of the first kind \eqref{integral:rep:Bessel}.
Here $\cR$ denotes the rank of the global symmetry group (including the R-symmetry),
$q_0$ is the electric charge conjugate to $\beta$%
\footnote{In a superconformal theory, the operator $H_L$
in \eqref{trace} equals the zero mode generator $L_0$ of the superconformal algebra.
The electric charge $q_0$ is the eigenvalue of $L_0$.} (see section \ref{ch:1:I-ext:vs:attractor}) 
and
\be
  S_{\text{Cardy}} = 2 \pi \sqrt{ \frac{c_l (\fn_I) q_0}{6}} \, .
\ee

Finally, there is a simple universal formula at leading order in $N$ for computing the
index from the twisted superpotential as a function of the chemical potentials $\Delta_I$,
\bea
\label{Z large N conjecture0:4d}
 \log Z(\Delta_I,\fn_I) = - \frac{3}{\pi} \, \wt\cW ( \Delta_I ) - \sum_{I} \left[ \left( \fn_I - \frac{\Delta_I}{\pi} \right) \frac{\partial \wt\cW ( \Delta_I )}{\partial \Delta_I} \right]
 = \frac{\pi^2}{6 \beta} c_r ( \Delta_I , \fn_I ) \, ,  
\eea
where the index $I$ runs over the bi-fundamental and adjoint fields in the quiver.
In the large $N$ limit the twisted superpotential can be written as
\bea\label{centralchargea0}
 \wt\cW ( \Delta_I )= \frac{16 \pi^3}{27 \beta} a ( \Delta_I) \, .
\eea
These formulae are valid for theories of D3-branes, where $\Tr R=\cO(1)$ and $c=a$ at large $N$ \cite{Henningson:1998gx}.
These topologically twisted theories have holographic duals in terms of black strings in AdS$_5 \times Y_5$,
where $Y_5$ are five-dimensional Sasaki-Einstein spaces \cite{Benini:2012cz,Benini:2013cda}.

There is a striking similarity with the results obtained for the large $N$ limit of
the topologically twisted index of three-dimensional $\cN=2$
field theories on $S^2 \times S^1$ (see subsection \ref{ch:1:sec:N=2:3D}), if we replace
\begin{equation*}
 \begin{array}{ccc}
  \text{central charge } a( \Delta_I) \; &\longleftrightarrow & \; \text{free energy on } S^3 \\
  \text{central charge } c_{r} \left( \Delta_I , \fn_I \right) \; &  \longleftrightarrow & \; \text{black hole entropy} \\
 c\text{-extremization} \; & \longleftrightarrow & \; \cI\text{-extremization} \, .
\end{array}
\end{equation*}
Indeed, in three dimensions, the very same formula \eqref{Z large N conjecture0:4d} holds
with the twisted superpotential given by the $S^3$ partition function $F_{S^3}$ of the gauge theory.
Notice that  $F_{S^3}$ is the natural replacement for $a$, both being monotonic along RG
flows \cite{Intriligator:2003jj,Jafferis:2011zi}. Furthermore, both of them can be computed, as a function of $\Delta_I$,
in terms of the volume of a family of Sasakian manifolds \cite{Gubser:1998vd,Butti:2005vn,Martelli:2005tp,Martelli:2006yb,Jafferis:2011zi}. 
In addition, in three dimensions, the dual AdS$_5$ black string is replaced by a dual AdS$_4$ black hole
and $\log Z$ computes the entropy of the black hole. As discussed in section \ref{sec:intro:I-extremization},
the Bekenstein-Hawking entropy is obtained by extremizing $\log Z$ with respect to $\Delta_I$ ($\cI$-extremization).
Similarly, as it was shown in \cite{Benini:2012cz,Benini:2013cda}, the exact central charge of the two-dimensional SCFT
is obtained by extremizing the trial right-moving central charge with respect to $\Delta_I$.
Given the relation \eqref{Z large N conjecture0:4d} we see that $c$-extremization corresponds to $\cI$-extremization.  

Notice also that our results \eqref{Z large N conjecture0:4d} and \eqref{centralchargea0} are compatible with  a very simple relation between the field theoretical quantities  $\Tr R^3 ( \Delta_I) $ and $c_{r} \left( \Delta_I , \fn_I \right)$ that is worthwhile to state separately,
\be
 \label{field theory}
 c_{r} \left( \Delta_I , \fn_I \right) = - 3 \Tr R^3 ( \Delta_I) - \pi \sum_{I} \left[ \left( \fn_I - \frac{\Delta_I}{\pi} \right) \frac{\partial \Tr R^3 ( \Delta_I)}{\partial \Delta_I} \right] \, .
\ee

\subsection[\texorpdfstring{$\cI$}{I}-extremization versus attractor mechanism]{$\fakebold{\cI}$-extremization versus attractor mechanism}
\label{ch:1:I-ext:vs:attractor}

We can ignore the linear relation among the chemical potentials and use a set of $\Delta_I$ such that
$\wt\cW (\Delta_I)$ is a homogeneous function of degree two (in three dimensions) or three (in four dimensions)
of the $\Delta_I$ alone. In this case, the index theorem simplifies to
\begin{equation}
 \label{Z large N conjecture2:intro}
 \log Z ( \Delta_I, \fn_I )= - \sum_{I} \fn_I \frac{\partial \wt\cW(\Delta_I)}{\partial \Delta_I} \, .
\end{equation}
This is due to the form of the differential operator in \eqref{Z large N conjecture0:3d}
[or \eqref{Z large N conjecture0:4d}] and $\sum_{I \in \cW_a} \fn_I  = 2$.

The $\cI$-extremization is then telling us that upon extremizing
\be
 \label{intro:Legendre:extr:vs:attr}
 \cI ( \Delta_I , \fn_I) =
 \sum_{I} \left( - \fn_I \frac{\partial \wt\cW(\Delta_I)}{\partial \Delta_I} - i \Delta_I q_I \right) \, ,
\ee
with respect to the chemical potentials $\Delta_I$ (under the constraint that $\sum_{I \in W_a} \Delta_I = 2 \pi$),
its value at the extremum $\bar\Delta_I$ precisely reproduces the black hole entropy.
Comparison of \eqref{intro:Legendre:extr:vs:attr} with the attractor equation \eqref{attractor:intro}
suggests the following relations
\bea
 \label{ch:1:identification:X:Delta}
 \Delta_I \propto X^\Lambda \, , \qquad
 \wt\cW ( \Delta_I ) \propto i \cF ( X^\Lambda ) \, ,
\eea
also valid before extremization.\footnote{A relation between the free energy $F_{S^3}$ and the prepotential
of the compactified theory was already suggested in \cite{Lee:2014rca}.}
Thus, in both three and four dimensions, the $\cI$-extremization principle corresponds to the attractor mechanism
\cite{DallAgata:2010ejj,Karndumri:2013iqa,Hristov:2014eza,Amariti:2016mnz,Klemm:2016kxw} on the supergravity side.

\paragraph*{ABJM on $\fakebold{S^2 \times S^1\,}$.} This relation certainly holds for the $\cN=6$ ABJM theory
which is holographically dual to M-theory on AdS$_4 \times S^7/\bZ_k$.
The twisted superpotential for this theory reads (see section \ref{ABJM:example:ch:2})
\be
  \wt\cW ( \Delta_a ) =
  \frac {2 \sqrt{2}}{3} k^{1/2} N^{3/2} \sqrt{\Delta_1 \Delta_2 \Delta_3 \Delta_4} \, .
\ee
which can be clearly mapped to the holomorphic prepotential
\begin{equation}
 \label{ch:1:F-sugra:ABJM}
 \cF (X^\Lambda) = -2 i \sqrt{ X^0 X^1 X^2 X^3} \, ,
\end{equation} 
of the so-called \emph{STU model} (consisting of three vector multiplets in addition to the gravity multiplet)
in four-dimensional $\cN=2$ gauged supergravity.

\paragraph*{D2$\fakebold{_k}$ on $\fakebold{S^2 \times S^1\,}$.}

In chapter \ref{ch:4}, we demonstrate another example of the relation \eqref{ch:1:identification:X:Delta}.
The instance of AdS$_4/$CFT$_3$ correspondence we shall study was finalized in \cite{Guarino:2015jca}
and states that the massive type IIA string theory on asymptotically AdS$_4 \times S^6$ backgrounds
admits a dual description in terms of an $\cN = 2$ Chern-Simons deformation (at level $k$) of the maximal
$\cN=8$ super Yang-Mills theory on the worldvolume of $N$ D2-branes \cite{Schwarz:2004yj,Guarino:2015jca}.
We will call this model the \emph{$\text{D2}_k$ theory}. Its twisted superpotential is given by (see section \ref{ssec:SYM-CS:index})
\bea
 \label{ch:1:twisted superpotential:SYM-CS}
 \wt\cW( \Delta_j ) =
 \frac{3^{13/6}}{5 \times 2^{8/3}}
 \left(1 - \frac{i}{\sqrt{3}} \right)
 k^{1/3} N^{5/3}
 \left( \Delta_1 \Delta_2 \Delta_3 \right)^{2/3} \, .
\eea
As we will show in chapter \ref{ch:4}, the effective prepotential describing the near-horizon geometries
constructed in the four-dimensional dyonic $\mathcal{N}=2$ gauged supergravity, that arises as a consistent truncation
of massive type IIA supergravity on $S^6$, reads
\bea
  \label{ch:1:IIA:prepotantial}
  \cF \left( X^I \right) = - i \frac{3^{3/2}}{4} \left(1-\frac{i}{\sqrt{3}}\right)
  c^{1/3} \left( X^1 X^2 X^3 \right)^{2/3} \ ,
\eea
where $c$ is the dyonic gauge parameter.
Quite remarkably, the above correspondence \eqref{ch:1:identification:X:Delta}
holds true including the imaginary part of the twisted superpotential
\eqref{ch:1:twisted superpotential:SYM-CS} and the prepotential \eqref{ch:1:IIA:prepotantial}.

Furthermore, we show that for a generic D2$_k$ theory with
\be
\sum_{a=1}^{|G|} k_a = |G| k \, ,
\ee
the logarithm of the topologically twisted index can be written as%
\footnote{A similar relation between the three-sphere free energy $- \log Z_{S^3} ( \Delta_I )$
and the anomaly coefficient $a ( \Delta_I )$ was found in \cite{Fluder:2015eoa}.}
\bea
 \label{intro:index:generic:c2d:a4d}
 \log Z \left( \fn_I, \Delta_I \right) =
 \frac{ 3^{7/6} \pi}{5 \times 2^{10/3}}
 \left( 1 - \frac{i}{\sqrt{3}} \right) 
 \left( n N \right)^{1/3}
 \frac{c_r \left( \fn_I , \Delta_I \right)}{a \left( \Delta_I \right)^{1/3}} \, ,
\eea
where $n \equiv |G| k$.
Here $a ( \Delta_I )$ is the trial conformal 't Hooft anomaly of a four-dimensional ``parent''
$\cN=1$ SCFT on $S^2 \times T^2$ and $c_r ( \fn_I , \Delta_I )$ is the trial right-moving central charge
of the $\cN=(0,2)$ theory obtained by a twisted compactification on $S^2$ \cite{Benini:2012cz,Benini:2013cda,Hosseini:2016cyf}.
Let us stress that, the r\^ole of the imaginary part of the partition function \eqref{intro:index:generic:c2d:a4d}
is to fix the electric charges (conjugate to $\Delta_I$) such that its value at the extremum is a \emph{real positive} quantity.
We will explain this better in section \ref{mIIA:Introduction}.

Given the interesting connection between the four-dimensional D3-brane theories and the three-dimensional D2$_k$ theories \eqref{intro:index:generic:c2d:a4d},
it would be intriguing to find the analogous relation on the supergravity side. In particular, one can expect a close connection
between the supergravity backgrounds discussed in chapter \ref{ch:4} and the black string solutions in five-dimensional
STU gauged supergravity found in \cite{Benini:2013cda}.

\paragraph*{$\fakebold{\cN=4}$ SYM on $\fakebold{S^2 \times T^2\,}$.} The AdS dual to topologically twisted
$\cN=4$ SYM is given by magnetically charged BPS black strings in type IIB string theory on AdS$_5 \times S^5$.
In the high-temperature limit, to leading order in $1/\beta$, we find that (see section \ref{twisted superpotential_SYM})
\be
 \label{ch:1:SYM:twisted superpotential}
 \wt\cW ( \Delta_a )
 = ( N^2 - 1 ) \frac{\Delta_1 \Delta_2 \Delta_3}{2 \beta} \, .
\ee
Once we take the high-temperature limit ($\beta \to 0$), we are shrinking a circle inside the torus and effectively dealing with
a three-dimensional field theory living on the twisted $S^2 \times S^1$ background.
Upon compactification geometric symmetries remain as global symmetries of the lower-dimensional theory.
Therefore, the theory in three dimensions has an extra global $\U(1)$ symmetry
whose chemical potential can be identified with $\beta$.

On the supergravity side the same story goes through.
The near-horizon geometry of the BPS black string is locally AdS$_3 \times S^2 \times S^5$.
The longitudinal direction along this black string lies within the AdS$_3$.
One could periodically identify the black string%
\footnote{The near-horizon geometry is then a BTZ black hole.}
and perform a dimensional reduction (as it was already done
in \cite{Hristov:2014eza}) to obtain a 4D black hole.
It interpolates between an AdS$_2 \times S^2$ near-horizon region and an asymptotic curved domain-wall.
The \emph{only} electric charge of the 4D black hole descends from the momentum on the circle.
The prepotential of this theory is given by
\bea
 \label{ch:1:stu:prepotential}
 \cF ( X^\Lambda ) = - \frac{X^1 X^2 X^3}{X^0} \, .
\eea 
Comparing \eqref{ch:1:SYM:twisted superpotential} with \eqref{ch:1:stu:prepotential}, we see that
\eqref{ch:1:identification:X:Delta} holds true if we identify $\beta$ with $X^0$.

\subsection[Black hole microstates in AdS\texorpdfstring{$_5$}{(5)}]{Black hole microstates in AdS$\fakebold{_5}$}

The derivation of the entropy of BPS electrically charged  rotating black holes in
AdS$_5\times S^5$ \cite{Gutowski:2004ez,Gutowski:2004yv,Chong:2005da,Chong:2005hr,Kunduri:2006ek}
in terms of states of the dual ${\cal N}=4$ $\SU(N)$ SYM theory is still an open problem.
Various attempts have been made in this direction \cite{Kinney:2005ej,Grant:2008sk,Chang:2013fba}
but none was really successful. One could consider the superconformal index \cite{Romelsberger:2005eg,Kinney:2005ej}
since it counts states preserving the same supersymmetries of the black holes and it depends on a number
of fugacities equal to the number of conserved charges of the black holes.%
\footnote{One can introduce five independent fugacities in the superconformal index
with a constraint $\sum_{I=1}^{3} \Delta_I + \sum_{i=1}^{2} \omega_i \in \bZ$.
On the gravity side, the BPS black holes have five conserved charges which satisfy a nonlinear constraint.}
However, due to a large cancellation between
bosonic and fermionic states, the index is a quantity of order one for generic values of the fugacities while the
entropy scales like $N^2$ \cite{Kinney:2005ej}. We also know that the supersymmetric partition function on
$S^3\times S^1$ is equal to the superconformal index only up to a multiplicative factor $\e^{-E_{\rm susy}}$,
where $E_{\rm susy}$ is the supersymmetric Casimir energy
\cite{Assel:2014paa,Lorenzen:2014pna,Assel:2015nca,Bobev:2015kza,Martelli:2015kuk,Genolini:2016sxe,Brunner:2016nyk}.
For $\cN=4$ SYM it reads (see for example \cite{Bobev:2015kza})
\be
 \label{ch:1:susy:Casimir}
 E_{\rm susy} = - i \pi  (N^2-1) \frac{\Delta_1 \Delta_2 \Delta_3}{\omega_1 \omega_2} \, ,
\ee
where $\Delta_{I}$ and $\omega_{i}$ are the chemical potentials for the Cartan generators of the R-symmetry and rotation,
respectively. They are subject to the constraint
\bea
 \label{ch:1:constraint:N=4:0}
 \sum_{I=1}^{3} \Delta_I + \sum_{i=1}^{2} \omega_i = 0 \, .
\eea

As can be seen from \eqref{ch:1:susy:Casimir}, the supersymmetric Casimir Energy is of order $N^2$ in the large $N$ limit;
however, it is not clear what the average energy of the vacuum should have to do with the entropy,
which is the degeneracy of ground states. In chapter \ref{ch:6}, we show that the entropy of rotating BPS
black holes in AdS$_5$ can be obtained by extremizing the Legendre transform of a quantity which \emph{formally}
equals $-E_{\rm susy}$ under the constraint
\bea
 \label{ch:1:constraint:N=4:1}
 \sum_{I=1}^{3} \Delta_I + \sum_{i=1}^{2} \omega_i = 1 \, .
\eea
We will give some preliminary discussion about the interpretation of this result in section \ref{sec:discussion}.

\chapter[Large \texorpdfstring{$N$}{N} matrix models for three-dimensional \texorpdfstring{$\mathcal{N}=2$}{N=2} theories]{Large \texorpdfstring{$\fakebold{N}$}{N} matrix models for three-dimensional \texorpdfstring{$\fakebold{\cN=2}$}{N=2} theories}
\label{ch:2}
\ifpdf
    \graphicspath{{Chapter2/Figs/Raster/}{Chapter2/Figs/PDF/}{Chapter2/Figs/}}
\else
    \graphicspath{{Chapter2/Figs/Vector/}{Chapter2/Figs/}}
\fi

\section{Introduction}

In this chapter we study the large $N$ behavior of the topologically twisted index introduced in chapter \ref{ch:1}
for three-dimensional  $\cN \geq 2$ gauge theories.
Here we extend the analysis of \cite{Benini:2015eyy,Benini:2016rke} to a larger class of $\cN \geq 2$ theories with an M-theory
or massive type IIA dual, containing bi-fundamental, adjoint and (anti-)fundamental chiral matter. Most of the theories
proposed in the literature are obtained by adding Chen-Simons terms \cite{Aharony:2008ug,Hanany:2008cd,Hanany:2008fj,Martelli:2008si,Franco:2008um,Davey:2009sr,Hanany:2009vx,Franco:2009sp}
or by flavoring \cite{Gaiotto:2009tk,Jafferis:2009th,Benini:2009qs}  four-dimensional quivers  describing D3-branes probing
Calabi-Yau three-fold (CY$_3$) singularities. We refer to these theories as having a four-dimensional parent.
They all have an M-theory phase where the index is expected to scale as $N^{3/2}$. The main motivation for studying the
large $N$ limit of the index for these theories comes indeed form the attempt to extend the result of \cite{Benini:2015eyy,Benini:2016rke}
to a larger class of black holes. However, the matrix model computing the
index reveals an interesting structure at large $N$ which deserves attention by itself. In particular, we will point out analogies
and relations with other matrix models appeared in the literature on three-dimensional $\cN \geq 2$ gauge theories.

The method for solving the BAEs is similar to that used in \cite{Herzog:2010hf,Jafferis:2011zi} for the large $N$ limit of
the partition function on $S^3$ in the M-theory limit and the one used for the partition function on $S^5$
of five-dimensional theories \cite{Jafferis:2012iv,Minahan:2013jwa,Minahan:2014hwa}. We take an ansatz for the
eigenvalues where the imaginary parts grow in the large $N$ limit as some power of $N$.
The solution to the BAEs in the large $N$ limit is then used to evaluate index using the residue theorem.  
In this last step we need to take into account (exponentially small) corrections to the large $N$ limit of
the BAEs which contribute to the index due to the singular logarithmic behavior of its integrand.

We focus on the limit where $N$ is much greater than the Chern-Simons couplings $k_a$.
For the class of quivers we are considering, this limit corresponds to an M-theory description when
$\sum_a k_a=0$ and a massive type IIA one when $\sum_a k_a\neq 0$.
We recover the known scalings $N^{3/2}$ and $N^{5/3}$ for the M-theory and massive type IIA phase, respectively.
Similarly to \cite{Jafferis:2011zi}, we find that, in order to have a consistent $N^{3/2}$ scaling of the index in
the M-theory phase, we have to impose some constraints on the quiver. In particular, quivers with a chiral four-dimensional
parent are not allowed, as in \cite{Jafferis:2011zi}. They are instead allowed in the massive type IIA phase.

In this chapter we give the general rules for constructing the twisted superpotential and the index for
a generic Yang-Mills-Chen-Simons theory with bi-fundamental, adjoint and fundamental fields and the example of ABJM theory.
Many other examples can be found in the next chapter, including models for well-known homogeneous
Sasaki-Einstein manifolds, suspended pinch point singularity (SPP), $N^{0,1,0}$, $Q^{1,1,1}$, $V^{5,2}$,
and various nontrivial checks of dualities.

The chapter is organized as follows. In section \ref{proof vanishing forces}
we give the general rules for constructing the twisted superpotential and the index for a generic Yang-Mills-Chen-Simons theory
with bi-fundamentals, adjoint and fundamental fields with $N^{3/2}$ scaling, and in section \ref{general rules N32} we derive
them. In section \ref{sec:freeenergy} we prove the identity of the twisted superpotential and the $S^3$ free energy at large $N$.
In section \ref{sec:index} we derive the index theorem that allows to express the index at large $N$ in terms of
the twisted superpotential and its derivatives. In section \ref{N53} we discuss the rules  for a $N^{5/3}$ scaling.
In section \ref{discussion} we give a discussion of some open issues.

\section[The large \texorpdfstring{$N$}{N} limit of the index]{The large $\fakebold{N}$ limit of the index}
\label{proof vanishing forces}

In this chapter we are interested in  the large $N$ limit of the topologically twisted index  for theories with unitary gauge groups and matter transforming in the fundamentals, bi-fundamentals and adjoint representation. As in \cite{Benini:2015eyy},  we evaluate the matrix model in two steps.
We first perform the summation over magnetic fluxes introducing a large cut-off $M$.\footnote{According to the rules in  \cite{Benini:2015noa}, the residues to take in \eqref{path} depend on the
sign of the Chern-Simons couplings. We can choose a set of co-vectors in the Jeffrey-Kirwan prescription such that the contribution comes from  residues with $\fm_a\leq 0$ for $k_a > 0$, residues with $\fm_a\geq 0$ for $k_a < 0$  and residues in the origin. We can then take a large positive integer $M$ and perform the summations in Eq.\,\eqref{path}, with $\fm_a \leq M-1 \, (k_a > 0)$ and $\fm_a \geq 1-M \, (k_a < 0)$.}
The result of this summation produces terms in the integrand of the form 
\begin{equation}
 \prod_{i=1}^{N} \frac{\left(\e^{i B_i^{(a)}}\right)^M}{\e^{i B_i^{(a)}} - 1} \, ,
\end{equation}
where we defined
\bea\label{BA expression}
 \e^{i \sign(k_a) B_i^{(a)}} & = \xi^{(a)} (x_i^{(a)})^{k_a}
 \prod_{\substack{\text{bi-fundamentals} \\ (a,b) \text{ and } (b,a) }} \prod_{j=1}^N \frac{\sqrt{\frac{x_i^{(a)}}{x_j^{(b)}} y_{(a,b)}}}{1 - \frac{x_i^{(a)}}{x_j^{(b)}} y_{(a,b)}}
 \frac{1 - \frac{x_j^{(b)}}{x_i^{(a)}} y_{(b,a)}}{\sqrt{\frac{x_j^{(b)}}{x_i^{(a)}} y_{(b,a)}}} \\
 & \hskip 1truecm \times \prod_{\substack{\text{fundamentals} \\ a}} \frac{\sqrt{x_i^{(a)} y_a}}{1 - x_i^{(a)} y_a} \,\, 
 \prod_{\substack{\text{anti-fundamentals} \\ a}} \frac{1 - \frac{1}{x_i^{(a)}} \tilde y_a}{\sqrt{\frac{1}{x_i^{(a)}} \tilde y_a}} \, ,
\eea
and adjoints are identified with bi-fundamentals connecting the same gauge group ($a=b$).
In this way the contributions from the residues at the origin have been moved to the solutions of the BAEs
\begin{align}\label{BAEs}
 \e^{i \sign(k_a) B_i^{(a)}}=1 \, .
\end{align}
It is convenient to use the variables $u_i^{(a)}$ and $\Delta_I$, defined modulo $2\pi$,\footnote{Notice that the index
is a holomorphic function of $y_I$ and $\xi$. There is no loss of generality in restricting to
the case of purely real chemical potentials $\Delta_I$ and $\Delta_m$ in \eqref{u Delta variable}.}
\be\label{u Delta variable}
x_i^{(a)}= \e^{i u_i^{(a)}} \;, \qquad\qquad y_I = \e^{i\Delta_I}\;, \qquad\qquad \xi^{(a)} = \e^{i\Delta_m^{(a)}} \, ,
\ee
and take the logarithm of the BAEs
\begin{align}\label{log BAEs0}
 0 & = \log \left[ \text{RHS of \eqref{BA expression}} \right] - 2 \pi i n_i^{(a)} 
 \, ,
\end{align}
where $n_i^{(a)}$ are integers that parameterize the angular ambiguities. The BAEs \eqref{log BAEs0} can be obtained as critical points of an effective twisted superpotential $\wt\cW(u_i^{(a)})$.

We then need to solve these auxiliary equations in the large $N$ limit. Once the distribution of poles in the integrand in the large $N$ limit has been found, we can finally evaluate the index by computing the residue of the
resummed integrand of \eqref{path} at the solutions of \eqref{log BAEs0}.  
In the final expression, the dependence on $M$ disappears.

We are interested in the properties of the topologically twisted index  in the large $N$ limit of theories with  an M-theory dual. 
We focus on quiver Chern-Simons-Yang-Mills gauge theories with gauge group
\be \label{quiver} \cG = \prod_{a=1}^{|G|} \U(N)_a \, ,\ee
and bi-fundamental, adjoint and fundamental chiral multiplets.  Most of the conjectured theories living on M2-branes probing $\text{CY}_4$ singularities
are of this form. Moreover, many of them are obtained by adding Chern-Simons terms and fundamental flavors to quivers appeared in the four-dimensional literature as 
describing D3-branes probing $\text{CY}_3$ singularities.  We refer to these theories as quivers {\it with a 4D parent}. In order to have a $\text{CY}_4$ moduli space, the  Chern-Simons couplings must satisfy
\be\label{CScontr}
\sum_{a=1}^{|G|} k_a = 0 \, .
\ee
The M-theory phase of these theories is obtained for $N\gg k_a$ and this is the limit we consider here. We expect the index to scale as $N^{3/2}$.

As in \cite{Benini:2015eyy}, we consider the following ansatz for  the large $N$ saddle-point eigenvalue distribution:
\be\label{ansatz alpha0}
u_i^{(a)} = i N^{1/2} t_i + v_i^{(a)} \, .
\ee
Notice that the imaginary parts of all the  $u_i^{(a)}$ are equal. In the large $N$ limit, we define the continuous functions $t_i = t(i/N)$ and $v_i^{(a)}= v^{(a)}(i/N)$ and  we introduce the density of eigenvalues 
\be \rho(t) = \frac1N\, \frac{di}{dt}\, , \ee
normalized as $ \int dt\, \rho(t) = 1$. 

The large $N$ limit of the twisted superpotential is performed in details in section \ref{twisted superpotential rules N32},
generalizing the analyses in \cite{Benini:2015eyy}. Here, we report the final result and some of the crucial subtleties.
We need to require the cancellations of long-range forces in the BAEs, as originally observed in a similar context in  \cite{Jafferis:2011zi}, and this imposes some
constraints on the quiver. Once these are satisfied,  the twisted superpotential $\wt\cW$ becomes a local functional of $\rho(t)$ and $v_i^{(a)}(t)$ and it scales as  $N^{3/2}$. 
The same constraints guarantee that the index itself scales as $N^{3/2}$.

\subsection{Cancellation of long-range forces}
\label{long-range}

As in \cite{Jafferis:2011zi}, when bi-fundamentals are present, we need to cancel  long-range forces in the BAEs.  These are detected by considering the force exerted by the eigenvalue $u_j^{(b)}$ on the eigenvalue $u_i^{(a)}$ in \eqref{log BAEs0}. They can grow with large powers of $N$ and need to be canceled by imposing constraint on the quiver and matter content if necessary. Since $u_j^{(b)}- u_i^{(a)} \sim \sqrt{N}$ for $i\neq j$,
when the long-range forces vanish, the BAEs and the twisted superpotential get only
contributions from $i\sim j$ and they become local functionals of $\rho(t)$ and $v_i^{(a)}(t)$.

Let us consider the effects of such long-range forces in the twisted superpotential $\wt\cW$. A single bi-fundamental field connecting gauge groups $a$ and $b$ contributes terms of the form
\be\label{N3}
\sum_{i<j} \frac{\left(u_i^{(a)} - u_j^{(b)}\right)^2}{4} - \sum_{i<j} \frac{\left(u_j^{(a)} - u_i^{(b)}\right)^2}{4}\, ,
\ee
to the twisted superpotential [see Eq.\,\eqref{forces chiral}].
In the large $N$ limit, they are of order $N^{5/2}$.
In order to cancel these terms, we are then forced, as  in \cite{Jafferis:2011zi}, to consider quivers where for each bi-fundamental connecting $a$ and $b$ there is also a bi-fundamental connecting $b$ and $a$. The contribution of the two bi-fundamentals then cancel out [see Eq.\,\eqref{reflection applied 34} and Eq.\,\eqref{reflection applied 12}].

From a pair of  bi-fundamentals, we  get another contribution to the twisted superpotential of the form  [see Eq.\,\eqref{bi-fundamentals forces2}]
\begin{equation}\label{N2ang}
- \frac{1}{2} \left[ \left(\Delta_{(a,b)} - \pi \right) + \left( \Delta_{(b,a)} - \pi \right) \right] \sum_{j\neq i}^N  \left(u_i^{(a)} -u_j^{(b)}\right) \text{sign}{(i-j)} \, .
\end{equation}
This term can be canceled by the contribution of the angular ambiguities in \eqref{log BAEs0} to the twisted superpotential $\wt\cW$
\begin{equation}
  2 \pi \sum_{i=1}^N   n_i^{(a)} u_i^{(a)}  \, ,
\end{equation}
provided that,\footnote{This is actually true only when $N$ is odd. For even $N$ we are left with a common factor $\pi \sum_{i=1}^N u_i^{(a)}$ which can be
reabsorbed in the definition of $\xi^{(a)}$.}
\begin{equation}\label{no long-range-forces Bethe}
\sum_{I \in a} (\pi  - \Delta_I) \in 2 \pi \bZ \, ,
\end{equation}
where the sum is taken over all bi-fundamental fields with one leg in the node $a$.\footnote{Adjoint fields are supposed to be counted twice.}
Since for any reasonable quiver the number of arrows entering a node is the same as the number of arrow leaving it, this equation is obviously equivalent to $\sum_{I \in a} \Delta_I \in 2 \pi \bZ$ and can be also written as 
\begin{equation}\label{no long-range-forces Bethe0}
\prod_{I \in a}  y_I =1 \, .
\end{equation}
This condition implies that the sum of the charges under all global symmetries of the bi-fundamental fields at each node must vanish. For quivers with a 4D parent, 
this is equivalent to the absence of anomalies for the global symmetries of the 4D theory. Taking the product over all the nodes in a quiver, we also get
\be\label{globalJ}
\Tr J =0 \, ,
\ee
for any global symmetry of the theory, where the trace is taken over all the bi-fundamental fermions.

There are also contributions to  the twisted superpotential  of $\cO(N^2)$.  The Chern-Simons terms give indeed
\be
\sum_a k_a  \sum_{i=1}^N  \frac{\left(u_i^{(a)}\right)^2}{2} \, .
\ee
However, the $\cO(N^2)$ term cancels out when the condition \eqref{CScontr} is satisfied. Finally, there is an $\cO(N^2)$ contributions of the fundamental fields  given by
\eqref{fundcontr}. This  vanishes if the {\it total} number of fundamental and anti-fundamental fields in the quiver is the same. 

We turn next to the large $N$ limit of the index. The vector multiplet contributes a term of $\cO(N^{5/2})$ [see Eq.\,\eqref{gauge entropy}]
\begin{equation}\label{indexscalingg}
 i \sum_{i<j}^N \left( u_i^{(a)} - u_j^{(a)} + \pi \right) \, .
\end{equation}
The contribution of $\cO(N^{5/2})$ of a chiral multiplet is  [see Eq.\,\eqref{forces chiral ba}]:
\begin{align}\label{indexscalingb}
 i \sum_{I \in a} \frac{\left(\fn_I - 1\right)}{2} \sum_{i<j}^N \left( u_i^{(a)} - u_j^{(a)} + \pi \right) \, .
\end{align}
To have a cancellation between terms of $\cO(N^{5/2})$ and $\cO(N^2)$ for each node $a$ we must have
\begin{equation}\label{no long-range forces}
 2 + \sum_{I \in a} \left(\fn_I - 1\right) = 0 \, .
\end{equation}
For a quiver with a 4D parent, this condition is equivalent to the absence of anomalies for the R-symmetry.
If we sum over all the nodes we also obtain the following constraint
\begin{equation}\label{globalR}
 {\quad \rule[-1.4em]{0pt}{3.4em} |G| + \sum_{I} \left(\fn_I - 1\right) = 0\, .\quad}
\end{equation}
The above equation is equivalent to $\Tr R = 0$ for any trial R-symmetry, where the trace is taken over all the bi-fundamental fermions and gauginos.

Summarizing, we can have  a $N^{3/2}$ scaling for the index if for each bi-fundamental connecting $a$ and $b$ there is also a bi-fundamental connecting $b$ and $a$,
the total number of fundamental and anti-fundamental fields in the quiver is equal,  Eq.\,\eqref{no long-range-forces Bethe0} and Eq.\,\eqref{no long-range forces}
are fulfilled. All these conditions are automatically satisfied for quivers with a toric vector-like 4D parent and also for other interesting models like \cite{Martelli:2009ga}.
However, they rule out many interesting chiral quivers appeared in the literature on M2-branes. We note a striking analogy with the conditions imposed in \cite{Jafferis:2011zi}. 

\subsection[Twisted superpotential at large \texorpdfstring{$N$}{N}]{Twisted superpotential at large $\fakebold{N}$}
\label{large N twisted superpotential rules}

In this section we give the general rules for constructing the  twisted superpotential for any $\cN \geq 2$ quiver gauge theory which respects the constraints \eqref{no long-range-forces Bethe0} and \eqref{no long-range forces}:

\begin{enumerate}
 \item Each group $a$ with CS level $k_a$ and chemical potential for the topological symmetry $\Delta_m^{(a)}$ contributes the term
 \begin{equation}\label{CS cV rule}
  - i k_a N^{3/2} \int \rd t\, \rho(t)\, t\, v_a(t) - i \Delta_m^{(a)} N^{3/2} \int \rd t\, \rho(t)\, t\, .
 \end{equation}
 \item A pair of bi-fundamental fields, one with chemical potential $\Delta_{(a,b)}$ and transforming in the $({\bf N},\overline{\bf N})$
 of $\U(N)_a \times \U(N)_b$ and the other with chemical potential $\Delta_{(b,a)}$ and transforming in the $(\overline{\bf N},{\bf N})$ of $\U(N)_a \times \U(N)_b$, contributes
 \begin{equation}\label{twisted superpotential bi-fundamental}
  i N^{3/2} \int \rd t\, \rho(t)^2 \left[g_+\left(\delta v(t) + \Delta_{(b,a)}\right) - g_-\left(\delta v(t) - \Delta_{(a,b)}\right)\right]\, ,
 \end{equation}
 where $\delta v(t) \equiv v_b (t) - v_a (t)$. Here, we introduced the polynomial functions
 \begin{equation}\label{gp gm:ch:2}
  g_\pm(u) = \frac{u^3}6 \mp \frac\pi2 u^2 + \frac{\pi^2}3 u \;,\qquad\qquad g_\pm'(u) = \frac{u^2}2 \mp \pi u + \frac{\pi^2}3 \, ,
 \end{equation}
 and we assumed them to be in the range
 \be
\label{inequalities for delta v0}
0 < \delta v + \Delta_{(b,a)} < 2\pi \;,\qquad\qquad\qquad -2\pi < \delta v - \Delta_{(a,b)} < 0 \, ,
\ee
which can be adjusted by choosing a specific determination for the $\Delta$ that  are  defined modulo $2\pi$. We will also assume, and 
this is certainly true if $\delta v$ assumes the value zero, that
\be
\label{general range of Delta_a0}
0 < \Delta_I < 2\pi \, .
\ee
 \item An adjoint field with chemical potential $\Delta_{(a,a)}$, contributes
 \begin{equation}\label{adjoint cV rule}
  i g_+(\Delta_{(a,a)}) N^{3/2} \int \rd t\, \rho(t)^2\, .
 \end{equation}
 \item A field $X_a$ with chemical potential $\Delta_a$ transforming in the fundamental of $\U(N)_a$, contributes
 \begin{equation}\label{fund cV rule}
  -\frac{i}{2} N^{3/2} \int \rd t\, \rho(t)\, |t| \Big[v_a(t) + \big( \Delta_a - \pi \big)\Big] \, ,
 \end{equation}
 while an anti-fundamental field with chemical potential $\tilde \Delta_a$ contributes\footnote{We also assume $0<v_a(t)+\Delta_a<2 \pi$ and $0<-v_a(t)+\tilde \Delta_a<2 \pi$.}
 \begin{equation}\label{anti-fund cV rule}
  \frac{i}{2} N^{3/2} \int \rd t\, \rho(t)\, |t| \Big[v_a(t) - \big( \tilde \Delta_a - \pi \big)\Big] \, .
 \end{equation}
\end{enumerate}

Adding all the previous contributions for all gauge groups and matter fields, we get a local functional $\wt\cW(\rho(t) ,v_a(t), \Delta_I)$ that we need to extremize
 with respect to the continuous functions
$\rho(t)$ and $v_a(t)$ with the constraint $\int \rd t \rho(t) =1$. Equivalently we can introduce a Lagrange multiplier $\mu$ and extremize
\be\label{bethefunctional}
\wt\cW\left(\rho(t) ,v_a(t), \Delta_I\right) - \mu \left (\int \rd t \rho(t) -1 \right)  \, .
\ee
This gives the large $N$ limit distribution of poles in the index matrix model.  

The solutions of the BAEs have a typical piece-wise structure.
 Eq.\,\eqref{bethefunctional} is the right functional to extremize when the conditions \eqref{inequalities for delta v0} are satisfied.
 This gives a central region where $\rho(t)$ and $v_a(t)$ vary with continuity as functions of $t$. 
When one of the $\delta v(t)$ associated with a pair of bi-fundamental hits the boundaries of the inequalities \eqref{inequalities for delta v0}, 
it remains frozen to a constant value  $\delta v = -\Delta_{(b,a)}$ (mod $2\pi$) or  $\delta v = \Delta_{(b,a)}$  (mod $2\pi$)
for larger (or smaller) values of $t$. This creates ``tail'' regions where one or more $\delta v$ are frozen and the functional  \eqref{bethefunctional}
is extremized with respect to the remaining variables. In the tails, the derivative of  \eqref{bethefunctional} with respect to the frozen variable is not zero
and it is compensated by subleading terms that we omitted.  To be precise, the equations of motion  [see Eq.\,\eqref{talis eom}] includes
subleading terms
 \begin{equation}\label{tails} \frac{\partial \wt\cW}{\partial (\delta v)} + i N \rho \left[ \Li_1 \left( \e^{i \left( \delta v + \Delta_{(b,a)} \right)} \right)
 - \Li_1 \left( \e^{i \left( \delta v - \Delta_{(a,b)} \right)} \right)\right] = 0 \, ,
 \end{equation}
which are negligible except on the tails, where $\delta v$ has exponentially small correction to the large $N$ constant value 
 \begin{equation}\label{tails2}
 \delta v = - \Delta_{(b,a)} + \e^{-N^{1/2} Y_{(b,a)}}  \, , \qquad 
 \delta v = \Delta_{(a,b)}   - \e^{-N^{1/2} Y_{(a,b)}} \, ,   \qquad  \text{mod } 2\pi \, .
 \end{equation}
The quantities $Y$ are determined by equation \eqref{tails} and contribute to the large $N$ limit of the index.

\subsubsection{The ABJM example}
\label{ABJM:example:ch:2}

As an example, we briefly review here the solution to the BAEs for the ABJM model found in \cite{Benini:2015eyy}.
The reader can find many other examples in chapter \ref{ch:3}.
ABJM is a Chern-Simons-matter theory with gauge group $\U(N)_k \times \U(N)_{-k}$,
with two pairs of  bi-fundamental fields $A_i$ and $B_j$ transforming in the 
representation $({\bf N},\overline{\bf N})$ and $(\overline{\bf N},{\bf N})$ of the gauge group, respectively, and superpotential
\be
W = \Tr \left( A_1B_1A_2B_2 - A_1B_2A_2B_1 \right) \, .
\ee
We assign chemical potentials $\Delta_{1,2}\in [0,2\pi]$ to $A_i$ and  $\Delta_{3,4}\in [0,2\pi]$ to $B_j$.
Invariance of the superpotential under the global symmetries
requires that $\sum_I \Delta_I \in 2\pi \mathbb{Z}$ (or equivalently $\prod_{I} y_I = 1$).
Conditions \eqref{no long-range-forces Bethe0} and \eqref{no long-range forces} are then automatically satisfied.
The twisted superpotential, for $k=1$,\footnote{There is a similar solution  for $k>1$ with $\wt\cW\rightarrow \wt\cW \sqrt{k}$.
However, we also need to take into account that, for $k>1$, there are further identifications among the $\Delta_I$
due to discrete $\mathbb{Z}_k$ symmetries of the quiver.} reads
\be
\wt\cW =  i N^{3/2} \int \rd t \left\{ t\, \rho(t)\, \delta v(t) + \rho(t)^2 \left[ \, \sum_{a=3,4} g_+ \left( \delta v(t) + \Delta_a \right) - \sum_{a=1,2} g_- \left( \delta v(t) - \Delta_a \right)\right] \right\} \, .
\ee
The solution for $\sum_I \Delta_I = 2 \pi$ and $\Delta_1\leq \Delta_2$, $\Delta_3\leq \Delta_4$  is as follows \cite{Benini:2015eyy}. We have a central region where 
\be
\begin{aligned}
\rho &= \frac{2\pi \mu + t(\Delta_3 \Delta_4 - \Delta_1 \Delta_2)}{(\Delta_1 + \Delta_3)(\Delta_2 + \Delta_3)(\Delta_1 + \Delta_4)(\Delta_2 + \Delta_4)} \, , \\[.5em]
\delta v &= \frac{\mu(\Delta_1 \Delta_2 - \Delta_3 \Delta_4) + t \sum_{a<b<c} \Delta_a \Delta_b \Delta_c }{ 2\pi \mu + t ( \Delta_3 \Delta_4 - \Delta_1 \Delta_2) } \, ,
\end{aligned}
\qquad\qquad -\frac{\mu}{\Delta_4}   < t < \frac{\mu}{\Delta_2} \, .
\ee
When $\delta v$ hits $-\Delta_3$ on the left the solution becomes
\be
\rho = \frac{\mu + t\Delta_3}{(\Delta_1 + \Delta_3)(\Delta_2 + \Delta_3)(\Delta_4 - \Delta_3)} \, , \,\,   \delta v = - \Delta_3 \;, \,\, \qquad  -\frac{\mu}{\Delta_3}  < t <-\frac{\mu}{\Delta_4}  \, ,
\ee
with the exponentially small correction  $Y_3 = (- t\Delta_4 -\mu)/(\Delta_4 - \Delta_3)$, while when  $\delta v$ hits $\Delta_1$ on the right the solution becomes
\be
\rho = \frac{\mu - t \Delta_1}{(\Delta_1 + \Delta_3)(\Delta_1 + \Delta_4)(\Delta_2 - \Delta_1)} \, ,\,\,
\delta v = \Delta_1 \;,\,\, \qquad \frac{\mu}{\Delta_2}  < t < \frac{\mu}{\Delta_1}  \, ,
\ee
with $Y_1 =(t\Delta_2 - \mu)/(\Delta_2 - \Delta_1)$.
Finally, the on-shell twisted superpotential is
\be
 \label{Vsolution sum 2pi -- end}
 \wt\cW \left(\rho(t), \delta v(t), \Delta_I \right) \big|_{\text{BAEs}}
 = \frac {2i}{3} \mu N^{3/2}  = \frac {2 i N^{3/2}}{3} \sqrt{ 2 \Delta_1 \Delta_2 \Delta_3 \Delta_4} \, .
\ee
There is also a solution for $\sum_I \Delta_I = 6 \pi$ which, however, is obtained by the previous one
by a discrete symmetry $\Delta_I \rightarrow 2 \pi -\Delta_I$ $\left(y_I\rightarrow y_I^{-1}\right)$. 

\subsection[The index at large \texorpdfstring{$N$}{N}]{The index at large $\fakebold{N}$}
\label{large N index rules}

We now turn to the large $N$ limit of the  index for an $\cN \geq 2$ quiver gauge theory without long-range forces. 
Here, we give the rules for constructing the index once we know the large $N$ solution $\rho(t),v_a(t)$ of the BAE,
which is obtained by extremizing \eqref{bethefunctional}. The final result scales as $N^{3/2}$. 
\begin{enumerate}
 \item For each group $a$, the contribution of the Vandermonde determinant is
 \begin{equation}\label{gaugecontribution}
 -\frac{\pi^2}{3} N^{3/2} \int \rd t\, \rho(t)^2\, .
 \end{equation}
 \item A $\U(1)_a$ topological symmetry with flux $\ft_a$ contributes
 \begin{equation}
 - \ft_a N^{3/2} \int \rd t\, \rho(t)\, t\, .
 \end{equation}
 \item A pair of bi-fundamental fields, one with magnetic flux $\fn_{(a,b)}$ and chemical potential $\Delta_{(a,b)}$ transforming
 in the $({\bf N},\overline{\bf N})$ of $\U(N)_a \times \U(N)_b$ and the other  with magnetic flux $\fn_{(b,a)}$ and chemical potential $\Delta_{(b,a)}$
 transforming in the $(\overline{\bf N},{\bf N})$ of $\U(N)_a \times \U(N)_b$, contributes
 \begin{equation}\label{bicontribution}
 - N^{3/2} \int \rd t\, \rho(t)^2 \left[(\fn_{(b,a)}-1)\, g'_+ \left(\delta v(t) + \Delta_{(b,a)}\right)
  + (\fn_{(a,b)}-1)\, g'_- \left(\delta v(t) - \Delta_{(a,b)}\right)\right]\, .
 \end{equation}
 \item An adjoint field with magnetic flux $\fn_{(a,a)}$ and chemical potential $\Delta_{(a,a)}$, contributes
 \begin{equation}
  - (\fn_{(a,a)}-1)\, g'_+ \left(\Delta_{(a,a)}\right) N^{3/2} \int \rd t\, \rho(t)^2\, .
 \end{equation}
 \item A field $X_a$ with magnetic flux $\fn_a$ transforming in the fundamental of $\U(N)_a$, contributes
 \begin{equation}
  \frac{1}{2} (\fn_a - 1) N^{3/2} \int \rd t\, \rho(t) |t| \, ,
 \end{equation}
 while an anti-fundamental field with magnetic flux $\tilde \fn_a$ contributes
 \begin{equation}
  \frac{1}{2} (\tilde\fn_a - 1) N^{3/2} \int \rd t\, \rho(t) |t| \, .
 \end{equation}
 \item The tails,  where $\delta v$ has a constant value, as in \eqref{tails}, contribute
 \be\label{tailcontribution}
 - \fn_{(b,a)} N^{3/2} \int_{\delta v \approx - \Delta_{(b,a)} (\text{mod } 2\pi)}  \rd t \, \rho(t) Y_{(b,a)} - \fn_{(a,b)} N^{3/2} \int_{\delta v \approx  \Delta_{(a,b)} (\text{mod } 2\pi)}  \rd t \, \rho(t) Y_{(a,b)} \, ,
 \ee
 where the integrals are taken on the tails regions.
    
\end{enumerate}

As an example, for ABJM, using the above solution of the BAEs, one obtains the simple expression \cite{Benini:2015eyy}
\be
 \label{index:ABJM:ch:2}
 \re\log Z (\fn_I , \Delta_I) = -\frac{N^{3/2} }{3} \sqrt{2 \Delta_1\Delta_2\Delta_3\Delta_4} \sum_{I=1}^{4} \frac{\fn_I}{\Delta_I} \, .
\ee

\section{Derivation of matrix model rules}
\label{general rules N32}

In this section we give a detail derivation of the rules presented in subsections
\ref{large N twisted superpotential rules} and \ref{large N index rules}
for finding the twisted superpotential and the index at large $N$.
This section is rather technical and can be skipped on a first reading.

We consider the following large $N$ saddle-point eigenvalue distribution ansatz
\be\label{ansatz alpha}
u_i^{(a)} = i N^{\alpha} t_i + v_i^{(a)} \, .
\ee
Note that the imaginary parts of the $u_i^{(a)}$ are equal. We also define
\be
\delta v_i = v_i^{(b)} - v_i^{(a)} \, .
\ee
At large $N$, we define the continuous functions $t_i = t(i/N)$ and $v_i^{(a)}= v^{(a)}(i/N)$ and  we introduce the density of eigenvalues 
\be \rho(t) = \frac1N\, \frac{di}{ \rd t}\, , \ee
normalized as $ \int  \rd t\, \rho(t) = 1$.
Furthermore, we impose the additional constraint
\be
\sum_{a=1}^{|G|} k_a = 0 \, ,
\ee
corresponding to quivers dual to M-theory on AdS$_4\times Y_7$ and $N^{3/2}$ scaling. 

\subsection[Twisted superpotential at large \texorpdfstring{$N$}{N}]{Twisted superpotential at large $\fakebold{N}$}
\label{twisted superpotential rules N32}

We may write the BAEs as
\begin{align}\label{log BAEs2}
 0 & = \log \left[ \text{RHS of \eqref{BA expression}} \right] - 2 \pi i n_i^{(a)} \, ,
\end{align}
where $n_i^{(a)}$ are integers that parameterize the angular ambiguities.
We define the ``twisted superpotential'' as the function whose critical points give the BAEs \eqref{log BAEs2}.
In  the large $N$ limit the twisted superpotential $\wt\cW$ will be the sum of various contributions,
\be
\wt\cW = \wt\cW^{\text{CS}} + \wt\cW^{\text{bi-fund}} + \wt\cW^{\text{adjoint}} + \wt\cW^{\text{(anti-)fund}} \, .
\ee
$\alpha$ will be determined to be $1/2$ by the competition between Chern-Simons terms and matter contribution.

\subsubsection{Chern-Simons contribution}
\label{Bethe CS N32}

Each group $a$ with CS level $k_a$ and topological chemical potential $\Delta_m^{(a)}$, contributes to the  finite $N$ twisted superpotential as
\begin{equation}
 \wt\cW^{\text{CS}} = \sum_{i=1}^N  \left[- \frac {k_a}{2} \left(u_i^{(a)}\right)^2 - \Delta_m^{(a)} u_i^{(a)} \right] \, .
\end{equation}
Given the large $N$ saddle-point eigenvalue distribution \eqref{ansatz alpha}, we find
\begin{equation}
 \wt\cW^{\text{CS}} = \frac{k_a}{2} N^{2\alpha} \sum_{i=1}^{N} t_i^2 - i N^{\alpha} \sum_{i=1}^{N} \left(k_a t_i v_i^{(a)} + \Delta_m^{(a)} t_i \right) \, .
\end{equation}
Summing over nodes the first term vanishes (since $\sum_{a=1}^{|G|} k_a = 0$).
Taking the continuum limit, we obtain
\begin{equation}
 \wt\cW^{\text{CS}} = - i k_a N^{1+\alpha} \int  \rd t\, \rho(t)\, t\, v_a(t) - i N^{1+\alpha} \Delta_m^{(a)} \int  \rd t\, \rho(t)\, t\, .
\end{equation}

\subsubsection{Bi-fundamental contribution}
\label{Bethe bi-fundamentals N32}

For a pair of bi-fundamental fields, one with chemical potential $\Delta_{(a,b)}$ transforming in the $({\bf N},\overline{\bf N})$
of $\U(N)_a \times \U(N)_b$ and one with chemical potential $\Delta_{(b,a)}$ transforming in the $(\overline{\bf N},{\bf N})$ of $\U(N)_a \times \U(N)_b$, the finite $N$  contribution to the twisted superpotential is given by
\be\begin{aligned}\label{cV bi-fundamentals}
 \wt\cW^{\text{bi-fund}} & = \sum_{\substack{\text{bi-fundamentals} \\ (b,a) \text{ and } (a,b)}} \sum_{i,j=1}^N \left[ \Li_2 \left( \e^{i \left(u_j^{(b)} - u_i^{(a)} + \Delta_{(b,a)}\right)} \right) - \Li_2 \left( \e^{i \left(u_j^{(b)} - u_i^{(a)} - \Delta_{(a,b)}\right)} \right) \right] \\
 & - \sum_{\substack{\text{bi-fundamentals} \\ (b,a) \text{ and } (a,b)}} \sum_{i,j=1}^N  \left[ \frac{\left( \Delta_{(b,a)} - \pi \right) + \left( \Delta_{(a,b)} - \pi \right) }{2} \left( u_j^{(b)} - u_i^{(a)} \right) \right] \, ,
\end{aligned}\ee
up to constants that do not depend on $u_j^{(b)}$, $u_i^{(a)}$.

We would like to remind the reader that all angular variables are defined modulo $2\pi$. Part of the ambiguity in $\Delta_I$ can be fixed by requiring that
\be
\label{inequalities for delta v}
0 < \delta v + \Delta_{(b,a)} < 2\pi \;,\qquad\qquad\qquad -2\pi < \delta v - \Delta_{(a,b)} < 0 \;.
\ee
The remaining ambiguity of simultaneous shifts $\delta v \to \delta v+2\pi$, $\Delta_{(a,b)} \to \Delta_{(a,b)} +2\pi$, $\Delta_{(b,a)} \to \Delta_{(b,a)}-2\pi$ can also be fixed by requiring that $\delta v(t)$ takes the value $0$ somewhere, if it vanishes at all, which we assume.
We then have
\be
\label{general range of Delta_a}
0 < \Delta_I < 2\pi \, .
\ee

To compute $\wt\cW^{\text{bi-fund}}$, we break
\be\begin{aligned}
\label{broken expression}
\sum_{i,j=1}^N \Li_2\left( \e^{i \left(u_j^{(b)} - u_i^{(a)}+ \Delta_{(b,a)}\right)} \right) & = \sum_{i>j} \Li_2\left( \e^{i \left(u_j^{(b)} - u_i^{(a)} + \Delta_{(b,a)}\right)} \right) + \sum_{i<j} \Li_2\left( \e^{i \left( u_j^{(b)} - u_i^{(a)} + \Delta_{(b,a)}\right)} \right) \\
& + \sum_{i=1}^N \Li_2\left( \e^{i \left(u_i^{(b)} - u_i^{(a)} + \Delta_{(b,a)}\right)} \right) \, .
\end{aligned}\ee
The crucial point here is that the last term is naively of $\cO(N)$ and thus subleading; however, we should keep it since its derivative is not subleading on part of the solution when $\delta v$ hits $\Delta_{(a,b)}$ or $-\Delta_{(b,a)}$.
Therefore, we keep
\be\label{talis eom}
N \int  \rd t\, \rho(t) \left[\Li_2 \left( \e^{i \left( \smallstrut \delta v(t) + \Delta_{(b,a)} \right)} \right) - \Li_2 \left( \e^{i \left( \smallstrut \delta v(t) - \Delta_{(a,b)} \right)} \right) \right] \; .
\ee
This will be important in  the {\em tail contribution} to the twisted superpotential.
The second term in \eqref{broken expression} is
\be
 \label{sum:i:less:j}
 \sum_{i<j} \Li_2\left( \e^{i \left(u_j^{(b)} - u_i^{(a)} + \Delta_{(b,a)} \right)} \right) = N^2 \int  \rd t\, \rho(t) \int_t  \rd t'\, \rho(t')\, \Li_2 \left( \e^{i \left( \smallstrut u_b(t') - u_a(t) + \Delta_{(b,a)} \right)} \right) \, .
\ee
We first write the dilogarithm function as a power series
\be
 \Li_2(\e^{iu}) = \sum_{k=1}^\infty \frac{\e^{iku}}{k^2} \, .
\ee
Then, we consider the integral
\bea
I_k & = \int_t  \rd t'\, \rho(t')\, \e^{ik \left( \smallstrut u_b(t') - u_a(t) + \Delta_{(b,a)} \right)} \\
&= \int_t  \rd t'\, \e^{-kN^\alpha (t'-t)} \sum_{j=0}^\infty \frac{(t'-t)^j}{j!} \partial_x^j \left[ \rho(x)\, \e^{ik \left( \smallstrut v_b(x) - v_a(t) + \Delta_{(b,a)} \right)} \right]_{x=t} \, ,
\eea
where in the second equality we have Taylor-expanded the integrand around the lower bound.
Doing the integration over $t'$ we see that the leading contribution is for $j=0$, thus
\be
I_k = \frac{\rho(t)\, \e^{ik \left( \smallstrut v_b(t) - v_a(t) + \Delta_{(b,a)} \right)} }{ k N^\alpha} + \cO(N^{-2\alpha}) \, .
\ee
Substituting we find
\be
\label{first summation}
\sum_{i<j} \Li_2 \left( \e^{i \left( u_j^{(b)} - u_i^{(a)} + \Delta_{(b,a)} \right)} \right) = N^{2-\alpha} \int  \rd t\, \Li_3 \left( \e^{i \left( \smallstrut \delta v(t) + \Delta_{(b,a)} \right)} \right)\, \rho(t)^2 + \cO(N^{2-2\alpha}) \;.
\ee
Next, we need to compute the first term in \eqref{broken expression}.
In order for the integral to be localized at the boundary, we need to invert the integrand.
Since $0<\re \left(u_j^{(b)} - u_i^{(a)} + \Delta_{(b,a)}\right)<2\pi$: 
\be\begin{aligned}
\label{reflection applied 34}
\Li_2 \left( \e^{i \left( u_j^{(b)} - u_i^{(a)} + \Delta_{(b,a)} \right)} \right) & = - \Li_2 \left( \e^{i \left( u_i^{(a)} - u_j^{(b)} - \Delta_{(b,a)} \right)} \right) + \frac{ \left(u_j^{(b)} - u_i^{(a)} + \Delta_{(b,a)}\right)^2}2 \\
& - \pi \left(u_j^{(b)} - u_i^{(a)} + \Delta_{(b,a)}\right) + \frac{\pi^2}3 \;.
\end{aligned}\ee
The summation $\sum_{i>j}$ of the first term in the latter expression is similar to \eqref{first summation} but with $-\Li_3\left( \e^{-i \left( \delta v(t) + \Delta_{(b,a)} \right)} \right)$ instead of $\Li_3$.
The two contributions may then be combined, using \eqref{reflection formulae},
\be\begin{aligned}
& N^{2-\alpha} \int  \rd t\,  \left[\Li_3 \left( \e^{i \left( \smallstrut \delta v(t) + \Delta_{(b,a)} \right)} \right) - \Li_3\left( \e^{-i \left( \delta v(t) + \Delta_{(b,a)} \right)} \right) \right]\, \rho(t)^2 \\
& = i N^{2-\alpha} \int  \rd t\,  g_+\left( \smallstrut \delta v(t) + \Delta_{(b,a)} \right) \, \rho(t)^2\, ,
\end{aligned}\ee
where we have introduced the polynomial function $g_+(u)$ defined in  Eq.\,\eqref{gp gm:ch:2}.

The second term in the first line of (\ref{cV bi-fundamentals}) can be treated similarly.
We now have $-2\pi< \re (u_j^{(b)} - u_i^{(a)} - \Delta_{(a,b)}) < 0$ and
\be\begin{aligned}
\label{reflection applied 12}
- \Li_2 \left( \e^{i \left(u_j^{(b)} - u_i^{(a)} - \Delta_{(a,b)}\right)} \right) & = \Li_2 \left( \e^{i \left(u_i^{(a)} - u_j^{(b)} + \Delta_{(a,b)} \right)} \right) - \frac{\left(u_j^{(b)} - u_i^{(a)} - \Delta_{(a,b)}\right)^2}2 \\
& - \pi \left(u_j^{(b)} - u_i^{(a)} - \Delta_{(a,b)} \right) - \frac{\pi^2}3 \, .
\end{aligned}\ee
As before, the result of the summation $\sum_{i>j}$ together with that of $\sum_{i<j}$ yields a cubic polynomial expression
\be\begin{aligned}
&  - i N^{2-\alpha} \int  \rd t\,  g_-\left( \smallstrut \delta v(t) - \Delta_{(a,b)} \right) \, \rho(t)^2\, ,
\end{aligned}\ee
where $g_-(u)$ is defined in  Eq.\,\eqref{gp gm:ch:2}.

The left over terms from \eqref{reflection applied 34} and \eqref{reflection applied 12}, throwing away the constants which do not affect the critical points, are
\be\label{bi-fundamentals forces}
 \left[ \left( \Delta_{(a,b)} - \pi \right) + \left( \Delta_{(b,a)} - \pi \right) \right] \sum_{i>j} \left( u_j^{(b)} - u_i^{(a)}\right) \, ,
\ee
which, combined with the second line in \eqref{cV bi-fundamentals}, gives
\be\label{bi-fundamentals forces2}
 \frac{1}{2} \left[ \left( \Delta_{(a,b)} - \pi \right) + \left( \Delta_{(b,a)} - \pi \right) \right] \sum_{i\neq j} \left( u_j^{(b)} - u_i^{(a)}\right) \text{sign} (i-j)\, .
\ee
This term can be precisely canceled by
\begin{equation}
 - 2 \pi \sum_{i=1}^N \left( n_i^{(b)} u_i^{(b)} - n_i^{(a)} u_i^{(a)} \right) \, ,
\end{equation}
provided that $\sum_{I \in a} \Delta_I \in 2 \pi \bZ $.\footnote{When $N \in 2 \bZ_{\geq 0} +1 $. For even $N$  one can include an extra $(-1)^{\fm}$ in the twisted partition function, which can be reabsorbed in the definition of the topological fugacity $\xi$, to compensate the overall factor of $\pi$.}

Notice that  a single bi-fundamental chiral multiplet, with chemical potential $\Delta_{(b,a)}$,
transforming in the representation $(\overline{\bf N},{\bf N})$ of $\U(N)_a \times \U(N)_b$ contributes to the twisted superpotential as
\begin{equation}\label{single bifund Bethe}
 \sum_{i,j=1}^N \left[ \Li_2 \left( \e^{i \left( u_j^{(b)} - u_i^{(a)} + \Delta_{(b,a)} \right)} \right) - \frac{\left( u_j^{(b)} - u_i^{(a)} + \Delta_{(b,a)} \right)^2}{4} \right] \, .
\end{equation}
Using Eq.\,\eqref{reflection applied 34} we find the following long-range terms
\be\label{forces chiral}
\sum_{i<j} \frac{\left(u_i^{(a)} - u_j^{(b)}\right)^2}{4} - \sum_{i<j} \frac{\left(u_j^{(a)} - u_i^{(b)}\right)^2}{4}\, .
\ee
In the large N limit, they are of order $N^{5/2}$ and cannot be canceled  for {\it chiral quivers}.

To find a nontrivial saddle-point the leading terms of order $N^{1+\alpha}$ and $N^{2-\alpha}$ have to be of the same order, so we need $\alpha = 1/2$.
Putting everything together we arrive at the final expression for the large $N$ contribution of the bi-fundamental fields to the twisted superpotential
\begin{equation}
 \wt\cW^{\text{bi-fund}} = i N^{3/2} \sum_{\substack{\text{bi-fundamentals} \\ (b,a) \text{ and } (a,b)}} \int  \rd t\, \rho(t)^2 \left[g_+\left(\delta v(t) + \Delta_{(b,a)}\right) - g_-\left(\delta v(t) - \Delta_{(a,b)}\right)\right]\, .
\end{equation}

In the sum over pairs of bi-fundamental fields $(b,a)$ and $(a,b)$, adjoint fields should be counted once and should come with an explicit factor of $1/2$.
Keeping this in mind and setting $v_b = v_a \, ,$ $\Delta_{(b,a)} =  \Delta_{(a,b)} =  \Delta_{(a,a)} \, ,$ we find the
contribution of fields transforming in the adjoint of the $a$th gauge group
with chemical potential $\Delta_{(a,a)}$ to the large $N$ twisted superpotential,
\begin{equation}
 \wt\cW^{\text{adjoint}} = i N^{3/2} \sum_{\substack{\text{adjoint} \\ (a,a) }} g_+\left(\Delta_{(a,a)}\right) \int  \rd t\, \rho(t)^2 \, .
\end{equation}

\subsubsection{Fundamental and anti-fundamental contribution}
\label{Bethe fundamentals N32}
 
The fundamental and anti-fundamental fields contribute to the large $N$  twisted superpotential as\footnote{Up to a factor $-\pi (\tilde n_a -n_a) u_i/2$ that cancels at this order  for 
total number of fundamentals equal to total number of anti-fundamentals, which we will need to assume for consistency.} 
\be\begin{aligned}
 \wt\cW^{\text{(anti-)fund}} & =
 \sum_{i=1}^{N} \left[ \sum_{\substack{\text{anti-fundamental} \\ a }} \Li_2 \left(\e^{i\left(-u_i ^{(a)} + \tilde \Delta_a \right)}\right) - \sum_{\substack{\text{fundamental} \\ a }} \Li_2 \left(\e^{i\left(-u_i ^{(a)} - \Delta_a\right)}\right) \right] \\
 & + \frac{1}{2} \sum_{i=1}^{N} \left[ \sum_{\substack{\text{anti-fundamental} \\ a }} \big( \tilde \Delta_a - \pi \big) u_i^{(a)} + \sum_{\substack{\text{fundamental} \\ a }} \big( \Delta_a - \pi \big) u_i^{(a)}\right] \\
 & - \frac{1}{4} \sum_{i=1}^{N} \left[\sum_{\substack{\text{anti-fundamental} \\ a }} \left(u_i^{(a)}\right)^2
 - \sum_{\substack{\text{fundamental} \\ a }} \left(u_i^{(a)}\right)^2 \right] \, .
\end{aligned}\ee
Let us denote the total number of (anti-)fundamental fields by $(\tilde n_a)\, n_a$.
Substituting in $\wt\cW^{\text{(anti-)fund}}$ the ansatz \eqref{ansatz alpha} and taking the continuum limit, the first line contributes
\be\begin{aligned}\label{fund Li2}
 & - \frac{\left( \tilde n_a - n_a \right)}{2} N^{2} \int_{t>0}  \rd t\, \rho(t)\, t^2 \\
 & + i N^{3/2} \int_{t>0}  \rd t\, \rho(t)\, t\, \left\{ \sum_{\substack{\text{anti-fundamental} \\ a }} \Big[ v_a(t) - \big( \tilde \Delta_{a} - \pi \big) \Big] - \sum_{\substack{\text{fundamental} \\ a }} \Big[ v_a(t) + \big( \Delta_{a} - \pi \big) \Big] \right\} \, ,
\end{aligned}\ee
while the second and the third lines give
\be\begin{aligned}\label{fund linear}
 & \frac{\left( \tilde n_a - n_a \right)}{4} N^{2} \int  \rd t\, \rho(t)\, t^2 \\
 & - \frac{i}{2} N^{3/2} \int  \rd t\, \rho(t)\, t\, \left\{ \sum_{\substack{\text{anti-fundamental} \\ a }} \Big[ v_a(t) - \big( \tilde \Delta_{a} - \pi \big) \Big] - \sum_{\substack{\text{fundamental} \\ a }} \Big[ v_a(t) + \big( \Delta_{a} - \pi \big) \Big] \right\} \, . %
\end{aligned}\ee
Combining Eq.\,\eqref{fund Li2} and Eq.\,\eqref{fund linear}, we obtain
\be\begin{aligned}\label{fundcontr}
 \wt\cW^{\text{(anti-)fund}} & = - \frac{\left( \tilde n_a - n_a \right)}{4} N^{2} \int  \rd t\, \rho(t)\, t\, |t| \\
 & + \frac{i}{2} N^{3/2} \sum_{\substack{\text{anti-fundamental} \\ a }} \int  \rd t\, \rho(t)\, |t|\, \Big[ v_a(t) - \big( \tilde \Delta_{a} - \pi \big) \Big] \\
 & - \frac{i}{2} N^{3/2} \sum_{\substack{\text{fundamental} \\ a }} \int  \rd t\, \rho(t)\, |t|\, \Big[ v_a(t) + \big( \Delta_{a} - \pi \big) \Big] \, .
\end{aligned}\ee
Summing over nodes the first term vanishes, demanding that
\be
\sum_{a=1}^{|G|} \left( \tilde n_a - n_a \right) = 0 \, .
\ee
We see that we need to consider quivers where the {\it total} number of fundamentals equal the {\it total} number of anti-fundamentals.
For each single node this number can be different. 

\subsection[The index at large \texorpdfstring{$N$}{N}]{The index at large $\fakebold{N}$}
\label{entropy rules N32}

We are interested in the large $N$ limit of the logarithm of the twisted partition function.

\subsubsection{Gauge contribution}
\label{entropy CS N32}

Given the expression for the matrix model in section \ref{index}, the Vandermonde determinant contributes to the logarithm of the index  as
\be\begin{aligned}\label{gauge entropy}
 \log \prod_{i \neq j} \left( 1- \frac{x_i^{(a)}}{x_j^{(a)}} \right) & =
 \log \prod_{i < j} \left( 1- \frac{x_j^{(a)}}{x_i^{(a)}} \right)^2 \left( - \frac{x_i^{(a)}}{x_j^{(a)}} \right) \\
 & = i \sum_{i<j}^N \left( u_i^{(a)} - u_j^{(a)} + \pi \right)
 - 2 \sum_{i<j}^N \Li_{1} \left(\e^{i \left( u_j^{(a)} - u_i^{(a)}\right)}\right) \, .
\end{aligned}\ee
The first term is of $\cO(N^2)$ and, therefore, a source of the long-range forces and will be  canceled by the contribution coming from the chiral multiplets.
The second term is treated as in section \ref{Bethe bi-fundamentals N32}, and gives
\be
 \re \log Z^{\text{gauge}} = - \frac{\pi^2}{3} N^{3/2}  \int  \rd t\, \rho(t)^2 + \cO(N)  \, .
\ee

\subsubsection{Topological symmetry contribution}
\label{entropy topological N32}

A $\U(1)_a$ topological symmetry with flux $\ft_a$ contributes as $i \sum_{i=1}^{N} u_i^{(a)} \ft_a \, .$
In the continuum limit, we get
\be
 \re \log Z^{\text{top}} = - \ft_a N^{3/2} \int  \rd t\, \rho(t)\, t + \cO(N)  \, .
\ee

\subsubsection{Bi-fundamental contribution}
\label{entropy bi-fundamentals N32}

We can rewrite the contribution  to the twisted index of a bi-fundamental chiral multiplet transforming in the $(\overline{\bf N}, {\bf N})$ of $\U(N)_a \times \U(N)_b$, with magnetic flux $\fn_{(b,a)}$ and chemical potential $\Delta_{(b,a)}$  as:\footnote{The phases can be neglected, as we will be interested in $\log |Z|$.}
\be\begin{aligned}
 & \prod_{i=1}^{N} \left( \frac{x_i^{(b)}}{x_i^{(a)}} \right)^{\frac{1}{2} \left( \fn_{(b,a)} - 1\right)}
 \left( 1 - y_{(b,a)} \frac{x_i^{(b)}}{x_i^{(a)}} \right)^{\fn_{(b,a)} - 1} \times \\
 & \prod_{i<j}^{N} (-1)^{\fn_{(a,b)} - 1} \left( \frac{x_i^{(a)} x_i^{(b)}}{x_j^{(a)} x_j^{(b)}} \right)^{\frac{1}{2} \left(\fn_{(b,a)} - 1\right)}
 \left( 1 - y_{(b,a)} \frac{x_j^{(b)}}{x_i^{(a)}} \right)^{\fn_{(b,a)}-1} \left( 1 - y_{(b,a)}^{-1} \frac{x_j^{(a)}}{x_i^{(b)}} \right)^{\fn_{(b,a)}-1} \, .
\end{aligned}\ee
The first term in $\prod_{i}$ is subleading and the second term only contributes in the tail where $\delta v \approx - \Delta_{(b,a)}$,
\be\begin{aligned}
 & N \left(\fn_{(b,a)} - 1\right) \int  \rd t\, \rho(t) \log \left( 1-\e^{i\left(\delta v + \Delta_{(b,a)}\right)} \right) \\
 & = - N^{3/2} \; \left(\fn_{(b,a)} - 1\right) \int_{\delta v \,\approx\, - \Delta_{(b,a)}} \hspace{-3em}  \rd t\, \rho(t) \, Y_{(b,a)}(t) + \cO(N)
\end{aligned}\ee
The first two terms in $\prod_{i<j}$ give a long-range force contribution to the index
\begin{equation}\label{forces chiral ba}
 \frac{i}{2} \left( \fn_{(b,a)} - 1 \right) \sum_{i<j} \left[ \left( u_i^{(a)} - u_j^{(a)} + \pi \right) + \left( u_i^{(b)} - u_j^{(b)} + \pi \right) \right] \, ,
\end{equation}
while the last two terms result in
\be\begin{aligned}
&  - N^{3/2} \left(\fn_{(b,a)}-1\right) \int  \rd t\, \rho(t)^2 \left[ \Li_2\left( \e^{i\left(\delta v+\Delta_{(b,a)}\right)} \right) + \Li_2 \left( \e^{-i \left(\delta v + \Delta_{(b,a)}\right)} \right) \right] + \cO(N) \\
 & = - N^{3/2} \left(\fn_{(b,a)}-1\right) \int  \rd t\, \rho(t)^2 g_+' \left( \delta v(t) + \Delta_{(b,a)} \right) + \cO(N) \, .
\end{aligned}\ee
A bi-fundamental field transforming in the $({\bf N}, \overline{\bf N})$ of $\U(N)_a \times \U(N)_b$, with magnetic flux $\fn_{(a,b)}$ and chemical potential $\Delta_{(a,b)}$ gives
the same contribution with the replacement $a\leftrightarrow b$ and $\delta v \rightarrow -\delta v$. 

The long-range force contribution of bi-fundamental fields at node $a$ cancels with the gauge contribution  in \eqref{gauge entropy},  provided that
\be
2 + \sum_{I \in a} (\fn_I - 1) = 0 \, ,
\ee
where the sum is taken over all chiral bi-fundamentals $I$ with an endpoint in $a$.

In picking the residues, we need to insert a Jacobian in the twisted index and evaluate everything else at the pole.
The matrix $\mathds{B}$ appearing in the Jacobian is $2 N \times 2 N$ with block form
\be
\label{Jacobian general}
\mathds{B} = \frac{\partial \big( \e^{i B_j^{(a)}}, \e^{i B_j^{(b)}} \big) }{ \partial( \log x_l^{(a)}, \log x_l^{(b)})} =
\mat{ x_l^{(a)} \dfrac{\partial \e^{iB_j^{(a)}}}{\partial x_l^{(a)}} & x_l^{(b)} \dfrac{ \partial \e^{i B_j^{(a)}}}{\partial x_l^{(b)}}  \\[1em]
x_l^{(a)} \dfrac{\partial \e^{i B_j^{(b)}}}{\partial x_l^{(a)}} & x_l^{(b)} \dfrac{ \partial \e^{i B_j^{(b)}}}{\partial x_l^{(b)}} }_{2 N \times 2 N} \, ,
\ee
and only contributes in the tails regions,\footnote{We refer the reader to \cite{Benini:2015eyy} for a detailed analysis of the Jacobian at large $N$.}%
\bea
-\log \det \mathds{B} = - N^{3/2} \sum_{\substack{\text{bi-fundamentals} \\ (b,a) \text{ and } (a,b)}}
 & \int_{\delta v \,\approx\, - \Delta_{(b,a)}} \hspace{-3em}  \rd t\, \rho(t) \, Y_{(b,a)}(t) \\
 & + \int_{\delta v \,\approx\,\Delta_{(a,b)}} \hspace{-2em}  \rd t\, \rho(t) \, Y_{(a,b)}(t) + \cO(N\log N) \, .
\eea

Summarizing, pairs of bi-fundamental fields contribute to the logarithm of the index  as
\bea\label{indexbi}
 \re \log Z^{\text{bi-fund}}_{\text{bulk}} = - N^{3/2} \sum_{\substack{\text{bi-fundamentals} \\ (b,a) \text{ and } (a,b)}}
 & \int  \rd t\, \rho(t)^2 \big[(\fn_{(b,a)}-1)\, g'_+ \left(\delta v(t) + \Delta_{(b,a)}\right) \\
 & + (\fn_{(a,b)}-1)\, g'_- \left(\delta v(t) - \Delta_{(a,b)}\right)\big]\, .
\eea
The tails contribution is also given by
\bea
 \re \log Z^{\text{bi-fund}}_{\text{talis}} = - N^{3/2} \sum_{\substack{\text{bi-fundamentals} \\ (b,a) \text{ and } (a,b)}}
 & \fn_{(b,a)}  \int_{\delta v \,\approx\, - \Delta_{(b,a)}} \hspace{-3em}  \rd t\, \rho(t) \, Y_{(b,a)}(t) \\
 & + \fn_{(a,b)} \int_{\delta v \,\approx\,\Delta_{(a,b)}} \hspace{-2em}  \rd t\, \rho(t) \, Y_{(a,b)}(t) \, .
\eea

A field transforming in the adjoint of the $a$th gauge group with magnetic flux $\fn_{(a,a)}$ and chemical potential $\Delta_{(a,a)}$ only contributes to the bulk index.
To find its contribution we need to include an explicit factor of $1/2$ in the  expression (\ref{indexbi}) and take
\begin{equation}
v_b = v_a \, , \qquad \qquad  \Delta_{(b,a)} =  \Delta_{(a,b)} =  \Delta_{(a,a)} \, , \qquad \qquad  \fn_{(b,a)} =  \fn_{(a,b)} =  \fn_{(a,a)} \, .
\end{equation}

\subsubsection{Fundamental and anti-fundamental contribution}
\label{entropy fundamentals N32}

The fundamental and anti-fundamental fields contribute to the logarithm of the index  as
\begin{align}
 & \log \prod_{i=1}^{N} \prod_{\substack{\text{anti-fundamental} \\ a }} \left(x_i^{(a)}\right)^{\frac12 \left(\tilde \fn_a - 1\right)}
 \left[ 1 - \tilde y_a \left( x_i^{(a)} \right)^{-1} \right]^{\tilde \fn_a - 1} \times \nn \\
 & \prod_{\substack{\text{fundamental} \\ a }} \left(x_i^{(a)}\right)^{\frac12 \left( \fn_a - 1\right)}
 \left[ 1 - y_a^{-1} \left( x_i^{(a)} \right)^{-1} \right]^{\fn_a - 1} \, .
\end{align}
Using the scaling ansatz \eqref{ansatz alpha}, in the continuum limit we get
\begin{align}
 & \log \prod_{i=1}^N \prod_{\substack{\text{anti-fundamental} \\ a }} \left(x_i^{(a)}\right)^{\frac12 \left(\tilde \fn_a - 1\right)}
 \prod_{\substack{\text{fundamental} \\ a }} \left(x_i^{(a)}\right)^{\frac12 \left( \fn_a - 1\right)} \nn \\
 & = - \frac{1}{2} N^{3/2} \left[\sum_{\substack{\text{anti-fundamental} \\ a }} \left(\tilde \fn_a - 1\right)
 + \sum_{\substack{\text{fundamental} \\ a }} \left( \fn_a - 1\right) \right] \int  \rd t\, \rho(t)\, t
 + \mathcal{O}(N) \, ,
\end{align}
and
\begin{align}
 &\log \prod_{i=1}^{N} \prod_{\substack{\text{anti-fundamental} \\ a }} \left[ 1 - \tilde y_a \left( x_i^{(a)} \right)^{-1} \right]^{\tilde \fn_a - 1}
 \prod_{\substack{\text{fundamental} \\ a }} \left[ 1 - y_a^{-1} \left( x_i^{(a)} \right)^{-1} \right]^{\fn_a - 1} \nn \\
 & = N^{3/2} \left[\sum_{\substack{\text{anti-fundamental} \\ a }} \left(\tilde \fn_a - 1\right)
 + \sum_{\substack{\text{fundamental} \\ a }} \left( \fn_a - 1\right) \right]
 \int_{t>0}  \rd t\, \rho(t)\, t + \mathcal{O}(N) \, .
\end{align}
Putting the above equations together we find:
\begin{equation}
 \re \log Z^{\text{(anti-)fund}} = \frac{1}{2} N^{3/2} \left[\sum_{\substack{\text{anti-fundamental} \\ a }} \left(\tilde \fn_a - 1\right)
 + \sum_{\substack{\text{fundamental} \\ a }} \left( \fn_a - 1\right) \right] \int  \rd t\, \rho(t)\, |t| \, .
\end{equation}

\section[Twisted superpotential versus free energy on \texorpdfstring{$S^3$}{S**3}]{Twisted superpotential versus free energy on $\fakebold{S^3}$}
\label{sec:freeenergy}

We would like to emphasize a remarkable connection of the large $N$ limit of the twisted superpotential,
which for us is an auxiliary quantity, with the large $N$ limit of the free energy $F_{S^3}$ on $S^3$
of the same ${\cal N}\geq 2$ theory.

Recall that the three-sphere free energy $F_{S^3}$ of an ${\cal N}=2$ theory is 
a function of trial R-charges $\Delta_I$  for the chiral fields \cite{Jafferis:2010un,Hama:2010av}.
They parameterize the curvature coupling of the supersymmetric Lagrangian on $S^3$. 
The $S^3$ free energy can be computed using localization and reduced to a matrix model \cite{Kapustin:2009kz}.
The large $N$ limit of the free energy, for $N \gg k_a$, has been computed in \cite{Herzog:2010hf,Jafferis:2011zi,Gulotta:2011vp,Gulotta:2011aa,Gulotta:2011si}
and scales as $N^{3/2}$. For example, the free energy for
ABJM with $k=1$ reads \cite{Jafferis:2011zi}
 \be F_{S^3} = \frac{4 \pi N^{3/2}} {3} \sqrt{ 2 \Delta_1 \Delta_2 \Delta_3 \Delta_4} \, .
 \ee
We notice a striking similarity with \eqref{Vsolution sum 2pi -- end}. This is not a coincidence and generalizes to other theories.
Indeed, remarkably, although the finite $N$ matrix models are quite different, for any ${\cal N}=2$ theory,
the large $N$ limit of the twisted superpotential becomes exactly equal to the large $N$ limit of the free energy on $S^3$.
We can indeed compare the rules for constructing the twisted superpotential with the rules for constructing the
large $N$ limit of $F_{S^3}$, which have been derived in \cite{Jafferis:2011zi}. By comparing the rules in
section \ref{large N twisted superpotential rules} with the rules given in section 2.2 of \cite{Jafferis:2011zi},
we observe that they are indeed the same up to a normalization.
For reader's convenience the map is explicitly given in Table.\,\ref{Bethe map}.
\begin{table}[h!!!!]
\centering
\begin{tabular}{@{}l l@{}} \toprule \toprule
 twisted superpotential & $S^3$ free energy\\
 \toprule
 $k_a$ & $- k_a$ \\
 $\mu$ & $\frac{\mu}{2}$ \\
 $v_a(t)$ & $\frac{v_a(t)}{2}$ \\
 $\rho(t)$ & $4 \rho(t)$ \\
 $\Delta_I$ & $\pi \Delta_I$ \\
 $\Delta_m$ & $- \pi \Delta_m$ \\
 $\wt\cW$ & $4 \pi i F_{S^3}$ \\
 $\wt\cW \big|_{\text{BAEs}}$ & $\frac{i \pi}{2} F_{S^3} \big|_{\text{On-shell}}$ \\[0.3cm]
 \toprule
 \toprule
\end{tabular}
\caption{The large $N$ twisted superpotential versus the $S^3$ free energy of \cite{Jafferis:2011zi}.}
\label{Bethe map}
\end{table}
The conditions for cancellation of long-range forces (and therefore the allowed models) are also remarkably similar.

It might be surprising that our chemical potentials for global symmetries are mapped to R-charges in the free energy.
However, remember that our $\Delta_I$ are angular variables. The invariance of the superpotential under the global
symmetries implies that 
\be
\prod_{I \in \text{matter fields}} y_I = 1 \, ,
\ee
in each term of the superpotential, which is equivalent to 
\be \label{2pi}
\sum_{I \in \text{matter fields}} \Delta_I = 2 \pi \ell \,  \qquad \ell\in \mathbb{Z} \, ,
\ee 
where now $\Delta_I$ are the index chemical potentials. Under the assumption $0<\Delta_I<2\pi$, few values of $\ell$
are actually allowed. In the ABJM model reviewed above, only $\ell=1$ and $\ell=3$ give sensible results,
with $\ell=3$ related to $\ell=1$ by a discrete symmetry of the model.
We found a similar result in all the examples we have checked, and we do believe indeed that
a solution to the BAEs only exists when  
\be\label{superpotential0} 
\sum_{I \in \text{matter fields}} \Delta_I = 2 \pi\, ,
\ee
for each term of the superpotential, up to solutions related by discrete symmetries.
$\Delta_I / \pi$ then behaves at all effects like  a trial R-symmetry of the theory and we can
compare the index chemical potentials in $\wt\cW$ with the R-charges in $F_{S^3}$.
 
\section{An index theorem for the twisted matrix model}\label{sec:index}

Under mild assumptions, the index at large $N$ can be actually extracted from the twisted superpotential with a simple formula.

\begin{theorem}\label{entropy theorem}
The index of any $\cN \geq 2$ quiver gauge theory which respects the constraints \eqref{no long-range-forces Bethe0} and \eqref{no long-range forces}, and satisfies in addition \eqref{superpotential0}, can be written as
\begin{equation}\label{Z large N conjecture}
 \re\log Z = - \frac{2}{\pi} \, \wt\cW(\Delta_I) \,
 - \sum_{I}\, \left[ \left(\fn_I - \frac{\Delta_I}{\pi}\right) \frac{\partial \wt\cW(\Delta_I)}{\partial \Delta_I}
  \right] \, ,
\end{equation}
where $\wt\cW (\Delta_I)$ is the extremal value of the functional \eqref{bethefunctional}
\begin{equation}\label{virial theorem}
 {\quad \rule[-1.4em]{0pt}{3.4em}
 \wt\cW(\Delta_I) \equiv - i \wt\cW \left(\rho(t), v_a(t), \Delta_I \right) \big|_{\text{BAEs}}
 = \frac{2}{3} \mu N^{3/2} \, ,}
\end{equation}
and $\mu$ is the  Lagrange multiplier appearing in \eqref{bethefunctional}.%
\footnote{The second identity in \eqref{virial theorem} is a consequence of a virial theorem for matrix models
(see appendix B of \cite{Gulotta:2011si}).}

\end{theorem}

\begin{proof}
We first replace the explicit factors of $\pi$, appearing in Eqs.\,\eqref{twisted superpotential bi-fundamental}-\eqref{anti-fund cV rule}, with a formal variable $\fakebold{\pi}$.
Note that the ``on-shell'' twisted superpotential $\wt\cW (\Delta_I)$ is a homogeneous function of $\Delta_I$ and $\fakebold{\pi}$ and therefore it satisfies
\begin{equation}
 \wt\cW(\lambda \Delta_I, \lambda \fakebold{\pi}) = \lambda^2 \, \wt\cW(\Delta_I, \fakebold{\pi}) \quad \Rightarrow \quad
 \frac{\partial \wt\cW(\Delta_I, \fakebold{\pi})}{\partial \fakebold{\pi}} =
 \frac{1}{\fakebold{\pi}} \left[ 2 \, \wt\cW(\Delta_I) -\sum_I  \Delta_I \frac{\partial \wt\cW(\Delta_I)}{\partial \Delta_I} \right]\, .
\end{equation}
Now, we consider a pair of bi-fundamental fields which contribute to the twisted superpotential according to \eqref{twisted superpotential bi-fundamental}.
The derivative of $\wt\cW(\Delta_I, \fakebold{\pi})$ with respect to  $\Delta_{(b,a)}$ and $\Delta_{(a,b)}$ is given by
 \begin{align}\label{proof bf}
  \sum_{I=(b,a),(a,b)} \fn_{I} \frac{\partial\wt\cW(\Delta_I, \fakebold{\pi})}{\partial\Delta_I} & =
  i N^{3/2} \int \rd t\, \rho(t)^2\, \left[ \fn_{(b,a)} g'_+ \left(\delta v(t) + \Delta_{(b,a)}\right) + \fn_{(a,b)} g'_- \left(\delta v(t) - \Delta_{(a,b)}\right)\right] \nn \\
  & +   \sum_{I=(b,a),(a,b)}  \fn_{I}  \underbrace{\frac{\partial \wt\cW}{\partial \rho} \frac{\partial \rho}{\partial \Delta_I}}_\text{vanishing on-shell}
  +   \sum_{I=(b,a),(a,b)} \fn_{I}  \underbrace{\frac{\partial \wt\cW}{\partial (\delta v)} \frac{\partial (\delta v)}{\partial \Delta_I}}_\text{tails contribution} \, .
 \end{align}
 The expression in the first line is precisely part of the contribution of a pair of bi-fundamentals \eqref{bicontribution} to the index.  
In the tails, using \eqref{tails}, we find
 \begin{equation}
 \frac{\partial (\delta v)}{\partial \Delta_{(b,a)}} = -1 \, , \qquad
 \frac{\partial (\delta v)}{\partial \Delta_{(a,b)}} = 1 \, , \qquad
 \frac{\partial \wt\cW}{\partial (\delta v)} = - i Y_{(b,a)} \rho \, , \qquad
 \frac{\partial \wt\cW}{\partial (\delta v)} = i Y_{(a,b)} \rho \, . 
 \end{equation}
Therefore, the last term in Eq.\,\eqref{proof bf} can be simplified to
 \begin{equation}
 i N^{3/2}  \fn_{(b,a)} \int_{\delta v \approx - \Delta_{(b,a)}} \rd t\, \rho(t)\, Y_{(b,a)}
 + i N^{3/2}  \fn_{(a,b)}  \int_{\delta v \approx \Delta_{(a,b)}} \rd t\, \rho(t)\, Y_{(a,b)} \, .
\end{equation}
This precisely gives the tail contribution \eqref{tailcontribution} to the index. 
Next, we take the derivative of the twisted superpotential with respect to  $\fakebold{\pi}$. It can be written as
 \begin{align}\label{pi derivative}
 \frac{\partial \wt\cW}{\partial \fakebold{\pi}}  & = - i N^{3/2} \int \rd t\, \rho(t)^2 \left[ g'_+ \left(\delta v(t) + \Delta_{(b,a)}\right) + g'_- \left(\delta v(t) - \Delta_{(a,b)}\right) \right] \nn \\
 & + i N^{3/2} \int \rd t\, \rho(t)^2\, \left[ \frac{2 \fakebold{\pi}^2}{3} - \frac{\fakebold{\pi}}{3} \left( \Delta_{(b,a)} + \Delta_{(a,b)} \right) \right] \, .
 \end{align}
 The expression in the first line completes the contribution of a pair of bi-fundamentals \eqref{bicontribution} to the index.  
The expression  in the second line, after summing  over all the bi-fundamental pairs,  can be written as
 \begin{equation}
\sum_{\text{pairs}} \left[ \frac{2 \fakebold{\pi}^2}{3} - \frac{\fakebold{\pi}}{3} \left( \Delta_{(b,a)} + \Delta_{(a,b)} \right)\right]
 = \frac{\fakebold{\pi}}{3} \sum_{I} \left( \fakebold{\pi} - \Delta_I \right)
 = \frac{\fakebold{\pi}^2}{3} |G| \, , 
 \end{equation}
which is precisely the contribution of the gauge fields  \eqref{gaugecontribution} to the index.  
Here, we used the condition
 \begin{equation}\label{pi constraint}
 \fakebold{\pi} |G| + \sum_{I} \left( \Delta_I - \fakebold{\pi} \right) = 0 \, .
 \end{equation}
This condition follows from the fact that, assuming  \eqref{superpotential0} for each superpotential term, $\Delta_I/\pi$ behaves as a trial R-symmetry, so that  \eqref{no long-range forces} yields
\begin{equation}\label{no long-range forces2}
 2 + \sum_{I \in a} \left(\frac{\Delta_I}{\pi} - 1\right) = 0 \, ,
\end{equation}
which, summed over all the nodes, since each bi-fundamental field belongs precisely to two nodes, gives \eqref{pi constraint}.
Condition \eqref{pi constraint} is indeed equivalent to ${\rm Tr} R=0$, where the trace is taken over the bi-fundamental  fermions and gauginos in the quiver and $R$ is an R-symmetry. 
Combining everything as in the right hand side of Eq.\,\eqref{Z large N conjecture} we obtain the contribution of gauge and bi-fundamental fields to the index.
The proof for all the other matter fields and the topological symmetry is straightforward.
\end{proof}

If we ignore the linear relation among the chemical potentials, we can always use a  set of  $\Delta_I$ such that $\wt\cW(\Delta_I)$ is a homogeneous function of degree two of the
$\Delta_I$ alone.\footnote{This is what happens in \eqref{Vsolution sum 2pi -- end} for ABJM. Recall that $\sum_i \Delta_i = 2 \pi$ so that the four $\Delta_i$ are not linearly independent.}
In this case, the index theorem simplifies to
\begin{equation}\label{Z large N conjecture2}
 {\quad \rule[-1.4em]{0pt}{3.4em}
 \re\log Z = 
 - \sum_{I}\, \fn_I \frac{\partial \wt\cW(\Delta_I)}{\partial \Delta_I}
  \, .
 \quad}
\end{equation}

\section[Theories with \texorpdfstring{$N^{5/3}$}{N**(5/3)} scaling of the index]{Theories with $\fakebold{N^{5/3}}$ scaling of the index}
\label{N53}

Chern-Simons quivers of the form \eqref{quiver} have a rich parameter space. If  condition \eqref{CScontr} is satisfied and $N\gg k_a$,  they have an M-theory weakly coupled dual.  In the t'Hooft limit, $N,k_a\gg 1$ with $N/k_a =\lambda_a$ fixed and large, they have a type IIA weakly coupled dual. When instead
\be\label{CScontr2}
\sum_{a=1}^{|G|} k_a \ne 0 \, ,
\ee
they probe massive type IIA \cite{Gaiotto:2009mv}. There is an interesting limit, given \eqref{CScontr2}, where again   $N\gg k_a$. The limit is no more an M-theory phase \cite{Aharony:2010af},
but rather an extreme phase of massive type IIA. Supergravity duals of this type of phases have been found in \cite{Petrini:2009ur,Lust:2009mb,Aharony:2010af,Tomasiello:2010zz,Guarino:2015jca,Fluder:2015eoa,Pang:2015vna,Pang:2015rwd}. The free energy scales as $N^{5/3}$ \cite{Aharony:2010af}. We now show that also the topologically twisted index scales in the 
same way. As it happens for the $S^3$ matrix model \cite{Jafferis:2011zi,Fluder:2015eoa}, we find a consistent large $N$ limit whenever the constraints \eqref{globalJ} and \eqref{globalR} are satisfied. 

The ansatz for the eigenvalues distribution is now, as in \cite{Jafferis:2011zi,Fluder:2015eoa},
 \begin{equation}\label{ansatz N53}
 u^{(a)} (t) = N^{\alpha} (i t + v(t))\, ,
\end{equation}
for some $0<\alpha<1$. The scaling is again dictated by the competition between the Chern-Simons terms,
now with \eqref{CScontr2}, and the gauge and bi-fundamental contributions.

\subsection{Long-range forces}
Since the eigenvalue distribution is the same for all gauge groups, the long-range forces \eqref{N3} cancel automatically. We see that, differently from before, we can have a consistent large $N$ limit also in the case of  chiral quivers. We also need  to cancel the long-range forces  \eqref{N2ang}.  They compensate each other  if condition  \eqref{no long-range-forces Bethe0}  is satisfied.
Since the eigenvalues are the same for all groups, it is actually  enough to sum over nodes and we obtain the milder constraint  \eqref{globalJ} on the flavor charges:
\begin{equation} \Tr J =0\, ,
\end{equation}
where the trace is taken over bi-fundamental fermions in the quiver. 

We obtain similar conditions by looking at the scaling of the twisted index.
As  in section \ref{proof vanishing forces}, vector multiplets and chiral bi-fundamental multiplets contribute
terms \eqref{indexscalingg} and \eqref{indexscalingb} which are of order $\cO(N^{7/3})$.
They compensate each other if condition \eqref{no long-range forces} is satisfied.
Since the eigenvalues are the same for all groups, it is again   enough to sum over nodes and we obtain the  constraint \eqref{globalR} on the flavor magnetic fluxes:
\begin{equation}
\Tr R =|G| + \sum_{I} \left(\fn_I - 1\right) = 0\, ,
\end{equation}
where the trace is taken over bi-fundamental fermions and gauginos in the quiver. 

Conditions $\Tr R=\Tr J=0$ are certainly satisfied for all quivers with a four-dimensional parent, even the chiral ones. 

\subsection[Twisted superpotential at large \texorpdfstring{$N$}{N}]{Twisted superpotential at large $\fakebold{N}$}
\label{general:rules:N53:Bethe}

Here, we discuss the $N^{5/3}$ contributions to the twisted superpotential.
Given the large $N$ behavior of the eigenvalues \eqref{ansatz N53},
the classical contribution to the large $N$ twisted superpotential is
\begin{equation}
 \frac{k_a}{2} N^{2 \alpha + 1} \int  \rd t\, \rho(t)\, \left[ t^2 - v(t)^2 - 2 i t\, v(t) \right]\, .
\end{equation}

A bi-fundamental field between $\U(N)_a$ and $\U(N)_b$, with chemical potential $\Delta_{(b,a)}$,
contributes to the twisted superpotential via \eqref{single bifund Bethe}. In order to find the large $N$ behavior
we only need to evaluate \eqref{sum:i:less:j} using the ansatz \eqref{ansatz N53}.
Let us focus on the following integral:
\bea
 I_k & = \int_t  \rd t' \rho(t') \e^{i k \left(u_b(t') - u_a(t) +\Delta_{(b,a)}\right)} \\ &
 = \int_t  \rd t'\, \e^{-k N^\alpha (t' - t)} \sum_{j=0}^{\infty} \frac{(t'-t)^j}{j!}
 \partial_x^j \left[\rho(x) \e^{i k \left[N^\alpha \left(v(x) - v(t)\right) + \Delta_{(b,a)}\right]}\right]_{x=t} \\
 & = \sum_{j=0}^{\infty} \left(k N^\alpha\right)^{-j-1}
 \partial_x^j \left[\rho(x) \e^{i k \left[N^\alpha \left(v(x) - v(t)\right) + \Delta_{(b,a)}\right]}\right]_{x=t} \, .
\eea
Extracting the leading large $N$ contribution of
\bea
 \partial_x^j \left[\rho(x) \e^{i k \left[N^\alpha \left(v(x) - v(t)\right) + \Delta_{(b,a)}\right]}\right]_{x=t} & \sim
 \left(i k N^\alpha\right)^j \left[v'(x)^j \rho(x) \e^{i k \left[N^\alpha \left(v(x) - v(t)\right) + \Delta_{(b,a)}\right]}\right]_{x=t} \\&
 =\left(i k N^\alpha\right)^j v'(t)^j \rho(t) \e^{i k \Delta_{(b,a)}}\, ,
\eea
we obtain
\begin{align}
 I_k = \frac{\e^{i k \Delta_{I}}}{k} \rho(t) N^{-\alpha} \sum_{j=0}^{\infty} \left[i v'(t)\right]^j
 = \frac{\e^{i k \Delta_{I}}}{k} N^{-\alpha} \frac{\rho(t)}{1- i v'(t)} \, .
\end{align}
Plugging this expression back into \eqref{sum:i:less:j}, we get
\begin{align}
 \label{sum:i:less:j:N53}
 \sum_{i<j}^N \Li_2 \left(\e^{i\left(u_j^{(b)} - u_i^{(a)} +\Delta_{(b,a)}\right)}\right) =
 N^{2-\alpha} \int \rd t \Li_3 \left(\e^{i \Delta_{(b,a)}}\right) \frac{\rho(t)^2}{1 - i v'(t)} \, .
\end{align}
Having \eqref{sum:i:less:j:N53} in our disposal the rest of the computation
is exactly the same as in section \ref{Bethe bi-fundamentals N32}.
Following the same steps, we find
\begin{equation}
 \wt\cW^{\text{bi-fund}} = i g_+\left(\Delta_{(b,a)}\right) N^{2 - \alpha} \int  \rd t\, \frac{\rho(t)^2}{1-i v'(t)} \, .
\end{equation}

The contribution of an adjoint field, with chemical potential $\Delta_{(a,a)}$, is derived in exactly the same fashion:
\begin{equation}
 \wt\cW^{\text{adjoint}} = i g_+\left(\Delta_{(a,a)}\right) N^{2 - \alpha} \int  \rd t\, \frac{\rho(t)^2}{1-i v'(t)} \, .
\end{equation}
To have a nontrivial saddle-point, we need $\alpha=1/3$ which ensures that the Chern-Simons terms and the matter
contributions scale with the same power of $N$.

The contribution of (anti-)fundamental fields to the twisted superpotential is given by (see section \ref{Bethe fundamentals N32}),
\be
 \wt\cW^{\text{(anti-)fund}} = \frac{\left( \tilde n_a - n_a \right)}{4} N^{5/3}\int  \rd t\, \rho(t)\, \sign(t) \left[ i t + v(t) \right]^2\, .
\ee
Notice that, when the {\it total} number of fundamental and anti-fundamental fields in the quiver are equal,
this contribution vanishes.

\subsection[The index at large \texorpdfstring{$N$}{N}]{The index at large $\fakebold{N}$}
\label{entropy rules N53}

In this section we discuss the $N^{5/3}$ contributions to the logarithm of the partition function.
Using the same methods given in subsections \ref{entropy rules N32} and \ref{general:rules:N53:Bethe},
we arrive at the following large $N$ expressions for the gauge and matter contributions:
\bea
 \label{ch:2:N53:index:rules}
 \log Z^{\text{gauge}} & = -\frac{\pi^2}{3} N^{5/3} \int  \rd t\, \frac{\rho(t)^2}{1-i v'(t)} \, , \\
 \log Z^{\text{bi-fund}}_{\text{bulk}} & = - (\fn_{(b,a)}-1)\, g'_+ \left(\Delta_{(b,a)}\right) N^{5/3}
 \int  \rd t\, \frac{\rho(t)^2}{1-i v'(t)} \, , \\
 \log Z^{\text{adjoint}}_{\text{bulk}} & = - (\fn_{(a,a)}-1)\, g'_+ \left(\Delta_{(a,a)}\right) N^{5/3}
 \int  \rd t\, \frac{\rho(t)^2}{1-i v'(t)} \, .
\eea
The contribution of (anti-)fundamental fields to the index, at large $N$,
is subleading and they just contribute through the twisted superpotential.

Notice that the relation with the $S^3$ free energy discussed in section \ref{sec:freeenergy}
and the {\em index theorem} of section \ref{sec:index} also hold for this class of
quiver gauge theories.\footnote{The coefficient $2/3$ in front of $\mu$ in Eq.\,\eqref{virial theorem}
must be replaced by $3/5$.}

\section{Discussion and conclusions}\label{discussion}

In this chapter we have studied the large $N$ behavior of the topologically twisted index for ${\cal N}=2$ gauge theories in three dimensions. We have focused on theories  with a conjectured M-theory or massive type IIA dual and examined the corresponding field theory phases, where holography predicts a $N^{3/2}$ or $N^{5/3}$ scaling for the path integral, respectively. We correctly reproduced this scaling  for a class of ${\cal N}=2$ theories and we also uncovered some surprising relations with apparently
different physical quantities.

The first surprise comes from the identification of the {\it effective twisted superpotential} $\wt\cW$ with the {\it $S^3$ free energy}
$F_{S^3}$ of the same  ${\cal N}=2$ gauge theory. Recall that, in our approach, the BAEs and the twisted superpotential
are auxiliary quantities determining the position of the poles in the matrix model in the large $N$ limit.
$\wt\cW$ depends on the chemical potentials for the flavor symmetries, satisfying \eqref{2pi},
while $F_{S^3}$ depends on trial R-charges, which parameterize the curvature couplings of the theory on $S^3$.
Both quantities, $\wt\cW$ and $F_{S^3}$ are determined in terms of a matrix model (auxiliary in the case of $\wt\cW$).
The two matrix models, and the corresponding equations of motion are different for finite $N$
but quite remarkably become indistinguishable in the large $N$ limit. Also the conditions to
be imposed on the quiver for the existence of a $N^{3/2}$ or $N^{5/3}$ scaling are the same.
Although the structure of the long-range forces and the mechanism for their cancellation are different, they rule out
quivers with chiral bi-fundamentals in the M-theory phase and impose the same conditions on flavor symmetries. 

This identification leads to a relation of the twisted superpotential $\wt\cW(\Delta_I)$ with the volume functional of Sasaki-Einstein manifolds.
The exact R-symmetry of a superconformal  ${\cal N}=2$ gauge theory can be found by extremizing $F_{S^3}(\Delta_I)$
with respect to the trial R-charges $\Delta_I$ \cite{Jafferis:2011zi} but $F_{S^3}(\Delta_I)$ makes sense for arbitrary
$\Delta_I$. The functional $F_{S^3}(\Delta_I)$ has a well-defined geometrical meaning for theories with an AdS$_4\times Y_7$ dual,
where $Y_7$ is a Sasaki-Einstein manifold. The  value of $F_{S^3}$  upon extremization is related to the (square root of the)
volume of $Y_7$. More generally, at least for a class of quivers corresponding to ${\cal N}=3$ and toric cones $C(Y_7)$,
the value of $F_{S^3}(\Delta_I)$, as a function of the trial R-symmetry parameterized by $\Delta_I$,  has been matched with the
(square root of the) volume of a family of Sasakian deformation of $Y_7$, as a function of the Reeb vector. For toric theories,
the volume can be parameterized in terms of a set of charges $\Delta_I$, that encode how the R-symmetry varies with
the Reeb vector, and it has been conjectured in \cite{Amariti:2011uw,Amariti:2012tj,Lee:2014rca}
to be a homogeneous quartic function of the $\Delta_I$, in agreement with the homogeneity properties of $\wt\cW$ and $F_{S^3}$. 
  
A second intriguing relation comes from the index theorem \eqref{Z large N conjecture}.
The original reason for studying the large $N$ limit of the topologically twisted index comes from
the counting of AdS$_4$ black holes microstates. The entropy of dyonic black holes asymptotic to AdS$_4\times S^7$
was successfully compared with the large $N$ limit of the index in \cite{Benini:2015eyy,Benini:2016rke},
when extremized with respect to the chemical potential $\Delta_I$. We expect that a similar relation holds
for dyonic BPS black holes asymptotic to AdS$_4\times Y_7$, for a generic Sasaki-Einstein manifold.
Given the very small number of black holes known, this statement is difficult to check.

The main motivation of our analysis comes certainly from the attempt to extend the result of
\cite{Benini:2015eyy,Benini:2016rke} to a larger class of black holes. The difficulty of doing so
is mainly the exiguous number of existing black holes solutions with an M-theory lift.
Few numerical examples are known in Sasaki-Einstein compactifications \cite{Halmagyi:2013sla}, 
mostly having Betti multiplets as massless vectors. Some interesting examples involves chiral quivers
and are therefore outside the range of our technical abilities at the moment. It is curious that
apparently well-defined chiral quivers, which passed quite nontrivial checks \cite{Benini:2011cma}, have an ill-defined 
large $N$ limit both for the $S^3$ free energy and the topologically twisted index in the M-theory phase.
It would be quite interesting to know whether this is just a technical problem and another saddle-point
with $N^{3/2}$ scaling exists, or the models are really ruled out.

\chapter{Necklace quivers, dualities, and Sasaki-Einstein spaces}
\label{ch:3}

\ifpdf
    \graphicspath{{Chapter3/Figs/Raster/}{Chapter3/Figs/PDF/}{Chapter3/Figs/}}
\else
    \graphicspath{{Chapter3/Figs/Vector/}{Chapter3/Figs/}}
\fi

\section{Introduction}

The large $N$ limit of the topologically twisted index $Z_{S^2 \times S^1}$
for three-dimensional $\cN \geq 2$ Chern-Simons-matter theories has been considered in the previous chapter.
Here, we compute the \emph{topological free energy}
\be
 \mathfrak{F} = \log |Z_{S^2 \times S^1}| \, ,
\ee
for various examples in order to demonstrate the use of our formulae.

We begin by studying quiver gauge theories that can be realized on M2-branes probing two asymptotically locally Euclidean
(ALE) singularities \cite{Porrati:1996xi}. These include the ADHM \cite{Atiyah:1978ri} and the Kronheimer-Nakajima
\cite{kronheimer1990yang} quivers, as well as some of the necklace quiver theories considered in \cite{Imamura:2008nn}.
We show that the topological free energy of such theories can be written as that of the ABJM theory times a numerical factor,
depending on the orders of the ALE singularities and the Chern-Simons level of ABJM.
In particular, we match the topological free energy between theories being related to each other by dualities,
including mirror symmetry \cite{Intriligator:1996ex} and $\SL(2,\BZ)$ duality \cite{Aharony:1997ju, Gaiotto:2008ak, Assel:2014awa}.
We then move to discuss theories which are holographically dual to the M-theory backgrounds ${\rm AdS}_4 \times Y_7$,
where $Y_7$ is a homogeneous Sasaki-Einstein manifold.
In particular, we calculate the topological free energy for $N^{0,1,0}$ \cite{Gaiotto:2009tk, Imamura:2011uj, Cheon:2011th} with $\cN = 3$ and suspended pinch point (SPP) singularity \cite{Hanany:2008fj},
$V^{5,2}$ \cite{Martelli:2009ga, Jafferis:2009th}, and $Q^{1,1,1}$ \cite{Benini:2009qs,Cremonesi:2010ae} with $\cN = 2$ supersymmetry.

\section[Quivers with \texorpdfstring{$\CN=4$}{N=4} supersymmetry]{Quivers with $\fakebold{\cN=4}$ supersymmetry}
\label{sec:N4susy}

In this section, we consider two quiver gauge theories with $\CN=4$ supersymmetry. As pointed out in \cite{Porrati:1996xi}, each of these theories can be realized in the worldvolume of M2-branes probing $\BC^2/\BZ_{n_1} \times \BC^2/\BZ_{n_2}$, for some positive integers $n_1$ and $n_2$.  We show below that the topological free energy of such theories can be written as $\sqrt{n_1 n_2/k}$ times that of the ABJM theory with Chern-Simons levels $(+k,-k)$.
We also match the index of a pair of theories which are mirror dual \cite{Intriligator:1996ex} to each other.
This serves as a check of the validity of our results.

\subsection{The ADHM quiver}
We consider $\U(N)$ gauge theory with one adjoint and $r$ fundamental hypermultiplets, whose $\CN=4$ quiver is given by
\be
\begin{aligned} \label{ADHMN4quiv}
\begin{tikzpicture}[font=\footnotesize, scale=0.9]
\begin{scope}[auto,%
  every node/.style={draw, minimum size=0.5cm}, node distance=2cm];
\node[circle]  (UN)  at (0.3,1.7) {$N$};
\node[rectangle, right=of UN] (Ur) {$r$};
\end{scope}
\draw(-0,2) arc (30:338:0.75cm);
\draw[black,solid,line width=0.1mm]  (UN) to (Ur) ; 
\end{tikzpicture}
\end{aligned} \ee
where the circular node denotes the $\U(N)$ gauge group; the square node denotes the $\SU(r)$ flavor symmetry; the loop around the circular node denotes the adjoint hypermultiplet; and the line between $N$ and $r$ denotes the fundamental hypermultiplet.  The vacuum equations of the Higgs branch of the theory were used in the construction of the instanton solutions by Atiyah, Drinfeld, Hitchin and Manin \cite{Atiyah:1978ri}.  This quiver gauge theory hence acquires the name ``ADHM quiver''.

In $\CN=2$ notation, this theory contains three adjoint chiral fields: $\phi_{1}, \, \phi_2, \, \phi_3$, where $\phi_{1,2}$ come from the $\CN=4$ adjoint hypermultiplet and $\phi_3$ comes from the $\CN=4$ vector multiplet, and fundamental chiral fields $Q^i_a$, $\tQ^a_i$ with $a = 1, \ldots, N$ and $i=1, \ldots, r$.  The superpotential is
\be
\begin{aligned}
W= \tQ^i_a (\phi_3)^a_{~b} Q^b_i + (\phi_3)^a_{~b}[ \phi_1, \phi_2]^b_{~a} \, .
\end{aligned} \ee
The $\CN=2$ quiver diagram is depicted below.
\be
\begin{aligned}
\begin{tikzpicture}[font=\footnotesize, scale=0.9]
\begin{scope}[auto,%
  every node/.style={draw, minimum size=0.5cm}, node distance=2cm];
\node[circle]  (UN)  at (0.3,1.7) {$N$};
\node[rectangle, right=of UN] (Ur) {$r$};
\end{scope}
\draw[decoration={markings, mark=at position 0.45 with {\arrow[scale=1.5]{>}}, mark=at position 0.5 with {\arrow[scale=1.5]{>}}, mark=at position 0.55 with {\arrow[scale=1.5]{>}}}, postaction={decorate}, shorten >=0.7pt] (-0,2) arc (30:340:0.75cm);
\draw[draw=black,solid,line width=0.2mm,->]  (UN) to[bend right=30] node[midway,below] {$Q$}node[midway,above] {}  (Ur) ; 
\draw[draw=black,solid,line width=0.2mm,<-]  (UN) to[bend left=30] node[midway,above] {$\tQ$} node[midway,above] {} (Ur) ;    
\node at (-2.2,1.7) {$\phi_{1,2,3}$};
\end{tikzpicture}
\end{aligned} \ee
The Higgs branch of this gauge theory describes the moduli space of $N$ $\SU(r)$ instantons on $\BC^2$ \cite{Atiyah:1978ri} and the Coulomb branch is isomorphic to the space $\Sym^N (\BC^2/\BZ_r)$ \cite{deBoer:1996mp}.  This theory can be realized on the worldvolume of $N$ M2-branes probing $\BC^2 \times \BC^2/\BZ_r$ singularity \cite{Porrati:1996xi}.

\subsubsection{A solution to the system of BAEs} \label{sec:solnADHM2pi}
Let us denote, respectively, by $\Delta$, $\tilde{\Delta}$, $\Delta_{\phi_{1,2,3}}$ the chemical potentials associated to the flavor symmetries of $Q$, $\tQ$, $\phi_{1,2,3}$, and by $\fn$, $\tilde{\fn}$, ${\fn}_{\phi_{1,2,3}}$ the corresponding fluxes associated with their flavor symmetries.  We denote also by  $\Delta_m$ the chemical potential associated with the topological charge of the gauge group $\U(N)$.

Given the rules of section \ref{large N twisted superpotential rules}, the twisted superpotential $\wt\cW$ for this model can be written as
\bea
 \label{ADHM Bethe}
 \frac{\wt\cW}{i N^{3/2}} & = \left( \sum_{i=1}^3 g_+(\Delta_{\phi_i})  \right) \int \rd t \, \rho(t)^2
 -\frac{r}{2} \left[(\Delta - \pi) + (\tilde{\Delta} - \pi) \right] \int \rd t \, |t| \, \rho(t) \\
 & + \Delta_m \int \rd t \, t\, \rho(t) -\mu \left(\int \rd t \, \rho(t)-1 \right) \, .
\eea
Taking the variational derivative of $\wt\cW$ with respect to $\rho(t)$, we obtain the BAE
\be
\begin{aligned} \label{BAEADHM1}
0 = 2 \rho(t) \sum_{i=1}^3 g_+(\Delta_{\phi_i}) - \frac{r}{2} |t| \left[ (\Delta - \pi) + (\tilde{\Delta} - \pi) \right] + \Delta_m t - \mu \, .
\end{aligned} \ee

We first look for a solution satisfying
\be \label{marginality}
\sum_{I \in W} \Delta_I = 2 \pi \, ,
\ee
where the sum is taken over all the fields in each monomial term $W$ in the superpotential.
We call this the \emph{marginality condition} on the superpotential.
This yields
\be
\begin{aligned}
{\Delta}+\tilde{\Delta} +{\Delta}_{\phi_3} = 2\pi\, , \qquad {\Delta}_{\phi_1}+{\Delta}_{\phi_2} +{\Delta}_{\phi_3} = 2\pi \, ,
\end{aligned} \ee
while the magnetic fluxes should satisfy
\be
{\fn}+\tilde{\fn} +{\fn}_{\phi_3} = 2 \, , \qquad {\fn}_{\phi_1}+{\fn}_{\phi_2} +{\fn}_{\phi_3} = 2 \, . 
\ee

For later convenience, let us normalize the chemical potential associated with the topological charge as follows:
\be
\begin{aligned}
\chi = \frac{2}{r} \Delta_m \, . \label{defchiADHM}
\end{aligned} \ee
Solving \eref{BAEADHM1}, we get
\be
\begin{aligned}
\rho(t) = \frac{2 \mu - r \Delta_{\phi_3}  \left| t\right| - r \chi t  }{2 \prod_{i=1}^3 \Delta_{\phi_i}} \, .
\end{aligned} \ee
The solution is supported on the interval $[t_-,t_+]$ with $t_- < 0 < t_+$, where $t_\pm$ can be determined from $\rho(t_\pm)=0$:
\be
\begin{aligned}
 t_{\pm}= \pm \frac{2\mu }{(\Delta_{\phi_3} \pm \chi) r} \, .
\end{aligned} \ee
The normalization $\int_{t_-}^{t_+} \rd t \, \rho(t) =1$ fixes
\be \label{ADHMmu}
\mu = \sqrt{\frac{r}{2} \Delta_{\phi_1} \Delta_{\phi_2}(\Delta_{\phi_3}+\chi)(\Delta_{\phi_3}-\chi)} \, .
\ee

\paragraph*{The solution in the other ranges.}
Let us consider
\be
\begin{aligned}
{\Delta}+\tilde{\Delta} +{\Delta}_{\phi_3} = 2  \pi \ell \, , \qquad {\Delta}_{\phi_1}+{\Delta}_{\phi_2} +{\Delta}_{\phi_3} = 2 \pi \ell \, , \qquad \text{where $\ell \in \BZ_{\geq 0}$} \, .
\end{aligned} \ee
For $\ell = 0$ and $\ell=3$, we have $\Delta=\tilde{\Delta}={\Delta}_{\phi_{1,2,3}}=0$ or $\Delta=\tilde{\Delta}={\Delta}_{\phi_{1,2,3}}=2\pi$, respectively. These are singular solutions.
For $\ell =2$, the solution can be mapped to the previous one (\ie{} $\ell = 1$) by a discrete symmetry
\be
\begin{aligned}
\Delta_I \rightarrow 2\pi - \Delta_I \,, ~\quad
\mu \rightarrow  - \mu\, , ~\quad \Delta_m \rightarrow - \Delta_m \, ,
\end{aligned} \ee
where the index $I$ labels matter fields in the theory.
From now on, we shall consider only the solution satisfying the marginality condition \eqref{marginality}.

\subsubsection[The index at large \texorpdfstring{$N$}{N}]{The index at large $\fakebold{N}$}

The topological free energy of the ADHM quiver can be derived from section \eref{large N index rules} as
\be
\begin{aligned} \label{index ADHM}
\frac{\mathfrak{F}_{\text{ADHM}}}{N^{3/2}} &= - \left[ \frac{\pi^2}{3} + \sum_{i=1}^3 (\fn_{\phi_i} -1) g'_+(\Delta_{\phi_i})  \right] \int \rd t \, \rho(t)^2 - \frac{r}{2} \ft \int \rd t \, t \, \rho(t) \\
& + \frac{r}{2} \left[ \left(\fn -1 \right) + \left(\tilde{\fn} -1 \right) \right]  \int \rd t \, |t|\, \rho(t) \, ,
\end{aligned} \ee
where $\ft$ is the magnetic flux conjugate to the variable $\chi$ defined in \eref{defchiADHM}.
Plugging the above solution back into \eqref{index ADHM}, we find that
\be
\begin{aligned}  \label{freeADHM}
\mathfrak{F}_{\text{ADHM}} = \sqrt{\frac{r}{k}}\, \mathfrak{F}_{\text{ABJM}_k} \, ,
\end{aligned} \ee
where, \cf\,\eqref{index:ABJM:ch:2},
\be
 \mathfrak{F}_{\text{ABJM}_k} = - \frac{k^{1/2}N^{3/2}}{3}
 \sqrt{2 \Delta_{A_1} \Delta_{A_2} \Delta_{B_1} \Delta_{B_2}}
 \left( \frac{\fn_{A_1}}{\Delta_{A_1}} + \frac{\fn_{A_2}}{\Delta_{A_2}}
 + \frac{\fn_{B_1}}{\Delta_{B_1}} + \frac{\fn_{B_2}}{\Delta_{B_2}} \right) \, .
\ee
The map of the parameters is as follows:
\be
\begin{aligned} \label{paramADHM}
\Delta_{A_1} & = \frac{1}{2} (\Delta_{\phi_3}-\chi) \, , \qquad \Delta_{A_2} =  \frac{1}{2}(\Delta_{\phi_3}+\chi) \, , \qquad \Delta_{B_1}= \Delta_{\phi_1} \, , \qquad \Delta_{B_2} =\Delta_{\phi_2} \, , \nn \\
\fn_{A_1} & = \frac{1}{2} (\fn_{\phi_3}-\ft)\, , \qquad \fn_{A_2} =  \frac{1}{2}(\fn_{\phi_3}+\ft) \, , \qquad \fn_{B_1}= \fn_{\phi_1} \, , \qquad \fn_{B_2} =\fn_{\phi_2} \, .
\end{aligned} \ee

The factor $\sqrt{r/k}$ in \eref{freeADHM} is the ratio between the orbifold order of $\Sym^N(\BC^2 \times \BC^2/\BZ_r)$ and that of $\Sym^N(\BC^2/\BZ_k)$; the former is the geometric branch of the ADHM theory and the latter is that of the ABJM theory with Chern-Simons levels $(+k, -k)$.

\subsection[The \texorpdfstring{$A_{n-1}$}{A[n-1]} Kronheimer-Nakajima quiver]{The $\fakebold{A_{n-1}}$ Kronheimer-Nakajima quiver}
We consider a necklace quiver with $\U(N)^n$ gauge group with a bi-fundamental hypermultiplet between the adjacent gauge groups and with $r$ flavors of fundamental hypermultiplets under the $n$-th gauge group.  The $\CN=4$ quiver is depicted below.
\be
\begin{aligned} \label{KNN4quiv}
\begin{tikzpicture}[baseline]
\def \n {6}
\def \radius {1.5cm}
\def \margin {16} 
\foreach \s in {1,...,5}
{
  \node[draw, circle] at ({360/\n * (\s - 2)}:\radius) {{\footnotesize $N$}};
  \draw[-, >=latex] ({360/\n * (\s - 3)+\margin}:\radius) 
    arc ({360/\n * (\s - 3)+\margin}:{360/\n * (\s-2)-\margin}:\radius);
}
\node[draw, circle] at ({360/3 * (3 - 1)}:\radius) {{\footnotesize $N$}};
\draw[dashed, >=latex] ({360/6 * (5 -2)+\margin}:\radius) 
    arc ({360/6 * (5 -2)+\margin}:{360/6 * (5-1)-\margin}:\radius);
\node[draw, rectangle] at (3,0) {{\footnotesize $r$}};
\draw[-, >=latex] (1+0.9,0) to (3-0.2,0);
\node[draw=none] at (0,-2.3) {{\footnotesize ($n$ circular nodes)}};
\end{tikzpicture}
\end{aligned} \ee
As proposed by Kronheimer and Nakajima \cite{kronheimer1990yang}, the vacuum equations for the Higgs branch of this theory describes the hyperK\"ahler quotient of the moduli space of $\SU(r)$ instantons on $\BC^2/\BZ_n$ with $\SU(r)$ left unbroken by the monodromy at infinity. We shall henceforth refer to this quiver as the ``Kronheimer-Nakajima quiver''.

The corresponding $\CN=2$ quiver diagram is
\be
\begin{aligned} \label{KNN2quiv}
\begin{tikzpicture}[baseline, font=\scriptsize]
\def \n {6}
\def \radius {1.5cm}
\def \margin {14} 
\foreach \s in {1,...,6}
{
  \node[draw, circle] at ({360/\n * (\s)}:\radius)  (\s) {$N$};
}
\foreach \s in {1,...,5}
{    
    \draw[black,-> ] (\s) edge [out={30+\s*60},in={-30+\s*60},loop,looseness=8] (\s);
} 
\draw[black,-> ] (6) edge [out={30+6*60-60},in={-30+6*60-40},loop,looseness=10] (6);
\foreach \s in {1,...,5}
{  
  \draw[decoration={markings, mark=at position 0.5 with {\arrow[scale=1.5]{>}}}, postaction={decorate}, shorten >=0.7pt]  ({360/\n * (\s - 3)+\margin}:\radius) to[bend left=60] node[midway,above] {}node[midway,below] {}  ({360/\n * (\s - 2)-\margin}:\radius) ; 
    \draw[decoration={markings, mark=at position 0.5 with {\arrow[scale=1.5]{<}}}, postaction={decorate}, shorten >=0.7pt]  ({360/\n * (\s - 3)+\margin}:\radius) to[bend right=60] node[midway,above] {}node[midway,below] {}  ({360/\n * (\s - 2)-\margin}:\radius) ;
}   
\draw[dashed, >=latex] ({360/6 * (5 -2)+\margin}:\radius) 
    arc ({360/6 * (5 -2)+\margin}:{360/6 * (5-1)-\margin}:\radius);
\node[draw, rectangle] at (3.5,0) {{\footnotesize $r$}};
\draw[decoration={markings, mark=at position 0.5 with {\arrow[scale=1.5]{>}}}, postaction={decorate}, shorten >=0.7pt] (1.5+0.35,0)  to[bend left=30] node[midway,above] {}node[midway,below] {}   (3.5-0.2,0);
\draw[decoration={markings, mark=at position 0.5 with {\arrow[scale=1.5]{<}}}, postaction={decorate}, shorten >=0.7pt] (1.5+0.35,0)  to[bend right=30] node[midway,above] {}node[midway,below] {}   (3.5-0.2,0);
\node[draw=none] at (0,-2.7) {{\footnotesize ($n$ circular nodes)}};
\end{tikzpicture}
\end{aligned} \ee
Let $Q_\alpha$ (with $\alpha=1,\ldots, n$) be the bi-fundamental field that goes from node $\alpha$ to node $\alpha+1$; $\tilde{Q}_\alpha$ be the bi-fundamental field that goes from node $\alpha+1$ to node $\alpha$; and $\phi_\alpha$ be the adjoint field under node $\alpha$.   Let us also denote by $q^i_{a}$ and $\tilde{q}^a_{i}$ the fundamental and anti-fundamental chiral multiplets under the $n$-th gauge group (with $a=1,\ldots, N$ and $i=1, \ldots, r$).  The superpotential is
\be
\begin{aligned}
W= \sum_{\alpha=1}^n  \Tr\left(Q_\alpha \phi_{\alpha+1} \tilde{Q}_\alpha- \tilde{Q}_{\alpha} \phi_{\alpha} Q_{\alpha}\right)+ \tilde{q}^a_i \, (\phi_n)^b_{~a}  \, {q}^i_b ~,
\end{aligned} \ee
where we identify $\phi_{n+1}=\phi_1$. From now on, the
index $\alpha$ labeling the nodes is taken modulo $n$
for any necklace quiver with $n$ nodes.

The Higgs branch of this gauge theory describes the moduli space of $N$ $\SU(r)$ instantons
on $\BC^2/\BZ_n$ such that the monodromy at infinity preserves $\SU(r)$
symmetry \cite{kronheimer1990yang}, and the Coulomb branch describes the
moduli space of $N$ $\SU(n)$ instantons on $\BC^2/\BZ_r$ such that the
monodromy at infinity preserves $\SU(n)$ symmetry \cite{deBoer:1996mp, Porrati:1996xi, Witten:2009xu, Mekareeya:2015bla}.  It can be indeed realized on the worldvolume of $N$ M2-branes probing $\BC^2/\BZ_n \times \BC^2/\BZ_r$ singularity \cite{Porrati:1996xi}.  Note also that 3D mirror symmetry exchanges the Kronheimer-Nakajima quiver \eref{KNN4quiv} with $r=1$ and $n=2$ and the ADHM quiver \eref{ADHMN4quiv} with $r=2$. 

\subsubsection{A solution to the system of BAEs} \label{sec:solnKN}
Let us denote respectively by $\Delta_{Q_\alpha}$, $\Delta_{\tQ_\alpha}$, $\Delta_{\phi_{\alpha}}$, $\Delta_{q}$, $\Delta_{\tilde{q}}$ the chemical potentials associated to the flavor symmetries of $Q_{\alpha}$, $\tQ_{\alpha}$, $\phi_{\alpha}$, $q$ and $\tilde{q}$, and by $\fn_{Q_\alpha}$, ${\fn}_{\tQ_\alpha}$, ${\fn}_{\phi_{\alpha}}$, $\fn_q$, $\fn_{\tilde{q}}$ the corresponding fluxes associated with their flavor symmetries.  We also denote by $\Delta_m^{(\alpha)}$ the chemical potential associated with the topological charge for gauge group $\alpha$ and by $\ft^{(\alpha)}$ the associated magnetic flux.

Given the rules of section \ref{large N twisted superpotential rules}, the twisted superpotential $\wt\cW$ for this model is given by
\be
\begin{aligned}
\frac{\wt\cW}{i N^{3/2}} & = \int \rd t \, \rho(t)^2 \sum_{\alpha=1}^n \left[ g_+ (\delta v^\alpha(t)+ {\Delta}_{\tQ_\alpha}) - g_- (\delta v^\alpha(t) - \Delta_{Q_\alpha})  + g_+(\Delta_{\phi_\alpha}) \right] \\
& - \frac{r}{2} \left[ \left(\Delta_q - \pi \right) + \left({\Delta}_{\tilde{q}} - \pi \right) \right] \int \rd t \, |t| \, \rho(t)
+ \left( \sum_{\alpha =1}^n \Delta_m^{(\alpha)}  \right) \int \rd t \, t \, \rho(t) \\
& -\mu \left(\int \rd t \, \rho(t)-1 \right) \, .
\end{aligned} \ee
where $\delta v^{\alpha} = v^{\alpha+1} - v^\alpha$ and we identify $\delta v^{n+1} = \delta v^{1}$. Taking the variational derivatives of $\wt\cW$ with respect to $\rho(t)$ and $\delta v^\alpha(t)$, we obtain the BAEs
\be
\begin{aligned} \label{BAEKN}
0 &= 2 \rho(t) \sum_{\alpha=1}^n \left[ g_+ (\delta v^\alpha(t)+ {\Delta}_{\tQ_\alpha}) - g_- (\delta v^\alpha(t) - \Delta_{Q_\alpha}) + g_+(\Delta_{\phi_\alpha})  \right]  \\
& - \frac{r}{2} |t| \left[ (\Delta_q - \pi) + ({\Delta}_{\tilde{q}} - \pi) \right]
+ \left( \sum_{\alpha =1}^n \Delta_m^{(\alpha)}  \right) t - \mu \, ,\\
0&= \rho(t) \Big[ g'_+( \delta v^\alpha(t)+ \Delta_{\tQ \alpha})
- g'_- ( \delta v^\alpha(t) - \Delta_{Q \alpha}) \\
& + g'_- ( \delta v^{\alpha-1}(t) - \Delta_{Q_{\alpha-1}})
- g'_+(\delta v^{\alpha-1}(t) + \Delta_{\tQ_{\alpha-1}}) \Big] \, , \qquad \alpha =1, \ldots, n \, .
\end{aligned} \ee

The superpotential imposes the following constraints on the chemical potentials of the various fields:
\be
\begin{aligned}
{\Delta}_q+{\Delta}_{\tilde{q}} +{\Delta}_{\phi_n} = 2\pi \, , \quad {\Delta}_{Q_\alpha}+\Delta_{\phi_{\alpha+1}}+{\Delta}_{\tQ_\alpha} = 2\pi \, , \quad {\Delta}_{\tQ_\alpha}+\Delta_{\phi_{\alpha}}+{\Delta}_{Q_\alpha} = 2\pi \, . 
\end{aligned} \ee

For notational convenience, we define
\be
\begin{aligned}\label{paraKN}
F_1= \sum_\alpha {\Delta}_{\tQ_\alpha}\, , \qquad F_3 = \Delta_{\phi_{n}} \, , \qquad {\Delta_m} = \frac{2}{r} \sum_{\alpha} \Delta_m^{(\alpha)}\, ,
\end{aligned} \ee
and
\be
\begin{aligned}
F_2 = 2\pi -F_1-F_3 \, .
\end{aligned} \ee
Solving the system of BAEs \eref{BAEKN}, we find that 
\be
\begin{aligned}
\rho(t) = \frac{2 \mu - r F_3  \left| t\right| -  r \Delta_m t  }{2 \prod_{i=1}^3 F_i} ~, \qquad \delta v^\alpha =  \frac{1}{n} F_1-\Delta_{\tQ_\alpha}~.
\end{aligned} \ee
The support $[t_-, t_+]$ of $\rho(t)$ is determined by $\rho(t_\pm)=0$. We get
\be
\begin{aligned}
 t_\pm = \pm \frac{2\mu }{(F_3 \pm \Delta_m)r}\, .
\end{aligned} \ee
The normalization $\int_{t_-}^{t_+} \rd t \, \rho(t) =1$ fixes
\be
\mu = \sqrt{\frac{3 r}{2} F_1 F_2 (F_3+\Delta_m) (F_3-\Delta_m)} \, . \label{muKN}
\ee

\subsubsection[The index at large \texorpdfstring{$N$}{N}]{The index at large $\fakebold{N}$}

From the rules given in section \ref{large N index rules}, the topological free energy of this quiver is given by
\be
\begin{aligned} \label{SKN}
\frac{\mathfrak{F}_{\text{KN}}}{N^{3/2}} & = - \frac{n \pi^2}{3} \int \rd t \, \rho(t)^2 - \left(\sum_{\alpha=1}^n \ft^{(\alpha)} \right)  \int \rd t \, t \, \rho(t)
+ \frac{r}{2} \left[ (\fn_q -1) + (\fn_{\tilde{q}} -1) \right]  \int \rd t \, |t|\, \rho(t) \\
& -  \int \rd t \, \rho(t)^2\, \sum_{\alpha=1}^n \bigg[(\fn_{\tQ_\alpha} -1) g'_+(\delta v^\alpha(t)+ \Delta_{\tQ_\alpha})
+ (\fn_{Q_\alpha} -1) g'_-(\delta v^\alpha(t) - \Delta_{Q_\alpha}) \bigg] \\
& - \sum_{\alpha=1}^n (\fn_{\phi_\alpha} -1) g'_+(\Delta_{\phi_\alpha}) \int \rd t \, \rho(t)^2 \, .
\end{aligned} \ee
Plugging the above solution back into \eref{SKN}, we find that the topological free energy depends only on the parameters $F_1$, $F_2$, $F_3$ given by \eref{paraKN} and their corresponding conjugate charges
\be
\begin{aligned}
\fn_1= \sum_\alpha {\fn}_{\tQ_\alpha} \, , \qquad \fn_3 =\fn_{\phi_{n}} \, , \qquad {\ft} = \frac{2}{r} \sum_{\alpha} \ft^{(\alpha)}  \, .
\end{aligned} \ee
Explicitly, we obtain
\be
\begin{aligned} \label{freeKN}
\mathfrak{F}_{\text{KN}} = \sqrt{\frac{nr}{k}}\, \mathfrak{F}_{\text{ABJM}_k}\, ,
\end{aligned} \ee
with the following map of the parameters
\be
\begin{aligned}
\Delta_{A_1} &= \frac{1}{2} (F_3-\Delta_m) \, , \qquad \Delta_{A_2} =  \frac{1}{2}(F_3+\Delta_m) \, , \qquad \Delta_{B_1}= F_1 \, , \qquad \Delta_{B_2} =F_2 \, , \\
\fn_{A_1} & = \frac{1}{2} (\fn_{3}-\ft) \, , \qquad \fn_{A_2} =  \frac{1}{2}(\fn_3+\ft) \, , \qquad \fn_{B_1}= \fn_{1} \, , \qquad \fn_{B_2} =\fn_{2} \, .
\end{aligned} \ee
Notice that, this is completely analogous to that of the ADHM quiver presented in \eref{paramADHM}.

The factor $\sqrt{nr/k}$ in \eref{freeADHM} is the ratio between the product of the orbifold orders in $\Sym^N(\BC^2/\BZ_n \times \BC^2/\BZ_r)$ and that of $\Sym^N(\BC^2/\BZ_k)$, where the former is the geometric branch of the Kronheimer-Nakajima theory and the latter is that of the ABJM theory with Chern-Simons levels $(+k, -k)$.

\paragraph*{Mirror symmetry.} The Kronheimer-Nakajima quiver \eref{KNN4quiv} with $r=1$ and $n=2$
is mirror dual \cite{Intriligator:1996ex} to the ADHM quiver \eref{ADHMN4quiv} with $r=2$.
From \eref{freeADHM} and \eref{freeKN}, the topological free energy of the two theories are indeed equal:
\be
\begin{aligned}
\mathfrak{F}_{\text{KN}} \Big |_{r=1, \, n=2} = \mathfrak{F}_{\text{ADHM}} \Big |_{r=2} \, .
\end{aligned} \ee

\section[Quivers with \texorpdfstring{$\CN=3$}{N=3} supersymmetry]{Quivers with $\fakebold{\cN=3}$ supersymmetry}
\label{sec:N3susy}

A crucial difference between the theories considered in this section and those with $\CN=4$ supersymmetry is that the solution to the BAEs of the former are divided into several regions and the final result of the topological free energy comes from the sum of the contributions of each region.
Such a feature of the solution was already present in the ABJM theory and was discussed extensively in \cite{Benini:2015eyy}. In subsection \ref{altCSN3}, we deal with the necklace quiver with alternating Chern-Simons levels and present the twisted superpotential, the BAEs and the procedure to solve them in detail.
The solutions for the other models in the following subsections can be derived in a similar fashion.

In subsections \ref{sec:necklacealtCS} and \ref{sec:necklacetwoCS}, we focus on theories whose geometric branch is a symmetric power of a product of two ALE singularities \cite{Imamura:2008nn, Cremonesi:2016nbo}.
Similarly to the preceding section, the topological free energy of such theories can be written as a numerical factor times the topological free energy of the ABJM theory, where the numerical factor equals to the square root of the ratio between the product of the orders of such singularities and the level of the ABJM theory.
Moreover, in a certain special case where the quiver is $\SL(2, \BZ)$ dual to a quiver with $\CN=4$ supersymmetry \cite{Gaiotto:2008ak, Assel:2012cj, Assel:2014awa, Cremonesi:2016nbo}, we match the topological free energy of two theories.

\subsection[The affine \texorpdfstring{$A_{2m-1}$}{A[2m-1]} quiver with alternating CS levels]{The affine $\fakebold{A_{2m-1}}$ quiver with alternating CS levels}
\label{sec:necklacealtCS}

We are interested in the necklace quiver with $n=2m$ nodes, each with $\U(N)$ gauge group, and alternating Chern-Simons levels:
\be
\begin{aligned} \label{quivaltCS}
k_\alpha= \begin{cases} +k & \text{if $\alpha$ is odd} \\
-k & \text{if $\alpha$ is even} \end{cases}
\end{aligned} \ee
The $\CN=2$ quiver diagram is depicted below.
\be
\begin{aligned}
\begin{tikzpicture}[baseline,font=\tiny, every loop/.style={min distance=10mm,looseness=10}, every node/.style={minimum size=0.9cm}]
\def \n {6}
\def \radius {1.5cm}
\def \margin {18} 
\foreach \s in {1}
{
    \node[draw, circle] at ({2*360/\n * (\s)}:\radius)  (\s) {$N_{+k}$};
    \draw[black,-> ] (\s) edge [out={30+2*\s*60},in={-30+2*\s*60},loop,looseness=5] (\s);
} 
\foreach \s in {1}
{    
    \node[draw, circle]  at ({2*360/\n*\s-360/\n}:\radius) (\s)  {$N_{-k}$};
    \draw[black,-> ] (\s) edge [out={-30+2*\s*60},in={-90+2*\s*60},loop,looseness=5] (\s);
} 
\foreach \s in {2,3}
{
    \node[draw, circle] at ({2*360/\n * (\s)}:\radius)  (\s) {$N_{+k}$};
    \draw[black,-> ] (\s) edge [out={30+2*\s*60},in={-30+2*\s*60},loop,looseness=5] (\s);
} 
\foreach \s in {2,3}
{    
    \node[draw, circle]  at ({2*360/\n*\s-360/\n}:\radius) (\s)  {$N_{-k}$};
    \draw[black,-> ] (\s) edge [out={-30+2*\s*60},in={-90+2*\s*60},loop,looseness=5] (\s);
} 
\foreach \s in {1,...,5}
{  
  \draw[decoration={markings, mark=at position 0.5 with {\arrow[scale=1.5]{>}}}, postaction={decorate}, shorten >=0.7pt]  ({360/\n * (\s - 3)+\margin}:\radius) to[bend left=60] node[midway,above] {}node[midway,below] {}  ({360/\n * (\s - 2)-\margin}:\radius) ; 
    \draw[decoration={markings, mark=at position 0.5 with {\arrow[scale=1.5]{<}}}, postaction={decorate}, shorten >=0.7pt]  ({360/\n * (\s - 3)+\margin}:\radius) to[bend right=60] node[midway,above] {}node[midway,below] {}  ({360/\n * (\s - 2)-\margin}:\radius) ;
}   
\draw[dashed, >=latex] ({360/6 * (5 -2)+\margin}:\radius) 
    arc ({360/6 * (5 -2)+\margin}:{360/6 * (5-1)-\margin}:\radius);
\end{tikzpicture}
\end{aligned} \ee
Let $Q_\alpha$ be the bi-fundamental field that goes from node $\alpha$ to node $\alpha+1$; $\tilde{Q}_\alpha$ be the bi-fundamental field that goes from node $\alpha+1$ to node $\alpha$; and $\phi_\alpha$ be the adjoint field under node $\alpha$.  The superpotential can be written as
\be
\begin{aligned}
W=\sum_{\alpha=1}^n  \Tr\left(Q_\alpha \phi_{\alpha+1} \tilde{Q}_\alpha- \tilde{Q}_{\alpha} \phi_{\alpha} Q_{\alpha}\right) +\frac{k}{2} \sum_{\alpha=1}^m \Tr \left(\phi_{2\alpha-1}^2- \phi_{2\alpha}^2\right) \, .
\end{aligned} \ee
After integrating out the massive adjoint fields, we have the superpotential
\be
\begin{aligned}
W= \frac{1}{k} \sum_{\alpha=1}^n (-1)^\alpha \Tr \left( Q_\alpha Q_{\alpha+1} \tQ_{\alpha+1} \tQ_\alpha \right) \, , \label{superpot}
\end{aligned} \ee
where we identify
\be
\begin{aligned}
Q_{n+1} = Q_{n} \, , \qquad \tQ_{n+1} = \tQ_{n} \, .
\end{aligned} \ee

\subsubsection{A solution to the system of BAEs} \label{altCSN3}
Let us denote respectively by $\Delta_{\alpha}$, $\tilde{\Delta}_{\alpha}$ the chemical potentials associated to the flavor symmetries of $Q_\alpha$ and $\tQ_\alpha$, and by $n_{\alpha}$, $\tilde{n}_{\alpha}$ the fluxes associated with the flavor symmetries of $Q_\alpha$ and $\tQ_\alpha$.

From the rules given in section \ref{large N twisted superpotential rules}, the twisted superpotential $\wt\cW$ can be written as
\be
\begin{aligned} \label{potaltCS}
\frac{\wt\cW}{i N^{3/2}} &= k \int \rd t \, t\, \rho(t) \sum_{\alpha=1}^m \delta v^{2\alpha-1}(t)
+ \int \rd t \, \rho(t)^2  \sum_{\alpha=1}^n \left[ g_+ \big(\delta v^\alpha(t)+ \tilde{\Delta}_{\alpha} \big) - g_- \big(\delta v^\alpha(t) - \Delta_{\alpha} \big) \right]  \\
& - \frac{i}{N^{1/2}} \int \rd t \, \rho(t) \sum_{\alpha=1}^m \bigg[ \Li_2 \left(\e^{i \big(\delta v^{2\alpha-1}(t)
+ \tilde{\Delta}_{2\alpha-1}\big)}\right) - \Li_2 \left(\e^{i \big(\delta v^{2\alpha-1}(t)- {\Delta}_{2\alpha-1}\big)}\right) \\
& + \Li_2 \left(\e^{i \big(\delta v^{2\alpha}(t)+ \tilde{\Delta}_{2\alpha}\big)}\right)
- \Li_2 \left(\e^{i \big(\delta v^{2\alpha}(t)- {\Delta}_{2\alpha}\big)}\right) \bigg] -\mu \left(\int \rd t \, \rho(t)-1 \right) \, ,
\end{aligned} \ee
where $\delta v^{\alpha}(t) = v^{\alpha+1}(t) - v^{\alpha}(t)$ and hence,
\be
\begin{aligned}
\sum_{\alpha=1}^n \delta v^{\alpha}(t) =0\, .
\end{aligned} \ee
Without loss of generality, we set the chemical potentials associated with topological symmetries to zero. The subleading terms in \eref{potaltCS} can be obtained by considering the node $2 \alpha-1$ (with $\alpha=1, \ldots, m$), where the fields with chemical potentials $\tilde{\Delta}_{2\alpha-1}$, $\Delta_{2\alpha-2}$ are incoming to that node and those with chemical potentials ${\Delta}_{2\alpha-1}$, $\tilde{\Delta}_{2\alpha-2}$ are outgoing of that node.
This explains the signs of such terms in \eref{potaltCS}.
These terms can be neglected when we compute the value of the twisted superpotential, since $\Li_2$ does not have divergences;
however, they play an important role when we deal with the derivatives of $\wt\cW$ because $\Li_1(e^{iu})$ diverges as $u \to 0$.

Taking the variational derivatives of $\wt\cW$ with respect to $\rho(t)$ and setting it to zero, we obtain
\be
\begin{aligned} \label{eq0}
0 &=k t \sum_{\alpha=1}^m \delta v^{2\alpha-1}(t) + 2 \rho(t) \sum_{\alpha=1}^n \left[ g_+ \big(\delta v^\alpha(t)+ \tilde{\Delta}_{\alpha}\big) - g_- \big(\delta v^\alpha(t) - \Delta_{\alpha}\big)  \right] - \mu \, .
\end{aligned} \ee
When $\delta v^\alpha \not\approx - \tilde{\Delta}_\alpha$ and $\delta v^\alpha \not\approx \Delta_\alpha$ for all $\alpha$, setting the variational derivatives of $\wt\cW$ with respect to $\delta v^\alpha(t)$ to zero yields
\be
\begin{aligned} \label{eqmid1}
0&= (-1)^{\alpha+1} k t + \rho(t) \Big[ g'_+\big( \delta v^\alpha(t) + \tilde{\Delta}_{\alpha}\big) - g'_-\big(\delta v^\alpha(t) - \Delta_{\alpha}\big) \\
& + g'_-\big(\delta v^{\alpha-1}(t) - \Delta_{{\alpha-1}}\big) - g'_+\big(\delta v^{\alpha-1}(t) + \tilde{\Delta}_{{\alpha-1}}\big)  \Big]\, , \qquad \alpha =1, \ldots, n \, .
\end{aligned} \ee
However, in the following, we also need to consider the cases in which  $\delta v^{2\alpha-1}(t) \approx -  \tilde{\Delta}_{2\alpha-1}$ and that in which $\delta v^{2\alpha-1}(t) \approx \Delta_{2\alpha-1}$, for all $\alpha =1, \ldots, m$.  
\bi
\item In the former case, taking $\delta v^{2\alpha-1}(t) = -  \tilde{\Delta}_{2\alpha-1}+\exp(- N^{1/2} \tilde{Y}_{2\alpha-1})$ and setting to zero the variational derivatives of $\wt\cW$ with respect to $\delta v^{2\alpha-1}(t)$ and $\delta v^{2\alpha}(t)$ yields
\be
\begin{aligned} \label{eqleft1}
0&=\tilde{Y}_{2\alpha-1}(t) + k t + \rho(t) \Big[ g'_+(0) - g'_-\big( -\tilde{\Delta}_{2\alpha-1} - \Delta_{2\alpha-1}\big) \\
& + g'_-\big(\delta v^{2\alpha-2}(t) - \Delta_{{2\alpha-2}}\big) - g'_+\big(\delta v^{2\alpha-2}(t) + \tilde{\Delta}_{2\alpha-2}\big)  \Big]\, , \\
0&=-\tilde{Y}_{2\alpha-1} (t)-k t + \rho(t) \Big[ g'_+\big(\delta v^{2\alpha}(t) + \tilde{\Delta}_{2\alpha}\big) - g'_-\big(\delta v^{2\alpha}(t) - \Delta_{2\alpha}\big) \\
& + g'_-\big(- \tilde{\Delta}_{2\alpha-1} - \Delta_{{2\alpha-1}}\big) - g'_+(0)  \Big]\, .
\end{aligned} \ee
\item In the latter case, taking $\delta v^{2\alpha-1}(t) ={\Delta}_{2\alpha-1}-\exp(- N^{1/2} {Y}_{2\alpha-1})$ and setting to zero the variational derivatives of $\wt\cW$ with respect to $\delta v^{2\alpha-1}(t)$ and $\delta v^{2\alpha}(t)$  yields
\be
\begin{aligned} \label{eqright1}
0&=-{Y}_{2\alpha-1}(t) + k t + \rho(t) \Big[ g'_+\big({\Delta}_{2\alpha-1} + \tilde{\Delta}_{{2\alpha-1}}\big) - g'_-(0) \\
& + g'_-\big(\delta v^{2\alpha-2}(t) - \Delta_{{2\alpha-2}}\big) - g'_+\big(\delta v^{2\alpha-2}(t) + \tilde{\Delta}_{2\alpha-2}\big)  \Big]\, , \\
0&={Y}_{2\alpha-1}(t) - k t + \rho(t) \Big[ g'_+\big(\delta v^{2\alpha}(t) + \tilde{\Delta}_{2\alpha}\big) - g'_-\big( \delta v^{2\alpha}(t) - \Delta_{2\alpha}\big) \\
& + g'_-(0) - g'_+\big({\Delta}_{2\alpha-1} + \tilde{\Delta}_{{2\alpha-1}}\big)  \Big]\, .
\end{aligned} \ee
\ei
We also impose the condition that the sum of the chemical potential for each term in the superpotential \eqref{superpot} is $2 \pi$,
\be
\begin{aligned}
\Delta_\alpha+\Delta_{\alpha+1}+\tilde{\Delta}_\alpha+\tilde{\Delta}_{\alpha+1} = 2\pi \, .
\end{aligned} \ee

For later convenience, we define the following notations
\be
\begin{aligned} \label{paraA2m}
F_1 = m \sum_{\alpha=1}^m \Delta_{2\alpha} \, , \qquad
F_2 = m  \sum_{\alpha=1}^m \Delta_{2\alpha-1} \, , \qquad
F_3= \Delta_1 + \tilde{\Delta}_1 \, .
\end{aligned} \ee

Let us now proceed to solve the BAEs. First, we solve \eref{eq0} and \eref{eqmid1} and obtain
\be
\begin{aligned}
\rho & =  \frac{m k t\left[F_1 F_3-F_2 (2\pi -F_3)\right] +2 \pi  \mu }{m F_3 (2 \pi -F_3 ) \left(2 \pi-F_1-F_2 \right) \left(F_1+F_2\right) } \, , \\[.5em]
\delta v^{2\alpha-1} &= \Delta_{2\alpha-1} -\frac{\left(F_1+F_2\right) F_3 \left[ \mu -m k t(2 \pi -F_3-F_1) \right]}{m k t \left[F_1 F_3-F_2 (2\pi -F_3)\right]+2 \pi  \mu } \, , \\[.5em]
\delta v^{2\alpha} &= \Delta_{2\alpha} -\frac{\left(F_1+F_2\right) (2\pi -F_3) \left[ \mu +m k t(F_3-F_2) \right]}{m k t \left[F_1 F_3-F_2 (2\pi -F_3)\right]+2 \pi  \mu } \, ,
\end{aligned}
\qquad\qquad t_< < t < t_> \, .
\ee
This solution is valid in the interval $[t_{<},t_>]$ where the end-points are determined from 
\be
\begin{aligned}
\delta v^{2\alpha-1}(t_<) = - \tilde{\Delta}_{2\alpha-1} \, , \qquad \delta v^{2\alpha-1}(t_>) = \Delta_{2\alpha-1}  \quad \forall \alpha=1, \ldots, m \, .
\end{aligned} \ee 
Explicitly, they are
\be
\begin{aligned}
t_< = - \frac{\mu}{k m F_1}\, ,\qquad 
t_> = \frac\mu{k m (2 \pi - F_1 -F_3)}\, .
\end{aligned} \ee
Next, we focus on the regions $[t_{\ll}, t_<] $ and $[t_>, t_{\gg}]$, where $\delta v^{2\alpha-1}(t) = -  \tilde{\Delta}_{2\alpha-1}$ for $ t \in [t_{\ll}, t_<]$ and $\delta v^{2\alpha-1}(t) = {\Delta}_{2\alpha-1}$ for $ t \in [t_{>}, t_\gg]$.

For the interval $[t_{\ll}, t_<] $, we solve \eref{eq0} and \eref{eqleft1} and obtain
\be
\begin{aligned}
\rho & = \frac{\mu +m \left(F_3 -F_3  \right) k t}{m F_3 \left(F_1+F_2-F_3\right) \left(2  \pi -F_1-F_2\right)} \, , \\[.5em]
\delta v^{2\alpha-1} &= -  \tilde{\Delta}_{2\alpha-1} \, ,\qquad  \delta v^{2\alpha} = F_3-F_1-F_2 + \Delta_{2\alpha} \, , \\[.5em]
\widetilde{Y}_{2\alpha-1}&=  -\frac{\mu +m kt F_1 }{m(F_3-F_1-F_2)} \, ,
\end{aligned}
\qquad\qquad t_\ll < t < t_< \, ,
\ee
where we determine the end-point $t_\ll$ by the condition $\rho(t_\ll)=0$:
\be
\begin{aligned}
t_\ll = - \frac{\mu}{km( F_3-F_2)} \, .
\end{aligned} \ee

For the interval $[t_>, t_\gg]$, we solve \eref{eq0} and \eref{eqright1} and obtain
\be
\begin{aligned}
\rho & =  \frac{\mu-m k t F_2}{m F_3 \left(F_1+F_2-F_3\right) \left(2\pi- F_1-F_2\right)} \, , \\[.5em]
\delta v^{2\alpha-1} &= \Delta_{2\alpha-1}\, , \qquad \delta v^{2\alpha} =  -F_1-F_2 + \Delta_{2\alpha} \, , \\[.5em]
{Y}_{2\alpha-1} &=  \frac{\mu-m kt(2\pi- F_1-F_3)}{m (2\pi - F_1 - F_2-F_3)} \, ,
\end{aligned}
\qquad\qquad t_> < t < t_\gg \, ,
\ee
where we determine the end-point $t_\gg$ by the condition $\rho(t_\gg)=0$:
\be
\begin{aligned}
t_\gg = \frac\mu{k m F_2}\, .
\end{aligned} \ee

To summarize, the above solution is divided into three regions, namely the {\it left tail} $[t_\ll, t_<]$, the {\it inner interval} $[t_<, t_>]$ and the {\it right tail} $[t_>, t_\gg]$.  These are depicted in the following diagram:
\begin{center}
\begin{tikzpicture}[scale=2]
\draw (-2.5,0) -- (2.5,0);
\draw (-2.5,-.05) -- (-2.5, .05); \draw (-1,-.05) -- (-1, .05); \draw (1,-.05) -- (1, .05); \draw (2.5,-.05) -- (2.5, .05);
\node [below] at (-2.5,0) {$t_\ll$}; \node [below] at (-2.5,-.3) {$\rho=0$};
\node [below] at (-1.,0) {$t_<$}; \node [below] at (-1.,-.3) { $\delta v^{2\alpha-1} = - \tilde{\Delta}_{2\alpha-1} \, \forall \alpha$}; 
\node [below] at (1,0) {$t_>$}; \node [below] at (1,-.3) { $\delta v^{2\alpha-1}=\Delta_{2\alpha-1}  \, \forall \alpha $}; 
\node [below] at (2.5,0) {$t_\gg$}; \node [below] at (2.5,-.3) {$\rho=0$};
\end{tikzpicture}
\end{center}
Finally, the normalization $\int_{t_\ll}^{t_\gg} \rd t \, \rho(t) =1$ fixes
\be
\label{solution sum 2pi -- end -- A2m}
\mu =m \sqrt{2k F_1  F_2 \left(F_3-F_2\right) \left(2 \pi -F_3- F_1 \right)} \, .
\ee

\subsubsection[The index at large \texorpdfstring{$N$}{N}]{The index at large $\fakebold{N}$}

Give the rules in section \ref{large N index rules}, the topological free energy of this theory is given by
\be
\begin{aligned}
\frac{\mathfrak{F}}{N^{3/2}} &= -\int \rd t \, \rho(t)^2 \Bigg\{ \frac{n \pi^2}{3} + \sum_{\alpha=1}^n \left[ (\tilde{n}_{\alpha} -1) g'_+(\delta v^\alpha+ \tilde{\Delta}_{\alpha})
+ (n_{\alpha} -1) g'_- (\delta v^\alpha - \Delta_{\alpha})\right]\Bigg\} \\
& \qquad - \sum_{\alpha=1}^m \tilde{n}_\alpha \int_{\delta v^{2\alpha-1} \approx -\tilde{\Delta}_{2\alpha-1}} \rd t \, \rho(t) \, \tilde{Y}_{2\alpha-1}(t)   - \sum_{\alpha=1}^m n_\alpha \int_{\delta v^{2\alpha-1} \approx {\Delta}_{2\alpha-1}}  \rd t \, \rho(t) \, {Y}_{2\alpha-1}(t) \, .
\end{aligned} \ee
The result depends only on the parameters $F_1$, $F_2$, $F_3$ and their corresponding flavor magnetic fluxes
\be
\begin{aligned}
{\frak n}_1 = m \sum_{\alpha=1}^m n_{2\alpha} \, , \qquad {\frak n}_2 = m \sum_{\alpha=1}^m n_{2\alpha-1} \, , \qquad {\frak n}_3 = n_1 + \tilde{n}_1 \, ,
\end{aligned} \ee
and can be written as
\be
\begin{aligned} \label{freealtCS}
\mathfrak{F} = m\, \mathfrak{F}_{\text{ABJM}_k} \, .
\end{aligned} \ee
The map of the parameters is as follows:
\be
\begin{aligned}
\Delta_{A_1} & = F_1 \, , \qquad \Delta_{A_2} =  F_2 \, , \qquad \Delta_{B_1}=F_3-F_2\, , \qquad \Delta_{B_2} = 2\pi -F_1- F_3 \, , \\
\fn_{A_1} & =  {\frak n}_1 \, , \qquad \fn_{A_2} = {\frak n}_2 \, , \qquad \fn_{B_1}= {\frak n}_3-{\frak n}_2 \, , \quad \fn_{B_2} = 2 - {\frak n}_1- {\frak n}_3 \, .
\end{aligned} \ee

Recall that the geometric branch of the moduli space of this theory is $\Sym^N (\BC^2/\BZ_m \times \BC^2/\BZ_{m})/\BZ_k$, whereas that of the ABJM theory is $\Sym^N (\BC^4/\BZ_k)$.
The square root of the relative orbifold orders of these two spaces explains the prefactor $m$ in \eref{freealtCS}.

\subsection[The affine \texorpdfstring{$A_{n-1}$}{A[n-1]} quiver with two adjacent CS levels of opposite signs]{The affine $\fakebold{A_{n-1}}$ quiver with two adjacent CS levels of opposite signs}
\label{sec:necklacetwoCS}

We are interested in the necklace quiver with $n$ nodes, each with $\U(N)$ gauge group, and the Chern-Simons levels:
\be
\begin{aligned}
k_\alpha= \begin{cases} +k & \text{if $\alpha=1$} \\
-k & \text{if $\alpha=2$}  \\
0 & \text{otherwise}
\end{cases}
\end{aligned} \ee 
The $\CN=2$ quiver diagram of this theory is
\be
\begin{aligned} \label{fig:CStwo}
\begin{tikzpicture}[baseline,font=\tiny, every loop/.style={min distance=10mm,looseness=10}, every node/.style={minimum size=0.9cm}]
\def \n {6}
\def \radius {1.5cm}
\def \margin {18} 
\foreach \s in {1}
{
    \node[draw, circle] at ({2*360/\n * (\s)}:\radius)  (\s) {$N_{+k}$};
    \draw[black,-> ] (\s) edge [out={30+2*\s*60},in={-30+2*\s*60},loop,looseness=5] (\s);
} 
\foreach \s in {1}
{    
    \node[draw, circle]  at ({2*360/\n*\s-360/\n}:\radius) (\s)  {$N_{-k}$};
    \draw[black,-> ] (\s) edge [out={-30+2*\s*60},in={-90+2*\s*60},loop,looseness=5] (\s);
} 
\foreach \s in {2,3}
{
    \node[draw, circle] at ({2*360/\n * (\s)}:\radius)  (\s) {$N$};
    \draw[black,-> ] (\s) edge [out={30+2*\s*60},in={-30+2*\s*60},loop,looseness=5] (\s);
} 
\foreach \s in {2,3}
{    
    \node[draw, circle]  at ({2*360/\n*\s-360/\n}:\radius) (\s)  {$N$};
    \draw[black,-> ] (\s) edge [out={-30+2*\s*60},in={-90+2*\s*60},loop,looseness=5] (\s);
} 
\foreach \s in {1,...,5}
{  
  \draw[decoration={markings, mark=at position 0.5 with {\arrow[scale=1.5]{>}}}, postaction={decorate}, shorten >=0.7pt]  ({360/\n * (\s - 3)+\margin}:\radius) to[bend left=60] node[midway,above] {}node[midway,below] {}  ({360/\n * (\s - 2)-\margin}:\radius) ; 
    \draw[decoration={markings, mark=at position 0.5 with {\arrow[scale=1.5]{<}}}, postaction={decorate}, shorten >=0.7pt]  ({360/\n * (\s - 3)+\margin}:\radius) to[bend right=60] node[midway,above] {}node[midway,below] {}  ({360/\n * (\s - 2)-\margin}:\radius) ;
}   
\draw[dashed, >=latex] ({360/6 * (5 -2)+\margin}:\radius) 
    arc ({360/6 * (5 -2)+\margin}:{360/6 * (5-1)-\margin}:\radius);
\end{tikzpicture}
\end{aligned} \ee
In the notation of the preceding subsection, the superpotential can be written as
\be
\begin{aligned}
W = \sum_{\alpha=1}^n  \Tr \left(Q_\alpha \phi_{\alpha+1} \tilde{Q}_\alpha- \tilde{Q}_{\alpha} \phi_{\alpha} Q_{\alpha}\right) + \frac{k}{2} \Tr \left(\phi_1^2 - \phi_2^2\right) \, .
\end{aligned} \ee
After integrating out the massive adjoint fields $\phi_1$ and $\phi_2$, we have the superpotential
\be
\begin{aligned}
W &=-\frac{1}{k} \Tr \left(Q_1 Q_2 \tQ_2 \tQ_1  - Q_1 \tilde{Q}_1 \tilde{Q}_n Q_n \right)+ \frac{1}{2k} \Tr \left[\left(Q_2 \tilde{Q}_2\right)^2 - \left(Q_n \tilde{Q}_n \right)^2 \right] \\
&+ \sum_{\alpha=2}^{n-1} \Tr \left(Q_\alpha \phi_{\alpha+1} \tilde{Q}_\alpha - \tilde{Q}_{\alpha+1} \phi_{\alpha+1} Q_{\alpha+1} \right), \label{superpot1}
\end{aligned} \ee

\subsubsection{A solution to the system of BAEs}
Let us denote respectively by $\Delta_{\alpha}$, $\tilde{\Delta}_{\alpha}$, $\Delta_{\phi_\alpha}$ the chemical potentials associated to the flavor symmetries of $Q_\alpha$, $\tQ_\alpha$, $\phi_\alpha$, and by $n_{\alpha}$, $\tilde{n}_{\alpha}$, ${n}_{\phi_\alpha}$ the corresponding fluxes associated with their flavor symmetries.  We also denote by $\Delta_m^{(\alpha)}$ the chemical potential associated with the topological charge corresponding to node $\alpha$ and $\ft^{(\alpha)}$ the corresponding magnetic flux.

The superpotential \eqref{superpot1} implies the following constraints
\be
\begin{aligned}
&\tilde{\Delta}_\alpha = \pi - \Delta_\alpha \quad \forall \alpha=1,\ldots,n \, , \\
&\Delta_{\phi_3} = \ldots =  \Delta_{\phi_n} = \pi \, . 
\end{aligned} \ee

The twisted superpotential for this particular model can be derived from the rules in section \ref{large N twisted superpotential rules}.
The procedure of solving the BAEs  is similar to that presented in section \ref{altCSN3}.
The solution can be separated into three regions, namely the left tail $[t_{\ll}, t_<]$, the inner interval $[t_{<},t_>]$ and the right tail $[t_>, t_{\gg}]$, where
\be
\begin{aligned}
t_< \text{ s.t. } \delta v^{1}(t_<) = -  \tilde{\Delta}_{1} \, , \qquad
 t_> \text{ s.t. } \delta v^{1}(t_>) = \Delta_{1} \, .
\end{aligned} \ee
The end-points $t_\ll$ and $t_\gg$ are the values where $\rho=0$ on the left and the right tails, respectively. Schematically:
\begin{center}
\begin{tikzpicture}[scale=2]
\draw (-2,0) -- (2,0);
\draw (-2.,-.05) -- (-2., .05); \draw (-0.7,-.05) -- (-0.7, .05); \draw (0.7,-.05) -- (0.7, .05); \draw (2,-.05) -- (2, .05);
\node [below] at (-2.,0) {$t_\ll$}; \node [below] at (-2.,-.3) {$\rho=0$};
\node [below] at (-0.7,0) {$t_<$}; \node [below] at (-0.7,-.3) {$\delta v^{1} = - \tilde{\Delta}_{1}$};
\node [below] at (0.7,0) {$t_>$}; \node [below] at (0.7,-.3) { $\delta v^{1}=\Delta_{1}$}; 
\node [below] at (2,0) {$t_\gg$}; \node [below] at (2,-.3) {$\rho=0$};
\end{tikzpicture}
\end{center}

It turns out that the solution depends on the following parameters:
\be
\begin{aligned} \label{paraF}
F_1 = \Delta_1 + \frac{1}{k} \sum_{\alpha=1}^n \Delta_m^{(\alpha)} \, , \qquad
F_2 = \frac{1}{n-1} \left[  \left( \sum_{\alpha=2}^n \Delta_{\alpha}  \right) - \frac{1}{k} \sum_{\alpha=1}^n \Delta_m^{(\alpha)} \right] \, .
\end{aligned} \ee 

The solution is as follows.
In the left tail $[t_\ll, t_<]$, we have
\be
\begin{aligned}
\rho & =  \frac{(n-1) \left[ \mu +(\pi -F_1 )k t\right]}{\pi  \left[n\pi -F_1-(n-1)F_2 \right] \left[\pi -F_1-(n-1)F_2\right]} \, , \\[.5em]
\delta v^{1} &= -  \tilde{\Delta}_{1} \, , \\[.5em]
\delta v^{\alpha} &=\Delta_\alpha + \left[ \pi-F_1-(n-1)F_2 \right] \, , \qquad \forall \, 2 \leq \alpha \leq n \, , \\[.5em]
\widetilde{Y}_{1} &=  \frac{(n-1)F_2k t+\mu }{\pi-F_1-(n-1)F_2 } \, .
\end{aligned}
\ee
In the inner interval $[t_<, t_>]$, we have
\be
\begin{aligned}
\rho & =  \frac{\left[(n-1) (F_1-F_2)\right] k t-n \mu }{\pi  \left[F_1+(n-1)F_2\right] \left[F_1+(n-1)F_2-n \pi \right]} \, , \\[.5em]
\delta v^{1} &= \frac{ \mu  n \Xi-(n-1) \Big[ F_1 [ \mu + (\pi+\Xi  )k t ]-F_2 [\mu- \{ (n-1)\pi-\Xi \} k t ] \Big]+(n-1) [F_1^2 +(n-1)F_2^2] k t }{(n-1)(F_1-F_2) k t- n \mu } \, , \\[.5em]
\delta v^{\alpha} &= \frac{\left[F_1+(n-1)F_2\right] \left[\mu +\left(\pi -F_1 \right) k t\right]}{\left[(n-1) F_1-(n-1)F_2\right] k t-n \mu } +\Delta_\alpha~, \quad \forall \, 2\leq \alpha \leq n \, ,
\end{aligned}
\ee
where $\Xi = \frac{1}{k} \sum_\alpha \Delta_m^{(\alpha)}$.
In the right tail $[t_>, t_\gg]$ we have
\be
\begin{aligned}
\rho &=  \frac{(n-1) \left(F_1 k t-\mu \right)}{\pi  \left[F_1+(n-1)F_2\right] \left[F_1+(n-1)F_2-(n-1) \pi \right]} \, , \\[.5em]
\delta v^{1} &=  {\Delta}_{1}  \, , \\[.5em]
\delta v^{\alpha} &= \Delta_\alpha +\frac{1}{n-1} \left[  \pi -F_1-(n-1)F_2 \right] \, , \\[.5em]
{Y}_{1} &=  \frac{\mu - (n-1)(\pi-F_2) k t}{F_1+(n-1)F_2-(n-1) \pi } \, .
\end{aligned}
\ee
The transition points are at
\be
\begin{aligned} \label{solution sum 2pi -- init -- An}
t_\ll = - \frac{\mu }{k (\pi - F_1)} \, ,\qquad t_< = - \frac{\mu }{k F_2} \, ,\qquad
t_> = \frac{\mu }{k(n-1) (\pi-F_2)} \, ,\qquad t_\gg = \frac{\mu }{k  F_1} \, .
\end{aligned} \ee
Finally, the normalization $\int_{t_\ll}^{t_\gg} \rd t \, \rho(t) =1$ fixes
\be
\mu = \sqrt{2(n-1)k  F_1 F_2 (\pi-F_1) ( \pi - F_2) } \, .
\ee

\subsubsection[The index at large \texorpdfstring{$N$}{N}]{The index at large $\fakebold{N}$}

The topological free energy of this theory can be derived from the rules give in section \eref{large N index rules}.
We find that the topological free energy of this quiver theory depends only on the parameters $F_1$, $F_2$
given by \eref{paraF} and their corresponding conjugate magnetic charges
\be
\begin{aligned}
{\frak n}_1 = n_1 + \frac{1}{k} \sum_{\alpha=1}^n \ft_\alpha\, , \qquad  {\frak n}_2 =\frac{1}{n-1} \left[ \left( \sum_{\alpha=2}^n n_{\alpha} \right)-  \frac{1}{k} \sum_{\alpha=1}^n \ft_\alpha \right]\, .
\end{aligned} \ee
The topological free energy can be written as,
\be
\begin{aligned}
\mathfrak{F}= \sqrt{n-1}\, \mathfrak{F}_{\text{ABJM}}\, . \label{freetwoaltCS}
\end{aligned} \ee
The map of the parameters is as follows,
\be
\begin{aligned}
\Delta_{A_1} & = F_1\, , \qquad \Delta_{A_2} = F_2\, , \qquad \Delta_{B_1}= \pi - F_1\, , \qquad \Delta_{B_2} = \pi - F_2\, , \\
\fn_{A_1} & =  \fn_1\, , \qquad \fn_{A_2} =\fn_2\, , \qquad \fn_{B_1}= 1 - \fn_1\, , \qquad \fn_{B_2} =1 - \fn_2\, .
\end{aligned} \ee
Indeed, for $n=2$, this theory becomes the ABJM theory and \eref{freetwoaltCS} reduces to $\mathfrak{F}_{\text{ABJM}}$, as expected.
Recall that the geometric branch of the moduli space of this theory is $\Sym^N (\BC^2 \times \BC^2/\BZ_{n-1})/\BZ_k$, whereas that of the ABJM theory is $\Sym^N (\BC^4/\BZ_k)$.
The square root of the relative orbifold orders of these two spaces explains the prefactor $\sqrt{n-1}$ in \eref{freetwoaltCS}.

Let us also comment on the number of the parameters which appears in the topological free energy of this model.
It can be seen from \eref{freetwoaltCS} that the topological free energy depends only on two parameters, $F_1$ and $F_2$ (or $\fn_1$ and $\fn_2$), instead of three, despite the fact that the geometric branch is associated with Calabi-Yau four-fold $\BC^2 \times \BC^2/\BZ_{n-1}$.
Indeed, in the $\CN=3$ description of the quiver, only $\U(1)^2$ (one mesonic and one topological symmetry) is manifest (see appendix C of \cite{Cremonesi:2016nbo}).
An extra mesonic symmetry that exchanges the holomorphic variables on $\BC^2$ and those on $\BC^2/\BZ_2$ is not present in the quiver description of this theory.

\paragraph*{$\fakebold{\SL(2,\BZ)}$ duality.} The affine $A_{n-1}$ quiver \eref{fig:CStwo} with $n$ gauge nodes and $k=1$ is $\SL(2,\BZ)$ dual to the $A_{n-2}$ Kronheimer-Nakajima quiver \eref{KNN4quiv} with $n-1$ gauge nodes and $r=1$.
This duality can be seen from the type IIB brane configuration as follows \cite{Hanany:1996ie, Aharony:1997ju, Gaiotto:2008ak, Assel:2014awa}.
The configuration of the Kronheimer-Nakajima quiver involves $N$ D3-branes wrapping $\BR^{1,2}_{0,1,2} \times S^{1}_6$ (where the subscripts indicate the direction in $\BR^{1,9}$);
$n-1$ NS5-branes wrapping $\BR^{1,2}_{0,1,2} \times \BR^3_{7,8,9}$ located at different positions along the circular $x^6$ direction;
and $r=1$ D5-branes wrapping $\BR^{1,2}_{0,1,2} \times \BR^3_{3,4,5}$ located along the circular $x^6$ direction within one of the NS5-brane intervals.
Applying an $\SL(2,\BZ)$ action on such a configuration, we can obtain a similar configuration except that the D5-brane becomes a $(1,1)\,5$-brane.
This is in fact the configuration for quiver \eref{fig:CStwo} with $n$ gauge nodes and $k=1$.
Indeed, in this case we can match the topological free energies \eref{freetwoaltCS} and \eref{freeKN}, as expected from the duality.

\subsection[The \texorpdfstring{$N^{0,1,0}/\BZ_k$}{N**(0,1,0)/Z[k]} theory]{The $\fakebold{N^{0,1,0}/\BZ_k}$ theory}
\label{sec:N010}

In this section we focus on the holographic dual of M-theory on ${\rm AdS}_4 \times N^{0,1,0}/\BZ_k$ \cite{Fabbri:1999hw, Billo:2000zr, Yee:2006ba}.
$N^{0,1,0}$ is a homogeneous Sasakian of dimension seven and defined as the coset $\SU(3)/\U(1)$. The manifold has the isometry $\SU(3) \times \SU(2)$.
The latter $\SU(2)$ is identified with the R-symmetry.
The description of the dual field theory was discussed in \cite{Gaiotto:2009tk, Imamura:2011uj, Cheon:2011th}.
This theory has $\CN=3$ supersymmetry and contains $\CG = \U(N)_{+k} \times \U(N)_{-k}$ gauge group with two bi-fundamental hypermultiplets and $r$ flavors of fundamental hypermultiplets under one of the gauge groups.
The $\CN=3$ quiver is depicted as follows:
\be
\begin{aligned}
\begin{tikzpicture}[font=\footnotesize, scale=0.9]
\begin{scope}[auto,%
  every node/.style={draw, minimum size=0.5cm}, node distance=2cm];
\node[circle]  (UN1)  at (0.3,1.7) {$N_{+k}$};
\node[circle, right=of UN1] (UN2) {$N_{-k}$};
\node[rectangle, left=of UN1] (Ur) {$r$};
\end{scope}
\draw[draw=blue,solid,line width=0.1mm]  (UN1) to[bend right=30]  (UN2) ; 
\draw[draw=red,solid,line width=0.1mm]  (UN1) to[bend left=30]  (UN2) ;    
\draw[draw=black,solid,line width=0.1mm]  (UN1) to (Ur) ;    
\end{tikzpicture}
\end{aligned} \ee
Note that for $k=0$, this theory becomes the Kronheimer-Nakajima quiver \eref{KNN4quiv} with $n=2$.

In $\CN=2$ notation, the quiver diagram for this theory is
\be
\begin{aligned}
\begin{tikzpicture}[font=\footnotesize, scale=0.8]
\begin{scope}[auto,%
  every node/.style={draw, minimum size=0.5cm}, node distance=2cm];
\node[circle] (USp2k) at (0., 0) {$N_{+k}$};
\node[circle, right=of USp2k] (BN)  {$N_{-k}$};
\node[rectangle, below=of USp2k] (Ur)  {$r$};
\end{scope}
\draw[draw=blue,solid,line width=0.2mm,<-]  (USp2k) to[bend right=15] node[midway,above] {$B_2 $}node[midway,above] {}  (BN) ;
\draw[draw=blue,solid,line width=0.2mm,->]  (USp2k) to[bend right=50] node[midway,above] {$A_1$}node[midway,above] {}  (BN) ; 
\draw[draw=red,solid,line width=0.2mm,<-]  (USp2k) to[bend left=15] node[midway,above] {$B_1$} node[midway,above] {} (BN) ;  
\draw[draw=red,solid,line width=0.2mm,->]  (USp2k) to[bend left=50] node[midway,above] {$A_2$} node[midway,above] {} (BN) ;    
\draw[black,-> ] (USp2k) edge [out={-150},in={150},loop,looseness=10] (USp2k) node at (-2.,1) {$\phi_1$} ;
\draw[black,-> ] (BN) edge [out={-30},in={30},loop,looseness=10] (BN) node at (5.8,1) {$\phi_2$};
\draw[draw=black,solid,line width=0.2mm,<-]  (USp2k) to[bend left=20] node[midway,right] {$\tilde{q}$} node[midway,above] {} (Ur) ;  
\draw[draw=black,solid,line width=0.2mm,->]  (USp2k) to[bend right=20] node[midway,left] {$q$} node[midway,above] {} (Ur) ;    
\end{tikzpicture}
\end{aligned} \ee
where the bi-fundamental chiral fields $(A_1, B_2)$ come from one of the $\CN=3$ hypermultiplet indicated in blue, and the bi-fundamental chiral fields $(A_2, B_1)$ come from the other $\CN=3$ hypermultiplet indicated in red. The superpotential is given by
\be
\begin{aligned} \label{WN010mass}
W = \Tr \left(A_1 \phi_2 B_2 - B_2 \phi_1 A_1 - A_2 \phi_2 B_1 + B_1 \phi_1 A_2 + \frac{k}{2} \phi_1^2 - \frac{k}{2} \phi_2^2 + \tilde{q} \phi_1 q \right) \, .
\end{aligned} \ee
Note that the bi-fundamental fields $A_1, A_2, B_1, B_2$ can be mapped to those in the Kronheimer-Nakajima quiver \eref{KNN2quiv} with $n=2$ as follows
\be
\begin{aligned}
A_1 \, \leftrightarrow \, Q_1\, , \qquad A_2 \, \leftrightarrow \, \tQ_2 \, , \qquad B_1 \, \leftrightarrow \, Q_2\, , \qquad B_2 \, \leftrightarrow \, \tQ_1 \, .
\end{aligned} \ee

Integrating out the massive adjoint fields $\phi_{1,2}$ in \eref{WN010mass}, we obtain the superpotential
\begin{equation}\label{superpotentialN010}
 W = \Tr\left[\left(\epsilon^{ij} B_i A_j - q \tilde q\right)^2-\left(\epsilon^{ij} A_i B_j\right)^2\right]\, .
\end{equation}

\subsubsection{A solution to the system of BAEs}
The twisted superpotential for this model can be derived from the rules in section \ref{large N twisted superpotential rules}.  The procedure of solving the BAEs  is similar to that presented in sections \ref{sec:solnKN} and \ref{altCSN3}.  In the following we present an explicit solution to the corresponding BAEs.

For brevity, let us write
\be
\begin{aligned} \label{shorthandAB}
\Delta_1 & = \Delta_{A_1}\, , \qquad \Delta_2 = \Delta_{A_2}\, , \qquad \Delta_3 = \Delta_{B_1}\, , \qquad \Delta_4 = \Delta_{B_2} \, , \\
\fn_1 & = \fn_{A_1}\, , \qquad \fn_2 = \fn_{A_2}\, ,\qquad \fn_3 = \fn_{B_1}\, , \qquad \fn_4 = \fn_{B_2}\, . 
\end{aligned} \ee
We look for a solution to the BAEs such that
\begin{align}\label{Delta constraintN010}
\Delta_{q}+ \Delta_{\tilde{q}}=\pi \, , \qquad \Delta_{1} + \Delta_{4} = \pi \, , \qquad \Delta_{2} + \Delta_{3} = \pi \, ,
\end{align}
and
\begin{align}\label{flux constraintN010}
 &\fn_{q}+\fn_{\tilde{q}}=1\, ,\qquad\qquad
 \fn_{1} + \fn_{4} = 1\, ,\qquad\qquad
 \fn_{2} + \fn_{3} = 1\, .
\end{align}

The solution can be separated into three regions, namely the left tail $[t_{\ll}, t_<]$, the inner interval $[t_{<},t_>]$ and the right tail $[t_>, t_{\gg}]$, where
\be
t_< \text{ s.t. } \delta v(t_<) = - \Delta_3 \,,\qquad\qquad t_> \text{ s.t. } \delta v(t_>) = \Delta_1 \,.
\ee
Then we define $t_\ll$ and $t_\gg$ as the values where $\rho=0$ and those bound the left and right tails. Schematically:
\begin{center}
\begin{tikzpicture}[scale=2]
\draw (-1.5,0) -- (1.5,0);
\draw (-1.5,-.05) -- (-1.5, .05); \draw (-0.5,-.05) -- (-0.5, .05); \draw (0.5,-.05) -- (0.5, .05); \draw (1.5,-.05) -- (1.5, .05);
\node [below] at (-1.5,0) {$t_\ll$}; \node [below] at (-1.5,-.3) {$\rho=0$};
\node [below] at (-0.5,0) {$t_<$}; \node [below] at (-0.5,-.3) {$\delta v = -\Delta_3$};
\node [below] at (-0.5,-.6) {$Y_3 = 0$};
\node [below] at (0.5,0) {$t_>$}; \node [below] at (0.5,-.3) {$\delta v=\Delta_1$};
\node [below] at (0.5,-.6) {$Y_1 = 0$};
\node [below] at (1.5,0) {$t_\gg$}; \node [below] at (1.5,-.3) {$\rho=0$};
\end{tikzpicture}
\end{center}

The solution is as follows.
In the left tail we have
\be
\begin{aligned}
\rho &= \frac{\mu + k t\Delta_3 - \frac{\pi}{2}r  |t|}{\pi (\Delta_1 + \Delta_3) (\Delta_4 - \Delta_3)} \, , \\[.5em]
\delta v &= - \Delta_3 \,,\qquad\qquad Y_3 = \frac{- kt\Delta_4 -\mu + \frac{\pi}{2}r |t|}{\Delta_4 - \Delta_3} \, ,
\end{aligned}
\qquad\qquad t_\ll < t < t_< \, .
\ee
In the inner interval we have
\be
\begin{aligned}
\rho &= \frac{2 \mu + k t(\Delta_3 - \Delta_1) - \pi r |t|}{\pi (\Delta_1 + \Delta_3)(\Delta_2 + \Delta_4)} \, , \\[.5em]
\delta v &= \frac{\left( \mu - \frac{\pi}{2} r |t| \right) (\Delta_1 - \Delta_3)
+ k t \left( \Delta_1 \Delta_4 + \Delta_2 \Delta_3 \right)}{2 \mu + k t (\Delta_3 - \Delta_1) - \pi r |t|} \, ,
\end{aligned}
\qquad\qquad t_< < t < t_> \, ,
\ee
and $\delta v'>0$. In the right tail we have
\be
\begin{aligned}
\rho &= \frac{\mu -k t \Delta_1 - \frac{\pi}{2}r |t|}{\pi (\Delta_1 + \Delta_3)(\Delta_2 - \Delta_1)} \, , \\[.5em]
\delta v &= \Delta_1 \,,\qquad\qquad Y_1 = \frac{kt\Delta_2 - \mu + \frac{\pi}{2}r  |t|}{\Delta_2 - \Delta_1} \, ,
\end{aligned}
\qquad\qquad t_> < t < t_\gg \, .
\ee
The transition points are at
\be
\begin{aligned}
\label{solution sum 2pi -- init -- N010}
& t_\ll = - \frac{2\mu}{\pi r + 2 k  \Delta_3} \, ,\qquad t_< = - \frac{2\mu}{\pi r + 2 k \Delta_4} \, ,\qquad
& t_> = \frac{2\mu}{\pi r + 2 k  \Delta_2} \, ,\qquad t_\gg = \frac{2\mu}{\pi r + 2 k  \Delta_1} \, .
\end{aligned} \ee
Finally, the normalization fixes
\be
 \label{solution sum 2pi -- end -- N010}
 \mu = \frac{1}{\sqrt{2}} \sqrt{\frac{\delta_1 \delta_2 \delta_3 \delta_4 (\Delta_1+\Delta_3) (\Delta_2+\Delta_4)}
 {k (\delta_1 + \delta_3) (\delta_2 + \delta_4) + r (\delta_1 \delta_4 + \delta_2 \delta_3)}} \, ,
\ee
where $\delta_a = \pi r + 2 k \Delta_a \, ,$ $\forall a = 1,2,3,4$.
For $k=0$, this expression indeed reduces to \eref{muKN} with 
\be
\begin{aligned}
F_1 & = 2 \pi c\, , \qquad F_2 = - c (\Delta_2 + \Delta_4) \, , \qquad F_3 = 2 \pi - c (\Delta_1 + \Delta_3) \, , \\
\Delta_m & = 2 \pi + c (\Delta_1 + \Delta_3) \, ,
\end{aligned} \ee
and $c =1/(2 \times 12^{1/3})$. Note that $F_1+F_2+F_3 = 2 \pi$, as required.

\subsubsection[The index at large \texorpdfstring{$N$}{N}]{The index at large $\fakebold{N}$}

The topological free energy of this theory can be computed from the rules in section \ref{large N index rules}.
The expression for the topological free energy is fairly long, so we will just give the formul\ae{} for $k=1$, $r=1$ and
\begin{equation}
 \Delta_3 = \Delta_4 = \Delta \, , \qquad  \fn_3 = \fn_4 = \fn \, .
\end{equation}
In this case, the topological free energy reads
\begin{align}
\mathfrak{F} & = - \frac{2 N^{3/2}}{3} \frac{\pi (\pi - 2 \Delta ) \left[ 4 (\pi -\Delta ) \Delta + 19 \pi^2 \right] \fn
 + \left( 8 \Delta^4 - 20 \pi \Delta^3 - 6 \pi^2 \Delta^2 + 37 \pi^3 \Delta + 33 \pi^4 \right)}{\left[ 4 (\pi - \Delta ) \Delta + 11 \pi^2 \right]^{3/2}} \, .
\end{align}

\section[Quivers with \texorpdfstring{$\CN=2$}{N=2} supersymmetry]{Quivers with $\fakebold{\cN=2}$ supersymmetry}
\label{sec:N2susy}

Let us now consider quiver gauge theories with $\CN=2$ supersymmetry.
We first discuss the SPP model. Then we move to study non-toric theories associated
with the Sasaki-Einstein seven manifold $V^{5,2}$.
There are two known models in this cases, one proposed by \cite{Martelli:2009ga} and the other by \cite{Jafferis:2009th}.
We show that the topological free energy of these models can
be matched with each other. We then move on to discuss flavored
toric theories \cite{Benini:2009qs}. The procedure in solving the BAEs
for these theories is similar to that for $\CN=3$ theories discussed in the preceding section.  

\subsection{The SPP theory}
\label{SPP}

We now consider the quiver gauge theory which describes the dynamics of
$N$ M2-branes at the SPP singularity. The quiver diagram is shown below.
\be
\begin{aligned}
\label{SPP_quiver}
\begin{tikzpicture}[font=\footnotesize, scale=0.9]
\begin{scope}[auto,
  every node/.style={draw, minimum size=0.5cm}, node distance=2cm];
\node[circle]  (UN)  at (0.35,1.7) {$N_{k_1}$};
\node[circle]  (UN2)  at (-1.2,-.7) {$N_{k_3}$};
\node[circle]  (UN3)  at (1.8,-.7) {$N_{k_2}$};
\end{scope}\draw[decoration={markings, mark=at position 0.45 with {\arrow[scale=1.5]{>}}}, postaction={decorate}, shorten >=0.7pt] (0.7,2.15) arc (-45:240:0.6cm);
\draw[draw=black,solid,line width=0.2mm,<-]  (UN) to[bend right=18] (UN2) ;
\draw[draw=black,solid,line width=0.2mm,->]  (UN) to[bend left=18]  (UN2) ;
\draw[draw=black,solid,line width=0.2mm,<-]  (UN) to[bend right=18] (UN3) ; 
\draw[draw=black,solid,line width=0.2mm,->]  (UN) to[bend left=18]  (UN3) ;  
\draw[draw=black,solid,line width=0.2mm,<-]  (UN2) to[bend right=18](UN3) ;
\draw[draw=black,solid,line width=0.2mm,->]  (UN2) to[bend left=18] (UN3) ;
\node at (1.5,1.0) {$A_1$};
\node at (1.0,0.6) {$A_2$};
\node at (0.3,-1.3) {$B_1$};
\node at (0.3,-0.6) {$B_2$};
\node at (-0.3,0.6) {$C_1$};
\node at (-0.9,1.0) {$C_2$};
\node at (0.4,3.5) {$X$};
\end{tikzpicture}
\end{aligned} \ee
The  Chern-Simons levels are $(k_1, k_2, k_3) = (2k,-k,-k)$ and the superpotential coupling is given by
\begin{equation}\label{superpotential}
 W = \Tr\left[X \left( A_1 A_2 - C_1 C_2 \right) - A_2 A_1 B_1 B_2 + C_2 C_1 B_2 B_1 \right]\, .
\end{equation}
The marginality condition on the superpotential \eqref{marginality} impose constraints on the
chemical potential of the various fields
\begin{align}
 &\Delta_{A}+\Delta_{B}= \pi \, ,\qquad\qquad
 \Delta_{B}+\Delta_{C}= \pi \, ,\qquad\qquad
 2 \Delta_{A}+\Delta_{X}= 2 \pi \, ,
\end{align}
where we have used the symmetry of the quiver to set $\Delta_{A_1} = \Delta_{A_2} = \Delta_{A}$, and so on.
Hence,
\begin{equation}
 \Delta_{B} = \Delta \, , \qquad \qquad \Delta_X = 2 \Delta \, , \qquad \qquad
 \Delta_A = \Delta_C = \pi - \Delta \, ,
\end{equation}
and
\begin{align}\label{flux constraint0}
 \fn_{B} = \fn \, , \qquad \qquad \fn_X = 2 \fn \, , \qquad \qquad
 \fn_A = \fn_C = 1 - \fn \, ,
\end{align}
where $\fn_I$ denotes the flavor magnetic flux of the field $I$. We assume $0\leq \Delta \leq 2\pi$ and we enforced condition \eqref{superpotential0}. One can check 
that all other solutions are related to the one we are presenting by a discrete symmetry of the quiver.\footnote{There is a solution for
\begin{align}
 &\Delta_{A}+\Delta_{B}= 3 \pi \, ,\qquad\qquad
 \Delta_{B}+\Delta_{C}= 3 \pi \, ,\qquad\qquad
 2 \Delta_{A}+\Delta_{X}= 4 \pi \, .
\end{align}
which is obtained, using the invariance of $Z$ under $y_I \to 1/y_I$, from \eqref{solution sum 2pi -- init}-\eqref{solution sum 2pi -- end} by performing the substitutions
\be
\mu \to - \mu \, ,\qquad k \to - k \, ,\qquad \Delta \to \pi - \Delta \, ,\qquad Y^{\pm} \to - Y^{\pm} \, .
\ee}

\subsubsection{A solution to the system of BAEs}

The theory under consideration is invariant under
\begin{equation}
 A \leftrightarrow C \, , \qquad \qquad \U(N)^{(2)} \leftrightarrow \U(N)^{(3)} \, .
\end{equation}
Let us assume that the saddle-point solution is also invariant under this $\bZ_2$ symmetry. Thus, we can choose
\begin{equation}
 v_i^{(1)} = v_i \, , \qquad \qquad v_i^{(2)} = v_i^{(3)} = w_i \, .
\end{equation}
Given the rules of section \ref{large N twisted superpotential rules}, the twisted superpotential reads
\be
\begin{aligned}\label{large N twisted superpotential}
 \frac{\wt\cW}{i N^{3/2}}&
 = 2 k \int \rd t \, t\, \rho(t)\, \delta v(t)
 + \int \rd t \, \rho(t)^2\, \Delta \left[ (\pi -\Delta ) (2 \pi -\Delta ) - 2 \delta v^2 \right] \\
 & - \mu \left(\int \rd t \, \rho(t)-1\right)
 - \frac{2 i}{N^{1/2}}\int \rd t \, \rho(t)\, \left[\pm \Li_2 \left(\e^{i\left[\delta v(t)\pm \left( \pi - \Delta\right) \right]}\right)\right] \, ,
\end{aligned}
\ee
where we defined
\begin{equation}
 \delta v(t) = w(t) - v(t) \, ,
\end{equation}
and we included the subleading terms giving rise to the equation of motion \eqref{tails}.
The eigenvalue density distribution $\rho(t)$, which maximizes the twisted superpotential, is a piece-wise function supported on $[t_{\ll}, t_{\gg}]$.
We define the inner interval as
\be
t_< \text{ s.t. } \delta v(t_<) = - \left( \pi - \Delta \right) \;,\qquad\qquad t_> \text{ s.t. } \delta v(t_>) = \pi - \Delta \, .
\ee
Schematically, we have:
\begin{center}
\begin{tikzpicture}[scale=2]
\draw (-2.,0) -- (2.,0);
\draw (-2,-.05) -- (-2, .05); \draw (-0.7,-.05) -- (-0.7, .05); \draw (0.7,-.05) -- (0.7, .05); \draw (2,-.05) -- (2, .05);
\node [below] at (-2,0) {$t_\ll$}; \node [below] at (-2,-.3) {$\rho=0$};
\node [below] at (-0.7,0) {$t_<$}; \node [below] at (-0.7,-.3) {$\delta v = - \left( \pi - \Delta \right)$}; \node [below] at (-0.7,-.6) {$Y^{-} = 0$};
\node [below] at (0.7,0) {$t_>$}; \node [below] at (0.7,-.3) {$\delta v = \pi - \Delta$}; \node [below] at (0.7,-.6) {$Y^{+} = 0$};
\node [below] at (2,0) {$t_\gg$}; \node [below] at (2,-.3) {$\rho=0$};
\end{tikzpicture}
\end{center}
The transition points are at
\be
\label{solution sum 2pi -- init}
t_\ll = -\frac{\mu }{2 k (\pi -\Delta )} \, ,\qquad t_< = - \frac{\mu }{k (2 \pi -\Delta )} \, ,\qquad t_> = \frac{\mu }{k (2 \pi -\Delta )} \, ,\qquad t_\gg = \frac{\mu }{2 k (\pi -\Delta )} \, .
\ee
In the left tail we have
\be
\begin{aligned}
\rho & = \frac{1}{2 \Delta^2} \left(\frac{\mu }{\pi -\Delta }+2 k t \right) \, , \qquad\quad
\delta v= - \left( \pi - \Delta \right) \, , \\[.5em]
Y^{-} & = -\frac{ \mu + k (2 \pi - \Delta ) t }{\Delta } \, ,
\end{aligned}
\qquad\qquad t_\ll < t < t_< \, .
\ee
In the inner interval we have
\be
\begin{aligned}
\rho & = \frac{\mu }{2 (\pi -\Delta ) (2 \pi -\Delta ) \Delta } \, ,\qquad\quad
\delta v = \frac{k (\pi -\Delta ) (2 \pi -\Delta ) t}{\mu } \, ,
\end{aligned}
\qquad\qquad t_< < t < t_> \, ,
\ee
and $\delta v'>0$. In the right tail we have
\be
\begin{aligned}
\rho & = \frac{1}{2 \Delta^2} \left( \frac{\mu }{\pi -\Delta }-2 k t \right) \, ,\qquad\quad
\delta v = \pi - \Delta \, , \\[.5em]
Y^{+} & = -\frac{\mu - k (2 \pi - \Delta ) t}{\Delta } \, ,
\end{aligned}
\qquad\qquad t_> < t < t_\gg \, .
\ee
Finally, the normalization fixes
\be
\label{solution sum 2pi -- end}
\mu = 2  k^{1/2} (\pi -\Delta ) (2 \pi -\Delta ) \sqrt{\frac{\Delta}{4 \pi -3 \Delta }} \, .
\ee
$\mu >0$ implies the following inequality
\begin{equation}
 0 < \Delta < \pi \, .
\end{equation}
For $k>1$  there can be discrete $\mathbb{Z}_k$ identifications among the chemical potential which can affect the final result.
We have not been too careful about them here.

\subsubsection[The index at large \texorpdfstring{$N$}{N}]{The index at large $\fakebold{N}$}

The rules of the large $N$ twisted index imply that the free energy functional is
\begin{equation}\label{Z large N functional}
  \begin{aligned}
  \mathfrak{F} &= - N^{3/2} \int \rd t \, \rho(t)^2 \left[ \Delta  (4 \pi -3 \Delta )+\fn \left(3 \Delta^2-6 \pi  \Delta + 2 \pi^2 - 2 \delta v^2 \right) \right] \quad \\
  \rule[-2em]{0pt}{1em} &\quad - N^{3/2} 2 ( 1- \fn) \int_{\delta v \,\approx\, - (\pi - \Delta)} \hspace{-2em} \rd t \, \rho(t) \, Y^{-}(t)
  - N^{3/2} 2 ( 1- \fn) \int_{\delta v \,\approx\, (\pi - \Delta)} \hspace{-2em} \rd t \, \rho(t) \, Y^{+}(t) \, .
 \end{aligned}
\end{equation}
We should take the solution to the BAEs, plug it back into the functional \eqref{Z large N functional} and compute the integral.
Doing so, we obtain the following expression for the topological free energy:
\be
\begin{aligned}
 \mathfrak{F} & = -\frac43 \frac{k^{1/2} N^{3/2} \left[\Delta \left(7 \Delta^2-18 \pi  \Delta +12 \pi^2\right) +
 \fn \left(-6 \Delta^3+19 \pi \Delta^2-18 \pi^2 \Delta +4 \pi^3\right) \right]}{(4 \pi -3 \Delta )^{3/2} \sqrt{\Delta }}
 \, .
\end{aligned}
\ee

\subsection[The \texorpdfstring{$V^{5,2}/\BZ_k$}{V**(5,2)/Z[k]} theory]{The $\fakebold{V^{5,2}/\BZ_k}$ theory}
\label{sec:V52}
In this subsection, we focus on field theories dual to ${\rm AdS}_4 \times V^{5,2}/\BZ_k$, where $V^{5,2}$ is a homogeneous Sasaki-Einstein seven-manifold known as a Stiefel manifold.   The latter can be described as the coset $V^{5,2} =\SO(5)/\SO(3)$, whose supergravity solution \cite{Fabbri:1999hw} possesses an $\SO(5) \times \U(1)_R$ isometry.  There are two known descriptions of such field theories; one proposed by Martelli and Sparks \cite{Martelli:2009ga} and the other proposed by Jafferis \cite{Jafferis:2009th}.  In the following, we refer to these theories as Model I and Model II, respectively.  Below we analyse the solutions to the BAEs in detail and show the equality between the topological free energy of two theories.

\subsubsection{Model I}
The description for Model I was first presented in \cite{Martelli:2009ga}.  The quiver diagram is depicted below.
\be
\begin{aligned}
\begin{tikzpicture}[baseline, font=\footnotesize, scale=0.8]
\begin{scope}[auto,%
  every node/.style={draw, minimum size=0.5cm}, node distance=2cm];
\node[circle] (USp2k) at (-0.1, 0) {$N_{+k}$};
\node[circle, right=of USp2k] (BN)  {$N_{-k}$};
\end{scope}
\draw[draw=blue,solid,line width=0.2mm,<-]  (USp2k) to[bend right=15] node[midway,above] {$B_2 $}node[midway,above] {}  (BN) ;
\draw[draw=blue,solid,line width=0.2mm,->]  (USp2k) to[bend right=50] node[midway,above] {$A_1$}node[midway,above] {}  (BN) ; 
\draw[draw=red,solid,line width=0.2mm,<-]  (USp2k) to[bend left=15] node[midway,above] {$B_1$} node[midway,above] {} (BN) ;  
\draw[draw=red,solid,line width=0.2mm,->]  (USp2k) to[bend left=50] node[midway,above] {$A_2$} node[midway,above] {} (BN) ;    
\draw[black,-> ] (USp2k) edge [out={-150},in={150},loop,looseness=10] (USp2k) node at (-2,1) {$\phi_1$} ;
\draw[black,-> ] (BN) edge [out={-30},in={30},loop,looseness=10] (BN) node at (5.8,1) {$\phi_2$};
\end{tikzpicture}
\end{aligned} \ee
with the superpotential
\begin{equation}
 W = \Tr\left[ \phi_1^3 + \phi_2^3 +\phi_1(A_1 B_2 + A_2 B_1) + \phi_2 (B_2 A_1+ B_1 A_2) \right] \, .
\end{equation}

\paragraph*{A solution to the BAEs.} Let us use the shorthand notation as in \eref{shorthandAB}. We look for a solution to BAEs, such that
\begin{align}\label{Delta constraint}
 \Delta_{\phi_i} + \Delta_{1} + \Delta_{4} =  2 \pi \, ,\qquad \qquad
 \Delta_{\phi_i} + \Delta_{2} + \Delta_{3} = 2 \pi \, , \qquad \qquad
 \Delta_{\phi_i} = \frac{2 \pi}{3} \, ,
\end{align}
and
\begin{align}\label{fluxconstraintV52}
 \fn_{\phi_i} + \fn_{1} + \fn_4 = 2\, ,\qquad \qquad
 \fn_{\phi_i} + \fn_{2} + \fn_3 = 2\, , \qquad \qquad
 \fn_{\phi_i} = \frac23 \, .
\end{align}

Observe that $\fn_{\phi_i}$ does not satisfy the quantization conditions $\fn_{\phi_i}\in \BZ$.
However, this problem can be cured easily by considering the twisted partition function
on a Riemann surface $\Sigma_\fg$ of genus $\fg$ times $S^1$ \cite{Benini:2016hjo}.
In this case, the flux constraints become
\begin{align}
 \fn_{\phi_i} + \fn_{1} + \fn_4 = 2(1-\fg)\, ,\qquad
 \fn_{\phi_i} + \fn_{2} + \fn_3 = 2(1-\fg)\, , \qquad
 \fn_{\phi_i} = \frac23 (1-\fg) \, .
\end{align}
By choosing $(1-\fg)$ to be an integer multiple of $3$, there always exists an integer solution to the above constraints.
As was pointed out in \cite{Benini:2016hjo}, the BAEs for the partition function
on $\Sigma_\fg \times S^1$ (with $\fg>1$) is the same as that for $\fg=0$.
We can therefore solve the BAEs in the usual way.
 
The inner interval $[t_<, t_>]$ is given by
\be
t_< \text{ s.t. } \delta v(t_<) = - \Delta_3 \, ,\qquad\qquad t_> \text{ s.t. } \delta v(t_>) = \Delta_1 \, .
\ee
Outside the inner interval, we find that $\delta v(t) =  \tilde v(t) - v(t)$ is frozen to the constant boundary value $- \Delta_3$ ($ \Delta_1$) and it defines the left (right) tail.
Schematically:
\begin{center}
\begin{tikzpicture}[scale=2]
\draw (-1.5,0) -- (1.5,0);
\draw (-1.5,-.05) -- (-1.5, .05); \draw (-0.5,-.05) -- (-0.5, .05); \draw (0.5,-.05) -- (0.5, .05); \draw (1.5,-.05) -- (1.5, .05);
\node [below] at (-1.5,0) {$t_\ll$}; \node [below] at (-1.5,-.3) {$\rho=0$};
\node [below] at (-0.5,0) {$t_<$}; \node [below] at (-0.5,-.3) {$\delta v = -\Delta_3$};
\node [below] at (-0.5,-.6) {$Y_3 = 0$};
\node [below] at (0.5,0) {$t_>$}; \node [below] at (0.5,-.3) {$\delta v=\Delta_1$};
\node [below] at (0.5,-.6) {$Y_1 = 0$};
\node [below] at (1.5,0) {$t_\gg$}; \node [below] at (1.5,-.3) {$\rho=0$};
\end{tikzpicture}
\end{center}
The solution is as follows. The transition points are at
\be
\label{solution sum 2pi -- init -- V52}
t_\ll = - \frac{\mu}{k \Delta_3} \,,\qquad\quad t_< = - \frac{\mu}{k \Delta_4} \,,\qquad\quad t_> = \frac{\mu }{k \left( \frac{4 \pi}{3} - \Delta_3 \right)} \,,\qquad\quad t_\gg = \frac{\mu }{k \left( \frac{4 \pi}{3} - \Delta_4 \right)} \,.
\ee
In the left tail we have
\be
\begin{aligned}
\rho &= \frac{\mu +k \Delta_3 t}{\frac{2 \pi}{3} \left(\Delta_3-\Delta_4+\frac{4 \pi }{3}\right) \left(\Delta_4-\Delta_3\right)} \, , \\[.5em]
\delta v &= - \Delta_3 \,,\qquad\qquad Y_3 = \frac{- k t\Delta_4 -\mu }{\Delta_4 - \Delta_3} \, ,
\end{aligned}
\qquad\qquad\qquad t_\ll < t < t_< \, .
\ee
In the inner interval we have
\be
\begin{aligned}
\rho &= \frac{2 \mu + k \left(\Delta_3+\Delta_4-\frac{4 \pi}{3}\right) t}{\frac{2 \pi }{3} \left[\left(\frac{4 \pi }{3}\right)^2-\left(\Delta_4 - \Delta_3\right)^2 \right]} \, , \\[.5em]
\delta v &= -\frac{\left(\Delta_3+\Delta_4-\frac{4 \pi }{3}\right) \mu -  \frac{4 \pi}{3}k  \left(\Delta_3+\Delta_4\right) t+\left(\Delta_3^2+\Delta_4^2\right) t}{2 \mu +k \left(\Delta_3+\Delta_4-\frac{4 \pi }{3}\right) t} \, ,
\end{aligned}
\qquad\qquad t_< < t < t_>
\ee
and $\delta v'>0$.
In the right tail we have
\be
\begin{aligned}
\rho &= \frac{\mu - k \Delta_1 t}{\frac{2 \pi}{3} \left(\Delta_3-\Delta_4+\frac{4 \pi }{3}\right) \left(\Delta_4-\Delta_3\right)} \, , \\[.5em]
\delta v &= \Delta_1 \,,\qquad\qquad Y_1 = \frac{- k t \left( \Delta_3 - \frac{4 \pi}{3}\right) - \mu}{\Delta_4 - \Delta_3} \, ,
\end{aligned}
\qquad\qquad\qquad t_> < t < t_\gg \, .
\ee
Finally, the normalization fixes
\be
\label{solution sum 2pi -- end -- V52}
\mu = \sqrt{k \left(\frac{4 \pi}{3}-\Delta_3\right) \Delta_3 \left(\frac{4 \pi }{3}-\Delta_4\right) \Delta_4} \, ,
\ee
with
\begin{equation}\label{Delta inequality -- V52}
 0 < \Delta_{3,4} < \frac{4 \pi}{3} \, .
\end{equation}
The solution satisfies
\begin{equation}
 \int \rd t \, \rho(t) \, \delta v(t) = 0 \, .
\end{equation}

We should take the solution to the BAEs and plug it back into the index. The higher genus index in the large $N$ limit
receives a simple modification, as discussed in \cite{Benini:2016hjo}, as follows,
\begin{align}
 \label{genentropyhigherg}
 \mathfrak{F}_{\fg \neq 1} (\fn_I) = (1 - \fg) \mathfrak{F}_{\fg=0} (\fn_I / (1-\fg)) \, .
\end{align}

We thus obtain the following expression for the topological free energy
\be
\begin{aligned} \label{entropy bi-fundamental}
\mathfrak{F}_{\fg \neq 1} & = -\frac{2}{3}(1-\fg) \frac{k^{1/2} N^{3/2}}{\sqrt{\left(\frac{4 \pi}{3} - \Delta_3\right) \Delta_3 \left(\frac{4 \pi}{3} - \Delta_4\right) \Delta_4}}
 \Bigg\{\left(\frac{4 \pi}{3}-\Delta_3\right) \Delta_3 \left(\frac{2 \pi}{3} - \Delta_4\right) \frac{\fn_4}{1-\fg} \\
 & + \Delta_4 \left[\left(\frac{2 \pi}{3}-\Delta_3\right) \left(\frac{4 \pi }{3} - \Delta_4\right) \frac{\fn_3}{1-\fg} - \frac{2 \Delta_3}{3} \left(\Delta_3+\Delta_4-\frac{8 \pi}{3}\right)\right]\Bigg\} \, .
\end{aligned} \ee
We check that the topological free energy indeed satisfies the index theorem for this model on $\Sigma_\fg \times S^1$:
\begin{equation} \label{indexhigherg}
 \mathfrak{F}_{\fg \neq 1} = (1-\fg) \left\{ - \frac{2}{\pi} \, \wt\cW(\Delta_I) \,
 - \sum_{I}\, \left[ \left(\frac{\fn_I}{1-\fg} - \frac{\Delta_I}{\pi}\right) \frac{\partial \wt\cW(\Delta_I)}{\partial \Delta_I}
  \right] \right\} \,,
\end{equation}
with
\begin{equation}
 \wt\cW(\Delta_I) = \frac{2}{3} \mu N^{3/2} \, .
\end{equation}

\subsubsection{Model II}
The description for Model II was first presented in \cite{Jafferis:2009th}. The quiver diagram is depicted below.
\be
\begin{aligned}
\begin{tikzpicture}[font=\footnotesize, scale=0.9]
\begin{scope}[auto,%
  every node/.style={draw, minimum size=0.5cm}, node distance=2cm];
\node[circle]  (UN)  at (0.3,1.7) {$N$};
\node[rectangle, right=of UN] (Ur) {$k$};
\end{scope}
\draw[decoration={markings, mark=at position 0.45 with {\arrow[scale=1.5]{>}}, mark=at position 0.5 with {\arrow[scale=1.5]{>}}, mark=at position 0.55 with {\arrow[scale=1.5]{>}}}, postaction={decorate}, shorten >=0.7pt] (-0,2) arc (30:343:0.75cm);
\draw[draw=black,solid,line width=0.2mm,->]  (UN) to[bend right=30] node[midway,below] {$Q$}node[midway,above] {}  (Ur) ; 
\draw[draw=black,solid,line width=0.2mm,<-]  (UN) to[bend left=30] node[midway,above] {$\tQ$} node[midway,above] {} (Ur) ;    
\node at (-2.2,1.7) {$\varphi_{1,2,3}$};
\end{tikzpicture}
\end{aligned} \ee
We start from the superpotential
\begin{equation}\label{superpotential fundamental}
 W = \Tr \left\{ \varphi_3 \left[ \varphi_1, \varphi_2 \right] + \sum_{j=1}^{k} q_j \left( \varphi_1^2 + \varphi_2^2 + \varphi_3^2 \right) \tilde q^j \right\} \, .
\end{equation}
The $\SO(5)$ symmetry of $V^{5,2}$ can be made manifest by using the following variables \cite{Cremonesi:2016nbo}:%
\footnote{Explicitly, the generators of the chiral ring in the vector representation of $\SO(5)$ are $X_{1,2,3}$ and $V_{\pm}$,
where $V_{\pm}$ are the monopole operators of magnetic charge $\pm 1$.}
\be
\begin{aligned}
X_1 =  \frac{1}{\sqrt{2}}(\varphi_1 + i \varphi_2)\, , \qquad X_2 = \frac{1}{\sqrt{2}}(\varphi_1 - i \varphi_2)\, , \qquad X_3 = i \varphi_3 \, .
\end{aligned} \ee
In terms of these new variables, the superpotential can be rewritten as
\be
\begin{aligned}
W = \Tr \left\{ X_3  [X_1 , X_2] +   \sum_{j=1}^k q_j ( X_1 X_2+ X_2 X_1 - X_3^2) \tilde q^j \right\}\, .
\end{aligned} \ee

\paragraph*{A solution to the system of BAEs.} The superpotential enforces
\begin{equation}
 \label{Delta constraint complex}
 \Delta_{X_1} + \Delta_{X_2} = \frac{4 \pi}{3}\, ,\qquad \qquad
 \Delta_{q_j} + \tilde \Delta_{q_j} = \frac{2 \pi}{3}\, ,\qquad \qquad
 \Delta_{X_3} = \frac{2 \pi}{3} \, ,
\end{equation}
and
\begin{equation}\label{flux constraint complex0}
 \fn_{X_1} + \fn_{X_2} = \frac{4}{3}\, ,\qquad \qquad
 \fn_{q_j} + \tilde \fn_{q_j} = \frac{2}{3}\, ,\qquad \qquad
 \fn_{X_3} = \frac{2}{3} \, .
\end{equation}

As in the previous subsection, the quantization conditions $\fn_I \in \BZ$ can be satisfied by considering the
twisted partition function on $\Sigma_\fg \times S^1$.  The flux constraints are modified to be
\begin{equation}\label{flux constraint complex}
 \fn_{X_1} + \fn_{X_2} = \frac{4}{3}(1-\fg) \, ,\qquad
 \fn_{q_j} + \tilde \fn_{q_j} = \frac{2}{3}(1-\fg) \, ,\qquad 
 \fn_{X_3} = \frac{2}{3}(1-\fg) \, .
\end{equation}
Here we choose $(1-\fg)$ to be an integer multiple of $3$. The solution to the BAEs are given below.

Setting to zero the variations with respect to $\rho(t)$, we find that the density is given by
\begin{equation}\label{density tm tp complex}
 \rho(t) = \frac{ \mu -\frac{2 \pi  k}{3} \left| t\right| + t \Delta_m}{\frac{2 \pi}{3} \left(\frac{4 \pi}{3} - \Delta_{X_1}\right) \Delta_{X_1}} \, .
\end{equation}
The support $[t_- , t_+]$ of $\rho(t)$ is determined by $\rho(t_\pm)=0$.  We obtain
\be
\begin{aligned}
 t_{\pm} =  \pm \frac{\mu }{\frac{2 \pi k}{3} \pm \Delta_m}\, .
\end{aligned} \ee
Requiring that $\int_{t_{-}}^{t_{+}} \, \rd t \, \rho(t) =1$, we have
\be
\begin{aligned}
 \mu = \sqrt{\frac{1}{k}\left(\frac{4 \pi}{3} - \Delta_{X_1}\right) \Delta_{X_1}
 \left[\left(\frac{2 \pi k}{3}\right)^2 - \Delta _m^2\right]} \, .
\end{aligned} \ee

The topological free energy may then be found using \eref{genentropyhigherg}.  We obtain
\be
\begin{aligned} \label{entropy complex} 
\mathfrak{F}_{\fg \neq 1} & =  \frac{2}{3} (1-\fg) \frac{N^{3/2}}{\sqrt{k \left(\frac{4 \pi}{3}
-\Delta_{X_1}\right) \Delta_{X_1} \left[\left(\frac{2 \pi k}{3}\right)^2-\Delta_m^2\right]}} \times \\
& \qquad \Bigg\{\Delta_{X_1} \left[- \frac{\ft}{1-\fg} \Delta_m \left(\frac{4 \pi}{3}-\Delta_{X_1}\right)
+ \left( \frac{2 \pi k}{3} \right)^{2} \left(\frac{\Delta_{X_1}}{\pi} - 2\right)+\frac{2 \Delta_m^2}{3}\right] \\
& \qquad - \left(\frac{2 \pi}{3}-\Delta_{X_1} \right) \frac{{\fn}_{X_1}}{1-\fg} \left[\left(\frac{2 \pi k}{3}\right)^2-\Delta_m^2\right] \Bigg\} \, .
\end{aligned} \ee
It can also be checked that this topological free energy satisfies \eref{indexhigherg}.

\paragraph*{Matching with Model I.} By taking
\begin{equation} \label{matchDeltaIandII}
\Delta_{X_1} = \Delta_3 \, ,  \qquad \Delta_m = k \left(\frac{2 \pi }{3} -\Delta _4 \right) \, ,  \qquad
\fn_{X_1} = \fn_3 \, , \qquad  \ft = k \left[\frac{2}{3}(1-\fg) -\fn _4 \right] \, ,
\end{equation}
we see that Eq.\,\eqref{entropy complex} reduces to Eq.\,\eqref{entropy bi-fundamental}.

\subsection{The flavored ABJM theory}
Let us consider the flavored ABJM models studied in \cite{Benini:2009qs,Cremonesi:2010ae}
\be
\begin{aligned}
\begin{tikzpicture}[baseline, font=\scriptsize, scale=0.8]
\begin{scope}[auto,%
  every node/.style={draw, minimum size=0.5cm}, node distance=4cm];
\node[circle] (UN1) at (0, 0) {$N_{+k}$};
\node[circle, right=of UN1] (UN2)  {$N_{-k}$};
\node[rectangle] at (3.2,2.2) (UNa1)  {$n_{a1}$};
\node[rectangle] at (3.2,3.5) (UNa2)  {$n_{a2}$};
\node[rectangle] at (3.2,-2.2) (UNb1)  {$n_{b1}$};
\node[rectangle] at (3.2,-3.5) (UNb2)  {$n_{b2}$};
\end{scope}
\draw[draw=red,solid,line width=0.2mm,<-]  (UN1) to[bend right=30] node[midway,above] {$B_2 $}node[midway,above] {}  (UN2) ;
\draw[draw=blue,solid,line width=0.2mm,->]  (UN1) to[bend right=-10] node[midway,above] {$A_1$}node[midway,above] {}  (UN2) ; 
\draw[draw=purple,solid,line width=0.2mm,<-]  (UN1) to[bend left=-10] node[midway,above] {$B_1$} node[midway,above] {} (UN2) ;  
\draw[draw=black!60!green,solid,line width=0.2mm,->]  (UN1) to[bend left=30] node[midway,above] {$A_2$} node[midway,above] {} (UN2) ;   
\draw[draw=purple,solid,line width=0.2mm,->]  (UN1)  to[bend right=30] node[midway,right] {}   (UNb1);
\draw[draw=purple,solid,line width=0.2mm,->]  (UNb1) to[bend right=30] node[midway,left] {} (UN2) ; 
\draw[draw=red,solid,line width=0.2mm,->]  (UN1)  to[bend right=30]  node[midway,right] {} (UNb2);
\draw[draw=red,solid,line width=0.2mm,->]  (UNb2) to[bend right=30] node[midway,left] {} (UN2); 
\draw[draw=blue,solid,line width=0.2mm,->]  (UN2)  to[bend right=30] node[pos=0.9,right] {}   (UNa1);
\draw[draw=blue,solid,line width=0.2mm,->]  (UNa1) to[bend right=30] node[pos=0.1,left] {} (UN1) ; 
\draw[draw=black!60!green,solid,line width=0.2mm,->]  (UN2)  to[bend right=30]  node[pos=0.9,right] {} (UNa2);
\draw[draw=black!60!green,solid,line width=0.2mm,->]  (UNa2) to[bend right=30] node[pos=0.1,left] {}  (UN1); 
\node at (4.5,-2.4) {$\tQ^{(1)}$};
\node at (2.,-2.4) {$Q^{(1)}$};
\node at (5.5,-2.8) {$\tQ^{(2)}$};
\node at (0.9,-2.9) {$Q^{(2)}$};
\node at (4.5,2.4) {$q^{(1)}$};
\node at (2.,2.4) {$\tilde{q}^{(1)}$};
\node at (5.5,2.8) {$q^{(2)}$};
\node at (1.,2.8) {$\tilde{q}^{(2)}$};
\end{tikzpicture}
\end{aligned} \ee
with the superpotential
\be
\begin{aligned}\label{supflvABJM}
W &= \Tr \left( A_1 B_1 A_2 B_2 - A_1 B_2 A_2 B_1  \right) + \\
& \quad \Tr\left[\sum_{j=1}^{n_{a1}} q_{j}^{(1)} A_1 \tilde q_{j}^{(1)}
 +\sum_{j=1}^{n_{a2}} q_{j}^{(2)} A_2 \tilde q_{j}^{(2)}
 +\sum_{j=1}^{n_{b1}} Q_{j}^{(1)} B_1 \tilde Q_{j}^{(1)}
 +\sum_{j=1}^{n_{b2}} Q_{j}^{(2)} B_2 \tilde Q_{j}^{(2)}\right]\, .
\end{aligned} \ee

We adopt the notation as in \eref{shorthandAB} and denote by
\be
\begin{aligned}
\Delta_{ai} = \Delta_{q^{(i)}}~, \qquad \tilde{\Delta}_{ai} = \Delta_{\tilde{q}^{(i)}}~, \qquad {\Delta}_{bi} = \Delta_{{Q}^{(i)}}~, \qquad \tilde{\Delta}_{bi} = \Delta_{\tilde{Q}^{(i)}}~,
\end{aligned} \ee
and similarly for $\fn_{ai}$ and $\fn_{bi}$.
The marginality of the superpotential implies that
\begin{align}
 &\Delta_1 + \Delta_{a1} + \tilde \Delta_{a1}= 2\pi \, ,\qquad\qquad
 \Delta_2 + \Delta_{a2} + \tilde \Delta_{a2}= 2\pi \, ,\nn\\&
 \Delta_3 + \Delta_{b1} + \tilde \Delta_{b1}= 2\pi \, ,\qquad\qquad
 \Delta_4 + \Delta_{b2} + \tilde \Delta_{b2}= 2\pi \, ,
\end{align}
and
\begin{align}\label{flux constraint}
 &\fn_1 + \fn_{a1} + \tilde \fn_{a1}=2\, ,\qquad\qquad
 \fn_2 + \fn_{a2} + \tilde \fn_{a2}=2\, ,\nn\\&
 \fn_2 + \fn_{b1} + \tilde \fn_{b1}=2\, ,\qquad\qquad
 \fn_4 + \fn_{b2} + \tilde \fn_{b2}=2\, .
\end{align}

\subsubsection{A solution to the system of BAEs}

The large $N$ expression for the twisted superpotential, using the rules given in section \ref{large N twisted superpotential rules}, can be written as
\begin{align}\label{large N twisted superpotential -- Q111}
\frac{\wt\cW}{i N^{3/2}}&
=\int \rd t \, \rho(t)^2\, \sum\nolimits^*_a\left[\pm g_{\pm} \left(\delta v(t) \pm \Delta_a\right)\right]
+\int \rd t \, t\, \rho(t)\, \left(\Delta_m^{(2)} - \Delta_m^{(1)} \right)\nn\\&
-\frac{1}{2}\int \rd t \, |t|\, \rho(t)\,\left[\sum\nolimits^*_f (\pm n_f)\delta v(t)- \sum_{i=1}^{2}\left(n_{ai}\Delta_{i} + n_{bi}\Delta_{i+2}\right)\right] \nn \\
& -\frac{i}{N^{1/2}}\int \rd t \, \rho(t)\, \sum\nolimits^*_a\left[\pm \Li_2 \left(\e^{i\left(\delta v(t)\pm \Delta_a\right)}\right)\right]
-\mu \left(\int \rd t \, \rho(t)-1\right) \, ,
\end{align}
where we introduced the notations
\begin{equation}
 \sum\nolimits^*_f=\sum_{\substack{f=a1,a2:+\\ f=b1,b2:-}}\, ,\qquad \qquad \sum\nolimits^*_a=\sum_{\substack{a=3,4:+\\ a=1,2:-}}\, .
\end{equation}

\paragraph*{The solution for $\fakebold{k=0}$ and $\fakebold{n_{a1}=n_{a2}=n\, ,\; n_{b1}=n_{b2}=0}$.}
As pointed out in \cite{Benini:2009qs}, this theory is dual to $\mathrm{AdS}_4 \times Q^{1,1,1}/\mathbb{Z}_n$.
The manifold $Q^{1,1,1}$ is defined by the coset
\be
\frac{\SU(2) \times \SU(2) \times \SU(2)}{\U(1) \times \U(1)} \, ,
\ee
and has the isometry
\be
\SU(2) \times \SU(2) \times \SU(2) \times \U(1) \, .
\ee
Using the symmetries of the quiver, we set for simplicity
\begin{align}
 \Delta_1=\Delta_2=\pi-\Delta_3=\pi-\Delta_4=\Delta\, .
\end{align}

Let $\Delta_m$ be the following linear combination of the topological chemical potentials of the two gauge groups:
\be
\begin{aligned} 
\Delta_m = \Delta_m^{(1)} - \Delta_m^{(2)}\, .
\end{aligned} \ee
Solving the BAEs, we obtain the following general solution
\be
\begin{aligned}
 \rho(t) & = -\frac{n \pi \left| t\right| + 2 \Delta_m\, t - 2 \mu}{\pi^3} \, , \\
 \delta v(t) & = \Delta + \frac{\pi \left(\mu - \Delta_m\, t \right)}{n \pi \left| t\right| + 2 \Delta_m\, t - 2 \mu}\, ,
\end{aligned} \ee
on the support $[t_- , t_+]$.  We determine $t_\pm$ from $\delta v (t_\pm)=-(\pi-\Delta)$:
\begin{equation}
 t_- = -\frac{\mu }{n \pi - \Delta_m}\, , \qquad \qquad t_+ = \frac{\mu }{n \pi + \Delta_m} \, .
\end{equation}
The normalization $\int_{t_-}^{t_+} \rd t \, \rho(t)=1$ fixes
\begin{equation}
 \mu = \frac{\pi}{\sqrt{n}} \frac{\left| n^2 \pi^2 - \Delta_m^2\right|}{\sqrt{3 n^2 \pi^2-\Delta_m^2}} \, .
\end{equation}
The solution satisfies,
\begin{equation}
 \int \rd t \, \rho(t)\, \delta v(t) = \Delta - \frac{2 n^2 \pi^3}{3 n^2 \pi^2 - \Delta_m^2} \, .
\end{equation}

\subsubsection[The index at large \texorpdfstring{$N$}{N}]{The index at large $\fakebold{N}$}

Given the rules in section \ref{large N index rules}, the topological free energy functional for this model reads
\bea
 \frac{\mathfrak{F}}{N^{3/2}} & = - \int \rd t \, \rho(t)^2 \bigg[ \frac{2\pi^2}3 + \sum\nolimits^*_{a} (\fn_a-1) g_\pm'\big( \delta v(t) \pm \Delta_a \big) \bigg] \\
  & - \frac{1}{2} \sum_{i=1}^2 \left(n_{ai} \fn_{i} + n_{bi} \fn_{i+2}\right) \int \rd t \, |t|\, \rho(t) - \left( \ft + \tilde \ft \right) \int \rd t \, t\, \rho(t) \\
  & - \sum_{a=1}^4 \fn_a \int_{\delta v \,\approx\, \varepsilon_a \Delta_a} \hspace{-2em} \rd t \, \rho(t) \, Y_a(t) \, ,
\eea 
where we have used the behavior
\be
\delta v(t) = \varepsilon_a \left( \Delta_a - \e^{- N^{1/2} Y_a(t)} \right)\, , \qquad\qquad \varepsilon_a = (1,1,-1,-1) \, ,
\ee
in the tails.
For the theory dual to AdS$_4 \times Q^{1,1,1}/\mathbb{Z}_n$ we find
\be
\begin{aligned}
\mathfrak{F} = -\frac23 \frac{N^{3/2}}{\sqrt{n} \left(3 \pi ^2 n^2 - \Delta_m^2\right)^{3/2}} \left[\pi \left(\ft + \tilde\ft \right) \left(\Delta_m^3-5 \pi^2 \Delta_m n^2\right) + \Delta_m^4 - 3 \pi ^2 n^2 \left(\Delta_m^2-2 \pi ^2 n^2\right)\right]\, .
\end{aligned} \ee

\subsection[\texorpdfstring{$\U(N)$}{U(N)} gauge theory with adjoints and fundamentals]{$\fakebold{\U(N)}$ gauge theory with adjoints and fundamentals}

In this section, we consider the following flavored toric quiver gauge theory \cite{Benini:2009qs}
\be
\begin{aligned}
\begin{tikzpicture}[baseline, font=\footnotesize, scale=0.9]
\begin{scope}[auto,%
  every node/.style={draw, minimum size=0.5cm}, node distance=2cm];
\node[circle]  (UN)  at (0.3,1.7) {$N$};
\node[rectangle, right=of UN] (Ur1) {$n_1$};
\node[rectangle, below=of UN] (Ur2) {$n_2$};
\node[rectangle, above=of UN] (Ur3) {$n_3$};
\end{scope}
\draw[decoration={markings, mark=at position 0.45 with {\arrow[scale=1.5]{>}}, mark=at position 0.5 with {\arrow[scale=1.5]{>}}, mark=at position 0.55 with {\arrow[scale=1.5]{>}}}, postaction={decorate}, shorten >=0.7pt] (-0,2) arc (30:340:0.75cm);
\draw[draw=black,solid,line width=0.2mm,->]  (UN) to[bend right=30] node[midway,below] {$q^{(1)}$}node[midway,above] {}  (Ur1) ; 
\draw[draw=black,solid,line width=0.2mm,<-]  (UN) to[bend left=30] node[midway,above] {$\tilde{q}^{(1)}$} node[midway,above] {} (Ur1) ;    
\draw[draw=black,solid,line width=0.2mm,->]  (UN) to[bend right=30] node[midway,left] {$q^{(2)}$}node[midway,above] {}  (Ur2) ; 
\draw[draw=black,solid,line width=0.2mm,<-]  (UN) to[bend left=30] node[midway,right] {$\tilde{q}^{(2)}$} node[midway,above] {} (Ur2) ;  
\draw[draw=black,solid,line width=0.2mm,->]  (UN) to[bend left=30] node[midway,left] {$q^{(3)}$}node[midway,above] {}  (Ur3) ; 
\draw[draw=black,solid,line width=0.2mm,<-]  (UN) to[bend right=30] node[midway,right] {$\tilde{q}^{(3)}$} node[midway,above] {} (Ur3) ;    
\node at (-2.2,1.7) {$\phi_{1,2,3}$};
\end{tikzpicture}
\end{aligned} \ee
with the superpotential
\begin{equation}\label{superpotentialU(n)}
 W = \Tr \left\{ \phi_1\left[\phi_2,\phi_3\right]
 + \sum_{j=1}^{n_1} q_j^{(1)} \phi_1 \tilde q_{j}^{(1)}
 + \sum_{j=1}^{n_2} q_j^{(2)} \phi_2 \tilde q_{j}^{(2)}
 + \sum_{j=1}^{n_3} q_j^{(3)} \phi_3 \tilde q_{j}^{(3)}\right\}\, .
\end{equation}

The marginality condition on the superpotential \eqref{superpotentialU(n)} implies that
\be
 \label{Delta constraintU(n)}
 \sum_{i=1}^{3} \Delta_{\phi_i} = 2 \pi\, ,  \qquad \Delta_{q_j^{(i)}} +\tilde \Delta_{q_j^{(i)}} + \Delta_{\phi_i} = 2 \pi \, ,
\ee
and
\be
 \label{flux constraintU(n)}
 \sum_{i=1}^3 \fn_{\phi_i} = 2\, ,\qquad\qquad  \fn_{q_j^{(i)}} + \tilde \fn_{q_j^{(i)}} + \fn_{\phi_i} = 2 \, .
\ee
Let $\Delta_m$ and $\ft$ be the chemical potential and the background flux for the topological symmetry associated with the $\U(N)$ gauge group.

\paragraph*{The index at large $\fakebold{N}$.} On the support of $\rho(t)$, the solution is
\begin{equation}\label{density tm tp}
 \rho(t) = \frac{2 \left(\mu +t\, \Delta_m \right)-\left| t\right|  \bar{\Delta}}{2\hat\Delta } \, ,
\end{equation}
where we defined
\be
\begin{aligned}
 \hat\Delta = \prod_{f=1}^{3} \Delta_{\phi_f}\, , \qquad \bar{\Delta} = \sum_{f=1}^{3} n_f \Delta_{\phi_f}\, .
\end{aligned} \ee
Let us denote by $[t_- , t_+]$ the support of $\rho(t)$.  We determine $t_\pm$ from the condition $\rho(t_\pm)=0$ and obtain
\be
\begin{aligned}
t_\pm = \pm \frac{2\mu}{\bar{\Delta} \mp 2 \Delta_m}\, .
\end{aligned} \ee
The normalization $\int_{t_-}^{t_+} \rd t \, \rho(t) =1$ fixes the Lagrange multiplier $\mu$,
\begin{align}
 \mu & = \sqrt{ \frac{\hat\Delta}{2 \bar{\Delta}} \left( \bar{\Delta} - 2 \Delta_m\right) \left( \bar{\Delta} + 2 \Delta_m\right)} \, .
\end{align}
Using the same methods presented earlier, we obtain the following expression for the topological free energy,
\begin{align}\label{freeU(N)}
 \mathfrak{F} & = -\frac{N^{3/2}}{3}  \sqrt{ \frac{\hat\Delta}{2 \bar{\Delta}} \left( \bar{\Delta} - 2 \Delta_m\right) \left( \bar{\Delta} + 2 \Delta_m\right)} \Bigg[ \hat{\fn} +\frac{\bar{\fn} \left(\bar{\Delta}^2 + 4 \Delta_m^2\right)}{\bar{\Delta} \left(\bar{\Delta}^2-4 \Delta _m^2\right)} -\frac{8 \Delta_m}{\bar{\Delta}^2-4 \Delta _m^2}  \Bigg] \, ,
\end{align}
where
\be
\begin{aligned}
\hat{\fn} = \sum_{i=1}^3 \frac{\fn_{\phi_i}}{\Delta_{\phi_i}}\, , \qquad \bar{\fn} = \sum_{i=1}^3 n_i \fn_{\phi_i}\, .
\end{aligned} \ee
When $n_1 = n_2 = 0$, and $n_3 = r$, the moduli space reduces to $\BC^2 \times \BC^2 / \BZ_r$ and
Eq.\,\eqref{freeU(N)} becomes the topological free energy of the ADHM quiver [see Eq.\,\eqref{freeADHM}].
This is consistent with the fact that this theory is dual to AdS$_4 \times S^7/\BZ_r$.

%

\chapter[Counting microstates of AdS\texorpdfstring{$_{4}$}{(4)} black holes in massive type IIA supergravity]{Counting microstates of AdS$\bm{_4}$ black holes in massive type IIA supergravity}
\label{ch:4}

\ifpdf
    \graphicspath{{Chapter4/Figs/Raster/}{Chapter4/Figs/PDF/}{Chapter4/Figs/}}
\else
    \graphicspath{{Chapter4/Figs/Vector/}{Chapter4/Figs/}}
\fi

\section{Introduction}
\label{mIIA:Introduction}

Extending the results of \cite{Benini:2015eyy,Benini:2016rke}, the large $N$ limit of general three-dimensional
Chern-Simons-matter-gauge theories with an M-theory or a massive type IIA dual was studied in chapters \ref{ch:2} and \ref{ch:3}.
For the special class of $\cN=2$ quiver gauge theories where the Chern-Simons levels
do not sum to zero the index has been shown to scale as $N^{5/3}$ in the large $N$ limit,
in agreement with a dual massive type IIA supergravity construction \cite{Aharony:2010af,Petrini:2009ur,Lust:2009mb,Tomasiello:2010zz,
Guarino:2015jca,Fluder:2015eoa,Pang:2015vna,Pang:2015rwd,Guarino:2016ynd,Guarino:2017eag,Guarino:2017pkw}
(see also \cite{Araujo:2016jlx,Araujo:2017hvi}).

Motivated by the above results, we look at four-dimensional $\cN = 8$ supergravity with a dyonically gauged $\ISO(7) = \SO(7) \ltimes \bR^{7}$ gauge group
that arises as a consistent truncation of massive type IIA supergravity \cite{Romans:1985tz} on a six-sphere \cite{Guarino:2015vca,Cassani:2016ncu}
and its further truncation to an $\cN = 2$ theory with an Abelian gauge group $\mathbb{R} \times \U(1)^3$.
The electric and magnetic gauge couplings $(g,m)$ are identified with the $S^6$ inverse radius and the ten-dimensional Romans mass $\hat{F}_{(0)}$, respectively.
In particular, we analyze the supersymmetry conditions for black holes in AdS$_4 \times S^6$,
with deformed metrics on the $S^6$, in the presence of nontrivial scalar fields.
We mainly focus on the near-horizon geometries which were also recently analyzed in \cite{Guarino:2017pkw}.
For our holographic purposes here we rederive these solutions in a different way and express the scalars and geometric data
in terms of the conserved electromagnetic charges. For the sake of clarity we focus primarily on the case of three magnetic
charges $\fn_j$ $(j=1,2,3)$ (with one constraint relating them) and equal electric charges $q_j=q\, ,\forall j=1,2,3$
with the possibility for different horizon geometries of the form AdS$_2 \times \Sigma_\fg$.

The particular model we analyze corresponds to the $\cN=2$ truncation of the $\cN=8$ theory \cite{Guarino:2015vca}
coupled to three vector multiplets ($n_{\rm V}=3$) and the universal hypermultiplet ($n_{\rm H}=1$) \cite{Guarino:2017pkw}.
We will call this model the \emph{dyonic STU model}. The route that we take to constructing the near-horizon geometries
is based on a supersymmetry preserving version of the Higgs mechanism worked out in \cite{Hristov:2010eu} for the case of
$\cN=2$ gauged supergravity. This allows us to truncate away in a BPS preserving way a full massive vector multiplet
(made from the merging of the massless hypermultiplet and one of the three massless vector multiplets)
that forms after the spontaneous breaking of one of the gauge symmetries
(corresponding to the $\mathbb{R}$ in $\mathbb{R} \times \U(1)^3$).
The remaining massless $\cN=2$ gauged supergravity contains only two vector multiplets and is described by the prepotential
 \bea
  \label{intro:prepotantial}
  \cF \left( X^I \right) = - i \frac{3^{3/2}}{4} \left(1-\frac{i}{\sqrt{3}}\right)
  c^{1/3} \left( X^1 X^2 X^3 \right)^{2/3} \ ,
 \eea
where the dyonic gauge parameter is the ratio $c \equiv m/g$.

The goal of the current work is to verify \eqref{d:micro} by a direct counting in the dual boundary
description in terms of a topologically twisted Chern-Simons-matter gauge theory
with level $k$ given by the quantized Romans mass, $m=\hat F_{(0)}=k / (2 \pi \ell_s)$.%
\footnote{$\ell_s$ is the string length.}

The SCFT dual to the background AdS$_4 \times S^6$ arises as an $\cN = 2$ Chern-Simons
deformation (at level $k$) of the maximal $\cN=8$ SYM
theory on the worldvolume of $N$ D$2$-branes \cite{Schwarz:2004yj,Guarino:2015jca}.
We will call this model the \emph{D2$_k$ theory}.
It has an adjoint vector multiplet (containing a real scalar and a complex fermion) with gauge group $\U(N)$ or $\SU(N)$
and three chiral multiplets $\phi_{j}$ $(j = 1, 2, 3)$ (containing a complex scalar and fermion).
To verify \eqref{d:micro} we evaluate the topologically twisted index for ${\rm D2}_k$.

Let us state the main result of this chapter. Upon extremizing $\cI (\Delta_j)$, at large $N$,
with respect to the chemical potentials $\Delta_j$ we show that its value at the extremum $\bar\Delta_j$
precisely reproduces the black hole entropy:
\be
 \label{main_result}
 \cI ( \bar\Delta_j ) \equiv \log Z ( \bar\Delta_j ) - i \sum_{j=1}^{3} \bar\Delta_j q_j
 = S_{\rm BH} ( \fn_j , q_j ) \, .
\ee
In the above equation appears three chemical potentials $\Delta_j$ and three electric charges $q_j$:
two for the global symmetries and one for the R-symmetry. The extremization equations are invariant
under a common shift of $q_j$'s, that corresponds to an electric charge for the R-symmetry, while $\cI$ is not.
We can fix the values of $q_j$'s by requiring that $\cI$ is \emph{real positive} \cite{Benini:2016rke}.
On the gravity side, there exists a BPS constraint that fixes one of the electric charges
in order to have a smooth black hole.
This argument thus gives an unambiguous prediction for the Bekenstein-Hawking entropy of the black hole.

Moreover, we demonstrate another example of the conjecture originally posed in \cite{Hosseini:2016tor}:
\bea
 \label{intro:extr:attractor}
 - \log Z_{S^3} \left( \Delta_j \right) & \propto \cF \left( X^j \right) \, , \\
 \cI\text{-extremization} & = \text{attractor mechanism} \, ,
\eea
where $Z_{S^3} \left( \Delta_j \right)$ denotes the $S^3$ partition function for ${\rm D2}_k$,
depending on trial R-charges $\Delta_j$ \cite{Fluder:2015eoa}:
\bea
 \label{intro:S3 free energy}
 \log Z_{S^3} = - \frac{3^{13/6} \pi }{5 \times 2^{5/3}} \left( 1 - \frac{i}{\sqrt{3}} \right) k^{1/3} N^{5/3}
 \left( \Delta_1 \Delta_2 \Delta_3 \right)^{2/3} \, .
\eea

The remainder of this chapter is arranged as follows.
In section \ref{ssec:largeN:limit:index}, we focus on the large $N$ limit of a class of three-dimensional supersymmetric Chern-Simons-matter
gauge theories arising from D$2$-branes probing generic Calabi-Yau three-fold (CY$_3$) singularities in the presence
of non-zero quantized Romans mass.
After deriving the formula \eqref{intro:index:generic:c2d:a4d}, we move to evaluate the twisted index
for the $\cN=2$ ${\rm D2}_k$ theory.
In section \ref{app:dyonic STU:detailed} we switch gears and review the four-dimensional $\cN=2$ dyonic STU model, as constructed in \cite{Guarino:2017pkw}.
In section \ref{sec:dyonic sugra} we discuss our supergravity solutions dual to a topologically twisted deformation of the ${\rm D}2_k$ theory.
This section contains the supersymmetric conditions for the existence of black hole solutions. 
We then proceed to analyze in more detail the exact UV and IR limits of the general equations, recovering
the asymptotic AdS$_4$ and the near-horizon AdS$_2 \times \Sigma_\fg$ geometries.
We finish this section by commenting on the general existence of full BPS flows between the UV and IR solutions that we have. 
In section \ref{sec:index vs entropy} we compare the field theory and the supergravity results,
and we show that the $\cI$-extremization correctly reproduces the black hole entropy.

Let us note that, the counting of microstates for black holes with constant scalar fields
--- equal fluxes along the exact R-symmetry of three-dimensional SCFTs ---
and horizon topology AdS$_2 \times \Sigma_\fg$, $(\fg>1)$ has been recently considered in \cite{Azzurli:2017kxo}.
While we were completing this work, we became aware of \cite{Benini:2017oxt}
which we understand has overlap with the results presented here.

\section[The large \texorpdfstring{$N$}{N} limit of the index for a generic theory]{The large $\fakebold{N}$ limit of the index for a generic theory}
\label{ssec:largeN:limit:index}

We focus on Chern-Simons quiver gauge theories with bi-fundamental and adjoint chiral multiplets
transforming in representations $\fR_I$ of $G$ and a number $|G|$ of $\U(N)$ gauge groups with equal Chern-Simons couplings $k_a = k$ $(a = 1, \ldots, |G|)$. 
We are interested in the large $N$ limit, $N \gg k_a$ with $\sum_{a=1}^{|G|}k_a \neq 0$, of the index for Chern-Simons-matter gauge theories with massive type IIA supergravity
duals $\mathrm{AdS}_4 \times \cS Y_5$ \cite{Aharony:2010af,Petrini:2009ur,Lust:2009mb,Tomasiello:2010zz,Guarino:2015jca,Fluder:2015eoa,Pang:2015vna,Pang:2015rwd,Guarino:2016ynd,Guarino:2017eag,Guarino:2017pkw}.
Here $\cS Y_5$ denotes the \emph{suspension} of a generic Sasaki-Einstein five-manifold $Y_5$:
$\rd s^2_{\cS Y_5} = \rd \alpha^2 + \sin^2 \alpha \, \rd s^2_{Y_5}$
and $\alpha \in [0 , \pi]$ with $\alpha = 0, \pi$ being isolated conical singularities.%
\footnote{The line element $\rd s^2_{\cS Y_5}$ is called the \emph{sine cone} over $Y_5$,
and is an Einstein metric admitting a Killing spinor.}
These theories describe the dynamics of $N$ D$2$-branes probing a generic Calabi-Yau three-fold (CY$_3$) singularity in the presence of a non-vanishing quantized Romans mass $m$ \cite{Gaiotto:2009mv}.

The twisted superpotential $\wt\cW$ for this class of theories reads (see section \ref{general:rules:N53:Bethe})
\bea
 \label{Bethe:potential:N53}
 \begin{split}
 \frac{\wt\cW \left( \rho(t) , v(t) , \Delta_I \right)}{N^{5/3}} & = n \int {\rm d}t\ \rho(t)\ \left\{ - i t\, v(t) + \frac{1}{2} \left[t^2 - v(t)^2\right] \right\}  \\
 & \phantom{=} + i \sum_{I} g_+ (\Delta_I) \int {\rm d}t \, \frac{\rho(t)^2}{1- i v'(t)} - i \mu \left( \int {\rm d}t\, \rho(t) - 1 \right)\, ,
 \end{split}
\eea
where
\bea
 \label{Romans:mass:k}
 n \equiv \sum_{a = 1}^{|G|} k_a = |G| k \, .
\eea
We need to extremize the local functional $\wt\cW\left(\rho(t), v(t) , \Delta_I \right)$
with respect to the continuous functions $\rho(t)$ and $v(t)$. 
The solution for $\sum_{I \in a} \Delta_I = 2 \pi$, for each term $W_a$ in the superpotential,
is as follows:\footnote{The support $[t_-, t_+]$ of $\rho(t)$
can be determined from the relations $\rho(t_\pm)=0$.} 
\bea
 \label{BAEs:sol:generic}
 v(t) & = -\frac{1}{\sqrt{3}} t \, , \\
 \rho(t) & = \frac{3^{1/6}}{2} \left[ \frac{n}{\sum_I g_+(\Delta_I)} \right]^{1/3}
 - \frac{2}{3^{3/2}} \left[ \frac{n}{\sum_I g_+(\Delta_I)} \right] t^2 \, , \\
 t_\pm & = \pm \frac{3^{5/6}}{2} \left[ \frac{\sum_I g_+ (\Delta_I)}{n}\right]^{1/3} \, , \\
 \mu & = \frac{\sqrt{3}}{4} \left( 1 - \frac{i}{\sqrt{3}} \right) n^{1/3}
 \left[ 3 \sum_{I} g_+ (\Delta_I) \right]^{2/3} \, .
\eea
One can explicitly check that
\bea
 \label{BethePot:on-shell}
 \wt\cW (\Delta_I) \equiv - i \wt\cW (\rho(t), v_a(t), \Delta_I) \big|_\text{BAEs} = \frac{3}{5} \mu N^{5/3}\, .
\eea
This is indeed equal to $- \log Z_{S^3}$, \cf\;Eq.\,(3.26) in \cite{Fluder:2015eoa}, up to a normalization.
Here $Z_{S^3}$ is the partition function of the same $\cN=2$ theory on the three-sphere \cite{Kapustin:2009kz,Jafferis:2010un,Hama:2010av}.

For this class of Chern-Simons-matter quiver gauge theories the topologically twisted index, at large $N$, is given by \eqref{ch:2:N53:index:rules}:
\bea
 \label{index:N53}
 \log Z  = - \left[ |G| \frac{\pi^2}{3} + \sum_{I} (\fn_I - 1) g'_+ (\Delta_I) \right] N^{5/3}
 \int {\rm d}t \, \frac{\rho(t)^2}{1-i v'(t)} \, .
\eea
Plugging the solution \eqref{BAEs:sol:generic} into the index \eqref{index:N53},
we obtain the following simple expression for the logarithm of the index
\bea
 \label{index:generic}
 \log Z \left( \fn_I, \Delta_I \right) = - \frac{3^{7/6}}{10} \left( 1 - \frac{i}{\sqrt{3}} \right) n^{1/3} N^{5/3}
 \frac{\sum_{I} \left[ \frac{3}{\pi} g_{+}(\Delta_I) +
 \left(\fn_I - \frac{\Delta_I}{\pi } \right) g_{+}'(\Delta_I) \right]}{\left[ \sum_{I} g_+(\Delta_I) \right]^{1/3}}
 \, .
\eea
Remarkably, it can be rewritten as
\bea
 \label{index:generic:c2d:a4d}
 \log Z \left( \fn_I, \Delta_I \right) =
 \frac{ 3^{7/6} \pi}{5 \times 2^{10/3}}
 \left(1 - \frac{i}{\sqrt{3}} \right)
 \left( n N \right)^{1/3}
 \frac{c_r \left( \fn_I , \Delta_I \right)}{a \left( \Delta_I \right)^{1/3}} \, .
\eea
Here $a \left( \Delta_I \right)$ is the trial $a$ central charge of the ``parent'' four-dimensional
$\cN = 1$ SCFT on $S^2 \times T^2$, with a partial topological A-twist on $S^2$,
and $c_r  \left( \fn_I , \Delta_I \right)$ is the trial right-moving central charge of the two-dimensional $\cN = (0,2)$
theory on $T^2$ obtained from the compactification on $S^2$ (see subsection \ref{ch:1:4d:N=1:intro}).%
\footnote{We refer the reader to \cite{Benini:2012cz,Benini:2013cda,Hosseini:2016cyf,Karndumri:2013dca,
Klemm:2016kxw,Amariti:2016mnz,Amariti:2017cyd,Amariti:2017iuz} for a detailed analysis of
superconformal theories obtained by twisted compactifications of four-dimensional $\cN = 1$
theories and their holographic realization.}
Notice that \eqref{index:generic:c2d:a4d} is consistent with \eqref{Z large N conjecture0:3d}.

\section[The index of D2\texorpdfstring{$_k$}{[k]} at large \texorpdfstring{$N$}{N}]{The index of D2$_{\fakebold{k}}$ at large $\fakebold{N}$}
\label{ssec:SYM-CS:index}

So far the discussion was completely general. Let us now focus on the $\cN = 2$ Chern-Simons
deformation of the maximal SYM theory in three dimensions \cite{Schwarz:2004yj,Guarino:2015jca}.
In $\cN = 2$ notation, the three-dimensional maximal SYM has an adjoint vector multiplet (containing a real scalar and a complex fermion)
with gauge group $\U(N)$ or $\SU(N)$ as well as three chiral multiplets $\phi_j$ $(j = 1, 2, 3)$ (containing a complex scalar and fermion).
This theory has $\U(1)_R \times \SU(3)$ symmetry, with $\SU(3)$ rotating the three complex scalar fields in the chiral multiplets.
The quiver diagram for this theory is depicted below.
\bea
 \label{SYM:k:quiver}
 \begin{tikzpicture}[font=\footnotesize, scale=0.9]
  \begin{scope}[auto,%
   every node/.style={draw, minimum size=0.5cm}, node distance=2cm];
  \node[circle]  (UN)  at (0.3,1.7) {$N$};
  \end{scope}
  \draw[decoration={markings, mark=at position 0.45 with {\arrow[scale=2.5]{>}}, mark=at position 0.5 with {\arrow[scale=2.5]{>}}, mark=at position 0.55 with {\arrow[scale=2.5]{>}}}, postaction={decorate}, shorten >=0.7pt] (-0,2) arc (30:341:0.75cm);
  \node at (-2.2,1.7) {$\phi_{1,2,3}$};
  \node at (0.7,1.2) {$k$};
 \end{tikzpicture}
 \nonumber
\eea
It has a cubic superpotential,
\bea
 W = \Tr \left( \phi_3 \left[ \phi_1 , \phi_2 \right] \right) \, .
\eea
We assign chemical potentials $\Delta_j \in [0 , 2 \pi]$ to the fields $\phi_{j}$.
The invariance of each monomial term in the superpotential under the global symmetries of the theory imposes the following constraints
on the chemical potentials $\Delta_j$ and the flavor magnetic fluxes $\fn_j$ associated with the fields $\phi_j$,
\bea
 \sum_{j = 1}^{3} \Delta_j \in 2 \pi \bZ \, , \qquad \qquad \qquad
 \sum_{j = 1}^{3} \fn_j = 2 \, ,
\eea
where the latter comes from supersymmetry.
Since $0 \leq \Delta_j \leq 2\pi$ we can only have $\sum_{j=1}^{3} \Delta_j=2 \pi s\, ,\forall s=0,1,2,3$.
The cases $s=0,3$ are singular while those for $s=2$ and $s=1$ are related by a discrete symmetry $\Delta_j = 2\pi - \Delta_j$.
Thus, without loss of generality, we will assume $\sum_{j = 1}^{3} \Delta_j = 2 \pi$.
We find that
\bea
 \sum _{j = 1}^3 g_+ \left( \Delta_j \right) & = \frac{1}{2} \Delta_1 \Delta_2 \Delta_3 \, , \\
 \sum _{j = 1}^3 g_+' \left( \Delta_j \right) & =
 \frac{1}{4} \left[ \left( \Delta_1^2 + \Delta_2^2 + \Delta_3^2 \right)
 - 2 \left(\Delta_1 \Delta_2 + \Delta_2 \Delta_3 + \Delta_1 \Delta_3 \right) \right] \, .
\eea
Finally, the ``on-shell'' value of the twisted superpotential \eqref{BethePot:on-shell} and the index \eqref{index:generic}, at large $N$, can be written as
\be
 \label{index:SYM-CS}
 \begin{aligned}
 \wt\cW( \Delta_j ) & =
 \frac{3^{13/6}}{5 \times 2^{8/3}}
 \left(1 - \frac{i}{\sqrt{3}} \right)
 k^{1/3} N^{5/3}
 \left( \Delta_1 \Delta_2 \Delta_3 \right)^{2/3} \, , \\
 \log Z ( \fn_j , \Delta_j ) & = - \frac{3^{7/6}}{5 \times 2^{5/3}}
 \left(1 - \frac{i}{\sqrt{3}} \right)
 k^{1/3} N^{5/3}
 ( \Delta_1 \Delta_2 \Delta_3 )^{2/3}
 \sum_{j=1}^3 \frac{\fn_j}{\Delta_j}
 \, ,
 \end{aligned}
\ee
which is valid for $\sum_{j=1}^{3} \Delta_j = 2 \pi$ and $0 \leq \Delta_j \leq 2 \pi$.
Note also that
\bea
 \label{index:SYM:k:attractor}
 \log Z ( \fn_j , \Delta_j ) = - \sum_{j = 1}^{3} \fn_j \frac{\partial \wt\cW(\Delta_j)}{\partial \Delta_j} 
 \, ,
\eea
as expected from the index theorem \eqref{Z large N conjecture2:intro}.

\section{Dyonic STU model}
\label{app:dyonic STU:detailed}


We look at the $\mathcal{N} = 2$ truncation of the dyonic $\ISO(7)$ gauged supergravity constructed recently in \cite{Guarino:2017pkw},
as the analogue of the STU model. We call this the ``dyonic STU model'' (see Fig.\;\ref{trunc_fig}).
It corresponds to picking the maximal Abelian subgroup of the original ISO$(7)$ gauge group
and arranging the resulting four Abelian vectors in an ${\cal N}=2$ gravity multiplet and three vector multiplets.
Due to the characteristics of the supergravity theory under consideration, there is the requirement that the four vectors couple to a
hypermultiplet, so that they effectively gauge some of the isometries of the scalar manifold.

\begin{figure}[ht]
\centering	
\begin{tikzpicture}[scale=1, every node/.style={scale=.9}]
	
	\node at (-2.3,0){10D};
	
	\draw[fill = white,thick, rounded corners] (-1.5,-.75) rectangle (1.5,.75);
	\node at (0,0){\small massive IIA};
	
	\draw[->,>=stealth] (0,-.85) -- (0,-3);
		
	\node at (-2.3,-3.9){4D};
		
	\draw[fill = white,thick, rounded corners] (-1.5,-4.65) rectangle (1.5,-3.15);
	\node at (0,-3.9){\small $\mathcal{N}=8$ ISO(7)};
		
	\draw[->,>=stealth, thick, dashed] (1.6,-3.9) -- (3.6,-3.9);	
				
	\draw[fill = white,thick, rounded corners] (3.7,-4.65) rectangle (7.1,-3.15);
	\node at (5.4,-3.9){\small $\mathcal{N}=2$ dyonic STU };

	\draw[->,>=stealth, ]  (4.7, -.4) -- (5.5, -.4);
	\node at (7.3,-.4){\small $S^6$ truncation \cite{Guarino:2015jca,Guarino:2015vca}};
	
	\draw[->,>=stealth, ,dashed, thick] (4.7,-1.2) -- (5.5,-1.2);
	\node at (7.45,-1.2){\small Cartan truncation \cite{Guarino:2017pkw}};
		
	\end{tikzpicture}
	\caption{Sequence of consistent truncations from massive type IIA supergravity in ten dimensions, to the dyonic STU model in four dimensions.}
	\label{trunc_fig}
\end{figure}
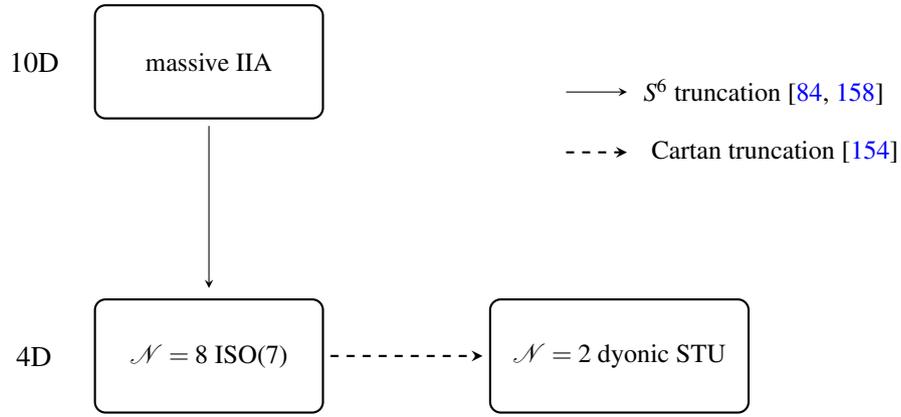

We start with the bosonic part of the Lagrangian for the dyonic STU model,
following the notation and conventions of the standard reference \cite{Andrianopoli:1996cm},
\begin{align}\label{dyonic Lagrangian}
\begin{split}
\frac{1}{\sqrt{-g}} {\cal L}_{\text{dyonic STU}} &= \frac{R}{2} - V_{g, m} - g_{i \bar{j}} \partial_{\mu} z^i \partial^{\mu} \bar{z}^{\bar{j}} - h_{u v} \nabla_{\mu} q^u \nabla^{\mu} q^v + \frac{1}{4} {\rm I}_{\Lambda \Sigma} H^{\Lambda \mu \nu} H^{\Sigma}{}_{\mu \nu} \\
 &+ \frac14 {\rm R}_{\Lambda \Sigma} H^{\Lambda \mu \nu} * H^{\Sigma}{}_{\mu \nu} - m \frac{\varepsilon^{\mu \nu \rho \sigma}}{4 \sqrt{-g}} B^0{}_{\mu \nu} \partial_{\rho} A_{0 \sigma} - g m \frac{\varepsilon^{\mu \nu \rho \sigma}}{32 \sqrt{-g}} B^0{}_{\mu \nu} B^0{}_{\rho \sigma}\ .
\end{split}
\end{align}
This is supplemented by a fermionic counterpart which we do not present here.
It is however instructive to look at the covariant derivative of the gravitino,
\begin{align}\label{gravitino derivative}
	\nabla_{\mu} \psi_{\nu A} = (\partial_{\mu} - \frac14 \omega_{\mu}^{ab} \gamma_{ab} + \frac{i}{2} A_{\mu}) \psi_{\nu A} + (\partial_{\mu} q^u \omega_u{}_A{}^B - \frac{i g}{2} \langle {\cal P}^x, {\cal A}_{\mu} \rangle \sigma^x{}_A{}^B ) \psi_{\nu B}\ .
\end{align}
Many of the above quantities require explanation, and in what follows we will discuss
independently several of the sectors of the theory. 

\subsection{Gravity multiplet}

The gravity multiplet consists of the graviton $g_{\mu \nu}$, a doublet of gravitini $\psi_{\mu A}$,
which transform into each other under the R-symmetry group $\U(1)_{\rm R} \times \SU(2)_{\rm R}$,
and a gauge field called the graviphoton with field strength $T_{\mu \nu}$.
Due to the presence of three additional vector multiplets in the theory the total number of gauge fields is four,
denoted by $A^{\Lambda}_{\mu} \, ,$ $\Lambda \in \{0,1,2,3\}$.
The graviphoton field strength is a scalar dependent linear combination of the four field strengths $F^{\Lambda}_{\mu \nu}$.
The theory we consider is gauged, meaning that some of the original global symmetries of the theory have been made local.

\subsection{Universal hypermultiplet}

An ${\cal N}=2$ hypermultiplet consists of four real scalars $q^u$ and two chiral fermions $\zeta_{\alpha}$ called hyperini.
The scalar moduli space is a quaternionic K\"{a}hler manifold with metric $h_{u v} (q)$
and three almost complex structures which further define three quaternionic two-forms that are covariantly constant
with respect to an $\SU(2)$ connection $\omega^x$, $x \in \{1,2,3\}$.
The particular model that comes from the truncation of ${\cal N}=8$ $\ISO(7)$ gauged supergravity has a single hypermutiplet,
which universally appears in various string compactifications, hence called the \emph{universal hypermultiplet}.
The moduli space is the coset space $\SU(2,1) / \U(2)$.
The metric, written in terms of real coordinates $\{\phi,\, \sigma,\, \zeta,\, \tilde{\zeta} \}$, is  
\begingroup
\renewcommand*{\arraystretch}{1.2}
\begin{equation}
h = 
\begin{pmatrix}
1 & 0 & 0 & 0 \\
0 & \frac{1}{4} \e^{4\phi} & - \frac{1}{8} \e^{4\phi} \tilde{\zeta} & \frac{1}{8} \e^{4\phi} \zeta \\
0 & - \frac{1}{8} \e^{4\phi} \tilde{\zeta} & \frac{1}{4} \e^{2\phi}(1 + \frac{1}{4} \e^{2\phi} \tilde{\zeta}^2) & -\frac{1}{16} \e^{4\phi} \zeta \tilde{\zeta} \\
0 &  \frac{1}{8} \e^{4\phi} \zeta & -\frac{1}{16} \e^{4\phi} \zeta \tilde{\zeta} &  \frac{1}{4} e^{2\phi}(1 + \frac{1}{4} \e^{2\phi} \zeta^2)
\end{pmatrix}.
\end{equation}
\endgroup
The isometry group $\SU(2,1)$ has eight generators; two of these are used for gauging in the model under consideration,
generating the group $\mathbb{R} \times \U(1)$. The corresponding Killing vectors are
\begin{align}
	k^{\mathbb{R}} = \partial_{\sigma}\ , \qquad k^{\U(1)} = -\tilde{\zeta} \partial_{\zeta} + \zeta \partial_{\tilde{\zeta}}\ .
\end{align}
One defines Killing vectors with index $\Lambda$ corresponding to each of the four gauge fields,
such that the hypermultiplet scalar covariant derivative that appears in \eqref{dyonic Lagrangian} reads
\begin{align}
	\nabla_{\mu} q^u \equiv \partial_{\mu} q^u - g \langle {\cal K}^u, {\cal A}_{\mu} \rangle =  \partial_{\mu} q^u - g k^u_{\Lambda} A^{\Lambda}_{\mu} + g k^{u, \Sigma} A_{\Sigma, \mu}\ ,
\end{align}
where $g$ is the gauge coupling constant and the operation $\langle. , .\rangle$ is the symplectic inner product
which will be discussed further when we move to the vector multiplet sector.
What is important to notice here is that we allow for the hypermultiplet isometries to be gauged not only by the ``ordinary''
electric fields $A^{\Lambda}_{\mu}$ but also by their dual magnetic fields $A_{\Lambda, \mu}$.
In the particular model here, the non-vanishing Killing vectors are
\begin{align}\label{app:magnetic killing}
	k_0 = k^{\mathbb{R}}\ , \quad k^0 = c  k^{\mathbb{R}}\ , \quad k_{1,2,3} = k^{\U(1)}\ , 
	\qquad c \equiv \frac{m}{g} \ ,
\end{align}
which means that the magnetic gauge field $A_{0, \mu}$ explicitly appears in the covariant derivative of the scalar $\sigma$
with an effective coupling constant $m$ related to the Romans mass of the massive type IIA supergravity. 

Note that although all four Abelian vectors participate in the gauging of the above isometries,
only two different isometries are actually being gauged: one corresponding to the non-compact group $\mathbb{R}$,
gauged by a linear combination of the electric and magnetic gauge fields $A^0$ and $A_0$,
and a U$(1)$ isometry gauged by the linear combination $A^1 + A^2 + A^3$.
These gaugings, via supersymmetry, generate a nontrivial scalar potential,
which has a critical point corresponding to an AdS$_4$ vacuum. 

One can also define moment maps (or momentum maps) associated with each isometry on the quaternionic K\"{a}hler manifold.
Using the metric and $\SU(2)$ connection on the universal hypermultiplet scalar manifold
(see \eg\;appendix D of \cite{Hristov:2012bk}) we find
\begin{align}
P_0 = \begin{pmatrix} 0, & 0, & - \frac{1}{2} \e^{2\phi} \end{pmatrix} ,& \qquad
P^0 = \begin{pmatrix} 0, & 0, &  -\frac{1}{2} c \e^{2\phi} \end{pmatrix} , \qquad \nn \\
P_{1,2,3} = \begin{pmatrix} \tilde{\zeta} \e^{\phi}, & - \zeta \e^{\phi}, &  1 - \frac{1}{4} (\zeta^2 + \tilde{\zeta}^2) \e^{2\phi}\end{pmatrix} ,& \qquad
P^{1,2,3} = \begin{pmatrix} 0, & 0, & 0 \end{pmatrix} .
\end{align} 
These are the moment maps that appear in the gravitino covariant derivative \eqref{gravitino derivative}
as a symplectic vector ${\cal P}^x = (P^{x, \Lambda}, P^x_{\Lambda})$.
Even in the absence of hypermultiplets the moment maps can be non-zero,
signifying that the R-symmetry rotating the gravitini is gauged. 

\subsection{STU vector multiplets}

Each ${\cal N} = 2$ vector multiplet consists of one gauge field, a doublet of chiral fermions $\lambda^A$ called gaugini,
and a complex scalar field $z$. We already mentioned that the STU model has three vector multiplets
and hence three complex scalars $z^i$, labeled by $s$, $t$, and $u$: $z^1 \equiv s$, $z^2 \equiv t$, $z^3 \equiv u$.
The complex scalars in the vector multiplets parameterize the special K\"{a}hler scalar (SK) manifold
${\cal M}_{{\rm SK}} = [{\rm SU}(1,1)/{\rm U}(1)]^3$
whose metric can be derived from a prepotential ${\cal F}$, which for the STU model is,
\begin{equation}\label{app:prepotential}
	{\cal F} = - 2 \sqrt{X^0 X^1 X^2 X^3} \ .
\end{equation}
$X^{\Lambda} = X^{\Lambda} (z^i)$ define the holomorphic sections ${\cal X} \equiv (X^{\Lambda}, F_{\Lambda})$ where
\begin{align}
	F_{\Lambda} \equiv \frac{\partial {\cal F}}{\partial X^{\Lambda}}\ .
\end{align}
${\cal X}$ transforms as a vector under electromagnetic duality or symplectic rotations
which leave the solutions of the theory invariant.
Other symplectic vectors are the Killing vectors ${\cal K}^u = (k^{u, \Lambda}, k^u_{\Lambda})$,
the moment maps ${\cal P}^x = (P^{x, \Lambda}, P^x_{\Lambda})$,
the gauge fields ${\cal A}_{\mu} = (A^{\Lambda}_{\mu},A_{\Lambda, \mu})$,
and finally the vector of magnetic $p^{\Lambda}$ and electric $e_{\Lambda}$ charges,
${\cal Q} = (p^{\Lambda}, e_{\Lambda})$, giving the name to the duality.

Returning to the holomorphic sections, we pick the standard parameterization
\begin{align}\label{app:scalar parameterization}
(X^0,\, X^1,\, X^2,\, X^3,\, F_0,\, F_1,\, F_2,\, F_3) = (- s t u,\, -s,\, -t,\, -u,\, 1,\, t u,\, s u,\, s t)\ . 
\end{align}
The metric on the moduli space follows from the K\"ahler potential,
\begin{align}
	\label{4d:Kahler potential}
	{\cal K} = - \log (i \langle {\cal X}, \bar{{\cal X}} \rangle) = - \log (i \bar{X}^{\Lambda} F_{\Lambda} - i X^{\Lambda} \bar{F}_{\Lambda}) = - \log (i (s-\bar{s}) (t-\bar{t}) (u-\bar{u}))\ , 
\end{align}
as $g_{i \bar{j}} \equiv \partial_i \partial_{\bar{j}} {\cal K}$ with $\partial_i = \partial/\partial z^i$. We therefore find that $g_{i \bar{j}}$ is diagonal
\begin{align}
	g_{s \bar{s}} = \frac{1}{4 (\im (s))^2}\ , \quad g_{t \bar{t}} = \frac{1}{4 (\im (t))^2}\ , \quad g_{u \bar{u}} = \frac{1}{4 (\im (u))^2}\ .  
\end{align}

Using the K\"ahler potential we introduce the rescaled sections 
\be
{\cal V} = \e^{{\cal K}/2} {\cal X} = (\e^{{\cal K}/2} X^{\Lambda}, \e^{{\cal K}/2} F_{\Lambda}) \equiv (L^{\Lambda}, M_{\Lambda})
\ee
and covariant derivatives 
\be
D_i {\cal V} = (f_i^{\Lambda}, h_{i, \Lambda}) \equiv \e^{{\cal K}/2} \left((\partial_i X^{\Lambda}+X^{\Lambda} \partial_i {\cal K}), (\partial_i F_{\Lambda}+F_{\Lambda} \partial_i {\cal K})\right).
\ee

Moving on to the kinetic terms for the vector fields, the magnetic gauging of the $\mathbb{R}$ isometry in \eqref{app:magnetic killing},
leads to the appearance of the magnetic field $A_{0, \mu}$ in the covariant derivative of the scalar field $\sigma$.
Consistency with supersymmetry then requires the introduction of an auxiliary tensor field $B^0_{\mu \nu}$ as derived
in \cite{deWit:2005ub,Samtleben:2008pe}. The Lagrangian \eqref{dyonic Lagrangian} therefore contains the modified field strengths
\begin{align}
	H^0_{\mu \nu} \equiv F^0_{\mu \nu} + \frac{1}{2} m B^0_{\mu\nu}\ , \qquad H^{i=1,2,3}_{\mu \nu} \equiv F^i_{\mu \nu}\ ,
\end{align}
where $F^{\Lambda}_{\mu \nu}$ are the field strengths of the electric potentials $A^{\Lambda}_{\mu}$.
The kinetic and theta term for the field strengths $H$ involve the scalar-dependent matrices,
${\rm I}_{\Lambda \Sigma} \equiv \im ({\cal N})_{\Lambda \Sigma}$ and
${\rm R}_{\Lambda \Sigma} \equiv \re ({\cal N})_{\Lambda \Sigma}$.
The matrix $\mathcal{N}$ can be computed from the prepotential via
\bea
\label{IIA:period:matrix}
 \cN_{\Lambda \Sigma} = \bar{F}_{\Lambda \Sigma}
 + 2 i \frac{\big( N_{\Lambda \Gamma} X^\Gamma \big)
 \big( N_{\Sigma \Delta}  X^\Delta \big)}
 {X^\Omega N_{\Omega \Psi}  X^\Psi} \, ,
\eea
where $F_{\Lambda \Sigma} \equiv \partial_\Lambda \partial_\Sigma F$ and
$N_{\Lambda \Sigma} \equiv \im ( F )_{\Lambda \Sigma}$.

\subsection{Scalar potential}

The last part of the Lagrangian \eqref{dyonic Lagrangian} left to discuss is the scalar potential $V_{g, m}$
which depends on the electric and magnetic gauge coupling constants $g$ and $m$ and is given by the general formula
\begin{align}
	V_{g, m} = g^2\left(4 h_{u v} \langle {\cal K}^u, {\cal V}\rangle  \langle {\cal K}^u, \bar{{\cal V}}\rangle  + g^{i \bar{j}} \langle {\cal P}^x, D_i {\cal V}\rangle  \langle {\cal P}^x, \bar{D}_{\bar{j}} \bar{{\cal V}}\rangle  - 3 \langle {\cal P}^x, {\cal V}\rangle \langle {\cal P}^x, \bar{{\cal V}}\rangle \right) \ .
\end{align}
$V_{g, m}$ can be further evaluated explicitly for the dyonic STU model but we will not need its expression. 

The theory is now fully specified by the data of the hypermultiplet moduli space,
the vector multiplet moduli space, derived from the prepotential ${\cal F}$ in \eqref{app:prepotential},
and the Killing vectors \eqref{app:magnetic killing} specifying the gauging.

\subsection{Tensor fields}

Due to the presence of the auxiliary tensor field $B^0$ (the other auxiliary fields can be immediately decoupled from the theory),
there is an additional constraint arising as an equation of motion for $B^0$,
\begin{align}
	G_{\Lambda, \mu\nu} = F_{\Lambda, \mu \nu} + \frac{1}{2} m B^0_{\mu\nu}\ ,
\end{align}
where $G_{\Lambda, \mu\nu}$ is the dual field strength defined by $G_{\Lambda} = (2/\sqrt{-g}) * \delta {\cal L}/\delta F^{\Lambda}$. This leads to
\begin{align}
	G_{\Lambda, \mu\nu} = \frac12 {\rm I}_{\Lambda \Sigma} H^{\Sigma}_{\mu \nu} + \frac{1}{4 \sqrt{-g}} \epsilon_{\mu\nu\rho\sigma} {\rm R}_{\Lambda \Sigma} H^{\Sigma, \rho\sigma}\ .
\end{align}
The appearance of the magnetic gauge field $A_0$ in the Lagrangian leads to the following equation of motion constraining the auxiliary tensor field
\begin{align}\label{app:B0fixing}
	\frac14 \epsilon^{\mu \nu \rho \sigma} \partial_{\mu} B^0_{\nu\rho} = -2 \sqrt{-g} h_{u v} k^{u, 0}\nabla^{\sigma} q^v\ ,
\end{align}
while the rest of the equations of motion are the standard Einstein--Maxwell equations (with sources)
and the scalar equations, stemming from \eqref{dyonic Lagrangian}.
We discuss these in great details in appendix \ref{IIA:appendix}.
Note that the BPS conditions together with the Maxwell equations imply the rest of the equations of motion.

\section[AdS\texorpdfstring{$_{4}$}{(4)} black holes in \texorpdfstring{$\mathcal{N}=2$}{N=2} dyonic STU gauged supergravity]{AdS$_{\fakebold{4}}$ black holes in $\fakebold{\mathcal{N}=2}$ dyonic STU  gauged supergravity}
\label{sec:dyonic sugra}



We now turn to the gravity duals of the field theories we have discussed so far.
Our aim is to find supersymmetric AdS$_4$ black hole solutions in the $\cN=2$ dyonic STU gauged supergravity.
We will do so in several steps, leaving all detailed calculations to appendix \ref{IIA:appendix};
for each step we find useful to dedicate a subsection.
First, we describe the black hole ansatz and supersymmetry equations derived by \cite{Halmagyi:2013sla,Klemm:2016wng}.
We then concentrate separately on the conditions for the asymptotic AdS$_4$ vacuum and the near-horizon AdS$_2 \times \Sg$ geometry. We manage to rewrite the near-horizon data in a particularly simple form in order to facilitate the match with the field theory. We finish the supergravity analysis by presenting arguments for the existence of a full BPS flow between the UV and IR geometry. Ultimately, the existence of the complete geometries is best justified by the successful entropy match with field theory. 

\subsection{Black hole ansatz and BPS conditions}
Static BPS AdS$_4$ black holes in general models with dyonic hypermultiplet gauging,
were considered in  \cite{Klemm:2016wng} generalizing earlier work of
\cite{Cacciatori:2009iz,DallAgata:2010ejj,Hristov:2010ri,Halmagyi:2013sla,Katmadas:2014faa,Halmagyi:2014qza}.
The reader can find all the details about the bosonic ansatz and BPS equations in appendix \ref{IIA:appendix}.
Here, for the sake of clarity, we repeat the form of the metric,
\begin{equation}
{\rm d} s^2 = - \e^{2U(r)} {\rm d} t^2 + \e^{-2U(r)} {\rm d} r^2 + \e^{2(\psi(r) - U(r))}{\rm d}\Omega_{\kappa}^2\,,
\end{equation}
where the radial functions $U(r)$, $\psi(r)$ and the choice of scalar curvature $\kappa$ for the horizon manifold,
uniquely specify the spacetime. Electric and magnetic charges, $e_{\Lambda} (r)$ and $p^{\Lambda} (r)$,
are present for each gauge field, and can have a radial dependence due to the fact that some of the hypermultiplet scalars
source the Maxwell equations. The spacetime symmetries also impose a purely radial dependence for the SK complex
scalars $s (r)$, $t (r)$, $u (r)$ and the QK real scalars $\phi (r)$, $\sigma (r)$, $\zeta (r)$, $\tilde{\zeta} (r)$,
as well as the phase $\alpha (r)$ of the Killing spinors that parameterize the fermionic symmetries of the black hole. 

We systematically write down the conditions for supersymmetry and equations of motion in appendix \ref{IIA:appendix},
while here we only discuss the most important points about the solution. In particular we find that we can already fix three of the four hypermultiplet scalars
\begin{equation}
	\zeta = \tilde{\zeta} = 0\ , \qquad \sigma = \text{const.} \, ,
\end{equation}
where the particular value of $\sigma$ is not physical as it is a gauge dependent quantity that drops out of all BPS equations.
The remaining hypermultiplet scalar however has in general a nontrivial radial profile governed by the equation
\begin{equation}\label{phi}
	\phi' = - g \kappa \lambda \e^{{\cal K}/2 - U} \im \left(\e^{-i \alpha} (X^0 - c F_0)\right)\ , 	
\end{equation}
where the K\"ahler potential $\e^{\cal K}$ and $\lambda = \pm 1$ are discussed in the appendix \ref{IIA:appendix}.
The scalar $\phi$ precisely sources the Maxwell equations, which read 
\begin{equation}\label{charge}
	p'^0 = c e'_0 = -c \e^{2 \psi - 3 U} \e^{4 \phi} \re  \left(\e^{-i \alpha} (X^0 - c F_0)\right)\ ,
\end{equation}
while all other charges $p^{1,2,3}$ and $e_{1,2,3}$ are truly conserved quantities.
These two equations highlight an important physical feature of the black holes in massive IIA supergravity:
due to the presence of charged hypermultiplet scalars there are massive vector fields that do not have conserved charges.
The charges of the massive vectors are not felt by the field theory, which explains why there were only
three different charges considered in the previous section. These are the magnetic charges $p^{1,2,3}$
as here we will further simplify our ansatz and put $e_{1,2,3}= e$ to be fixed by the magnetic charges.
However, one still needs to solve consistently the BPS equations for the massive vector fields,
which presents a particularly hard obstacle computationally, and has prevented people from writing down
exact analytic solutions for black holes with massive vector fields before \cite{Halmagyi:2013sla}.

\subsection{Constant scalars, analytic UV and IR geometries}
Let us now concentrate on the two important end-points of the full black hole flow: the asymptotic UV space AdS$_4$ and the IR fixed point, AdS$_2 \times \Sg$. Due to the symmetries of these spaces the scalars and charges are constant there, which means \eqref{phi}-\eqref{charge} can be further constrained by setting their left-hand sides to zero. This immediately implies 
\begin{equation}
	X^0 - c F_0 = 0 \Rightarrow X^0 = (- c)^{2/3} (X^1 X^2 X^3)^{1/3}\ , \qquad stu = -c\ , 
\end{equation}
which presents a strong constraint of the vector multiplet moduli space. In fact the remaining scalars (\eg\;freezing $s$ in favor of $t$ and $u$) are consistent with the simplified prepotential\footnote{Note that directly substituting $X^0$ in the original prepotential \eqref{app:prepotential} leads to a different normalization. Such a different prefactor does not lead to a change in physical quantities, but we prefer to comply with the correct normalization of the kinetic terms as imposed by the choice of parameterization in \eqref{app:scalar parameterization}.}
\begin{equation}\label{new_prepot}
	{\cal F}^\star = - \frac{3}{2} (-c)^{1/3} (X^1 X^2 X^3)^{2/3}\ .
\end{equation}

Notice that in this constant scalar case, the BPS equations automatically lead us to an effective truncation of the theory to a subsector, by ``freezing'' some of the fields. In particular, we see that the massive vector field has ``eaten up'' the Goldstone boson $\sigma$, and together with the massive scalars $\zeta$, $\tilde{\zeta}$, $\phi$ and the complex combination of $s t u$ can be integrated out of the model. This corresponds to a supersymmetry preserving version of the Higgs mechanism discussed in \cite{Hristov:2010eu} and a truncation\footnote{Note that strictly speaking we have not proven that this is a consistent truncation as the proof in \cite{Hristov:2010eu} only considered electrically gauged hypermultiplets. For the analogous proof in the general dyonic case one needs to use the full superconformal formalism of \cite{deWit:2011gk} where the general theory is properly defined. However, here we never need to go to such lengths since we use the Higgs mechanism to clarify the physical picture, not as a guiding principle in deriving the BPS equations.} to an ${\cal N}=2$ theory with two massless vector multiplets and no hypermultiplets. The remainders of the gauged hypermultiplet are constant parameters gauging the R-symmetry, known as Fayet-Iliopoulos terms, $\xi_I = P^{x=3}_I$, $I \in \{1,2,3 \}$. Therefore the effective, or truncated, prepotential ${\cal F}^\star$ is indeed the prepotential defining the Higgsed theory. This mechanism is in fact the reason why we are able to write down exact analytic solutions in the UV and IR limits where the constant scalar assumptions holds. Note that one could have in principle performed this truncation of the full theory looking for full black hole solutions there. However, this turns out to be a too strong constraint; in particular we will see that in the UV we have 
\begin{equation}
	\langle \e^{2 \phi} \rangle_{\rm UV} = 2 c^{-2/3}\ ,
\end{equation}
while in the IR in general
\begin{equation}
	\langle \e^{2 \phi}\rangle_{\rm IR} = \frac{2 c^{-2/3}}{3 (H^1 H^2 H^3)^{1/3}}\ ,
\end{equation}
with $H^I$ particular functions of the charges. Imposing the constraint that $\phi$ is constant throughout
the flow $\phi_{\rm UV} = \phi_{\rm IR}$ leads to a black hole solution with only a subset of all possible charges.
This is the so called \emph{universal twist} solution (defined only for hyperbolic Riemann surfaces)
dating back to \cite{Romans:1991nq,Caldarelli:1998hg}.
This class of black holes studied for massive IIA supergravity on $S^6$ in \cite{Guarino:2017eag}
and recently described holographically in \cite{Azzurli:2017kxo} (see also \cite{Bobev:2017uzs}).

\subsubsection[Asymptotic AdS\texorpdfstring{$_4$}{(4)} vacuum]{Asymptotic AdS$_4$ vacuum}

The black hole is asymptotically locally AdS$_4$ (it is often called {\it magnetic} AdS$_4$ in the literature \cite{Hristov:2011ye}).
The dual boundary theory is a relevant deformation of the ${\rm D}2_k$ theory,
partially twisted by the presence of the magnetic charges.
In this section we analyze the exact AdS$_4$ vacuum, which constrains
the scalar fields to obey the maximally supersymmetric conditions derived in \cite{Hristov:2010eu}.
These conditions, as shown in more details in appendix \ref{IIA:appendix},
not only constrain the scalars to be constant with $\zeta = \tilde{\zeta} = 0$ and $s t u= -c$ but further impose the particular vacuum expectation values
\begin{align}\label{vacuum}
\begin{split}
\langle s\rangle _{{\rm AdS}_4} = \langle t\rangle _{{\rm AdS}_4} = \langle u\rangle _{{\rm AdS}_4} = (-c)^{1/3},\\
\langle \e^{2 \phi}\rangle _{{\rm AdS}_4} = 2 c^{-2/3}\ ,
\end{split}
\end{align} 
which can be checked to explicitly solve all the equations \eqref{app:BHequations} at $r \rightarrow \infty$. The metric functions in this limit become
\begin{equation}
\label{AdS4:radius}
	\lim_{r \rightarrow \infty} (r \e^{-\psi}) = \lim_{r \rightarrow \infty} \e^{-U} = \frac{L_{{\rm AdS}_4}}{r}\ , \quad L_{{\rm AdS}_4} = \frac{c^{1/6}}{3^{1/4} g}\ , 
\end{equation}
as already found in \cite{Guarino:2017eag}.

\subsubsection{Near-horizon geometry and attractor mechanism}
The attractor mechanism for static supersymmetric asymptotically AdS$_4$ black holes was studied in detail in \cite{Klemm:2016wng}, generalizing the results of \cite{DallAgata:2010ejj} to cases with general hypermultiplet gaugings.
The near-horizon geometry is of the direct product type AdS$_2 \times \Sg$ and preserves four real supercharges, double the amount preserved by the full black hole geometry. We solve carefully all equations in appendix \ref{IIA:appendix},
while here we present an alternative derivation which, although incomplete as we explain in due course, is more suitable for the comparison with field theory.

The near-horizon metric functions are given by
\begin{equation}
U = \log(r/L_{{\rm AdS}_2}) \, , \qquad \psi = \log(L_{\Sigma_\fg} \cdot r / L_{{\rm AdS}_2}) \, ,
\end{equation}
where $L_{{\rm AdS}_2}$ is the radius of AdS$_2$ and $L_{\Sigma_\fg}$ that of the surface $\Sigma_\fg$.

We start with the BPS condition coming from the topological twist for the magnetic charges (valid not only on the horizon but everywhere in spacetime)
\begin{equation}
 g \sum_{I=1}^3 p^I = - \kappa\ ,
\end{equation}
with $\kappa$ the unit curvature of the internal manifold on the horizon
($\kappa = +1$ for $S^2$ and $\kappa = -1$ for $\Sigma_{\mathfrak{g}>1}$).
The general attractor equations imply in particular that the horizon radius is given by
\begin{align}\label{entropy extremization}
	L^2_{\Sg} =  i \kappa \frac{\cal Z}{\cal L} = - i \frac{ \sum_I ( e_I X^I - p^I F_I) }{g (X^1+X^2+X^3)}\ . 
\end{align}
where in the last equality we already used the model specific information that $X^0 - c F_0 = 0$ which implies $ X^0 = (- c)^{2/3} (X^1 X^2 X^3)^{1/3}$. Notice that the same equation is found by directly using the truncated prepotential, ${\cal F}^\star$, since by construction  
\begin{align}
\begin{split}
F_I ( X^0 &= (- c)^{2/3} (X^1 X^2 X^3)^{1/3}) = F^\star_I\ , \quad \forall I \in \{1,2,3\}\ , \\
\Rightarrow L^2_{\Sg} &=  i \kappa \frac{\cal Z^\star}{\cal L^\star}\ .
\end{split}
\end{align}
This shows that we can equally well use the truncated prepotential for this attractor equation. To solve it, we define the weighted sections $\hat{X}^I \equiv X^I/\sum_J X^J$ such that $\sum_I \hat{X}^I = 1$, and find
\begin{align}
\label{sugra:index}
	{\sum_{I=1}^{3}} \left( p^I \hat{F}^\star_I  -  e_I \hat{X}^I \right) = g L^2_{\Sg}\ , 
\end{align}
where we used the shorthand notation $F_I^\star (\hat{X}^I) \equiv \hat{F}^\star_I$. This expression is extremized at the horizon
\begin{align}\label{sugra:extremization}
	\left. \partial_{\hat{X}^J}\left[ \tsumI (p^I \hat{F}^\star_I - e_I \hat{X}^I) \right] \right|_{\hat{X}_{\rm horizon}} = 0\ , 
\end{align}
fixing the weighted sections, $\hat{X}^I_{\rm horizon} \equiv H^I$ in terms of the electric and magnetic charges. 

Let us now concentrate on what we call ``purely magnetic'' solution, \ie\;let us work under the assumption that we only have independent magnetic charges and all electric charges are equal $e_I = e$. The equations simplify to
\begin{align}\label{extremization}
	\left. \partial_{\hat{X}^J} \left( \tsumI p^I \hat{F}^\star_I \right)\right|_{H^I}  = 0\ , 
\end{align}
given $\sum_I \hat{X}^I = 1$. We find the following solutions:
\begin{equation}\label{Xsolution}
 3 H^I = 1 \pm \sum_{J, K} \frac{\left| \epsilon_{I J K} \big( p^J - p^K \big) \right|}{2 \sqrt{\big( \sqrt{\Theta} \pm p^I \big)^2 - p^J p^K}}\, ,
\end{equation}
where the $\pm$ signs are not correlated so we have four solutions.
Here $\epsilon_{I J K}$ is the Levi--Civita symbol and  
\begin{equation}\label{theta}
\Theta (p) \equiv \left(p^1\right)^2+\left(p^2\right)^2+\left(p^3\right)^2 - \left( p^1 p^2 + p^1 p^3 + p^2 p^3 \right) .
\end{equation}
The sign ambiguities are to be resolved in the full geometry as proper normalization of the scalar kinetic terms require that $\im(s, t, u) > 0$ everywhere in spacetime, including the horizon values. It is now straightforward to derive the physical scalars from the weighted sections $H^I$,
\begin{equation}
	s = \frac{\e^{i \pi/3} c^{1/3} H^1}{(H^1 H^2 H^3)^{1/3}}\ , \quad t = \frac{\e^{i \pi/3} c^{1/3} H^2}{(H^1 H^2 H^3)^{1/3}}\ , \quad u = \frac{\e^{i \pi/3} c^{1/3} H^3}{(H^1 H^2 H^3)^{1/3}}\ .
\end{equation}
At first it might seem that there is an ambiguity in the attractor equation, since at the moment we have allowed for an arbitrary parameter $e$ which sets the value of the three equal electric charges. This is however misleading, because we have in fact not yet solved the original equation \eqref{entropy extremization}. The electric charges there play the crucial r\^ole of making sure the radius of the horizon is indeed a positive real quantity,
\begin{align}
\begin{split}
\frac{\cal Z}{\cal L} &= -\frac{\kappa}{g} \left( (-1)^{4/3} c^{1/3} (H^1 H^2 H^3)^{2/3} \tsumI(p^I/H^I) - \tsumI e_I H^I \right) \\
&= -\frac{\kappa}{g} \left(\e^{-2 i \pi/3} c^{1/3} (H^1 H^2 H^3)^{2/3} \tsumI  (p^I/H^I )  - e \right)  \\
&= - i \kappa L^2_{\Sg}\ .
\end{split}
\end{align}
The imaginary part of the last equation fixes the radius of the Riemann surface,
\begin{equation}\label{area}
	L^2_{\Sg} = -\frac{\sqrt{3}}{2 g} c^{1/3} (H^1 H^2 H^3)^{2/3} \sum_{I=1}^3 \frac{p^I}{H^I} \, 
\end{equation}
while the real part fixes the value of the electric charges,
\begin{equation}
	e = \frac{1}{2} c^{1/3} (H^1 H^2 H^3)^{2/3} \sum_{I=1}^3 \frac{p^I}{H^I}  = - \frac{g}{\sqrt{3}} L^2_{\Sg} \ .
\end{equation}
However, \eqref{entropy extremization} can only get us this far, and one needs to solve the other near-horizon equations in order to write down the full solutions,
as we have done in appendix \ref{IIA:appendix}. This way one can fix the massive vector charges $p^0$, $e_0$, as well as the hypermultiplet scalar $\phi$:
\begin{equation}
	p^0 = c e_0 = \frac{g c^{1/3}}{3 \sqrt{3} (H^1 H^2 H^3)^{1/3}} L^2_{\Sg}\ ,  \qquad \e^{2 \phi} = \frac{2 c^{-2/3}}{3 (H^1 H^2 H^3)^{1/3}}\ .
\end{equation}
The AdS$_2$ radius is also fixed from the remaining near-horizon BPS equations analyzed in appendix \ref{IIA:appendix},
and it can also be expressed in terms of the functions $H^I$ as
\begin{equation}
	L_{{\rm AdS}_2} = \frac{3^{3/4} c^{1/6} (H^1 H^2 H^3)^{1/3}}{2 g}\ .
\end{equation}
Finally, for completeness, we write the Bekenstein-Hawking entropy for black holes with spherical horiozn $(\kappa=+1)$:\footnote{A precise
counting of microstates for $\fg=0$ case implies matching of the index and the entropy for all values of $\fg$ (see section 6 of \cite{Benini:2016hjo}).}
\be
 \label{BH:entropy:final}
 S_{\rm BH} = \frac{\text{Area}}{4 G_{\rm N}} = \frac{\pi L_{S^2}^2}{G_{\rm N}}
 = - \frac{\pi \sqrt{3}}{2 g G_{\rm N}} c^{1/3} (H^1 H^2 H^3)^{2/3} \sum_{I=1}^{3} \frac{p^I}{H^I} \, .
\ee

\subsection{Existence of full black hole flows}
The main challenge in constructing the full black hole spacetime interpolating between the UV and IR geometries we presented above, is the nontrivial massive vector field we need to consider. We have seen that in the constant scalar case we can effectively decouple the massive vector multiplet but this is not the case for the full flow, if we wish to have the most general spacetime. For the BPS equations, it is useful to define the function
\begin{equation}
	\gamma (r) \equiv c F_0 - X^0 = c + s(r) t(r) u(r)\ ,
\end{equation} 
which vanishes both in the UV and the IR. The function $\gamma (r)$ is in principle fixed by the BPS equations determining the scalars $s$, $t$, and $u$, and in turn governs the flow of the hypermultiplet scalar field $\phi$ as well as the massive vector charge $p^0$ via \eqref{phi} and \eqref{charge}, respectively. The remaining first order BPS equations involve also the metric functions $U$ and $\psi$, as well as the Killing spinor phase $\alpha$ while the conserved charges $e_{1,2,3}$ and $p^{1,2,3}$ remain constant and have been fixed already at the horizon. Therefore we have a total of eight coupled differential equations for eight independent variables\footnote{Note that in the ``purely magnetic'' ansatz the phase of the complex scalars has been fixed, therefore we count $s$, $t$, $u$ as each is carrying a single degree of freedom.} $\{ s$, $t$, $u$, $\phi$, $p^0$, $U$, $\psi$, $\alpha \}$.
All these fields have been uniquely fixed in the UV and IR as shown above and more carefully in appendix \ref{IIA:appendix}.
A similar set of equations with running hypermultiplet scalars has been considered in \cite{Halmagyi:2013sla} with the result that one can always connect the UV and IR solutions with a full numerical flow, whenever the number of free parameters matches the number of first order differential equations, as is also the case here. It is of course interesting to find such solutions explicitly but we leave this for a future investigation as the main scope here is the field theory match of our results, to which we turn now.

\section{Comparison of index and entropy}
\label{sec:index vs entropy}

Now we are in a position to confront the topologically twisted index of ${\rm D2}_k$,
to leading order in $N$, \eqref{index:SYM-CS} with the Bekenstein-Hawking entropy \eqref{BH:entropy:final}.
Let us first note that the relations between SCFT parameters $(N,k)$ and their supergravity duals in massive type IIA,
to leading order in the large $N$ limit, read%
\footnote{See for example \cite{Fluder:2015eoa,Guarino:2015jca}.}
\bea
 \label{free energy}
 \frac{m^{1/3} g^{-7/3}}{4 G_{\rm N}} = \frac{3^{2/3}}{2^{2/3} 5} k^{1/3} N^{5/3} \, , \qquad \qquad
 \frac{
 m
 }
 {g}
 = \left( \frac{3}{16 \pi^{3}} \right)^{1/5} k N^{1/5} \, .
\eea

From here on we set $q_j = q\, ,\forall j=1,2,3$.
The topologically twisted index of ${\rm D2}_k$ \eqref{index:SYM-CS} as a function of $\Delta_{2,3}$ is extremized for
\bea
 \label{hatDelta}
 \frac{3 \bar\Delta_{2}}{2 \pi} = 1 \mp \frac{\left|\fn_3 - \fn_1\right|}{\sqrt{\big( \sqrt{\Theta} \pm \fn_2 \big)^2 - \fn_1 \fn_3}} \, , \qquad \qquad
 \frac{3 \bar\Delta_{3}}{2 \pi} = 1 \mp \frac{\left|\fn_1 - \fn_2\right|}{\sqrt{\big( \sqrt{\Theta} \pm \fn_3 \big)^2 - \fn_1 \fn_2}} \, ,
\eea
where we defined the quantity
\bea
 \Theta \equiv \fn_1^2 + \fn_2^2 + \fn_3^2 - \left( \fn_1 \fn_2 + \fn_1 \fn_3 + \fn_2 \fn_3 \right) \, ,
\eea
which is symmetric under permutations of $\fn_j$. Upon identifying
\bea
 \label{Delta:Sections}
 \frac{\bar\Delta_j}{2 \pi} & = H^j \, , \\
 \fn_j & = 2 g p^j \, , \qquad q_j = - \frac{e_j}{2 g G_{\rm N}} \, , \quad \text{ for } \quad j=1,2,3 \, ,
\eea
\eqref{hatDelta} are precisely the values of the weighted holomorphic sections $\hat{X}^j$ at the horizon \eqref{Xsolution}.
The constraint $\sum_{j} \Delta_j \in 2 \pi \bZ$ is consistent with $\sum_j \hat{X}^j = 1$ valid in the bulk.
Plugging the values for the critical points \eqref{hatDelta} back into the Legendre transform of the partition function \eqref{index:SYM-CS},
and employing \eqref{free energy} 
we finally arrive at the conclusion that \eqref{main_result} holds true.
We thus found a precise statistical mechanical interpretation of the black hole entropy \eqref{BH:entropy:final}.
Obviously, the above analyses goes through for the most general case with three unequal electric charges
and different horizon topologies \cite{Benini:2016rke}.

It is worth stressing that the imaginary part of the partition function \eqref{index:SYM-CS} uniquely fixes
the value of the electric charges $q_j = q\, , \forall j=1,2,3$ such that its value at the critical point is a real positive quantity
in agreement with the supergravity attractor mechanism and the general expectations in \cite{Benini:2016rke}.
This precise holographic match therefore presents a new and successful check on the $\cI$-extremization principle
(see section \ref{sec:intro:I-extremization})
in the presence of a nontrivial phase which is new with respect to previous examples such as the index of ABJM.

%

\chapter[The Cardy limit of the topologically twisted index and black strings in AdS\texorpdfstring{$_{5}$}{(5)}]{The Cardy limit of the topologically twisted index and black strings in AdS$\fakebold{_5}$}
\label{ch:5}

\ifpdf
    \graphicspath{{Chapter5/Figs/Raster/}{Chapter5/Figs/PDF/}{Chapter5/Figs/}}
\else
    \graphicspath{{Chapter5/Figs/Vector/}{Chapter5/Figs/}}
\fi

\section{Introduction}

The large $N$ limit of general three-dimensional quivers with an AdS dual was studied in the previous chapters.
In this chapter we study the asymptotic behavior of the index, at finite $N$, for four-dimensional $\cN =1$ gauge theories.
With an eye on holography we also evaluate the index in the large $N$ limit.
We focus, in particular, on the class of $\cN =1$ theories  arising from D3-branes probing Calabi-Yau singularities,
which have a well-known holographic dual in terms of compactifications on Sasaki-Einstein manifolds (see section \ref{ch:1:intro:magnetic:BH}).

The explicit evaluation of the topologically twisted index is a hard task, even in the large $N$ limit.
However, the index greatly simplifies if we identify the modulus
$\tau=i \beta/2\pi$ of the torus $T^2$ with a \emph{fictitious} inverse temperature $\beta$,
and take the limit $\beta \to 0$. We will call this the \emph{high-temperature limit}.
Our finding implies a Cardy-like behavior of the topologically twisted index,
which is related to the modular properties of the elliptic genus \cite{Kawai:1993jk,Benjamin:2015hsa}.
Analogous behaviors for other partition functions have been found in
\cite{DiPietro:2014bca,Ardehali:2015hya,Ardehali:2015bla,Lorenzen:2014pna,Assel:2015nca,Bobev:2015kza,Genolini:2016sxe,DiPietro:2016ond,Brunner:2016nyk,Shaghoulian:2015kta,Shaghoulian:2015lcn,Closset:2017bse}.%

The rest of the chapter is organized as follows.
In section \ref{SYM} we analyze the high-temperature limit of the index for $\cN=4$ super Yang-Mills
while in section \ref{Klebanov-Witten} we discuss the example of the conifold.
Then in section \ref{high-temp limit of the index} we derive the formulae \eqref{tHoof:anomaly:0},
\eqref{index theorem:2d central charge0}, \eqref{Z large N conjecture0:4d} and \eqref{centralchargea0}.

\section[\texorpdfstring{$\cN = 4$}{N=4} super Yang-Mills]{$\fakebold{\cN = 4}$ super Yang-Mills}
\label{SYM}

We first consider the twisted compactification of four-dimensional $\cN = 4$ super Yang-Mills (SYM) with gauge group $\SU(N)$ on $S^2$.
At low energies, it results in a family of two-dimensional theories with $\cN = (0,2)$ supersymmetry depending on the twisting parameters $\fn$  \cite{Benini:2012cz,Benini:2013cda}.
The theory describes the dynamics of $N$ D3-branes wrapped on $S^2$ and can be pictured as the quiver gauge theory given in \eqref{SYM:quiver}.
\be
\begin{aligned}
\label{SYM:quiver}
\begin{tikzpicture}[font=\footnotesize, scale=0.9]
\begin{scope}[auto,%
  every node/.style={draw, minimum size=0.5cm}, node distance=2cm];
\node[circle]  (UN)  at (0.3,1.7) {$N$};
\end{scope}
\draw[decoration={markings, mark=at position 0.45 with {\arrow[scale=1.5]{>}}, mark=at position 0.5 with {\arrow[scale=1.5]{>}}, mark=at position 0.55 with {\arrow[scale=1.5]{>}}}, postaction={decorate}, shorten >=0.7pt] (-0,2) arc (30:341:0.75cm);
\node at (-2.2,1.7) {$\phi_{1,2,3}$};
\end{tikzpicture}
\end{aligned}
\ee
The superpotential 
\be
 \label{SYM:superpotential}
 W = \Tr \left( \phi_3 \left[ \phi_1, \phi_2 \right] \right) 
\ee
imposes the following constraints on the chemical potentials $\Delta_a$ and the flavor magnetic fluxes $\fn_a$ associated with the fields $\phi_a$,
\be
 \label{SYM:constraints}
 \sum_{a = 1}^{3} \Delta_a \in 2 \pi \mathbb{Z} \, , \qquad \qquad \sum_{a = 1}^{3} \fn_a = 2 \, .
 \ee 
The topologically twisted index for the $\SU(N)$ SYM theory is given by%
\footnote{We do not isolate the vacuum contribution --- the so called supersymmetric Casimir energy --- from the index (see section 3.3 of \cite{Closset:2017bse}).}
\begin{equation}
 \label{SYM path integral_constraint}
 Z = \frac{\cA}{N!} \;
 \sum_{\substack{\fm \,\in\, \mathbb{Z}^N \, , \\ \sum_i \fm_i = 0}} \; \int_\cC \;
 \prod_{i=1}^{N - 1} \frac{\rd x_i}{2 \pi i x_i} \prod_{j \neq i}^{N} \frac{\theta_1\left( \frac{x_i}{x_j} ; q\right)}{i \eta(q)}
 \prod_{a=1}^{3} \left[ \frac{i \eta(q)}{\theta_1\left( \frac{x_i}{x_j} y_a ; q\right)} \right]^{\fm_i - \fm_j - \fn_a + 1} \, ,
\end{equation}
where we defined the quantity
\be
\label{SYM:A}
\cA = \eta(q)^{2 (N-1)} \prod_{a=1}^{3} \left[ \frac{i \eta(q)}{\theta_1\left(y_a ; q\right)} \right]^{(N-1) (1 - \fn_a)} \, .
\ee
Here, we already imposed the $\SU(N)$ constraint $\prod_{i = 1}^{N} x_i = 1$.
Instead of performing a constrained sum over gauge magnetic fluxes we introduce 
the Lagrange multiplier $w$ and consider an unconstrained sum.
Thus, the index reads
\begin{equation}
 \label{SYM path integral}
 Z = \frac{\cA}{N!} \;
 \sum_{\fm \,\in\, \mathbb{Z}^N} \; \int_\cB \frac{\rd w}{2 \pi i w} w^{\sum_{i = 1}^{N} \fm_i} \; \int_\cC \;
 \prod_{i=1}^{N - 1} \frac{\rd x_i}{2 \pi i x_i} \prod_{j \neq i}^{N} \frac{\theta_1\left( \frac{x_i}{x_j} ; q\right)}{i \eta(q)}
 \prod_{a=1}^{3} \left[ \frac{i \eta(q)}{\theta_1\left( \frac{x_i}{x_j} y_a ; q\right)} \right]^{\fm_i - \fm_j - \fn_a + 1} \, .
\end{equation}
In order to evaluate \eqref{SYM path integral}, we employ the strategy introduced in section \ref{ch:1:twisted index:solving}.
The Jeffrey-Kirwan residue picks a middle-dimensional contour in $\left(\mathbb{C}^{*}\right)^N$.
We can then take a large positive integer $M$ and resum the contributions $\fm \leq M-1$.
Performing the summations we get
\be
\begin{aligned}
Z = \frac{\cA}{N!} \; \int_\cB \frac{\rd w}{2 \pi i w} \;
\int_\cC \; \prod_{i=1}^{N - 1} \frac{\rd x_i}{2 \pi i x_i} \; \prod_{i=1}^{N} \frac{\left( \e^{i B_i} \right)^M}{\e^{i B_i} - 1}
\prod_{j \neq i}^{N} \frac{\theta_1\left( \frac{x_i}{x_j} ; q\right)}{i \eta(q)}
\prod_{a=1}^{3} \left[ \frac{i \eta(q)}{\theta_1\left( \frac{x_i}{x_j} y_a ; q\right)} \right]^{ 1- \fn_a} \, ,
\label{SYM:index:2}
\end{aligned}
\ee
where we defined
\be
\label{iB}
\e^{i B_i} = w \prod_{j = 1}^{N} \prod_{a=1}^{3} \frac{ \theta_1 \left( \frac{x_j}{x_i} y_a ; q \right)}
{\theta_1 \left( \frac{x_i}{x_j} y_a ; q \right)} \, .
\ee
In picking the residues, we need to insert a Jacobian in the partition function and evaluate everything else at the poles,
which are located at the solutions to the BAEs,
\be
\label{BAE}
 \e^{i B_i} = 1 \, ,
\ee
such that the off-diagonal vector multiplet contribution does not vanish. We consider \eqref{BAE} as a system of 
$N$ independent equations with respect to $N$ independent variables $\{x_1,\dots,x_{N-1},w\}$. 
In the final expression, the dependence on the cut-off $M$ disappears and we find
\be
Z = \cA \sum_{I \in\mathrm{BAEs}}\frac{1}{\mathrm{det}\mathds{B}} \prod_{j \neq i}^{N}\frac{\theta_1\left(\frac{x_i}{x_j};q \right)}{i\eta(q)}
\prod_{a = 1}^3\left[ \frac{i\eta(q)}{\theta_1\left(\frac{x_i}{x_j}y_a;q \right)} \right]^{1-\fn_a} \, ,
\label{index:SYM:bethe}
\ee
where the summation is over all solutions $I$ to the BAEs \eqref{BAE}.
The matrix $\mathds{B}$ appearing in the Jacobian has the following form
\be
\mathds{B} = \frac{\partial \left( \e^{i B_1}, \ldots, \e^{i B_N} \right)}
{\partial \left( \log x_1, \ldots, \log x_{N-1}, \log w \right)} \, .
\ee

\subsection{Twisted superpotential at high temperature}
\label{twisted superpotential_SYM}

In this section we study the \emph{high-temperature} limit $(q \to 1)$ of the twisted superpotential.
Let us start by considering the BAEs \eqref{BAE} at high temperature.
Taking the logarithm of the BAEs \eqref{BAE}, we obtain
\be
\begin{aligned}
\label{SYM:BAE:logarithm}
0  = - 2 \pi i n_i + \log w - \sum_{j = 1}^{N} \sum_{a=1}^{3} \left\{ \log\left[ \theta_1\left(\frac{x_i}{x_j}y_a;q \right)\right]
- \log\left[ \theta_1\left(\frac{x_j}{x_i}y_a;q \right)\right] \right\} \, ,
\end{aligned}
\ee
where $n_i$ is an integer that parameterizes the angular ambiguity.
It is convenient to use the variables $u_i$, $\Delta_a$, $v$, defined modulo $2 \pi$:
\be
x_i = \e^{i u_i}\, ,\qquad \qquad y_a = \e^{i \Delta_a} \, ,\qquad \qquad w = \e^{i v} \, .
\ee
Then, using the asymptotic formul\ae\;\eqref{dedekind:hight:S} and \eqref{theta:hight:S}
we obtain the high-temperature limit of the BAEs \eqref{SYM:BAE:logarithm}, up to exponentially suppressed corrections,
\be
\begin{aligned}
0= - 2 \pi i n_i + i v + \frac{1}{\beta}\sum_{j = 1}^{N} \sum_{a=1}^3 \left[ F' \left(u_i - u_j + \Delta_a\right) - F' \left(u_j - u_i + \Delta_a\right) \right] \, ,
\label{BAE:SYM:hight}
\end{aligned}
\ee
where $i / (2 \pi \tau) = 1/\beta$ is the formal ``temperature'' variable.
Here, we have introduced the polynomial functions
\be
\begin{aligned}
F(u) = \frac{u^3}{6} - \frac{1}{2}\pi u^2 \mathrm{sign}[\re(u)] + \frac{\pi^2}{3} u \, , \quad \quad 
F'(u) = \frac{u^2}{2} - \pi u \sign[\re(u)] + \frac{\pi^2}{3} \, .
\label{F:function}
\end{aligned}
\ee

The high-temperature limit of the twisted superpotential can be found directly
by integrating the BAEs \eqref{BAE:SYM:hight} with respect to $u_i$ and summing over $i$.
It reads
\bea
 \label{Bethe:potential:SYM:hight}
 \wt\cW (\{u_i\}) & = \sum_{i = 1}^{N} \left( 2 \pi n_i - v \right) u_i
 + \frac{i (N - 1)}{\beta} \sum_{a = 1}^{3} F \left( \Delta_a \right)
 \\ & + \frac{i}{2 \beta}\sum_{i \neq j}^{N} \sum_{a = 1}^{3}
 \left[ F \left( u_i - u_j + \Delta_a \right) + F \left( u_j - u_i + \Delta_a \right) \right]
 \, .
\eea
It is easy to check that the BAEs \eqref{BAE:SYM:hight} can be obtained as critical points of the above twisted superpotential. We introduced a $\Delta_a$-dependent 
integration constant in order to have precisely one contribution  $F \left( u_i - u_j + \Delta_a \right)$ for each component of the adjoint multiplet.

It is natural to restrict the $\Delta_a$ to the fundamental domain. In the high-temperature limit,  we can assume that $\Delta_a$ are real and $0 < \Delta_a < 2 \pi$.
Moreover,  since \eqref{SYM:constraints} must hold, $\sum_{a=1}^3 \Delta_a$ can only be $0,2\pi,4\pi$ or  $6\pi$.
We have checked that $\sum_{a=1}^3 \Delta_a = 0, 6\pi$ lead to a singular index, and
those for $2\pi$ and $4\pi$ are related by a discrete symmetry of the index  \ie\;$y_a \to 1 / y_a \left( \Delta_a \to 2 \pi - \Delta_a \right)$.
Thus, without loss of generality, we will assume $\sum_{a=1}^{3} \Delta_a = 2 \pi$ in the following. 
 
\paragraph*{The solution for $\fakebold{\sum_{a} \Delta_a = 2 \pi}$.} 
We seek for solutions to the BAEs \eqref{BAE:SYM:hight} assuming that
\be
 0 < \re \left( u_j - u_i \right) + \Delta_{a} < 2 \pi \, , \qquad \forall \quad i, j, a \, .
\ee
Thus, the high-temperature limit of the BAEs \eqref{BAE:SYM:hight} takes the simple form
\be
\begin{aligned}
\frac{2}{\beta}\sum_{a = 1}^{3} \left( \Delta_a-\pi \right)\sum_{k = 1}^{N} \left( u_j - u_k \right) = i \left( 2\pi n_j-v \right) \, ,
\quad \mbox{for} \quad j=1,2, \ldots, N \, .
\end{aligned}
\ee
Imposing the constraints $\sum_{a = 1}^{3} \Delta_a = 2 \pi$ for the chemical potentials as well as $\SU(N)$ constraint $\sum_{i=1}^{N} u_i=0$ 
we obtain the following set of equations
\be
\begin{aligned}
 \label{BAE:SYM:hight:simp}
 \frac{i N}{\beta} u_j & = n_j-\frac{v}{2\pi}\, ,\quad \mbox{for}\quad  j=1,\dots,N-1 \, , \\ 
 -\frac{i N}{\beta}\sum_{j=1}^{N-1} u_j & = n_N-\frac{v}{2\pi} \, .
 \end{aligned}
\ee
Summing up all equations we obtain the solution for $v$, which is given by 
\be
v=\frac{2\pi}{N}\sum_{i=1}^N n_i\, .
\label{SYM:solution:v}
\ee
The solution for eigenvalues $u_i$ reads 
\be
u_i=-\frac{i \beta}{N}\left( n_i-\frac{1}{N}\sum_{i=1}^N n_i \right)\,.
\label{SYM:solution:u}
\ee
Notice that, the tracelessness condition is automatically satisfied in this case.

To proceed further, we need to provide an estimate on the value of the constants $n_i$.
Whenever any two integers are equal $n_i = n_j$, we find that the off-diagonal vector multiplet contribution to the index,
which is an elliptic generalization of the Vandermonde determinant, vanishes.
Moreover, the high-temperature expansion \eqref{theta:hight:S} breaks down as subleading terms start blowing up.
Hence, we should make another ansatz for the phases $n_i$ such that
\be
n_i-n_j\neq 0 ~ ~ \mbox{mod} ~ ~ N \, .
\label{condition:n}
\ee
To understand how much freedom we have, let us first note that 
eigenvalues $u_i$ are variables defined on the torus $T^2$ and thus they should be periodic in $\beta$.
Due to \eqref{SYM:solution:u}, this means that integers $n_i$ are defined modulo $N$ and hence, without loss 
of generality, we can consider only integers lying in the domain $\left[ 1,N \right]$ with the condition 
\eqref{condition:n} modified to $n_i\neq n_j \, ,\forall~ i,j$.
This leaves us with the only choice $n_i=i$ and its permutations.

Substituting \eqref{SYM:solution:u} and \eqref{SYM:solution:v}
into the twisted superpotential \eqref{Bethe:potential:SYM:hight}, we obtain
\be
 \label{SYM:onshel:twisted superpotential}
 \wt\cW( \Delta_a ) \equiv - i \wt\cW (\{u_i\}, \Delta_a) \big|_{\text{BAEs}}
 = \frac{\left( N^2 - 1 \right)}{\beta} \sum_{a = 1}^{3} F \left( \Delta_a \right)
 = \frac{\left( N^2 - 1 \right)}{2 \beta} \Delta_1 \Delta_2 \Delta_3 \, ,
\ee
up to terms $\cO(\beta)$.

There is an interesting relation between the ``on-shell'' twisted superpotential \eqref{SYM:onshel:twisted superpotential} and the central charge of
the UV four-dimensional theory. Note that, given the constraint $\sum_{a=1}^{3} \Delta_{a} = 2 \pi$, the quantities $\Delta_a$ can be used to parameterize the most general R-symmetry of the theory
\be R(\Delta_a) = \sum_{a=1}^3  \Delta_a  \frac{R_a}{2\pi} \, ,\ee
where $R_a$ gives charge 2 to $\phi_a$ and zero to $\phi_b$ with $b\neq a$. Observe also that the cubic R-symmetry 't Hooft anomaly is given by
\be
\begin{aligned}
 \Tr R^3 \left( \Delta_a \right) & = \left(N^2 - 1\right) \left[ 1 + \sum _{a = 1}^{3} \left( \frac{\Delta_a}{\pi} - 1 \right)^3 \right]
 = \frac{3 \left( N^2 - 1 \right)}{\pi^3} \Delta_1 \Delta_2 \Delta_3 \, ,
\end{aligned}
\ee
where the trace is taken over the fermions of the theory.
Therefore, the ``on-shell'' value of the twisted superpotential \eqref{SYM:onshel:twisted superpotential} can be rewritten as
\be
\label{on:shell:Bethe:pot:hight:SYM}
\wt\cW(\Delta_a ) = \frac{\pi^3}{6 \beta} \Tr R^3 \left(\Delta_a \right) = \frac{16 \pi^3}{27 \beta} a \left(\Delta_a \right) \, ,
\ee
where in the second equality we used the relation \eqref{generalac}. Notice that the linear R-symmetry 't Hooft anomaly is zero for $\cN = 4$ SYM. 

\subsection{The topologically twisted index at high temperature}
\label{The index at high temperature_SYM}

We are interested in the high-temperature limit of the logarithm of the partition function \eqref{index:SYM:bethe}.
We shall use the asymptotic expansions \eqref{dedekind:hight:S} and \eqref{theta:hight:S} in order to calculate the vector and hypermultiplet contributions to the twisted index in the $\beta \to 0$ limit.

The contribution of the off-diagonal vector multiplets can be computed as
\begin{align}
\log \prod_{i \neq j}^N \left[ \frac{\theta_1\left( \frac{x_i}{x_j} ; q\right)}{i \eta(q)}\right] =
- \frac{1}{\beta} \sum_{i \neq j}^{N} F' \left( u_i - u_j \right) - \frac{ i N (N - 1) \pi}{2} \, ,
\end{align}
in the asymptotic limit $q \to 1\;(\beta \to 0)$. The contribution of the matter fields is instead
\be
\begin{aligned}
 \log \prod_{i \neq j}^N \prod_{a=1}^3 \left[ \frac{i\eta(q)}{\theta_1\left(\frac{x_i}{x_j}y_a;q \right)} \right]^{1-\fn_a} & =
 - \frac{1}{\beta} \sum_{i \neq j}^{N} \sum_{a = 1}^{3} \left[ \left(\fn_a - 1\right) F' \left( u_i - u_j + \Delta_a \right) \right] \\
 & + \frac{ i N (N-1) \pi}{2} \sum_{a = 1}^{3} \left( 1 -  \fn_a \right) \, , \qquad \textmd{as } \; \beta \to 0 \, .
\end{aligned}
\ee
The prefactor $\cA$ in the partition function \eqref{SYM:A} at high temperature contributes
\be
\begin{aligned}
 \log \left\{ \eta(q)^{2 (N-1)} \prod_{a=1}^{3} \left[ \frac{i \eta(q)}{\theta_1\left(y_a ; q\right)} \right]^{(N-1) (1 - \fn_a)} \right\} & =
 - \frac{N-1}{\beta} \left[ \frac{\pi^2}{3} + \sum_{a = 1}^{3} \left( \fn_a - 1 \right) F' (\Delta_a) \right] \\
 & - (N-1) \left[ \log \left( \frac{\beta}{2 \pi}\right) - \frac{i \pi}{2} \sum_{a = 1}^{3} \left( 1 -  \fn_a \right) \right] \, .
\end{aligned}
\ee
The last term to work out is  $ - \log \det \mathds{B}$.
The matrix $\mathds{B}$, imposing $\e^{i B_i} = 1$, reads 
\be
\begin{aligned}
\mathds{B} = \frac{\partial \left( B_1, \ldots, B_N \right)}{\partial \left( u_1, \ldots, u_{N-1}, v \right)}
\, , \qquad \textmd{as } \; \beta \to 0 \, ,
\label{Jacobian:SYM}
\end{aligned}
\ee
and has the following entries 
\be
\begin{aligned}
\frac{\partial B_k}{\partial u_j} & = \frac{2\pi i}{\beta}N\delta_{kj} \, ,\quad \mbox{for} \quad k,j=1,2,\dots,N-1\, , \\
\frac{\partial B_N}{\partial u_k} & = -\frac{2\pi i}{\beta}N \, ,\quad \frac{\partial B_k}{\partial v}=1 \, ,
\quad \mbox{for} \quad k = 1, 2, \dots,N-1\, ,\\
\frac{\partial B_N}{\partial v} &= 1 \, .
\end{aligned}
\ee
Here, we have already imposed the constraint $\sum_{a = 1}^{3} \Delta_a = 2 \pi$.
Therefore, we obtain
\be
- \log \det \mathds{B} = (N-1) \left[ \log \left( \frac{\beta}{2 \pi} \right) - \frac{i \pi}{2}\right] - N\log N  \, .
\ee
Putting everything together we can write the high-temperature limit of the twisted index, at finite $N$,
\be
\begin{aligned}
\label{SYM:index:final:hight}
 \log Z & = - \frac{1}{\beta} \sum_{i \neq j}^{N} \left[ F' \left( u_i - u_j \right) + \sum_{a=1}^3 \left( \fn_a - 1\right) F' \left(u_i - u_j + \Delta_a \right) \right] \\
 & - \frac{N - 1}{\beta} \left[ \frac{\pi^2}{3} + \sum_{a=1}^3 \left( \fn_a - 1 \right) F' \left( \Delta_a \right) \right] - N \log N \, ,
\end{aligned}
\ee
up to exponentially suppressed corrections. We may then evaluate the index by
substituting the pole configurations \eqref{SYM:solution:u} back into the functional \eqref{SYM:index:final:hight} to get,
\be
\begin{aligned}
 \label{SYM:logZ:bethe}
  \log Z & = - \frac{N^2 - 1}{\beta} \left[ \frac{\pi^2}{3} + \sum_{a=1}^3 \left( \fn_a - 1 \right) F' \left( \Delta_a \right) \right] - N \log N  \\
& = - \frac{N^2 - 1}{2 \beta} \sum_{\substack{ a < b \\ (\neq c) }} \Delta_a \Delta_b \fn_c - N \log N \, ,
\end{aligned}
\ee
which, to leading order in $1 / \beta$, can be rewritten in a more intriguing form:
\be
\label{SYM:index:final:hight:bethe}
 \log Z = - \sum_{a=1}^{3} \fn_a \frac{\partial \wt\cW(\Delta_a)}{\partial \Delta_a} \, .
\ee

We can relate the index to the trial left-moving central charge of the two-dimensional $\cN = (0, 2)$ theory on $T^2$. The latter reads \cite{Benini:2012cz,Benini:2013cda}
\be
c_l = c_r - k \, ,
\ee
where $k$ is the gravitational anomaly
\be
k = - \Tr \gamma_3 = - \left( N^2 - 1 \right) \left[ 1 + \sum_{a = 1}^{3} \left( \fn_a - 1 \right) \right] = 0 \, ,
\ee
and $c_r$ is the trial right-moving central charge
\be
\begin{aligned}
 \label{c2d:anomaly}
 c_{r} \left( \Delta_a , \fn_a \right) = - 3 \Tr  \gamma_3 R^2 \left( \Delta_a \right)
 &= - 3 \left( N^2 - 1 \right) \left[ 1 + \sum_{a=1}^{3} \left( \fn_a - 1 \right) \left( \frac{\Delta_a}{\pi} - 1 \right)^2 \right] \,  \\
 &= - \frac{3 \left( N^2 - 1 \right)}{\pi^2}  \sum_{\substack{ a < b \\ (\neq c) }} \Delta_a \Delta_b \fn_c \, .
\end{aligned}
\ee
Here, the trace is taken over the fermions and $\gamma_3$ is the chirality operator in two dimensions.
In the twisted compactification, each of the chiral fields $\phi_a$ give rise to two-dimensional fermions.
The difference between the number of fermions of opposite chiralities is $\fn_a-1$, thus explaining the above formulae.
We used $\Delta_a/\pi$ to parameterize the trial R-symmetry.   
We find that the index is given by
\be
\label{SYM:2d:4d:relation}
 \log Z = \frac{\pi^2}{6 \beta} c_{r} \left(\Delta_a , \fn_a \right)
 = - \frac{16 \pi^3}{27\beta}  \sum_{a=1}^{3} \fn_a \frac{\partial a(\Delta_a) }{\partial \Delta_a} \, .
\ee
As shown in \cite{Benini:2012cz,Benini:2013cda}, the \emph{exact} central charge of the two-dimensional CFT is obtained by extremizing
$c_{r} ( \Delta_a , \fn_a )$ with respect to the $\Delta_a$.
Given the above relation \eqref{SYM:2d:4d:relation}, we see that this is equivalent to extremizing the $\log Z$ at high temperature.
As a function of $\Delta_{1,2}$ the trial central charge $c_r ( \Delta_a , \fn_a )$ is extremized for
\be
 \frac{\Delta_{a}}{2 \pi} = \frac{2 \fn_a (\fn_a - 1)}{\Theta} \, , \qquad a=1,2 \, ,
\ee
where we defined the quantity
\be
 \Theta = \fn_1^2 + \fn_2^2 + \fn_3^2 - 2 (\fn_1 \fn_2 + \fn_1 \fn_3 + \fn_2 \fn_3) \, .
\ee
At the critical point the function takes the value
\be
 c_r(\fn_a) = 12 (N^2 - 1) \frac{\fn_1 \fn_2 \fn_3 }{\Theta} \, .
\ee

\subsection{Towards quantum black hole entropy}

As we discussed in chapter \ref{ch:1}, a four-dimensional black hole asymptotic to a curved domain wall
can be regarded as a Kaluza-Klein compactification of a black string in AdS$_5$.
The AdS$_3$ near-horizon region is dual to the Ramond sector of the $(0,2)$ SCFT
which lives on the dimensionally reduced D3 worldvolume.
The equivariant elliptic genus of the 4D black hole $Z_{\text{BH}} = Z_{\text{CFT}}$ as a Ramond sector trace reads%
\footnote{We give the fermions periodic boundary conditions.}
\be
 \label{ch:1:elliptic genus}
 Z_{\text{BH}} (y_a , q) = \Tr_{\text{R}} (-1)^F q^{L_0} \prod_a y_a^{J_a} \, .
\ee
The microcanonical density of states (up to exponentially suppressed contributions) is then given by \eqref{ch:1:micro:density}:%
\footnote{The precise choice of $\beta$ contour is not important if we only concern ourselves with the asymptotic expansion of the
$\beta$ integral for $ c_r(\fn_a) |q_0| \to + \infty$.}
\bea
 \label{N=4:SYM:micro:density}
 d_{\text{micro}} (\fn_a , q_0) & = \frac{i}{N^{N}} \int \frac{{\rm d} \beta}{2 \pi}
 \int \frac{{\rm d}^3 \Delta_a}{(2 \pi)^3} \, \delta\Big(2 \pi - \sum_{a} \Delta_a\Big) \,
 \e^{\frac{\pi^2}{6 \beta} c_r(\Delta_a , \fn_a) + \beta q_0} \, .
\eea
Without loss of generality we assume that $\fn_1 < 0$, $\fn_2  > 0$, $\fn_3 < 0$ (thus $c_r (\fn_a) > 0$) and $q_0 > 0$.
We first perform the integral over $\Delta_3$.
The integral over $\Delta_1$ diverges for real values.
This can be avoided by rotating the integration contour to $0 + i \bR$.
The integral \eqref{N=4:SYM:micro:density} is now Gaussian on $\Delta_{1,2}$, leading to
\bea
 d_{\text{micro}} (\fn_a, q_0) & =
 \frac{i}{2 \pi^2 N^N (N^2 - 1) \sqrt{\Theta}}
 \int {\rm d} \beta \, \e^{\frac{\pi^2}{6 \beta} c_r(\fn_a) + \beta q_0} \beta \, .
\eea
The integral over $\beta$ is of Bessel type and it gives
\be
 \label{ch:1:micro:Bessel}
 d_{\text{micro}} (\fn_a, q_0)
 =  \frac{\pi}{N^N (N^2 - 1) \sqrt{\Theta}}
 \left( \frac{\pi}{3} \frac{c_r (\fn_a)}{S_{\text{Cardy}}} \right)^{2} \,
 I_{2} \left( S_{\text{Cardy}} \right) \, ,
\ee
where
\be
 S_{\text{Cardy}} = 2 \pi \sqrt{\frac{c_r (\fn_a) q_0}{6}} \, .
\ee
Eq.\,\eqref{ch:1:micro:Bessel} captures all power-law and logarithmic corrections to the leading Bekenstein-Hawking
entropy \emph{exactly} to all orders.
$I_\nu(z)$ is the standard modified Bessel function of the first kind
and has the following integral representation:
\be
 \label{integral:rep:Bessel}
 I_\nu (z) = \frac{1}{2 \pi i} \left( \frac{z}{2} \right)^\nu
 \int_{\epsilon - i \infty}^{\epsilon + i \infty}
 \rd t \, \e^{t + \frac{z^2}{4 t}} t^{- \nu - 1} \, , \qquad \text{for } \quad \re(\nu) > 0 \, , \epsilon >0 \, . 
\ee
Furthermore, the asymptotics of $I_\nu (z)$ for large $\re(\nu)$ is given by
\be
 \label{asymp:Bessel}
 I_\nu (z) \sim \frac{\e^{z}}{\sqrt{2 \pi z}} \left[ 1 - \frac{(\mu - 1)}{8 z} + \frac{(\mu - 1)(\mu - 3^2)}{2! (8z)^2}
 - \frac{(\mu - 1)(\mu - 3^2)(\mu - 5^2)}{3! (8z)^3} + \ldots \right] ,
\ee
where $\mu = 4 \nu^2$.
The exponential term gives the Cardy formula and $S_{\text{Cardy}}$ can be
identified with the Bekenstein-Hawking entropy of the 4D black hole.

Let us note that we find strong similarities between the result of this section and those for asymptotically flat black holes
in \emph{ungauged} supergravity. In particular, the quantum entropy \cite{Sen:2008yk,Sen:2008vm} of BPS black holes
in $\cN=2$ supergravity coupled to vector multiplets and hypermultiplets is schematically of the form \cite{Dabholkar:2011ec}
\be
 \label{flat:BH:entropy}
 \cE(\fn) (S_{\text{Cardy}})^{- \nu} I_{\nu} \left( S_\text{Cardy} \right)
 \, .
\ee
The prefactor $\cE(\fn)$ only depends on the magnetic charges $\fn$ and not on the electric charges $q$.
This is indeed the first term in the Rademacher expansion, which is an exact formula for the Fourier coefficients $d_{\text{micro}}$ of modular forms.
The higher terms in the expansion are exponentially suppressed with respect to the terms in \eqref{flat:BH:entropy} and thus are nonperturbative.
It would be interesting to find the Rademacher expansion of the twisted index \eqref{index:SYM:bethe} and compare it with \eqref{ch:1:micro:Bessel}.

\section{The Klebanov-Witten theory}
\label{Klebanov-Witten}

In this section we study the Klebanov-Witten theory \cite{Klebanov:1998hh} describing the low energy dynamics of $N$ D3-branes at the conifold singularity.
This is the Calabi-Yau cone over the homogeneous Sasaki-Einstein five-manifold $T^{1,1}$ which can be described as the coset $\SU(2) \times \SU(2) / \U(1)$.
This theory has $\cN = 1$ supersymmetry and has $\SU(N) \times \SU(N)$ gauge group with bi-fundamental chiral multiplets $A_i$ and $B_j$, $i , j = 1, 2$, transforming in the $\left( {\bf N}, \overline{\bf N} \right)$ and $\left( \overline{\bf N}, {\bf N} \right)$ representations of the two gauge groups. This can be pictured as
\be
\begin{aligned}
\label{KW:quiver}
\begin{tikzpicture}[baseline, font=\footnotesize, scale=0.8]
\begin{scope}[auto,%
  every node/.style={draw, minimum size=0.5cm}, node distance=2cm];
\node[circle] (USp2k) at (-0.1, 0) {$N$};
\node[circle, right=of USp2k] (BN)  {$N$};
\end{scope}
\draw[decoration={markings, mark=at position 0.9 with {\arrow[scale=1.5]{>}}, mark=at position 0.95 with {\arrow[scale=1.5]{>}}}, postaction={decorate}, shorten >=0.7pt]  (USp2k) to[bend left=40] node[midway,above] {$A_{i}$} node[midway,above] {} (BN) ; 
\draw[decoration={markings, mark=at position 0.1 with {\arrow[scale=1.5]{<}}, mark=at position 0.15 with {\arrow[scale=1.5]{<}}}, postaction={decorate}, shorten >=0.7pt]  (USp2k) to[bend right=40] node[midway,above] {$B_{j}$}node[midway,above] {}  (BN) ;  
\end{tikzpicture}
\end{aligned}
\ee
It has a quartic superpotential,
\be
\label{KW:superpotential}
W = \Tr \left( A_1 B_1 A_2 B_2 - A_1 B_2 A_2 B_1 \right) \, .
\ee
We assign chemical potentials $\Delta_{1,2} \in (0, 2 \pi)$ to $A_i$ and $\Delta_{3,4} \in (0, 2 \pi)$ to $B_i$.
Invariance of the superpotential under the global symmetries requires
\begin{align}
\label{KW:constraints}
\sum_{I = 1}^{4} \Delta_I \in 2 \pi \mathbb{Z} \, , \qquad \qquad \sum_{I = 1}^{4} \fn_I = 2 \, .
\end{align}
For the Klebanov-Witten theory, the topologically twisted index can be written as
\begin{multline}
\label{initial Z full}
Z = \frac1{(N!)^2} \sum_{\fm, \wt\fm \in \bZ^N} \int_\cB \; \frac{\rd w}{2 \pi i w}\frac{\rd \wt w}{2\pi i \wt w } w^{\sum_{i = 1}^N  \fm_i}\, 
\wt w^{ \sum_{i=1}^N\wt\fm_i  }  \times \\
\times \int_\cC \; \prod_{i=1}^{N-1}  \frac{\rd x_i}{2\pi i x_i} \,  \frac{\rd \tilde x_i}{2\pi i \tilde x_i} \eta(q)^{4 (N-1)}
\prod_{i\neq j}^N \left[ \frac{\theta_1\left( \frac{x_i}{x_j} ; q\right)}{i \eta(q)} \frac{\theta_1\left( \frac{\tilde x_i}{\tilde x_j} ; q\right)}{i \eta(q)}\right] \times \\
\times \prod_{i,j=1}^N \prod_{a=1,2}
\left[ \frac{i \eta(q)}{\theta_1\left( \frac{x_i}{\tilde x_j} y_a ; q\right)} \right]^{\fm_i - \wt\fm_j - \fn_a +1}
\prod _{b=3,4} \left[ \frac{i \eta(q)}{\theta_1\left( \frac{\tilde x_j}{x_i} y_b ; q\right)} \right]^{\wt\fm_j - \fm_i - \fn_b +1} \, .
\end{multline}
Here, we assumed that eigenvalues $x_i$ and $\tx_i$ satisfy the $\SU(N)$ constraint $\prod_{i=1}^N x_i=\prod_{i=1}^N \tx_i=1$. 
Hence, the integration is performed over $(N-1)$ variables instead of $N$. In order to impose 
the $\SU(N)$ constraints for the magnetic fluxes, \ie\;
\be
\begin{aligned}
\label{KW:tracelessness condition}
\sum_{i = 1}^{N} \fm_i = \sum_{i = 1}^{N} \wt \fm_i = 0 \, ,
\end{aligned}
\ee
we have introduced two Lagrange multipliers $w = \e^{i v}$ and $\wt w=\e^{i\tv}$.
Now, we can resum over gauge magnetic fluxes $\fm_i \leq M - 1$ and $\wt \fm_j \geq 1 - M$ 
for some large positive integer cut-off $M$. We obtain 
\be
\begin{aligned}
Z = \frac1{(N!)^2}  \int_\cB & \; \frac{\rd w}{2 \pi i w} \, \frac{\rd \wt w}{2\pi i \wt w } \int_\cC \; \prod_{i=1}^{N-1} \frac{\rd x_i}{2\pi i x_i} \, \frac{\rd \tilde x_i}{2\pi i \tilde x_i}
\prod_{i\neq j}^N \left[ \frac{\theta_1\left( \frac{x_i}{x_j} ; q\right)}{i \eta(q)} \frac{\theta_1\left( \frac{\tilde x_i}{\tilde x_j} ; q\right)}{i \eta(q)}\right] \times \\
& \times \cP \, \prod_{i=1}^N \frac{ (\e^{iB_i})^M}{ \e^{iB_i} -1 } \prod_{j=1}^N \frac{ (\e^{i\wt B_j})^M}{ \e^{i\wt B_j} - 1} \, ,
\end{aligned}
\ee
where we defined the quantities
\be
\label{KW:A}
\cP = \eta(q)^{4 (N-1)} \prod_{i,j=1}^N \, \prod_{a=1,2} \left[ \frac{i \eta(q)}{\theta_1\left( \frac{x_i}{\tilde x_j} y_a ; q\right)} \right]^{1-\fn_a}
\prod_{b=3,4} \left[ \frac{i \eta(q)}{\theta_1\left( \frac{\tilde x_j}{x_i} y_b ; q\right)} \right]^{1-\fn_b} \, ,
\ee
and
\be
\begin{aligned}
\label{BA expressions}
\e^{iB_i} = w\prod_{j=1}^N \frac{\prod _{b=3,4} \theta_1\left( \frac{\tilde x_j}{x_i} y_b ; q\right)}
{\prod _{a=1,2} \theta_1\left( \frac{x_i}{\tilde x_j} y_a ; q\right)}  \, ,\qquad
\e^{i\wt B_j} = \wt w^{-1} \,\prod_{i=1}^N   \frac{\prod _{b=3,4} \theta_1\left( \frac{\tilde x_j}{x_i} y_b ; q\right)}{\prod _{a=1,2}
\theta_1\left( \frac{x_i}{\tilde x_j} y_a ; q\right)} \, .
\end{aligned}
\ee
Then, similarly to the case of $\mathcal{N}=4$ SYM, the following BAEs
\be
\label{KW:BAEs}
\e^{i B_i} = 1 \, ,\qquad \qquad \qquad \e^{i \wt B_j} = 1 \, .
\ee
determine the poles of the integrand.
In order to calculate the index we simply insert a Jacobian 
of the transformation from $\{\log x_i,\log \tx_i, \log w,\log \wt w\}$ to $\{\e^{i B_i},\e^{i \wt{B}_i}\}$ variables and evaluate everything else  
at the solutions to BAEs. In the final expression, the dependence on the cut-off $M$ disappears.
We can then write the partition function as,
\be
\label{index:KW:bethe}
Z = \sum_{I \in\mathrm{BAEs}}\frac{1}{\mathrm{det}\mathds{B}}
\prod_{i\neq j}^N \left[ \frac{\theta_1\left( \frac{x_i}{x_j} ; q\right)}{i \eta(q)} \frac{\theta_1\left( \frac{\tilde x_i}{\tilde x_j} ; q\right)}{i \eta(q)}\right] \cP \, ,
\ee
where $\mathds{B}$ is a $2N \times 2N$ matrix
\be
\label{Jacobian general:KW}
\mathds{B} = \frac{\partial \big( \e^{iB_1},\dots,\e^{i B_N}, \e^{i\wt B_1},\dots, \e^{i\wt B_N} \big) }
{ \partial\left( \log x_1,\dots,\log x_{N-1},\log w, \log \tx_1,\dots, \log \tx_{N-1}, \log \wt w\right)} \, .
\ee

\subsection{Twisted superpotential at high temperature}
\label{The twisted superpotential_KW}

Let us now look at the twisted superpotential at high temperature, \ie\;$\beta \to 0$ limit.
Taking the logarithm of the BAEs \eqref{KW:BAEs} we obtain
\be
\begin{aligned}
\label{KW:BAE:logarithm}
0 & = - 2 \pi i n_i + \log w - \sum_{j = 1}^{N} \left\{ \sum_{a = 1,2} \log\left[ \theta_1\left( \frac{x_i}{\tilde x_j} y_a ; q\right)\right]
- \sum_{b = 3,4} \log\left[ \theta_1\left( \frac{\tilde x_j}{x_i} y_b ; q\right)\right]  \right\} \, , \\
0 & = - 2 \pi i \tilde n_j - \log \wt w - \sum_{i = 1}^{N} \left\{ \sum_{a = 1,2} \log\left[ \theta_1\left( \frac{x_i}{\tilde x_j} y_a ; q\right) \right]
- \sum_{b = 3,4} \log\left[ \theta_1\left( \frac{\tilde x_j}{x_i} y_b ; q\right) \right]  \right\} \, ,
\end{aligned}
\ee
where $n_i\, , \tilde n_j$ are integers that parameterize the angular ambiguities.
In order to compute the high-temperature limit of the above BAEs, we go to the variables $u_i \, , \tilde u_j \, , \Delta_I \, , v\, , \tv$,
defined modulo $2 \pi$,
and employ the asymptotic expansions \eqref{dedekind:hight:S} and \eqref{theta:hight:S}.
We find
\be
\begin{aligned}
\label{BAE:KW:hight}
0 & = -2\pi i n_i +iv+ \frac{1}{\beta} \sum_{j = 1}^{N} \left[ \sum_{a = 1,2} F' \left(u_i - \tilde u_j + \Delta_a\right) 
- \sum_{b = 3,4} F' \left(\tilde u_j - u_i + \Delta_b\right)  \right] 
\, , \\
0 & = -2\pi i \tilde n_j - i \tv +  \frac{1}{\beta} \sum_{i = 1}^{N} \left[ \sum_{a = 1,2} F' \left(u_i - \tilde u_j + \Delta_a\right) 
- \sum_{b = 3,4} F' \left(\tilde u_j - u_i + \Delta_b\right)\right] 
\, ,
\end{aligned}
\ee
where the polynomial function $F' (u)$ is defined in \eqref{F:function}.
The BAEs \eqref{BAE:KW:hight} can be obtained as critical points of the twisted superpotential
\be
\label{KW:twisted superpotential}
\begin{aligned}
 \wt\cW(\{u_i, \tilde u_i\}, \Delta_I) & = 2\pi \sum_{i = 1}^{N} \left( n_i u_i - \tilde n_i \tilde u_i \right) -
  \sum_{i=1}^{N} \left( v\, u_i+\tv\,\tu_i \right) \\
  & + \frac{i}{\beta} \sum_{i,j=1}^{N} \left[ \sum_{a=1,2} F(u_i - \tilde u_j + \Delta_a) + \sum_{b=3,4} F(\tilde u_j - u_i + \Delta_b) \right]\, .
\end{aligned}
\ee

We next turn to find solutions to the BAEs \eqref{BAE:KW:hight}. The constraints (\ref{KW:constraints}) imply that 
$\sum_{I=1}^4 \Delta_I$ can only be $0,2\pi,4\pi,6\pi$ or  $8\pi$.
For the conifold theory, it turns out that the solutions with $\sum_{I=1}^4 \Delta_I = 0, 8\pi$ lead to a singular index,
those for $2 \pi$ and $6 \pi$ are related by a discrete symmetry of the index, \ie\;$y_I \to 1 / y_I \left( \Delta_I \to 2 \pi - \Delta_I \right)$, and
there are no consistent solutions for $\sum_{I=1}^4 \Delta_I = 4\pi$.
Thus, without loss of generality, we  assume again $\sum_{I=1}^4 \Delta_I = 2 \pi$ in the following. 

\paragraph*{The solution for $\fakebold{\sum_I \Delta_I = 2 \pi}$.}
We assume that
\be
\begin{aligned}
0 < \re \left( \tilde u_j - u_i \right) + \Delta_{3,4} < 2 \pi \, , \qquad \qquad - 2 \pi < \re \left( \tilde u_j - u_i \right) - \Delta_{1,2} < 0 \, , \qquad \forall \quad i, j \, .
\end{aligned}
\ee
Hence, the BAEs \eqref{BAE:KW:hight} become
\be
\begin{aligned}
\label{BAE:KW:hight:simplified}
0 & = - 2 \pi i n_j + i\,v- \frac{1}{\beta} \sum_{k = 1}^{N} \left[ \Delta_1 \Delta_2 - \Delta_3 \Delta_4 - 2 \pi  \left(\tilde u_k - u_j \right) \right] \, , \\
0 & = - 2 \pi i \tilde n_k -i\,\tv -\frac{1}{\beta} \sum_{j = 1}^{N} \left[ \Delta_1 \Delta_2 - \Delta_3 \Delta_4 - 2 \pi  \left(\tilde u_k - u_j \right)  \right] \, .
\end{aligned}
\ee
Here, we have already imposed the constraint $\sum_{I=1}^{4} \Delta_I = 2 \pi$.
Imposing the $\SU(N)$ constraints for $u_i\,,\, \tu_i$ we can rewrite the BAEs in the following form 
\begin{subequations}
 \begin{align}
  \label{BAE:KW:1}
  \frac{iN}{\beta}\,u_j &=n_j-\frac{v}{2\pi}+\frac{iN}{2\pi\beta}\left( \Delta_3\,\Delta_4-\Delta_1\,\Delta_2 \right)\, ,\quad \mbox{for}\quad j=1,\dots,N-1\, ,\\
  \label{BAE:KW:2}
  -\frac{iN}{\beta}\,\sum_{j=1}^{N-1} u_j &=n_N-\frac{v}{2\pi}+\frac{iN}{2\pi\beta}\left( \Delta_3\,\Delta_4-\Delta_1\,\Delta_2 \right)\, , \\
  \label{BAE:KW:3}
  \frac{iN}{\beta}\,\tu_j &=- \tilde n_j-\frac{\tv}{2\pi}-\frac{iN}{2\pi\beta}\left( \Delta_3\,\Delta_4-\Delta_1\,\Delta_2 \right)\, ,\quad \mbox{for}\quad j=1,\dots,N-1\, ,\\
  \label{BAE:KW:4}
  -\frac{iN}{\beta}\,\sum_{j=1}^{N-1} \tu_j &=-\tilde n_N-\frac{\tv}{2\pi}-\frac{iN}{2\pi\beta}\left( \Delta_3\,\Delta_4-\Delta_1\,\Delta_2 \right)\, .
 \end{align}
\end{subequations}
Equations \eqref{BAE:KW:1} and \eqref{BAE:KW:3} can be considered as equations defining $u_i$ and $\tu_i$. In order to find $v$ and $\tv$ we need
to sum $(N-1)$ equations \eqref{BAE:KW:1} with \eqref{BAE:KW:2} and equations  \eqref{BAE:KW:3} with \eqref{BAE:KW:4}. This leads to 
\be
\begin{aligned}
\label{KW:sol:final}
v & =\frac{iN}{\beta}\left( \Delta_3\,\Delta_4-\Delta_1\,\Delta_2 \right)+\frac{2\pi}{N}\sum_{j=1}^N n_j \, , \qquad
u_j =-\frac{i\beta}{N}\left( n_j-\frac{1}{N}\sum_{i=1}^N n_i \right) \, , \\
\tilde v & =-\frac{iN}{\beta}\left( \Delta_3\,\Delta_4-\Delta_1\,\Delta_2 \right)-\frac{2\pi}{N}\sum_{j=1}^N \tilde n_j \, ,\qquad
\tilde u_j =\frac{i\beta}{N}\left( \tilde n_j-\frac{1}{N}\sum_{i=1}^N \tilde n_i \right)\, .
\end{aligned}
\ee
According to our prescription, all solutions which lead to zeros of the off-diagonal vector multiplets should be avoided.
Therefore, the allowed parameter space for integers $n_i$ and $\tilde n_i$ is determined by
\be
\label{KW:integers:constraint}
 n_j - n_i \neq 0 \quad \text{ mod }~ N \, , \qquad \qquad \tilde n_j - \tilde n_i \neq 0 \quad \text{ mod }~ N \, .
\ee
Given the solution \eqref{KW:sol:final} to the BAEs, the integers $n_i$ and $\tilde n_i$ are defined modulo $N$ due to the $\beta$-periodicity of eigenvalues on $T^2$.
Thus we are left with $\{n_i, \tilde n_i\} \in \left[ 1,N \right]$.
The only possible choice is then given by $n_i = \tilde n_i = i$ and its permutations.

Finally, plugging the solution \eqref{KW:sol:final} to the BAEs back into \eqref{KW:twisted superpotential},
we obtain the ``on-shell" value of the twisted superpotential
\be
\label{SU(N):KW:on-shell:twisted superpotential}
\wt\cW(\Delta_I) = \frac{N^2}{2\beta} \sum_{a < b < c} \Delta_a \Delta_b \Delta_c \, ,
\ee
up to terms $\cO(\beta)$.
The relation between the ``on-shell'' twisted superpotential and the four-dimensional conformal anomaly
coefficients also holds for the conifold theory. The R-symmetry 't Hooft anomalies can be expressed as
\be
\begin{aligned}
\label{KWR3}
 \Tr R \left( \Delta_I \right) & = 2 \left(N^2 - 1\right) + N^2 \sum _{I = 1}^{4} \left( \frac{\Delta_I}{\pi} - 1 \right) = - 2 \, , \\
 \Tr R^3 \left( \Delta_I \right) & = 2 \left(N^2 - 1\right) + N^2 \sum _{I = 1}^{4} \left( \frac{\Delta_I}{\pi} - 1 \right)^3
 = \frac{3 N^2}{\pi^3} \sum_{a < b < c} \Delta_a \Delta_b \Delta_c - 2 \, ,
\end{aligned}
\ee
where we used $\Delta_I / \pi$ to parameterize the trial R-symmetry of the theory.
Hence, Eq.\,\eqref{SU(N):KW:on-shell:twisted superpotential} can be rewritten as
\be
\label{SU(N):KW:on-shell:twisted superpotential:TrR3}
\wt\cW(\Delta_I) = \frac{\pi^3}{6 \beta} \left[ \Tr R^3 \left( \Delta_I \right) - \Tr R \left( \Delta_I \right) \right]
= \frac{16 \pi^3}{27 \beta}  \left[ 3 c \left( \Delta_I \right) - 2 a( \Delta_I) \right] \, .
\ee
Here, we employed Eq.\,\eqref{generalac} to write the second equality.

\subsection{The topologically twisted index at high temperature}
\label{The index at high temperature_KW}

The twisted index, at high temperature, can be computed by evaluating the contribution of the saddle-point configurations to \eqref{index:KW:bethe}.
The procedure for computing the index is very similar to that presented in section \ref{The index at high temperature_SYM}.
The off-diagonal vector multiplet contributes
\begin{align}
\log \prod_{i\neq j}^N \left[ \frac{\theta_1\left( \frac{x_i}{x_j} ; q\right)}{i \eta(q)} \frac{\theta_1\left( \frac{\tilde x_i}{\tilde x_j} ; q\right)}{i \eta(q)}\right] =
- \frac{1}{\beta} \sum_{i \neq j}^{N} \left[ F' \left( u_i - u_j \right) + F' \left( \tu_i - \tu_j \right) \right] - i N (N - 1) \pi \, .
\end{align}
The quantity $\cP$, Eq.\,\eqref{KW:A}, contributes
\be
\begin{aligned}
 \log \cP & = - \frac{1}{\beta} \bigg\{ \frac{2\pi^2}{3} (N-1)
 + \sum_{i , j = 1}^{N} \sum_{ \substack{I = 1,2: + \\ I = 3, 4: -}} (\fn_I - 1) F' \left[ \pm \left( u_i - \tilde u_j \pm \Delta_I \right) \right] \bigg\} \\  
 & + \frac{i N^2 \pi}{2} \sum_{I = 1}^{4} \left( 1 - \fn_I \right) - 2 (N-1) \log \left( \frac{\beta}{2 \pi} \right) \, .
\end{aligned}
\ee
The Jacobian \eqref{Jacobian general:KW} has the following entries
\begin{subequations}
 \begin{align}
  \frac{\partial B_k}{\partial u_j} & = -\frac{\partial \wt B_k}{\partial \tu_j} =\frac{2\pi i}{\beta}N\delta_{kj} \, ,\quad \mbox{for} \quad k,j=1,2,\dots,N-1\, ,\\
  \frac{\partial B_N}{\partial u_k} & = -\frac{\partial \wt B_N}{\partial \tu_k}=-\frac{2\pi i}{\beta}N \, ,\quad 
  \frac{\partial B_k}{\partial v}=-\frac{\partial \wt B_k}{\partial \tv}=1 \, ,
  \quad \mbox{for} \quad k = 1, 2, \dots,N-1\, ,\\
  \frac{\partial B_N}{\partial v} &=- \frac{\partial \wt B_N}{\partial \tv}=1\, ,\quad
  \frac{\partial B_k}{\partial \tu_j}  = \frac{\partial \wt B_k}{\partial u_j} = \frac{\partial B_k}{\partial \tv} = \frac{\partial \wt B_k}{\partial v} =0\, ,\quad
  \mbox{for}\quad k,j=1,\dots,N\, .
 \end{align}
\end{subequations}

Now, it is straightforward to find the determinant of the matrix $\mathds{B}$:
\be
- \log \det \mathds{B} =2 (N-1) \left[ \log \left( \frac{\beta}{2 \pi} \right) - \frac{i \pi}{2}\right] -2 N \log N+\pi i N  \, .
\ee

The high-temperature limit of the index, at finite $N$, may then be written as
\be
\begin{aligned}
 \label{KW:index:final:hight}
 \log Z & = - \frac{1}{\beta} \bigg\{ \sum_{i \neq j}^{N} \left[ F' \left( u_i - u_j \right) + F' \left( \tu_i -\tu_j \right) \right]
 + \frac{2 \pi^2}{3} (N-1)
 \\ & + \sum_{i,j=1}^{N} \sum_{ \substack{I = 1,2: + \\ I = 3, 4: -}} (\fn_I - 1) F' \left[ \pm \left( u_i - \tilde u_j \pm \Delta_I \right) \right] \bigg\}
 - 2 N \log N + \pi i ( N + 1 ) \, .
\end{aligned}
\ee
Plugging the solution \eqref{KW:sol:final} to the BAEs back into the index \eqref{KW:index:final:hight} we find
\be
\begin{aligned}
\label{KW:index:final:hight:bethe:2d central charge}
\log Z = - \frac{N^2}{2 \beta} \sum_{ \substack{a<b \\ (\neq c)}} \Delta_a \Delta_b \fn_c  + \frac{2 \pi^2}{3 \beta} - 2 N \log N + \pi i ( N + 1 ) \, .
\end{aligned}
\ee
As in the case of $\cN=4$ SYM we can also write, to leading order in $1 / \beta$,
\be
\label{conifold:index:final:hight:bethe2}
 \log Z = \frac{\pi^2}{6 \beta} c_{l} \left(\Delta_I , \fn_I\right)
 = - \frac{16 \pi^3}{27 \beta}  \sum_{I=1}^{4} \fn_I \frac{\partial a (\Delta_I)}{\partial \Delta_I} \, ,
\ee
where the second equality is written assuming that $N$ is large.
Here, $c_l$ is the left-moving central charge of the two-dimensional $\cN = (0, 2)$ SCFT obtained by the twisted compactification on $S^2$.
This is related to the trial right-moving central charge $c_r$ by the gravitational anomaly, \ie\;$c_l = c_r - k$.
The central charge $c_r$ takes contribution from the 2D massless fermions, the gauginos
and the zero modes of the chiral fields (the difference between the number of modes of opposite chirality being  $\fn_I-1$) \cite{Benini:2012cz,Benini:2013cda,Benini:2015bwz},
\be
\begin{aligned}
 \label{c2d:anomalyconifold}
 c_{r} \left( \Delta_I , \fn_I \right) = - 3 \Tr  \gamma_3 R^2 \left( \Delta_I \right)
 &= - 3 \left[ 2 \left( N^2 -1 \right) + N^2 \sum_{I = 1}^{4} \left( \fn_I - 1 \right) \left( \frac{\Delta_I}{\pi} - 1 \right)^2 \right] \, ,
\end{aligned}
\ee
while the gravitational anomaly $k$ reads
\be
 k = - \Tr  \gamma_3 = - 2 \left( N^2 -1 \right) - N^2 \sum_{I = 1}^{4} \left( \fn_I - 1 \right) = 2 \, .
\ee

The extremization of $c_{r} \left( \Delta_I , \fn_I \right)$ with respect to the $\Delta_I$ reproduces
the exact central charge of the two-dimensional CFT \cite{Benini:2012cz,Benini:2013cda}. Notice that all the non-anomalous symmetries,
including the baryonic one, enter in the formula \eqref{c2d:anomalyconifold}, which depends on
three independent fluxes and three independent fugacities. As pointed out in \cite{Benini:2015bwz},
the inclusion of baryonic charges is crucial when performing $c$-extremization. 

For later convenience we introduce the following combinations of parameters
\be
 \Upsilon = 2 \sum_{I=1}^4 \left[ \fn_I^2 (\fn_I - 1) - \frac{\fn_1 \fn_2 \fn_3 \fn_4}{\fn_I} \right] \, ,
 \qquad \qquad \Pi = \sum_{ \substack{a<b \\ (\neq c)}} \fn_a \fn_b \fn_c^2 \, .
\ee
The trial central charge $c_r$ as a function of $\Delta_{1,2,3}$ is extremized at
\be
 \frac{\Delta_I}{2 \pi} = \frac{1}{\Upsilon} \left[ 2 \left( \fn_I^3 - \frac{\fn_1 \fn_2 \fn_3 \fn_4}{\fn_I} \right)
 - \fn_I \sum _{J=1}^4 \fn_J^2 \right] , \qquad I=1,2,3 \, .
\ee
At the critical point the function takes the value
\be
 c_r (\fn_I) = 12 \left( \frac{N^2 \Pi}{\Upsilon} + \frac12 \right) \, .
\ee
Let us note that in the saddle-point approximation we can write the density of states $d_{\text{micro}}$ as
\be
 d_{\text{micro}} (\fn_I , q_0) = \frac{\pi}{N^{2N + 3} \sqrt{2 \Upsilon}}
 \left( \frac{\pi}{3} \frac{c_l (\fn_I)}{S_{\text{Cardy}}} \right)^{5/2}
 I_{5 / 2} (S_{\text{Cardy}}) \, ,
\ee
where
\be
 S_{\text{Cardy}} = 2 \pi \sqrt{\frac{c_l(\fn_I) q_0}{6}} \, .
\ee

\section{High-temperature limit of a generic theory}                                         
\label{high-temp limit of the index}

We can easily generalize the previous discussion to the case of general four-dimensional $\cN = 1$ SCFTs.
Our goal is to compute the partition function of $\cN = 1$ gauge theories
on $S^2 \times T^2$ with  a partial topological A-twist along $S^2$.
We identify, as before, the modulus of the torus with the \emph{fictitious} inverse temperature $\beta$, and we are interested in the \emph{high-temperature} limit $(\beta \to 0)$ of the index.
As we take $\beta$ to zero, we can use the asymptotic expansions \eqref{dedekind:hight:S}
and \eqref{theta:hight:S} for the elliptic functions appearing in the supersymmetric path integral \eqref{path integral index}. 
We focus on quiver gauge theories with bi-fundamental and adjoint chiral multiplets and a number $|G|$ of $\SU(N)^{(a)}$ gauge groups.
Eigenvalues $u_i^{(a)}$ and  gauge magnetic fluxes $\fm_i^{(a)}$ have to satisfy the tracelessness condition, \ie\;
\be
\begin{aligned}
\label{tracelessness condition}
\sum_{i = 1}^{N} u_i^{(a)} = 0 \, , \qquad \qquad \sum_{i = 1}^{N} \fm_i^{(a)} = 0 \, .
\end{aligned}
\ee
The magnetic fluxes and the chemical potentials for the global symmetries of the theory fulfill the constraints \eqref{supconstraints:intro} and \eqref{supconstraints2:intro}.
We also assume that $0<\Delta_I<2 \pi$. 

As in the previous examples, the solution to the BAEs is given by
\be
\label{BAEs:sol:SU(N)}
u_i^{(a)} = \cO \left( \beta \right) \, , \qquad \forall \quad i, a \, ,
\ee
and exists (up to discrete symmetries) only for  $\sum_{I \in W} \Delta_I  = 2 \pi$, for each monomial term $W$ in the superpotential, as we checked in many examples.
Due to this constraint, $\Delta_I / \pi$ behaves at all effects like a trial R-symmetry of the theory.

\subsection{Twisted superpotential at high temperature}
\label{twisted superpotential at high-temp}

The general rules for constructing the \emph{high-temperature}
``on-shell" twisted superpotential, \ie\;$\wt\cW(\{u_i^{(a)}\} , \Delta_I) \big|_{\text{BAEs}}$,
of $\cN = 1$ quiver gauge theories to leading order in $1 / \beta$ are:
\begin{enumerate}
 \item A bi-fundamental field with chemical potential $\Delta_{(a,b)}$ transforming in the $({\bf N},\overline{\bf N})$ representation of $\SU(N)_a \times \SU(N)_b$, contributes
 \begin{equation}
 \label{twisted:superpotential:bi-fundamental}
  \frac{i N^2}{\beta} F \left( \Delta_{(a,b)} \right) \, ,
 \end{equation}
 where the function $F$ is defined in \eqref{F:function}.
 \item An adjoint field with chemical potential $\Delta_{(a,a)}$ contributes
 \begin{equation}
 \label{twisted superpotential  adjoint}
 \frac{i \left( N^2 - 1 \right)}{\beta} F \left( \Delta_{(a,a)} \right)  \, .
 \end{equation}
\end{enumerate}

\subsection{The topologically twisted index at high temperature}
\label{The index at high-temp}

Using the dominant solution \eqref{BAEs:sol:SU(N)} to the BAEs we can proceed to compute the topologically twisted index.
Here are the rules for constructing the logarithm of the index at high temperature to leading order in $1 / \beta$:
\begin{enumerate}
 \item For each group $a$, the contribution of the off-diagonal vector multiplet is
 \begin{equation}
  \label{index off-diag vector}
  - \frac{\left(N^2 - 1 \right)}{\beta} \frac{\pi^2}{3} \, .
 \end{equation}
 \item A bi-fundamental field with magnetic flux $\fn_{(a,b)}$ and chemical potential $\Delta_{(a,b)}$ transforming
 in the $({\bf N},\overline{\bf N})$ representation of $\SU(N)_a \times \SU(N)_b$, contributes
 \begin{equation}
 \label{index bi-fundamental}
 - \frac{N^2}{\beta} \left(\fn_{(a,b)} - 1 \right) F' \left( \Delta_{(a,b)} \right) \, .
 \end{equation}
 \item An adjoint field with magnetic flux $\fn_{(a,a)}$ and chemical potential $\Delta_{(a,a)}$, contributes
 \begin{equation}
 \label{index adjoint}
 - \frac{N^2 - 1}{\beta} \left(\fn_{(a,a)} - 1 \right) F' \left( \Delta_{(a,a)} \right) \, .
 \end{equation}    
\end{enumerate}

\subsection{An index theorem for the twisted matrix model}
\label{4d index theorem}

The high-temperature behavior of the index, to leading order in $1 / \beta$ and $N$,
can be captured by a simple universal formula involving the twisted superpotential and its derivatives.
Let us recall some of the essential ingredients that we need in the following.

The R-symmetry 't Hooft anomaly of UV four-dimensional $\cN = 1$ SCFTs is given by
\be
\begin{aligned}
\label{tHoof:linear:cubic:anomalies}
\Tr R^{\alpha} (\Delta_I) & = |G| \text{ dim } \SU(N) + \sum _{I} \text{ dim }\fR_I \left( \frac{\Delta_I}{\pi} - 1 \right)^{\alpha} \, ,
\end{aligned}
\ee
where the trace is taken over all the bi-fundamental fermions and gauginos
and $\text{dim }\fR_I$ is the dimension of the respective matter representation with R-charge $\Delta_I / \pi$.
On the other hand, the trial right-moving central charge of the IR two-dimensional $\cN = (0, 2)$ SCFT on $T^2$ can be computed from the spectrum of massless fermions \cite{Benini:2012cz,Benini:2013cda,Benini:2015bwz}.
These are gauginos with chirality $\gamma_3=1$ for all the gauge groups and  fermionic zero modes for each chiral field,
with a difference between the number  of fermions of opposite chiralities equal to $\fn_I-1$.
The central charge is related by the $\cN = 2$ superconformal algebra to the R-symmetry anomaly \cite{Benini:2012cz,Benini:2013cda},
and is given by
\be
\begin{aligned}
 \label{c2d:anomaly0}
 c_{r} \left( \Delta_I , \fn_I \right) & = - 3 \Tr  \gamma_3 R^2 \left( \Delta_I \right) 
 = - 3 \left[ |G| \text{ dim }\SU(N) + \sum_{I} \text{ dim }\fR_I \left( \fn_I - 1 \right) \left( \frac{\Delta_I}{\pi} - 1 \right)^2 \right]  \, .
\end{aligned}
\ee
By an explicit calculation we see that Eq.\,\eqref{c2d:anomaly0} can be rewritten as
\be
\begin{aligned}
\label{QFTrelation}
 c_{r} \left( \Delta_I, \fn_I \right) & = - 3 \Tr R^3 \left( \Delta_I \right) - \pi \sum_{I} \left[ \left( \fn_I - \frac{\Delta_I}{\pi} \right) \frac{\partial \Tr R^3 \left( \Delta_I \right)}{\partial \Delta_I} \right]
\, ,
\end{aligned}
\ee
where we used the relation \eqref{tHoof:linear:cubic:anomalies}.%
\footnote{Notice that, in evaluating the right hand side of \eqref{QFTrelation},
we can consider all the $\Delta_I$  as independent variables and impose the constraints $\sum_{I \in W} \Delta_I  = 2 \pi$
only after differentiation.
This is due to the form of the differential operator in \eqref{QFTrelation} and the constraints $\sum_{I \in W} \fn_I  = 2$.}
Moreover, the trial left-moving central charge of the two-dimensional $\cN = (0, 2)$ theory reads
\be
\label{c2d:left:right:gr}
c_l = c_r - k \, , 
\ee
where $k$ is the gravitational anomaly and is given by
\be
\label{gr:anomaly:theorem}
 k = - \Tr \gamma_3 = - |G| \text{ dim }\SU(N) - \sum_{I} \text{ dim }\fR_I \left( \fn_I - 1 \right) \, .
\ee

For theories of D3-branes with an AdS dual, to leading order in $N$, the linear R-symmetry 't Hooft anomaly of the four-dimensional theory vanishes, \ie\;$\Tr R= \cO(1)$ and $a=c$ \cite{Henningson:1998gx}.
Using the parameterization of a trial R-symmetry in terms of $\Delta_I / \pi$,  this is equivalent to
\be
\label{index theorem:constraint0}
\pi |G| + \sum_{I} \left( \Delta_I - \pi \right) = 0 \, ,
\ee
where the sum is taken over all matter fields (bi-fundamental and adjoint) in the quiver.
Similarly,  we have
\be
\label{index theorem:constraint1}
|G| + \sum_{I} \left( \fn_I - 1 \right) = 0 \, .
\ee
This is simply $k = - \Tr \gamma_3 = \cO(1)$, to leading order in $N$.

The {\it index theorem} can be expressed as:

\begin{theorem}
\label{theorem:1}
The topologically twisted index of any $\cN = 1$ $\SU(N)$ quiver gauge theory placed on  $S^2 \times T^2$ to leading order in $1/\beta$ is given  by
\be
 \label{index theorem:2d central charge}
 \log Z \left( \Delta_I, \fn_I \right) = \frac{\pi^2}{6 \beta} c_{l} \left( \Delta_I, \fn_I \right) \, ,
\ee
which is Cardy's universal formula for the asymptotic density of states in a CFT$_2$ \cite{Cardy:1986ie}. We can write the  extremal value of  the twisted superpotential $\wt\cW ( \Delta_I )$   as
\be
\begin{aligned}
 \label{index theorem: twisted superpotential:finiteN}
 \wt\cW ( \Delta_I ) \equiv - i \wt\cW ( \{u_i^{(a)}\}, \Delta_I )\big|_{\text{BAEs}}
 & = \frac{\pi^3}{6 \beta} \left[ \Tr R^3 ( \Delta_I ) - \Tr R ( \Delta_I ) \right]
 =  \frac{16 \pi^3}{27 \beta} \left[ 3 c ( \Delta_I ) - 2 a ( \Delta_I) \right] \, .
\end{aligned}
\ee
For theories of D3-branes at large $N$, the index can be recast as 
\be
\begin{aligned}
 \label{index theorem:attractor}
 \log Z \left( \Delta_I, \fn_I \right) & = - \frac{3}{\pi} \wt\cW ( \Delta_I )
 - \sum_{I} \left[ \left( \fn_I - \frac{\Delta_I}{\pi} \right) \frac{\partial \wt\cW( \Delta_I )}{\partial \Delta_I} \right]
 = \frac{\pi^2}{6 \beta} c_{r} \left( \Delta_I, \fn_I \right) \, ,
\end{aligned}
\ee
where $\wt\cW ( \Delta_I )$ reads
\be
\begin{aligned}
 \label{index theorem: twisted superpotential:largeN}
 \wt\cW ( \Delta_I )
 = \frac{16 \pi^3}{27 \beta} a ( \Delta_I) \, .
\end{aligned}
\ee
\end{theorem}
\begin{proof} Observe first that again we can consider all the $\Delta_I$ in \eqref{index theorem:attractor} as independent variables and impose the constraints $\sum_{I \in W} \Delta_I  = 2 \pi$ only after differentiation. This is due to the form of the differential operator in \eqref{index theorem:attractor} and $\sum_{I \in W} \fn_I  = 2$.
To prove the first equality in \eqref{index theorem:attractor}, we promote the explicit factors of $\pi$, appearing in \eqref{twisted:superpotential:bi-fundamental} and \eqref{twisted superpotential adjoint}, to a formal variable $\fakebold{\pi}$.
Notice that the ``on-shell'' twisted superpotential $\wt\cW(\Delta_I, \fakebold{\pi})$, at large $N$, is a homogeneous function of $\Delta_I$ and $\fakebold{\pi}$, \ie\;
\begin{equation}
  \wt\cW(\lambda \Delta_I, \lambda \fakebold{\pi}) = \lambda^3 \, \wt\cW(\Delta_I, \fakebold{\pi}) \, .
\end{equation}
Hence,
\begin{equation}\label{hom}
 \frac{\partial \wt\cW(\Delta_I, \fakebold{\pi})}{\partial \fakebold{\pi}} =
 \frac{1}{\fakebold{\pi}} \left[ 3 \, \wt\cW(\Delta_I) -\sum_I  \Delta_I \frac{\partial \wt\cW(\Delta_I)}{\partial \Delta_I} \right]\, .
\end{equation}
Now, we consider a generic quiver gauge theory with matters in bi-fundamental and adjoint representations of the gauge group.
They contribute to the twisted superpotential $\wt\cW (\Delta_I, \fakebold{\pi})$ as written in \eqref{twisted:superpotential:bi-fundamental} and \eqref{twisted superpotential adjoint}, respectively.
Let us calculate the derivative of $\wt\cW(\Delta_I, \fakebold{\pi})$ with respect to $\Delta_{I}$:
 \be
 \label{proof bf:Delta}
 - \sum_{I} \fn_I \frac{\partial \wt\cW(\Delta_I, \fakebold{\pi})}{\partial \Delta_I} =
 - \frac{N^2}{\beta} \sum_{I} \fn_I\; F' \left( \Delta_{I} \right) \, .
 \ee
Next, we take the derivative of the twisted superpotential with respect to $\fakebold{\pi}$:
\begin{align}
 \label{proof bf:pi}
 - \sum_{I} \frac{\partial \wt\cW(\Delta_I, \fakebold{\pi})}{\partial \fakebold{\pi}} =
 \frac{N^2}{\beta} \sum_{I} \; F' \left( \Delta_{I} \right)
 - \frac{N^2}{\beta} \sum_{I} \left( \frac{\fakebold{\pi}^2}{3} - \frac{\fakebold{\pi}}{3} \Delta_{I} \right) \, .
\end{align}
Using \eqref{hom} and combining \eqref{proof bf:Delta} with the first term of \eqref{proof bf:pi} as in the right hand side of \eqref{index theorem:attractor}, 
we obtain the contribution of matter fields \eqref{index bi-fundamental} and \eqref{index adjoint} to the index.
The contribution of the second term in \eqref{proof bf:pi} to \eqref{index theorem:attractor} can be written as
\be
 - \frac{N^2}{\beta} \frac{\fakebold{\pi}}{3} \sum_I \left( \fakebold{\pi} - \Delta_I \right) =
 - \frac{N^2}{\beta} \frac{\fakebold{\pi}^2}{3} |G| \, ,
\ee
where we used the constraint \eqref{index theorem:constraint0}.
This is precisely the contribution of the off-diagonal vector multiplets \eqref{index off-diag vector} to the index at large $N$.

Parameterizing the trial R-symmetry in terms of $\Delta_I / \pi$,
we can prove \eqref{index theorem: twisted superpotential:finiteN}:
\be
\begin{aligned}
 \wt\cW ( \Delta_I )
 & = \frac{1}{\beta} \sum_{I } \text{ dim } \fR_I \; F \left( \Delta_I \right)
 = \frac{1}{6 \beta} \sum_{I} \text{ dim } \fR_I \left[ \left(\Delta_I - \pi \right)^3 - \pi^2 \left( \Delta_I - \pi \right) \right] \\
 & = \frac{\pi^3}{6 \beta} \left[\sum _{I} \text{ dim } \fR_I \left( \frac{\Delta_I}{\pi} - 1 \right)^3
 - \sum _{I} \text{ dim } \fR_I \left( \frac{\Delta_I}{\pi} - 1 \right)\right] \\
 & = \frac{\pi^3}{6 \beta} \left[ \Tr R^3 \left( \Delta_I \right) - \Tr R \left( \Delta_I \right) \right] \, ,
\end{aligned}
\ee
which at large $N$, due to \eqref{index theorem:constraint0}, is equal to \eqref{index theorem: twisted superpotential:largeN}.

Finally, we need to show that the high-temperature limit of the index is given by the Cardy formula \eqref{index theorem:2d central charge}.
Bi-fundamental and adjoint fields contribute to the index according to \eqref{index bi-fundamental} and \eqref{index adjoint}, respectively.
We thus have
\be
\begin{aligned}
 \log Z \left( \Delta_I, \fn_I \right) & = - \frac{1}{\beta} \left[ \frac{\pi^2}{3} |G| \text{ dim }\SU(N) + \sum_{I} \text{ dim }\fR_I \; \left( \fn_I - 1 \right) F' \left( \Delta_I \right) \right] \\
 & = - \frac{\pi^2}{6 \beta} \left\{ 2 |G| \text{ dim }\SU(N) + \sum_{I} \text{ dim }\fR_I \; \left( \fn_I - 1 \right) \left[ 3 \left( \frac{\Delta_I}{\pi} - 1 \right)^2 - 1 \right] \right\} \\
 & = \frac{\pi^2}{6 \beta} \left[ c_{r} \left( \Delta_I, \fn_I \right) + \Tr \gamma_3 \right]
 = \frac{\pi^2}{6 \beta} c_l \left( \Delta_I, \fn_I \right) \, ,
\end{aligned}
\ee
where we used \eqref{c2d:anomaly0} and \eqref{c2d:left:right:gr} in the third and the fourth equality, respectively.
For quiver gauge theories fulfilling the constraint \eqref{index theorem:constraint1} the above formula reduces to the second equality in \eqref{index theorem:attractor} at large $N$. This completes the proof.
\end{proof}

It is worth stressing  that, when using formula \eqref{index theorem:attractor}, the linear  relations among the $\Delta_I$ can be imposed after differentiation.
It is always possible, ignoring some linear relations, to parameterize $\wt\cW (\Delta_I)$ such that it becomes a homogeneous function of degree 3 in the chemical potentials $\Delta_I$ \cite{Benvenuti:2006xg}.
With this parameterization the index theorem becomes
\be
 \log Z \left( \Delta_I, \fn_I \right) = - \sum_{I} \fn_I \frac{\partial \wt\cW ( \Delta_I )}{\partial \Delta_I} \, .
\ee
As we have seen, this is indeed the case for  $\cN = 4$ SYM and the Klebanov-Witten theory.
We note that our result is very similar to that obtained for the large $N$ limit of the topologically twisted index
of three-dimensional $\cN \geq 2$ Yang-Mills-Chern-Simons-matter theories placed on $S^2 \times S^1$.

%

\chapter[An extremization principle for the entropy of BPS black holes in AdS\texorpdfstring{$_5$}{(5)}]{An extremization principle for the entropy of BPS black holes in AdS$\bm{_5}$}
\label{ch:6}

\ifpdf
    \graphicspath{{Chapter6/Figs/Raster/}{Chapter6/Figs/PDF/}{Chapter6/Figs/}}
\else
    \graphicspath{{Chapter6/Figs/Vector/}{Chapter6/Figs/}}
\fi

\section{Introduction}
\label{sec:introduction}

As we have shown in the previous chapters, the entropy of a class of dyonic BPS black holes in AdS$_4$
with magnetic and electric charges $(\fn_i, q_i)$ can be obtained as the Legendre transform
of the topologically twisted index $Z (\fn_i , \Delta_i)$, which is a function of magnetic fluxes $\fn_i$
and chemical potentials $\Delta_i$ for the global symmetries of the dual field theory:
\be
 S_{\text{BH}} (\fn_i , q_i) = \log Z (\fn_i , \bar\Delta_i) - i \sum_i q_i \bar \Delta_i \, .
\ee
Here $\bar\Delta_i$ is the extremum of $\cI(\Delta_i)=\log Z(\fn_i,\Delta_i)  - i \sum_i q_i \Delta_i$.
This procedure dubbed ${\cal I}-$extremization as we explained in details in section \ref{sec:intro:I-extremization}.

It is natural to ask what would be the analogous of this construction in five dimensions.
In this chapter, we humbly look at the gravity side of the story and try to understand
what kind of extremization can reproduce the entropy of the supersymmetric rotating black holes.
Unfortunately, the details of the attractor mechanism for rotating black holes in five-dimensional
gauged supergravity are not known but we can nevertheless find an extremization principle for the entropy.
The final result is quite  surprising and intriguing. 

We consider  the class of supersymmetric rotating black holes found and studied in \cite{Gutowski:2004ez,Gutowski:2004yv,Chong:2005da,Chong:2005hr,Kunduri:2006ek}. They are asymptotic to  AdS$_5\times S^5$
and depend on three electric charges  $Q_I$ $(I = 1,2,3)$, associated with rotations in $S^5$, and two angular momenta $J_\phi, J_\psi$ in AdS$_5$.
Supersymmetry actually requires a constraint among the charges and only four of them are independent.
We show that the Bekenstein-Hawking entropy of the black holes can be obtained as the Legendre transform
of the quantity\footnote{Notice that one can write the very same entropy as the result of a different extremization
in the context of the Sen's entropy functional \cite{Sen:2005wa,Morales:2006gm}.
The two extremizations are over different quantities and use different charges.}  
\be
\begin{aligned}
  \label{modified:Casimir:SUSY:N=4}
 E= - i \pi  N^2  \frac{\Delta_1 \Delta_2 \Delta_3}{\omega_1 \omega_2} \, ,
\end{aligned} \ee
where $\Delta_I$ are  chemical potentials conjugated to the electric charges $Q_I$ and $\omega_{1,2}$
chemical potentials conjugated to the angular momenta $J_\phi, J_\psi$.
The constraint among charges is reflected in the following constraint among chemical potentials,
\be
\begin{aligned}
 \label{constraint:N=4}
 \Delta_1 + \Delta_2 + \Delta_3 + \omega_1 + \omega_2 = 1 \, .
\end{aligned}
\ee
To further motivate the result  \eqref{modified:Casimir:SUSY:N=4} we shall consider the case of equal angular momenta $J_\psi=J_\phi$.
In this limit, the black hole has an enhanced $\SU(2)\times \U(1)$ isometry
and it can be reduced along the $\U(1)$ to a static dyonic black hole in four dimensions.
We show that, upon dimensional reduction, the extremization problem
based on \eqref{modified:Casimir:SUSY:N=4} coincides with the attractor mechanism in four dimensions,
which is well understood for static BPS black holes \cite{Cacciatori:2009iz,DallAgata:2010ejj,Hristov:2010ri,Katmadas:2014faa,Halmagyi:2014qza,Klemm:2016wng}. 

It is curious to observe that the expression \eqref{modified:Casimir:SUSY:N=4} is {\it formally} identical
to the supersymmetric Casimir energy for ${\cal N}=4$ SYM, as derived, for example, in \cite{Bobev:2015kza} and reviewed in appendix \ref{AppC}.
It appears in the relation $Z_{{\cal N}=4}=\e^{- E} I$ between the partition function $Z_{{\cal N}=4}$
on $S^3\times S^1$ and the superconformal index $I$. Both the partition function and the superconformal index are functions
of a set of chemical potentials $\Delta_I$ $(I=1,2,3)$ and $\omega_i$ $(i = 1,2)$ associated with the R-symmetry generators $\U(1)^3\in \SO(6)$ and
the two angular momenta $\U(1)^2\in \SO(4)$, respectively. Since the symmetries that appear in the game must commute with
the preserved supercharge, the index and the partition function are actually functions of only four independent chemical potentials,
precisely as our quantity $E$. The constraint among chemical potentials is usually imposed as $\sum_{I=1}^3 \Delta_I +\sum_{i=1}^2 \omega_i =0$.
Since chemical potentials in our notations are periodic of period 1, our constraint \eqref{constraint:N=4} reflects a different
choice for the angular ambiguities. We comment about  the interpretation of this result in the discussion section,
leaving the proper understanding to future work.
 
The chapter is  organized as follows.
In section \ref{sec:5D gauged sugra} we give a short overview of $\cN=2$ $\text{D}=5$ FI gauged supergravity.
In section \ref{sec:AdS5 black holes} we review the basic features of the BPS rotating black holes of interest.
In section \ref{sec:Casimir and entropy}, we show that the Bekenstein-Hawking entropy of
the black hole can be obtained as the Legendre transform of the quantity \eqref{modified:Casimir:SUSY:N=4}.
In section  \ref{sec:a limiting case}, we perform the dimensional reduction of
the black holes with equal angular momenta down to four dimensions and we prove that
the extremization of \eqref{modified:Casimir:SUSY:N=4} is equivalent to
the attractor mechanism for four-dimensional static BPS black holes in gauged supergravity.
We conclude in section \ref{sec:discussion} with discussions. 

\section[\texorpdfstring{$\cN=2$}{N=2}, \texorpdfstring{${\rm D}=5$}{D=5} gauged supergravity]{$\fakebold{\cN=2}$, $\fakebold{{\rm D}=5}$ gauged supergravity}
\label{sec:5D gauged sugra}

The theory we shall be considering, following the conventions of \cite{Looyestijn:2010pb},
is the five-dimensional $\cN = 2$ FI gauged supergravity coupled to $n_{\rm V}$ vector multiplets.
It is based on a homogeneous real cubic polynomial
\bea
 \cV \left( L^I \right) = \frac16 C_{I J K} L^I L^J L^K \, ,
\eea
where $I, J , K = 1, \ldots , n_{{\rm V}}$ and $C_{I J K}$ is a fully symmetric third-rank tensor appearing in the Chern-Simons term.
Here, $L^I(\varphi^i)$ are real scalars satisfying the constraint $\cV=1$.
The action for the bosonic sector reads \cite{Gunaydin:1983bi,Gunaydin:1984ak}
\bea
 S^{(5)} = \int_{\bR^{4,1}} & \bigg[ \frac12 R^{(5)} \star_5 1 - \frac12 G_{I J} \rd L^I \wedge \star_5 \rd L^J
 - \frac12 G_{I J} F^I \wedge \star_5 F^J \\
 & - \frac{1}{12} C_{I J K} F^I \wedge F^J \wedge A^K + \chi^2 V \star_5 1 \bigg] \, ,
\eea
where $R^{(5)}$ is the Ricci scalar, $F^I \equiv \rd A^I$ is the Maxwell field strength, and $G_{I J}$ can be written in terms of $\cV$,
\bea
 \label{metric:gauge:kinetic:V}
 G_{I J} = - \frac12 \partial_I \partial_J \log \cV \big|_{\cV = 1} \, .
\eea
We also set $8 \pi G^{(5)}_{\text{N}} = 1$.
Furthermore, it is convenient to define
\bea
 \label{real:sections}
 L_I \equiv \frac16 C_{I J K} L^J L^K \, .
\eea
Therefore, we find that
\bea
 \label{metric:gauge:kinetic:sections}
 G_{I J} = \frac92 L_I L_J- \frac12 C_{I J K} L^K \, , \qquad \qquad L^I L_I = 1 \, ,
\eea
and
\bea
 L_I = \frac{2}{3} G_{I J} L^J \, , \qquad \qquad L^I = \frac{3}{2} G^{I J} L_J \, ,
\eea
where $G_{I K} G^{K J} = \delta^I_J$.
The inverse of $G_{I J}$ is given by
\bea
 G^{I J} = 2 L^I L^J - 6 C^{I J K} L_K \, ,
\eea
where $C^{I J K} \equiv C_{I J K}$. We then have
\bea
 L^I = \frac92 C^{I J K} L_J L_K \, .
\eea
The metric on the scalar manifold is defined by
\bea
 g_{i j} = \partial_i L^I \partial_j L^J G_{I J} \big|_{\cV = 1} \, ,
\eea
and the scalar potential reads
\bea
 V (L) = V_I V_J \left(6 L^I L^J - \frac92 g^{i j} \partial_i L^I \partial_j L^J \right) \, .
\eea
Here, $V_I$ are FI constants which are related to the vacuum value $\bar L_I$ of the scalars $L_I$,
\be
 \bar L_I = \xi^{-1} V_I \, ,
\ee
where $\xi^3 = \frac92 C^{I J K} V_I V_J V_K$ and the AdS$_5$ radius of curvature is given by $g^{-1} \equiv \left( \xi \chi \right)^{-1}$.
A useful relation of very special geometry is,
\bea
 g^{i j} \partial_i L^I \partial_j L^J = G^{I J} - \frac23 L^I L^J \, .
\eea
Thus,
\bea
 V(L) = 9 V_I V_J \left( L^I L^J - \frac12 G^{I J} \right) \, .
\eea

\section[Supersymmetric AdS\texorpdfstring{$_5$}{(5)} black holes in \texorpdfstring{$\U(1)^3$}{U(1)**3} gauged supergravity]{Supersymmetric AdS$\fakebold{_5}$ black holes in $\fakebold{\U(1)^3}$ gauged supergravity}
\label{sec:AdS5 black holes}

In this section we will briefly review a class of supersymmetric, asymptotically AdS, black holes
\cite{Gutowski:2004ez,Gutowski:2004yv,Chong:2005hr,Chong:2005da,Kunduri:2006ek}
of $D = 5$ $\U(1)^3$ gauged supergravity \cite{Gunaydin:1983bi,Gunaydin:1984ak}.
They can be embedded in type IIB supergravity as an asymptotically AdS$_5 \times S^5$ solution
which is exactly the decoupling limit of the rotating D3-brane \cite{Cvetic:1999xp}.
When lifted to type IIB supergravity they preserve only two real supercharges \cite{Gauntlett:2004cm}.
They are characterized by their mass, three electric charges and two  angular momenta with a constraint, and are
holographically dual to $1/16$ BPS states of $\cN = 4$ $\SU(N)$ SYM on $S^3 \times \mathbb{R}$ at large $N$.

We shall primarily be interested in the so-called $\cN = 2$ gauged STU model $(n_{\text{V}} = 3)$.
The only nonvanishing triple intersection numbers are $C_{1 2 3} = 1$ (and cyclic permutations). 
The bosonic sector of the theory comprises three gauge fields which correspond
to the Cartan subalgebra of the $\SO(6)$ isometry of $S^5$,
the metric, and three real scalar fields subject to the constraint
\be
 \label{XI:constraint}
 L^1 L^2 L^3 = 1 \, .
\ee
They take vacuum values $\bar L^I = 1$. 
The five-dimensional black hole metric can be written as \cite{Kunduri:2006ek}
\be
 \label{AdS5:generic:metric}
 \rd s^2 = - (H_1 H_2 H_3)^{-2/3} \left( \rd t + \omega_\psi \rd \psi + \omega_\phi \rd \phi \right)^2
 + \left( H_1 H_2 H_3 \right)^{1/3} h_{m n} \rd x^m \rd x^n \, ,
\ee
where
\be
 H_I = 1 + \frac{\sqrt{\Xi_a \Xi_b} \left( 1 + g^2 \mu_I \right) - \Xi_a \cos^2 \theta - \Xi_b \sin^2 \theta }{g^2 r^2} \, ,
\ee
\be
\begin{aligned}
 h_{m n} \rd x^m \rd x^n & = r^2 \bigg\{ \frac{\rd r^2}{\Delta_r} + \frac{\rd \theta^2}{\Delta_\theta}
 + \frac{\cos^2\theta}{\Xi_b^2} \left[ \Xi_b + \cos^2\theta \left( \rho^2 g^2 + 2 \left( 1+ b g \right) \left( a + b \right) g \right) \right] \rd \psi^2 \\
 & + \frac{\sin^2\theta}{\Xi_a^2} \left[ \Xi_a + \sin^2\theta \left( \rho^2 g^2 + 2 \left( 1+ a g \right) \left( a + b \right) g \right) \right] \rd \phi^2 \\
 & + \frac{2 \sin^2\theta \cos^2\theta}{\Xi_a \Xi_b} \left[ \rho^2 g^2 + 2 \left( a + b \right) g + \left( a + b \right)^2 g^2 \right] \rd \psi \rd \phi \bigg\} \, ,
\end{aligned} \ee
\be
\begin{aligned}
 \Delta_r & = r^2 \left[ g^2 r^2 + \left( 1 + a g + b g \right)^2 \right] \, , \qquad \qquad \; \;
 \Delta_\theta = \Xi_a \cos^2\theta + \Xi_b \sin^2\theta \, , \\
 \Xi_a & = 1 - a^2 g^2 \, , \qquad \Xi_b = 1 - b^2 g^2 \, , \qquad \qquad \rho^2 = r^2 + a^2 \cos^2\theta + b^2 \sin^2\theta \, ,
\end{aligned} \ee
\be
\begin{aligned}
 \omega_\psi & = - \frac{g \cos^2\theta}{r^2 \Xi_b} \left[ \rho^4 + \left( 2 r_m^2 + b^2 \right) \rho^2 + \frac12 \left( \beta_2 - a^2 b^2 + g^{-2} \left( a^2 - b^2 \right) \right) \right] \, , \\
 \omega_\phi & = - \frac{g \sin^2\theta}{r^2 \Xi_a} \left[ \rho^4 + \left( 2 r_m^2 + a^2 \right) \rho^2 + \frac12 \left( \beta_2 - a^2 b^2 - g^{-2} \left( a^2 - b^2 \right) \right) \right] \, ,
\end{aligned} \ee
and
\be
\begin{aligned}
 r_m^2 & = g^{-1} (a +b) + a b \, , \\
 \beta_2 & = \Xi_a \Xi_b \left( \mu_1 \mu_2 + \mu_1 \mu_3 + \mu_2 \mu_3 \right)
 - \frac{2 \sqrt{\Xi_a \Xi_b} \left(1 - \sqrt{\Xi_a \Xi_b} \right)}{g^2} \left( \mu_1 + \mu_2 + \mu_3 \right) + \frac{3 \left(1 - \sqrt{\Xi_a \Xi_b} \right)^2}{g^4} \, .
\end{aligned} \ee
The gauge coupling constant $g$ is fixed in terms of the AdS$_5$ radius of curvature, $g = 1 / \ell$.
The coordinates are $\left( t, r, \theta, \phi, \psi \right)$ where $r > 0$ corresponds to the exterior of the black hole, $0 \leq \theta \leq \pi / 2$ and $0 \leq \phi, \psi \leq 2 \pi$.
The scalars read
\be
 L^I = \frac{\left( H_1 H_2 H_3 \right)^{1/3}}{H_I} \, ,
\ee
while the gauge potentials are given by
\be
 A^I = H_I^{-1} \left( \rd t + \omega_\psi \rd \psi + \omega_\phi \rd \phi \right) + U_\psi^I \rd \psi + U_\phi^I \rd \phi \, ,
\ee
where
\be
\begin{aligned}
 U^I_\psi & = \frac{g \cos^2\theta}{\Xi_b} \left[ \rho^2 + 2 r_m^2 + b^2 - \sqrt{\Xi_a \Xi_b} \mu_I + g^{-2} \left( 1 - \sqrt{\Xi_a \Xi_b} \right) \right] \, , \\
 U_\phi^I & = \frac{g \sin^2\theta}{\Xi_a} \left[ \rho^2 + 2 r_m^2 + a^2 - \sqrt{\Xi_a \Xi_b} \mu_I + g^{-2} \left( 1 - \sqrt{\Xi_a \Xi_b} \right) \right] \, .
\end{aligned} \ee
The black hole is labeled by five parameters: $\mu_{1,2,3}, a, b$ where $g^{-1} > a$, $b \geq 0$.
Only four parameters are independent due to the constraint
\be
 \label{gravity:constraint}
 \mu_1 + \mu_2 + \mu_3 = \frac{1}{\sqrt{\Xi_a \Xi_b}} \left[ 2 r_m^2 + 3 g^{-2} \left( 1 - \sqrt{\Xi_a \Xi_b} \right) \right] \, .
\ee
Furthermore, regularity of the scalars for $r \geq 0$ entails that
\be
 g^2 \mu_I > \sqrt{\frac{\Xi_b}{\Xi_a}} - 1 \geq 0 \, ,
\ee
when $a \geq b$. If $a < b$ then the same expression remains valid with $a$ and $b$ interchanged.

\subsection[The asymptotic AdS\texorpdfstring{$_5$}{(5)} vacuum]{The asymptotic AdS$\bm{_5}$ vacuum}

This solution is expressed in the co-rotating frame.
The change of coordinates  $t = \wb t$, $\phi = \wb\phi - g \wb t$, $\psi = \wb\psi - g \wb t$, and $y^2 = r^2 + 2 r_m^2 / 3$
transforms the metric to a static frame at infinity. In order to bring the metric into
a manifestly asymptotically AdS$_5$ spacetime (in the global sense) as $y \to \infty$
we make the following change of coordinates \cite{Kunduri:2006ek}
\be
\begin{aligned}
  \Xi_a Y^2 \sin^2\Theta = \left( y^2 + a^2 \right) \sin^2 \theta \, , \qquad \qquad \Xi_b Y^2 \cos^2\Theta = \left( y^2 + b^2 \right) \cos^2 \theta \, .
\end{aligned} \ee
One gets the line element
\be
\begin{aligned}
 \rd s^2 \simeq - g^2 Y^2 \, \rd\wb{t}^2 + \frac{\rd Y^2}{g^2 Y^2} + Y^2 \left( \rd\Theta^2 + \sin^2\Theta \, \rd\wb\phi^2 + \cos^2\Theta \, \rd\wb\psi^2 \right) \, .
\end{aligned} \ee
The black hole has the Einstein universe $\bR \times S^3$ as its conformal boundary.
In the asymptotically static coordinates, the supersymmetric Killing vector field reads
\be
\begin{aligned}
 V = \frac{\partial}{\partial \wb{t}} + g \frac{\partial}{\partial\wb{\phi}} + g \frac{\partial}{\partial\wb{\psi}} \, ,
\end{aligned} \ee
which is timelike everywhere outside the black hole and is null on the conformal boundary.

\subsection{Properties of the solution}

It is convenient to define the following polynomials,
\be
 \gamma_1 = \mu_1 + \mu_2 + \mu_3 \, , \qquad \gamma_2 = \mu_1 \mu_2 + \mu_1 \mu_3 + \mu_2 \mu_3 \, , \qquad \gamma_3 = \mu_1 \mu_2 \mu_3 \, .
\ee
The black hole carries three $\U(1)^3 \subset \SO(6)$ electric charges in $S^5$ which are given by
\be
\begin{aligned}
 \label{BH:Q:5d:U(1)^3}
 Q_I = \frac{\pi}{4 G_{\rm N}^{(5)}} \left[ \mu_I + \frac{g^2}{2} \left( \gamma_2 - \frac{2 \gamma_3}{\mu_I} \right) \right] \, , \quad
 \text{ for } \quad I = 1, 2, 3 \, ,
\end{aligned} \ee
and two $\U(1)^2 \subset \SO(4)$ angular momenta in AdS$_5$ that read
\be
\begin{aligned}
 \label{BH:J:5d:U(1)^3}
 J_\psi & = \frac{\pi}{4 G_{\rm N}^{(5)}} \left[ \frac{g \gamma_2}{2} + g^3 \gamma_3 + g^{-3} \big(\sqrt{\Xi_a / \Xi_b} - 1 \big) \cJ \right] \, , \\
 J_\phi & = \frac{\pi}{4 G_{\rm N}^{(5)}} \left[ \frac{g \gamma_2}{2} + g^3 \gamma_3 + g^{-3} \big(\sqrt{\Xi_b / \Xi_a} - 1 \big) \cJ \right] \, .
\end{aligned} \ee
Here, $G_{\rm N}^{(5)}$ is the five-dimensional Newton constant and we defined
\be
 \cJ \equiv \prod_{I = 1}^{3} \left( 1 + g^2 \mu_I \right) \, .
\ee
The mass of the black holes is determined by the BPS condition
\be
\begin{aligned}
 \label{BH:BPS relation}
 M = g |J_\phi| + g |J_\psi| + |Q_1| + |Q_2| + |Q_3| \, ,
\end{aligned} \ee
which yields
\be
\begin{aligned}
 \label{BH:M:5d:U(1)^3}
 M = \frac{\pi}{4 G_{\rm N}^{(5)}} \left[
 \gamma_1 + \frac{3 g^2 \gamma_2}{2} + 2 g^4 \gamma_3
 + \frac{\left( \sqrt{\Xi_a} - \sqrt{\Xi_b} \right)^2}{g^2 \sqrt{\Xi_a \Xi_b}} \cJ \right] \, . ~~
\end{aligned} \ee
The solution has a regular event horizon at $r_{\rm h} = 0$ only for nonzero angular momenta in AdS$_5$.
The angular velocities of the horizon, measured with respect to the azimuthal coordinates $\psi$ and $\phi$
of the asymptotically static frame at infinity, are
\be
 \Omega_\psi = \Omega_\phi = g \, .
\ee
The near-horizon geometry is a fibration of AdS$_2$ over a non-homogeneously squashed $S^3$ \cite{Kunduri:2007qy} with area
\be
\begin{aligned}
 \label{BH:area:5d:U(1)^3}
 \text{Area} = 2 \pi^2
 \sqrt{ \gamma_3 \left( 1 + g^2 \gamma_1 \right) - \frac{g^2 \gamma_2^2}{4}
 - \frac{\left( \sqrt{\Xi_a} - \sqrt{\Xi_b} \right)^2}{g^6 \sqrt{\Xi_a \Xi_b}} \cJ} \, . ~~
\end{aligned} \ee
Positivity of the expression within the square root ensures the absence of closed causal curves near $r = 0$.
The Bekenstein-Hawking entropy of the black hole is proportional to its horizon area
and can be compactly written in terms of the physical charges as \cite{Kim:2006he}
\be
\begin{aligned}
 \label{BH:entropy:5d:U(1)^3}
 S_{\text{BH}} & = \frac{\text{Area}}{4 G_{\rm N}^{(5)}}
 = \frac{2 \pi}{g} \sqrt{Q_1 Q_2 + Q_2 Q_3 + Q_1 Q_3 - \frac{\pi}{4 G_{\rm N}^{(5)} g} \left( J_\phi + J_\psi \right)} \, .
\end{aligned} \ee
Finally, let
\be
\begin{aligned}
 \cX_I = \left( 1 + g^2 \mu_I \right) \sqrt{\Xi_a \Xi_b} - \Delta_\theta \, .
\end{aligned} \ee
The values of the scalar fields at the horizon read
\be
\begin{aligned}
 L^I (r_{\rm h}) = \frac{\left( \cX_1 \cX_2 \cX_3 \right)^{1/3}}{\cX_I} \, .
\end{aligned} \ee

In the next section we will obtain  the Bekenstein-Hawking entropy \eqref{BH:entropy:5d:U(1)^3} of the BPS black hole
from an extremization principle.

\section{An extremization principle for the  entropy}
\label{sec:Casimir and entropy}

We shall now extremize the quantity  \eqref{modified:Casimir:SUSY:N=4}, 
and show that the extremum precisely reproduces the entropy of the multi-charge BPS black hole discussed in the previous section.

Let us first introduce some notation that facilitates the comparison with supergravity:
\be
\begin{aligned}
 X^I \equiv \Delta_I \, , \qquad \qquad X^0_{\pm} \equiv \omega_1 \pm \omega_2 \, ,
\end{aligned} \ee
where $I = 1, 2 , 3$.
We shall also use $J^{\pm} \equiv J_\phi \pm J_\psi$,
\be
\begin{aligned}
 J^{+} & = \frac{\pi}{4 G_{\text{N}}^{(5)}} \left[ g \gamma_2 + 2 g^3 \gamma_3 + \frac{\left( \sqrt{\Xi_a} - \sqrt{\Xi_b} \right)^2}{g^3 \sqrt{\Xi_a \Xi_b}} \cJ \right] \, , \\
 J^{-} & = \frac{\pi}{4 G_{\text{N}}^{(5)}} \frac{\Xi_b - \Xi_a}{g^3 \sqrt{\Xi_a \Xi_b}} \cJ \, .
\end{aligned} \ee
Thus, we can rewrite the quantity  \eqref{modified:Casimir:SUSY:N=4}  as
\be
\begin{aligned}
\label{eq:Casimir:rewrittenX}
 E= - \frac{2 \pi^2 i}{g^3 G_{{\rm N}}^{(5)}} \frac{X^1 X^2 X^3}{\left( X^0_{+} \right)^2 - \left( X^0_{-} \right)^2} \, ,
\end{aligned} \ee
where we used the standard relation between gravitational and QFT parameters in the large $N$ limit,
\be
\begin{aligned}
 \frac{\pi}{2 g^3 G_{{\rm N}}^{(5)}} = N^2 \, .
\end{aligned} \ee
In the following we set the unit of the AdS$_5$ curvature $g = 1$.
The entropy of the BPS black hole, at leading order, can be obtained by extremizing the quantity\footnote{This is not the only possible choice of signs. There
are various sign ambiguities in the superconformal index literature as well as in the black hole one that should be fixed in a proper comparison between gravity and field theory.}
\be
\begin{aligned}
 \label{rotating:attractor:4d}
 \cI_{\text{sugra}} = - E\left( X_{\pm}^0, X^I \right) + 2 \pi i \sum_{I = 1}^{3} Q_I X^I - \pi i \left( J^{+} X_{+}^0 + J^{-} X^0_{-} \right) 
  \, ,
\end{aligned} \ee
with respect to $X^I$, $X_{\pm}^0$ and subject to the constraint \eqref{constraint:N=4},
\be
\begin{aligned}
\label{eq:constraintX}
 X_{+}^0 + \sum_{I = 1}^{3} X^I = 1 \, .
\end{aligned} \ee
At this stage, we find it more convenient to work in the basis $z^{\alpha}$ $(\alpha = 0 , 1 , 2 , 3)$ which is related to $\left( X_{\pm}^0 , X^I \right)$ by
\be
\begin{aligned}
 X_{-}^{0} = \frac{z^0}{1 + z^1 + z^2 + z^3} \, , \qquad  X_{+}^{0} = \frac{1}{1 + z^1 + z^2 + z^3} \, ,
 \qquad X^{1,2,3} = \frac{z^{1,2,3}}{1 + z^1 + z^2 + z^3} \, .
\end{aligned} \ee
Hence, in terms of the variables $z^\alpha$ the extremization equations can be written as
\be
\begin{aligned}
 \label{extremization:scalars}
 & \left[ (z^0)^2 -1 \right] \left\{ \left[ (z^0)^2 - 1 \right] c^i + \frac{z^1 z^2 z^3}{z^i} \right\} - 2 z^1 z^2 z^3 = 0 \, , \quad \text{ for } \quad i = 1, 2, 3 \, ,\\
 & c^0 \left[ (z^0)^2 - 1 \right]^2 - 2 z^0 z^1 z^2 z^3 = 0 \, ,
\end{aligned} \ee
where we defined the constants
\be
\begin{aligned}
 c^0 = \frac{\cJ \left( \Xi_b - \Xi_a \right)}{8 \sqrt{\Xi_a \Xi_b}} \, , \qquad \qquad
 c^i = \frac{\cJ}{4} \left( \frac{1}{1 + \mu_i} - \frac{\Xi_b + \Xi_a}{2 \sqrt{\Xi_a \Xi_b}} \right) \, .
\end{aligned} \ee
With an explicit computation one can check that the value of $\cI_{\text{sugra}}\left( z^\alpha \right)$
at the critical point precisely reproduces the entropy of the black hole,
\be
\begin{aligned}
\label{eq:EntropyMatch}
  \cI_{\text{sugra}} \big|_{\text{crit}} \left( J^{\pm} , Q_I \right) = S_\text{BH} \left( J^{\pm} , Q_I \right) \, .
 \end{aligned} \ee
It is remarkable that the solution to the extremization equations \eqref{extremization:scalars} is complex; however,
at the saddle-point it becomes a real function of the black hole charges.
We conclude that  the extremization of the quantity  \eqref{modified:Casimir:SUSY:N=4}  yields exactly the
Bekenstein-Hawking entropy of the $1 / 16$ BPS black holes in AdS$_5 \times S^5$.

So far the discussion was completely general.
In the next section, we will analyze the case $J_\phi = J_\psi$, for which the solution to
the extremization equations takes a remarkably simple form.

\section[Dimensional reduction in the limiting case: \texorpdfstring{$J_\phi = J_\psi$}{J[phi]=J[psi]}]{Dimensional reduction in the limiting case: $\fakebold{J_\phi = J_\psi}$}
\label{sec:a limiting case}

We gain some important insight by considering the dimensional reduction of the
five-dimensional BPS black holes when the two angular momenta are equal.
The black hole metric on the squashed sphere then has an enhanced isometry $\SU(2) \times \U(1) \subset \SO(4)$.
If we choose the appropriate Hopf coordinates we can dimensionally reduce the solution along the $\U(1)$
down to four-dimensional gauged supergravity.
As discussed in \cite{Hristov:2014eza}, it turns out that such a dimensional
reduction makes sense not only for asymptotically flat solutions where
first discovered in \cite{Gaiotto:2005gf,Gaiotto:2005xt,Behrndt:2005he,Banerjee:2011ts}
but also for the asymptotically AdS solutions in the gauged supergravity considered here.
A crucial difference is that the lower-dimensional vacuum will no longer be maximally symmetric
but will instead be of the hyperscaling-violating Lifshitz (hvLif) type \cite{Hristov:2014eza}.

The reason for looking at  the limit $J_\phi = J_\psi$ is simple: due to the $\SU(2)$ symmetry  the lower-dimensional solution
is guaranteed to be static and the horizon metric is a direct product AdS$_2 \times S^2$ geometry, as will be shown in due course.
Since the attractor mechanism for static BPS black holes in four-dimensional $\cN = 2$ gauged supergravity has been completely
understood \cite{Cacciatori:2009iz,Hristov:2010ri,DallAgata:2010ejj} we can fit the reduced solution in this framework.

\subsection{The near-horizon geometry}
\label{ssec:nh geometry}

We begin by taking the near-horizon limit, $r \to 0$, of the BPS black hole solution presented in section \ref{sec:AdS5 black holes}.
We set $a = b$, corresponding to the equal angular momenta $(J_{\phi} = J_{\psi})$,
and adopt the notation $\Xi_a = \Xi_b \equiv \Xi$.

Let us first introduce the following coordinates,
\be
 \psi \equiv \frac12 (\chi + \varphi) \, , \qquad \qquad \phi \equiv \frac12 (\chi - \varphi) \, , \qquad \qquad \theta \equiv \frac12 \vartheta \, ,
\ee
where $\vartheta, \varphi, \chi$ are the Euler angles of $S^3$ with $0 \leq \vartheta \leq \pi$, $0 \leq \varphi < 2 \pi$, $0 \leq \chi < 4\pi$.
The near-horizon geometry then reads
\be
\begin{aligned}
 \label{J1:J2:near-horizon:AdS2:M3}
 \rd s^2 & = R_{\text{AdS}_2}^2 \rd s^2_{\text{AdS}_2}
 + \gamma_3^{1/3} \, \rd s^2_{\cM_3} \, , \\
 L^I & = \frac{\gamma_3^{1/3}}{\mu_I} \, , \qquad \qquad
 A^I = \frac{\gamma_3^{1/3}}{\mu_I} R_{\text{AdS}_2} \, \tilde r \, \rd \tilde{t}
 + g \left( \gamma_1 - \mu_I - \frac{\gamma_2}{2 \mu_I} \right) \, \rd \gamma \, ,
\end{aligned} \ee
where we defined
\be
\begin{aligned}
 \rd s^2_{\text{AdS}_2} & = - \tilde r^2 \rd \tilde t^2 + \frac{\rd \tilde r^2}{\tilde r^2} \, , \qquad \qquad
 R_{\text{AdS}_2}^2 = \frac{\gamma_3^{1/3} }{4 (1 + g^2 \gamma_1 )} \, , \\
 \tilde r & = \frac{r^2}{4 R_{\text{AdS}_2}^2} \, , \qquad \qquad \qquad
 \tilde t = \frac{2}{\Xi \sqrt{\gamma_3^{1/3} (1 + g^2 \gamma_1) }} \, t \, ,
\end{aligned} \ee
\be
\begin{aligned}
 \label{J1:J2:near-horizon:M3}
 \rd s^2_{\cM_3} & = \rd s^2_{S^3}
 - \left[ \Gamma^2 \gamma_3^{1/3} - \frac{a g (4 + 5 a g)}{\Xi} \right] \rd \gamma^2
 + 2 R_{\text{AdS}_2} \, \Gamma \tilde r \, \rd \tilde t \, \rd \gamma \, , \\
 \Gamma & = \frac{g \left( 3 a^4 + 4 a^2 r_m^2 + \beta_2 \right)}{2 \Xi^2 \gamma_3^{2/3}} \, ,
\end{aligned} \ee
and
\be
\begin{aligned}
 \rd s^2_{S^3} =  \frac14 \left( \rd\vartheta^2 + \rd\varphi^2 + \rd \chi^2 + 2 \cos \vartheta \, \rd\varphi \, \rd\chi \right)
 = \frac14 \sum_{i = 1}^{3} \sigma_i \, , \qquad
 \rd\gamma = \frac{\sigma_3}{2} \, .
\end{aligned} \ee
Here, $\sigma_i$ $(i = 1 , 2, 3)$ are the right-invariant $\SU(2)$ one-forms,
\be
\begin{aligned}
 \sigma_1 & = - \sin \chi \, \rd\vartheta + \cos \chi \sin \vartheta \, \rd\varphi \, , \\
 \sigma_2 & = \cos \chi \, \rd\vartheta + \sin \chi \sin \vartheta \, \rd\varphi \, , \\
 \sigma_3 & = \rd\chi + \cos \vartheta \, \rd\varphi \, .
\end{aligned} \ee
Notice that $\rd s^2_{S^2} = \sigma_1^2 + \sigma_2^2 = \rd \vartheta^2 + \sin^2 \vartheta \, \rd \varphi^2$.
Due to the constraint \eqref{gravity:constraint} we can simplify
\be
 \frac{a (4 + 5 a g)}{\Xi} = g \gamma_1 \, , \qquad \qquad
 \Gamma = \frac{g \gamma_2}{2 \gamma_3^{2/3}} \, .
\ee
Upon a further rescaling of the time coordinate
\be
 \tilde t = - \frac{1}{2} \sqrt{4 - \frac{g^2 \gamma_2^2}{\gamma_3 (1 + g^2 \gamma_1)}} \, \tau \, ,
\ee
the near-horizon metric with squashed AdS$_2 \times_w S^3$ geometry and the gauge fields can be brought to the form:
\bea
 \label{J1:J2:final:near-horizon}
 \rd s^2_{(5)} & = R^2_{\text{AdS}_2}\ \rd s^2_{\text{AdS}_2}
 + \frac{R_{S^2}^2}{4} \left[ \rd s^2_{S^2} + \upsilon \left( \sigma_3 - \alpha \, \tilde r \rd \tau \right)^2 \right] \, , \\
 L^I & = \frac{\gamma_3^{1/3}}{\mu_I} \, , \qquad \qquad
 A^I_{(5)} = e^I \, \tilde r \, \rd \tau - f^I \, \sigma_3 \, .
\eea
Here, we defined the constants
\bea
 \alpha & = \frac{g \gamma_2}{2 \left( 1 + g^2 \gamma_1 \right)  \sqrt{\gamma_3 \upsilon}} \, , \qquad \qquad
 && R_{S^2}^2 = \gamma_3^{1/3} \, , \\
 e^I & = - \frac{\sqrt{\gamma_3 \upsilon}}{2 \mu_I (1 + g^2 \gamma_1) } \, ,
 && f^I = \frac{g}{4} ( \mu_I - \gamma_1 ) + \frac{g \gamma_3}{4 \mu_I^2} \, ,
\eea
and
\be
 \upsilon = 1 + g^2 \gamma_1 - \frac{g^2 \gamma_2^2}{4 \gamma_3} \, .
\ee
Note that we added the subscript $(5)$ in order to emphasize that these are five-dimensional quantities
which will next be related to a solution in four dimensions via dimensional reduction along the $\chi$ direction.

\subsection[Dimensional reduction on the Hopf fibres of squashed \texorpdfstring{$S^3$}{S**3}]{Dimensional reduction on the Hopf fibres of squashed $\fakebold{S^3}$}
\label{ssec:reduction along S3}

In five-dimensional supergravity theories, including $n_{\rm V}$ Abelian gauge fields $A_{(5)}^I$
and real scalar fields $L^I$ $(I = 1 , \ldots , n_{\rm V})$ coupled to gravity, the rules for reducing the bosonic fields are the following \cite{Andrianopoli:2004im,Behrndt:2005he,Cardoso:2007rg,Aharony:2008rx,Looyestijn:2010pb}:
\be
\begin{aligned}
 \label{reduction:rules:5dBH}
 \rd s_{(5)}^2 & = \e^{2 \phi} \, \rd s_{(4)}^2 + \e^{- 4 \phi} \, \left( \rd x^5 - A_{(4)}^0 \right)^2 \, , \qquad && \rd x^5 = \rd \chi \, , \\
 A_{(5)}^I & = A_{(4)}^I + \re z^I \left( \rd x^5 - A_{(4)}^0 \right) \, , \\
 L^I & = \e^{2 \phi} \im z^I \, , && \e^{-6 \phi} = \frac16 C_{I J K}  \im z^I  \im z^J  \im z^K \, ,
\end{aligned} \ee
where $\rd s_{(4)}^2$ denotes the four-dimensional line element,
the $A_{(4)}^\Lambda$ $( \Lambda = 0, I )$ are the four-dimensional Abelian gauge fields
and $z^I = X^I / X^0$ are the complex scalar fields in four dimensions.
The four-dimensional theory has $n_{\rm V}$ Abelian vector multiplets,
parameterizing a special K\"ahler manifold $\cM$ with metric $g_{i \bar{j}}$,
in addition to the gravity multiplet (thus a total of $n_{\rm V}+1$ gauge fields and $n_{\rm V}$ complex scalars).
The action of the bosonic part of the 4D $\cN=2$ FI gauged supergravity reads \cite{Andrianopoli:1996cm,Louis:2002ny}%
\footnote{\label{convention:action:metric}We follow
the conventions of \cite{Looyestijn:2010pb}, which is different from \cite{Andrianopoli:1996cm}
by factors of two in the gauge kinetic terms and the scalar potential $V(z, \bar{z} )$.
One can swap between the conventions by rescaling the four-dimensional metric $g_{\mu \nu} \to \frac12 g_{\mu \nu}$ and then multiplying the action by $2$.
This will modify the definition of the symplectic-dual gauge field strength $G_\Lambda$ by a factor of 2, see \eqref{symplectic-dual:F}.}
\bea
 S^{(4)} = \int_{\bR^{3,1}} & \bigg[ \frac{1}{2} R^{(4)} \star_4 1 + \frac14 \im \cN_{\Lambda \Sigma} F^\Lambda \wedge \star_4 F^{\Sigma}
 + \frac14 \re \cN_{\Lambda \Sigma} F^{\Lambda} \wedge F^{\Sigma} \\
 & - g_{i \bar{j}} D z^{i} \wedge \star_4 D \bar{z}^{\bar{j}}
 - V(z, \bar{z}) \star_4 1 \bigg] \, ,
\eea
where $i , \bar{j}= 1 , I$ and we already set $8 \pi G_{{\rm N}}^{(4)} = 1$.
Here $V$ is the scalar potential of the theory,
\bea
 V(z, \bar{z}) = 
 2 g^2 \left( U^{\Lambda \Sigma} - 3 \e^{\cK} \bar{X}^\Lambda X^\Sigma \right) \xi_\Lambda \xi_\Sigma \, ,
\eea
where $\xi_\Lambda$ are the constant quaternionic moment maps (known as FI parameters) and
\bea
 U^{\Lambda \Sigma} = - \frac{1}{2} \left( \im \cN \right)^{-1 | \Lambda \Sigma } - \e^{\cK} \bar{X}^\Lambda X^\Sigma \, .
\eea
The matrix $\cN_{\Lambda \Sigma}$ of the gauge kinetic term is given by \eqref{IIA:period:matrix}.
The special geometry prepotential $\cF \left( X^\Lambda \right)$ reads
\be
 \label{4d:prepotential:cubic}
 \cF ( X^\Lambda ) = - \frac16 \frac{C_{I J K} X^I X^J X^K}{X^0} = - \frac{X^1 X^2 X^3}{X^0} \, ,
\ee 
where in the second equality we employed the five-dimensional supergravity data for the STU model from section \ref{sec:AdS5 black holes}.
The K\"ahler potential is given by
\be
 \label{4d:Kahler:prepotential}
 \e^{- \cK (z , \bar{z})} = i ( \bar{X}^\Lambda F_\Lambda - X^\Lambda \bar{F}_\Lambda )
 = \frac{4 i}{3} C_{I J K} \im z^I \im z^J \im z^K
 = 8 \e^{- 6 \phi} \, ,
\ee
where due to the symmetries of the theory we can set $X^0 = 1$.
In the last equality we used \eqref{reduction:rules:5dBH}.
The K\"ahler metric can be written as
\be
 \label{Kahler:metric:scalars}
 g_{I J} = \partial_I \partial_{J} \cK(z , \bar{z})
 = - \frac{1}{4 \e^{- 6 \phi}} \left( C_{I J} - \frac{C_I C_J}{4 \e^{- 6 \phi}} \right) \, ,
\ee
where we introduced the following notation
\bea
 C_{I J} = C_{I J K} \im z^K \, , \qquad \qquad C_I = C_{I J K} \im z^J \im z^K \, .
\eea

In $\cN = 2$ gauged supergravity in four dimensions the $\U(1)_R$ symmetry, rotating the gravitini,
is gauged by a linear combination of the (now four) Abelian gauge fields.
Three of the FI parameters $g_\Lambda$ can be directly read off from the five-dimensional theory: $g_1 = g_2 = g_3 = 1$.%
\footnote{In consistent models one can always apply an
electric-magnetic duality transformation so that the corresponding gauging becomes purely electric, \ie, $g^\Lambda = 0$.}
The last coefficient, $g_0$, measuring how the Kaluza-Klein gauge potential $A^0_{(4)}$
enters the R-symmetry, can be left arbitrary for the moment.
This can be achieved by a Scherk-Schwarz reduction when allowing a particular reduction ansatz for the gravitino
as explained in \cite{Andrianopoli:2004im,Andrianopoli:2004xu,Andrianopoli:2005jv,Looyestijn:2010pb,Hristov:2014eba}.
The prepotential \eqref{4d:prepotential:cubic} and the FI parameters uniquely specify the four-dimensional $\cN = 2$ gauged supergravity Lagrangian and BPS variations.

Now, we can proceed with the explicit reduction of the line element \eqref{J1:J2:final:near-horizon} on the Hopf fibres of $S^3$ viewed as a $\U(1)$ bundle over $S^2 \cong \bC\bP^1$.
We thus identify $x^5$ with $\chi$. The four-dimensional solution takes the form\footnote{We have rescaled the time coordinate,
$\tilde{\tau} \equiv - R_{S^2} R_{\text{AdS}_2} \sqrt{\upsilon} \, \tau / 2$,
in order to put the AdS$_2$ part of the metric in the canonical coordinates.}
\be
\begin{aligned}
\label{eq:4dsolution}
 \rd s_{(4)}^2 & = - \e^{2 U} \, \rd \tilde{\tau}^2 + \e^{-2 U} \rd r^2+ \e^{2 \left( V - U \right)}
 \left( \rd \vartheta^2 + \sin^2\vartheta \, \rd \varphi^2 \right) \, , \\
 A_{(4)}^0 & = \tilde{q}_{(4)}^0(r) \, \rd \tilde{\tau} - \cos\vartheta \, \rd \varphi \, ,
 \qquad \qquad A_{(4)}^I = \tilde{q}_{(4)}^I(r) \, \rd \tilde{\tau} \, ,
\end{aligned} \ee
where
\be
\begin{aligned}
 \e^{U} & = \frac{\sqrt{2}}{R_{\text{AdS}_2} R_{S^2}^{1/2} \upsilon^{1/4}} \, r \, ,
 && \e^{V} = \frac{R_{S^2}}{2 R_{\text{AdS}_2}} \, r \, , \\
 \tilde{q}_{(4)}^0(r) & = - \frac{2 \alpha}{R^2_{\text{AdS}_2} R_{S^2} \upsilon^{1/2}} \, r \, ,
 \qquad && \tilde{q}_{(4)}^I(r) = - \frac{2 \left( e^I - f^I \alpha \right)}{R^2_{\text{AdS}_2} R_{S^2} \upsilon^{1/2}} \, r \, .
\end{aligned} \ee
The complex scalars are given by
\be
\begin{aligned}
 \label{4d:physical:scalars}
 z^I = - f^I + \frac{i}{2} R_{S^2} \upsilon^{1/2} L^I
 = - f^I + \frac{i}{2} \frac{\upsilon^{1/2} \gamma_3^{1/2}}{\mu_I} \, .
\end{aligned} \ee
Employing \eqref{em:charges:4d:scalars} we can compute the conserved electric charges.
After some work they read
\be
\begin{aligned}
 \label{J1:J2:em:charges:4d:final}
 q_{0} & = \frac{g}{8} \left( \gamma_2 + 2 g^2 \gamma_3 \right)
 = \frac{G_{\text{N}}^{(5)}}{\pi} J_{\phi} \, , \\
 q_I & = - \frac{1}{4} \left[ \mu_I + \frac{g^2}{2} \left( \gamma_2 - \frac{2 \gamma_3}{\mu_I} \right) \right] 
 = - \frac{G_{\text{N}}^{(5)}}{\pi} Q_I \, .
 \end{aligned} \ee
This is in agreement with \cite{Astefanesei:2011pz}.
The magnetic charges of the four-dimensional solution can be directly read off from the spherical components of the gauge fields $A^\Lambda_{(4)}$ \eqref{eq:4dsolution},
\be
\begin{aligned}
 p^0 = 1 \, , \qquad \qquad p^I = 0 \, .
\end{aligned} \ee
The entropy of the four-dimensional black hole precisely equals the entropy of the rotating black hole in five dimensions,
\be
\begin{aligned}
 \label{4d:entropy:5d}
 S_{\text{BH}}^{(4)} = \frac{\text{Area}^{(4)}}{4 G^{(4)}_{\rm N}} = \frac{\pi \e^{2 \left( V - U \right)}}{G^{(4)}_{\rm N}} = \frac{\pi^2 R_{S^2}^3 \upsilon^{1/2}}{2 G^{(5)}_{\rm N}} = S_{\text{BH}}^{(5)} \, ,
\end{aligned} \ee
upon using the standard relation 
\be
\begin{aligned}
 \frac{1}{G_{\rm N}^{(4)}} = \frac{4 \pi}{G_{\rm N}^{(5)}}\, .
\end{aligned} \ee

\subsection{Attractor mechanism in four dimensions}
\label{ssec:the attractor mechanism}

The BPS equations for the near-horizon solution \eqref{ansatz:metric:gauge field:4d}
with constant scalar fields $z^i$ imply that \cite{DallAgata:2010ejj}:\footnote{From comparing \eqref{attractor:initial:4d} with equations (3.5) and (3.8) in \cite{DallAgata:2010ejj},
we see that they differ by a factor of $2$. This is due to our different convention for the action (see footnote \ref{convention:action:metric}).}
\be
\begin{aligned}
 \label{attractor:initial:4d}
 \cZ - 2 i \e^{2 \left( V - U \right)} \cL = 0 \, , \qquad \qquad
 D_j \left( \cZ - 2 i \e^{2 \left( V - U \right)} \cL \right) = 0 \, ,
\end{aligned}
\ee
where $\cZ$ and $\cL$, are respectively the central charge and the superpotential of the black hole
defined in \eqref{intro:central charge:superpotential}, and $D_j = \partial_j + \frac12 \partial_j \cK$.
They can be rewritten as
\be
\begin{aligned}
 \label{attractor:4d}
 \partial_j \frac{\cZ}{\cL} = 0 \, , \qquad \qquad
 - i \frac{\cZ}{\cL} = 2 \e^{2 \left( V - U \right)} \, .
\end{aligned}
\ee
This is the attractor mechanism discussed in section \ref{ch:1:attractor:mechanism}.

We can extremize the quantity $- i \frac{\cZ}{\cL}$ under the following gauge fixing constraint,
which precisely corresponds to \eqref{constraint:N=4},
\be
\begin{aligned}
 \label{eq:gaugefixing}
 g_0 X^0 + \sum_{I = 1}^{3} X^I = 1 \, ,
\end{aligned} \ee
where we plugged in the explicit values for the FI parameters, \ie, $g_1 = g_2 = g_3 = 1$.
The real part of the sections $X^\Lambda$ are constrained in the range $0 < X^\Lambda < 1$.
We find that
\be
\begin{aligned}
 \label{J1:J2:attractor:equations}
 \partial_I \left[ \sum_{I = 1}^{3} X^I \left( q_I - \frac{q_0}{g_0} \right) + \frac{q_0}{g_0}
 - \frac{g_0^2 X^1 X^2 X^3}{\left( 1 - X^1 - X^2 - X^3 \right)^2} \right] = 0 \, , \quad \text{ for } \quad I = 1 , 2 , 3 \, ,
\end{aligned} \ee
where $\partial_I \equiv \partial / \partial X^I$.
The sections at the horizon are obtained from
\be
\begin{aligned}
 X^{0} = \frac{1}{g_0 \left( 1 + z^1 + z^2 + z^3 \right)} \, , \qquad \qquad X^{1,2,3} = \frac{z^{1,2,3}}{1 + z^1 + z^2 + z^3} \, .
\end{aligned} \ee

We are now in a position to determine the value of the FI parameter $g_0$.
Partial topological A-twist along $S^2$ ensures that $\cN = 2$ supersymmetry is preserved in four dimensions \cite{Witten:1991zz}.
This leads to the following Dirac-like quantization condition \cite{Cacciatori:2009iz,DallAgata:2010ejj,Hristov:2010ri}:
\be
\begin{aligned}
 g_{\Lambda} p^{\Lambda}  = 1 = g_0 p^0 \, ,
\end{aligned} \ee
which fixes the value of $g_0 = 1$.
It is straightforward to check that, substituting the values for the physical scalars at the horizon \eqref{4d:physical:scalars},
the charges \eqref{J1:J2:em:charges:4d:final}, and setting $g_0 = 1$, Eq.\,\eqref{J1:J2:attractor:equations} is fulfilled.
The scalars $\bar{z}^i(r_{\rm h})$ at the horizon are determined in terms of the black hole charges $q_I$ by virtue of the attractor equations:
\be
\begin{aligned}
 q_I - q_0 = \left( 2 + \frac{1}{\bar{z}^I} \right) \bar{z}^1 \bar{z}^2 \bar{z}^3 \, , \qquad \text{ for } \quad I = 1 , 2 , 3 \, .
\end{aligned} \ee
The value of $- i \frac{\cZ}{\cL}$ at the critical point yields,
\be
\begin{aligned}
 - i \frac{\cZ}{\cL} \bigg|_{\rm crit} \left( q_\Lambda \right) = 2 \e^{2 \left( V - U \right)}
 = \frac{2 G^{(4)}_{\rm N}}{\pi} S^{(4)}_{\rm BH} \left( q_\Lambda \right)
 \, .
\end{aligned} \ee
The holding of the four-dimensional BPS attractor mechanism for
the dimensionally reduced near-horizon geometry \eqref{eq:4dsolution}
proves that the dimensional reduction preserves the full
amount of supersymmetries originally present in five dimensions.

Due to the very suggestive form of the attractor equations \eqref{J1:J2:attractor:equations}
it is now not hard to compare them with the five-dimensional extremization.

\subsection{Comparison with five-dimensional extremization}
\label{ssec:ESUSY:extremization}

Consider the quantity $E$ in \eqref{eq:Casimir:rewrittenX}
rewritten in terms of the chemical potentials for $J^{\pm}$ and $Q_{1,2,3}$.
Recall that we are focusing on the case with equal angular momenta, \ie, $J_{\psi} = J_{\phi}$ (so $J^{-} = 0)$.
Extremizing \eqref{rotating:attractor:4d} with respect to $X^0_{-}$ fixes the value of $X^0_{-} = 0$.
Thus, the black hole entropy is obtained by extremizing the quantity
\be
\begin{aligned}
 \label{static:attractor:4d}
 \cI_{\text{sugra}} \big|_{J^- = 0} =   \frac{2 \pi^2 i}{G_{{\rm N}}^{(5)}} \frac{X^1 X^2 X^3}{\left( X^0_{+} \right)^2} + 2 \pi i \sum_{I = 1}^{3} Q_I X^I  - \pi i J^{+} X_{+}^0 
\, ,
\end{aligned} \ee
subject to the constraint \eqref{eq:constraintX}.
Identifying $X^{0}$ in \eqref{J1:J2:attractor:equations}
with $X^{0}_{+}$ in \eqref{static:attractor:4d},
and using $g_0 = 1$ together with \eqref{J1:J2:em:charges:4d:final},
we find that the extremization of $\cI_{\text{sugra}}$ 
corresponds to the four-dimensional attractor mechanism on the gravity side and they lead
to the same entropy.

\section{Discussion}
\label{sec:discussion}

We have shown that the entropy of a supersymmetric rotating black hole in AdS$_5$ with electric charges $Q_I$ $(I = 1, 2 , 3)$
and angular momenta $J_\phi \equiv J_1$, $J_\psi \equiv J_2$ can be obtained as the Legendre
transform of the quantity $-E$ in \eqref{modified:Casimir:SUSY:N=4}:
\be\label{E-extremization}
 S_{\text{BH}}(Q_I,J_i) = - E(\Delta_I,\omega_i)  +  2 \pi i  \bigg( \sum_{I=1}^3 Q_I \Delta_I -  \sum_{i = 1}^2 J_i \omega_i \bigg) \bigg |_{\bar\Delta_I, \bar\omega_i} \, ,
\ee
where $\bar\Delta_I$ and  $\bar\omega_i$ are the extrema of the functional on the right hand side.

The result is quite intriguing and deserves a better explanation and understanding. We leave a more careful analysis for the future. For the moment, let us just make few preliminary observations.

The quantity $E$ can be interpreted  as a combination of 't Hooft anomaly polynomials that arise studying the partition function
$Z_{{\cal N}=4}(\Delta_I,\omega_i)$ on $S^3 \times S^1$ or the superconformal index
$I(\Delta_I,\omega_i)$ for ${\cal N}=4$ SYM \cite{Bobev:2015kza,Brunner:2016nyk}.
Some explicit expressions are given in appendix \ref{AppC}.
In particular, $E$ is  {\it formally} equal to the supersymmetric Casimir energy of ${\cal N}=4$ SYM as a function of the chemical potentials (see for example equation (4.50) in \cite{Bobev:2015kza} and appendix \ref{AppC}).
However, this analogy is only formal since we are imposing the constraint \eqref{constraint:N=4}. Chemical potentials are only defined modulo 1, so the constraint to be imposed on them also suffers from angular ambiguities.
Consistency of the index and partition function just requires  $\sum_{I=1}^3 \Delta_I +\sum_{i=1}^2 \omega_i \in \mathbb{Z}$. To recover the known expressions for the supersymmetric Casimir energy and for consistency with gauge anomaly cancellations \cite{Assel:2015nca,Bobev:2015kza},
one needs to impose  $\sum_{I=1}^3 \Delta_I +\sum_{i=1}^2 \omega_i =0$, and this contrasts with \eqref{constraint:N=4}. 

It would be tempting to interpret  the Legendre transform \eqref{E-extremization} as a result of the saddle-point approximation
of a Laplace integral of $Z_{{\cal N}=4}$ in the limit of large charges (large $N$).\footnote{We are ignoring here potential sign ambiguities in the definition of charges.}
Ignoring angular ambiguities, $E$ is  the leading contribution at order $N^2$
of the logarithm of the partition function $Z_{{\cal N}=4}$ on $S^3\times S^1$.
Indeed, $\log Z_{{\cal N}=4} = - E + \log I$ \cite{Assel:2014paa,Lorenzen:2014pna,Assel:2015nca,Bobev:2015kza,Martelli:2015kuk,Genolini:2016sxe,Brunner:2016nyk}
and the index is a quantity independent of $N$ for generic values for the chemical potentials  \cite{Kinney:2005ej}.
In these terms, the result would be completely analogous to the connection between asymptotically
AdS$_4\times S^7$ (or AdS$_4\times S^6$) back hole entropy and the topologically twisted index of ABJM (or D2$_k$)
\cite{Benini:2015eyy,Benini:2016rke,Hosseini:2017fjo,Benini:2017oxt} .

The appearance of the supersymmetric Casimir energy can be surprising since the entropy counts the degeneracy of ground states of the system. 
However, the dimensional reduction to four dimensions performed in section \ref{sec:a limiting case} offers a different perspective on this point.
The dimensionally reduced black hole is static but not asymptotically AdS. Let us assume that we can still use holography.
In the dimensional reduction,  a magnetic flux $p^0$ is turned on for the graviphoton. This means that supersymmetry
is preserved with a topological twist.  The same should be true for the boundary theory.
It is then tempting to speculate that, upon dimensional reduction, the partition function $Z_{{\cal N}=4}$
becomes the topologically twisted index of the boundary three-dimensional theory \cite{Benini:2015noa}.
The supersymmetric Casimir energy, which is the leading contribution  of $\log Z_{{\cal N}=4}$ at large $N$
then becomes the leading contribution of the three-dimensional topologically twisted index and the latter is known to
correctly account  for the microstates of four-dimensional black holes.

The above discussion ignores completely the angular ambiguities and the role of the  constraint \eqref{constraint:N=4},
which should be further investigated. For sure, the result of the extremization of $E$ is complex and lies in the region
where the chemical potentials satisfy \eqref{constraint:N=4}. Unfortunately, we are not aware of a general discussion
of the possible regularizations  of $Z_{{\cal N}=4}$ that takes into account the angular ambiguities. Moreover, there is some recent claim \cite{Papadimitriou:2017kzw,An:2017ihs} 
of the presence of an anomaly in the supersymmetry transformations leading to a modification of the BPS condition in gravity
that would be interesting to investigate further in this context. 

Both the constraint \eqref{constraint:N=4} and the analogous of the more traditional one $\sum_{I=1}^3 \Delta_I +\sum_{i=1}^2 \omega_i =0$
have been used in the literature to explore different features of $Z_{{\cal N}=4}$ or the index.
The traditional constraint has been used in the analysis of the high-temperature limit of the index
\cite{Ardehali:2015hya,Ardehali:2015bla} (see also \cite{Shaghoulian:2016gol,DiPietro:2016ond}) and in the study of factorization properties \cite{Nieri:2015yia}.
On the other hand, the importance of \eqref{constraint:N=4} has been stressed in  \cite{Brunner:2016nyk}
where the constraint has been used to extract the supersymmetric Casimir energy directly from the superconformal index.\footnote{Interestingly, the same constraint is also used in relating
the universal part of supersymmetric R\'enyi entropy to an equivariant integral of the anomaly polynomial \cite{Yankielowicz:2017xkf}.} See appendix \ref{AppC} for more details.
In the low temperature limit, which can be obtained by rescaling $\Delta_I \to \beta\Delta_I, \omega_i \to \beta \omega_i$ and taking large $\beta$, the angular ambiguity in the constraint disappears.
 
Finally, it is worth mentioning that angular ambiguities also played a prominent r\^ole in the evaluation of the saddle-point
for the topologically twisted index of ABJM and D2$_k$, and the comparison with the entropy of AdS$_4$ black holes.


%

\chapter{Discussion and future directions}
\label{ch:7}

\ifpdf
    \graphicspath{{Chapter7/Figs/Raster/}{Chapter7/Figs/PDF/}{Chapter7/Figs/}}
\else
    \graphicspath{{Chapter7/Figs/Vector/}{Chapter7/Figs/}}
\fi

The main goal of this dissertation was to give an explanation for the microscopic origin of the
Bekenstein-Hawking entropy for a class of static BPS black holes (strings) in AdS$_{4,5}$.
The specific theories we focused on are consistent truncations of string or M-theory.
To this aim, we studied the topologically twisted index for a variety of 3D $\cN=2$
and 4D $\cN=1$ gauge theories in the large $N$ limit.
The index is a function of background magnetic fluxes and chemical potentials for the global symmetries of the theory,
and can be reduced to a matrix model using the technique of supersymmetric localization \cite{Benini:2015noa}.
Using the method introduced in \cite{Benini:2015eyy}, we solved a number of such matrix models.
These computations reveal the characteristic $N^{3/2} / N^{5/3} / N^{2}$ scaling of the number of degrees of
freedom on $N$ coincident $\text{M}2 / \text{D}2_k / \text{D}3$-branes.

An obvious follow-up is to find new examples of AdS$_4$ M-theory and massive type IIA black holes directly
in eleven or ten dimensions (see, for example, \cite{Katmadas:2015ima}) or in some other consistent truncations
of eleven-dimensional or ten-dimensional supergravity where to test our results.
In particular, in chapter \ref{ch:2} we evaluated the index for field theories dual to a variety of $\cN \geq 2$
M-theory backgrounds, including such well-known solutions as AdS$_4 \times N^{0,1,0}/\bZ_k$,
AdS$_4 \times V^{5,2}/\bZ_k$, and AdS$_4 \times Q^{1,1,1}/\bZ_k$.
One feature of these quivers compared to ABJM is the presence of many baryonic symmetries
that couple to the vector multiplets arising from nontrivial two-cycles (and thus by Poincar\'e duality five-cycles)
in the Sasaki-Einstein 7-manifold. As is evident from our analyses,
such background fluxes for baryonic symmetries do not show up in the large $N$ limit of the index,
while they affect the details of the black hole entropy.
Solving this apparent puzzle provides rather intricate tests of the proposed AdS$_4 /$CFT$_3$ dualities.

We studied the 3D matrix models in the limit where $N \gg k_a$ and the Chern-Simons levels $k_a$ are kept fixed.
It would be most interesting to develop a new method to study the topologically twisted index beyond the leading large $N$ contribution.
Recent attempts in doing so can be found in \cite{Liu:2017vll}, where the $\log N$ coefficient of the index
was extracted using numerical techniques.
There are two main subtleties in going beyond the numerical methods employed therein:
first, the tail contributions seem to prevent one from a systematic large $N$ expansion of the index;
secondly, the imaginary part of the index needs a better understanding. As we saw in the D2$_k$ theory,
the latter was very important in order to make the black hole entropy a real quantity.


We briefly discussed, in section \ref{ch1:susy background}, the refinement of the index by the angular momentum on $S^2$.
In the path integral formulation this corresponds to turning on an $\Omega$-background on $S^2$.
The holographic dual description of this setup is given by a dyonic, rotating supersymmetric black hole in AdS$_4$,
preserving (at least) two real supercharges.
Another interesting generalization of our results is to solve such matrix models \eqref{conclusions:path:refined}.

%



\begin{appendices} 

\chapter{Special functions}

In this appendix we review the special functions and their properties which we used in this dissertation.

\section{Polylogarithms} 

The polylogarithm function $\Li_n(z)$ is defined by a power series
\be
\Li_n(z) = \sum_{k=1}^\infty \frac{z^k}{k^n} \, ,
\ee
in the complex plane over the open unit disk, and by analytic continuation outside the disk.
For $z = 1$ the polylogarithm reduces to the Riemann zeta function
\be
\Li_n(1) = \zeta(n)\, , \qquad \text{for} \quad \re n >1 \, .
\ee
The polylogarithm for $n=0$ and $n=1$ is
\be
\Li_0(z) = \frac{z}{1-z}\, , \qquad \qquad  \Li_1(z) = - \log(1-z) \, .
\ee
Notice that $\Li_0(z)$ and $\Li_1(z)$ diverge at $z=1$.
For $n\geq 1$, the functions have a branch point at $z=1$ and we shall take the principal determination with a cut $[1,+\infty)$ along the real axis.
The polylogarithms fulfill the following relations
\be
\partial_u \Li_n(\e^{iu}) = i \Li_{n-1}(\e^{iu}) \;,\qquad\qquad \Li_n(\e^{iu}) = i \int_{+i\infty}^u \Li_{n-1}(\e^{iu'})\, \rd u' \;.
\ee
The functions $\Li_n(\e^{iu})$ are periodic under $u \to u+2\pi$ and have branch cut discontinuities along the vertical line $[0, -i\infty)$ and its images.
For $0< \re u < 2\pi$, polylogarithms satisfy the following inversion formul\ae{}\footnote{The inversion formul\ae{} in the domain $-2\pi < \re u < 0$ are obtaind by sending $u \to -u$.}
\bea
\label{reflection formulae}
\Li_0(\e^{iu}) + \Li_0(\e^{-iu}) &= -1 \\
\Li_1(\e^{iu}) - \Li_1(\e^{-iu}) &= -iu + i\pi \\
\Li_2(\e^{iu}) + \Li_2(\e^{-iu}) &= \frac{u^2}2 - \pi u + \frac{\pi^2}3 \\
\Li_3(\e^{iu}) - \Li_3(\e^{-iu}) &= \frac i6 u^3 - \frac{i\pi}2 u^2 + \frac{i\pi^2}3 u \, .
\eea
One can find the formul\ae{} in the other regions by periodicity.

\section{Eta and theta functions}
\label{Elliptic functions}

The Dedekind eta function is defined by
\be
\eta(q) = \eta(\tau) = q^{\frac{1}{24}} \prod_{n=1}^{\infty} \left(1 - q^{n} \right) \, , \qquad \qquad \im \tau > 0 \, ,
\label{dedekind:eta}
\ee
where $q = \e^{2 \pi i \tau}$.
It has the following modular properties
\be\label{modular_eta}
\eta(\tau + 1) = \e^{\frac{i \pi }{12}} \, \eta (\tau)\, , \qquad \qquad \eta\left( - \frac{1}{\tau} \right) = \sqrt{- i \tau} \, \eta (\tau) \, .
\ee
The Jacobi theta function reads
\bea
\theta_1 (x;q) = \theta_1 (u;\tau) & = - i q^{\frac18} x^{\frac12} \prod_{k=1}^{\infty} \left( 1- q^k \right) \left( 1-x q^k \right) \left( 1- x^{-1} q^{k-1} \right) \\
& = -i \sum_{n \in \mathbb{Z}} (-1)^n \e^{i u \left( n+ \frac12 \right)} \e^{\pi i \tau \left( n+ \frac12 \right)^2}\, ,
\label{theta:function}
\eea
where $x = \e^{i u}$ and $q$ is as before.
The function $\theta_1 (u; \tau)$ has simple zeros in $u$ at $u = 2 \pi \mathbb{Z} + 2 \pi \tau \mathbb{Z}$ and no poles.
Its modular properties are,
\be\label{modular_theta}
\theta_1\left(u;\tau+1\right) = \e^{\frac{i \pi}{4}} \, \theta_1\left(u;\tau\right)\, , \qquad \qquad
\theta_1\left( \frac{u}{\tau};-\frac{1}{\tau}\right) = - i \sqrt{- i \tau} \, \e^{\frac{i u^2}{4 \pi \tau}} \, \theta_1\left( u; \tau\right) \, .
\ee
We also note the following useful formula,
\be
\label{theta:function:shift}
\theta_1\left( q^m x;q\right) = (-1)^{-m}\,x^{-m} q^{-\frac{m^2}{2}}\theta_1(x;q) \, , \qquad \qquad m \in \bZ \, .
\ee

The asymptotic behavior of the $\eta(q)$ and $\theta_1(x; q)$ as $q \to 1$ can be derived by using their modular properties.
To this purpose, we first need to perform an $S$-transformation, \ie\;$\tau\to -1/\tau$, and then expand the
resulting functions in series of $q$, which is now a small parameter in the $\tau\to i0$ limit. 

Let us start with the Dedekind $\eta$-function.
The action of modular transformation is written in (\ref{modular_eta}). 
We will identify the ``inverse temperature'' $\beta$ with the modular parameter $\tau$ of the torus: $\tau=i\beta/2\pi$.
Then, expanding the $S$-transformed $\eta$-function we get 
\be
 \label{dedekind:hight:S}
  \log\left[\eta(\tau)\right] = -\frac{1}{2}\log\left( -i \tau \right) + \log\left[ \eta\left( -\frac{1}{\tau}\right)\right]
   = - \frac{1}{2}\log\left( \frac{\beta}{2\pi}\right)-\frac{\pi^2}{6\beta} + \cO\left( \e^{- 1 / \beta} \right) \, .
\ee
Similarly, we can consider the asymptotic expansion of the Jacobi $\theta$-function:
\bea
 \log\left[\theta_1(u;\tau) \right] & =
 \frac{i \pi}{2} - \frac{1}{2}\log\left( -i \tau \right) - \frac{i u^2 }{4 \pi  \tau}
 + \log\left[ \theta_1\left(\frac{u}{\tau}; -\frac{1}{\tau}\right)\right] \\
 & = -\frac{\pi^2}{2\beta} - \frac{u^2}{2\beta}-\frac{1}{2} \log\left( \frac{\beta}{2\pi} \right)
 + \log\left[2 \sinh\left(\frac{\pi u}{\beta} \right)\right]
 + \cO \left( \e^{- 1 / \beta} \right) \, ,
\eea
Writing $2\sinh\left(\frac{\pi u}{\beta} \right) = \e^{\pi u/\beta}\left( 1 - \e^{- 2 \pi u/\beta} \right)$,
we have the following expansion
\be
 \log\left[2\sinh\left(\frac{\pi u}{\beta} \right)\right]
 = \frac{\pi}{\beta} u \sign\left[ \re (u) \right]
 - \sum_{k=1}^{\infty} \frac{1}{k} \, \e^{- \frac{2 k \pi}{\beta} u \sign\left[ \re (u) \right]} \, .
\ee
Putting all pieces together, we find
\bea
\label{theta:hight:S}
\log\left[ \theta_1(u;\tau) \right] = -\frac{\pi^2}{2\beta} - \frac{u^2}{2\beta}-\frac{1}{2} \log\left( \frac{\beta}{2\pi} \right)
+ \frac{\pi}{\beta} u \sign\left[ \re(u) \right]
+ \cO \left( \e^{- 1 / \beta} \right) \, .
\eea

%

\chapter[BPS black holes in \texorpdfstring{$\cN=2$}{N=2} dyonic STU gauged supergravity]{BPS black holes in $\fakebold{\cN=2}$ dyonic STU gauged supergravity}
\label{IIA:appendix}


In this appendix we write down a black hole ansatz and derive the corresponding BPS equations.
The black hole can be embedded in massive type IIA supergravity and is asymptotic to
the $\mathcal{N} = 2$ supersymmetric AdS$_4 \times S^6$ background of \cite{Guarino:2015jca}.

\section{Black hole ansatz}
We are interested in supersymmetric asymptotically AdS$_4$ black holes,
which in \cite{Klemm:2016wng} were considered for general models with hypermultiplets and dyonic gaugings,
extending earlier works \cite{Cacciatori:2009iz,DallAgata:2010ejj,Hristov:2010ri,Halmagyi:2013sla,Katmadas:2014faa,Halmagyi:2014qza}.
Reviewing these results, we write down the bosonic field ansatz and the final form of the supersymmetry conditions to be solved,
which also imply all equations of motion. The metric is given by
\begin{equation}
{\rm d} s^2 = - \e^{2U(r)} {\rm d} t^2 + \e^{-2U(r)} {\rm d} r^2 + \e^{2(\psi(r) - U(r))}{\rm d}\Omega_{\kappa}^2\ ,\label{eq:ansatzmet} 
\end{equation}
where ${\rm d}\Omega_{\kappa}^2={\rm d}\theta^2+f_{\kappa}^2(\theta){\rm d}\varphi^2$ defines the metric on a surface $\Sigma_\fg$ of constant scalar curvature
$2\kappa$, with $\kappa\in\{+1,-1\}$, and
\begin{equation}
f_\kappa(\theta) = \frac{1}{\sqrt{\kappa}} \sin(\sqrt{\kappa}\theta) = 
\left\{\begin{array}{c@{\quad}l} \sin\theta\, & \kappa=+1\,, \\                                             
                                             \sinh\theta\, & \kappa=-1\,. \end{array}\right. 
\end{equation} 
The scalar fields depend only on the radial coordinate $r$, while the electric and magnetic gauge fields $(A^\Lambda$, $A_{\Lambda})$ and the tensor fields $(B^\Lambda$, $B_\Lambda)$ are given by
\begin{equation}
A^\Lambda = A_t^\Lambda {\rm d} t- \kappa p^\Lambda f^\prime_\kappa(\theta){\rm d}\phi\ , \qquad
A_\Lambda = A_{\Lambda t}{\rm d} t - \kappa e_\Lambda f^\prime_\kappa(\theta){\rm d}\phi\ ,
\end{equation}
\begin{equation}
B^\Lambda = 2\kappa p^{\prime\,\Lambda} f^\prime_\kappa(\theta){\rm d} r\wedge{\rm d}\phi\ , \qquad 
B_\Lambda = -2\kappa e^{\prime}_\Lambda f^\prime_\kappa(\theta){\rm d} r\wedge{\rm d}\phi\ .
\end{equation}
In the theory we consider the only relevant tensor field is $B^0$, as the rest can be consistently decoupled.
The magnetic and electric charges $(p^\Lambda, e_\Lambda)$ are defined as
\begin{equation}
p^{\Lambda} \equiv \frac1{\mbox{vol}(\Sigma_\fg)}\int_{\Sigma_\fg} F^{\Lambda}\ , \quad   e_{\Lambda} \equiv \frac1{\mbox{vol}(\Sigma_\fg)}\int_{\Sigma_\fg} G_{\Lambda}\ , \quad
\mbox{vol}(\Sigma_\fg) = \int f_\kappa(\theta){\rm d}\theta\wedge{\rm d}\phi\ .
\label{eq:charges}
\end{equation}
Note that the charges can depend on the radial coordinate in general,
since the Maxwell equations are sourced by the hypermultiplet scalars due to the gauging.

\section{BPS and Maxwell equations}
The above ansatz is subject to a set of conditions required for supersymmetry, which after a number of manipulations can be recast into a set of algebraic and first order differential equations, given by (3.74) in \cite{Klemm:2016wng} in a manifestly symplectic covariant way,
\begin{align}
\begin{split}
\label{app:BHequations}
{\cal E}  &= 0\ , \\
\psi^{\prime} &= -2\kappa \e^{-U} \im(\e^{-i\alpha}\mathcal{L})\ , \\
\alpha^{\prime} + A_r &= 2\kappa \e^{-U} \re(\e^{-i\alpha} {\cal L})\ , \\
q^{\prime\,u} &= \kappa \e^{-U} h^{uv}\im (\e^{-i\alpha}\partial_v {\cal L})\ ,  \\
{\cal Q}^\prime &= -4 \e^{2\psi - 3U} {\cal H}\Omega \re(\e^{-i\alpha}{\cal V})\ ,\\
\end{split}
\end{align}
where
\begin{align}
	{\cal E} \equiv 2 \e^{2\psi}\left(\e^{-U} \im(\e^{-i\alpha}{\cal V})\right)^{\prime} - \kappa \e^{2(\psi -
U)}\Omega{\cal M} {\cal Q}^x {\cal P}^x + 4 \e^{2\psi-U}(\alpha^{\prime} + A_r)\re
(\e^{-i\alpha} {\cal V}) +{\cal Q}\ .
\end{align}
As earlier introduced, ${\cal Q} = (p^{\Lambda}, e_{\Lambda})$, ${\cal P}^x = (P^{x, \Lambda}, P^x_{\Lambda})$
and ${\cal K}^u = (k^{u, \Lambda}, k^u_{\Lambda})$. $A_\mu$ is the $\U(1)$ K\"ahler connection,
and $\alpha$ an a priori arbitrary phase of the Killing spinor,
which depends only on the radial coordinate (derivatives with respect to which are given by primes). Furthermore,
\begin{align}\label{app:quantities}
{\cal Q}^x & \equiv g \langle {\cal P}^x, {\cal Q}\rangle  = g P^x_{\Lambda} p^{\Lambda} - g P^{x, \Lambda} e_{\Lambda}\ , \quad
{\cal W}^x  \equiv g \langle {\cal P}^x, {\cal V}\rangle  = g  P^x_{\Lambda} L^{\Lambda} - g P^{x, \Lambda} M_{\Lambda}\ ,
\nn \\[.2cm]
{\cal Z} & \equiv \langle {\cal Q}, {\cal V}\rangle  = e_{\Lambda} L^{\Lambda} - p^{\Lambda} M_{\Lambda}\ , \quad
{\cal L}  \equiv {\textstyle \sum_x} {\cal Q}^x {\cal W}^x \ , 
\end{align}
and
\begin{equation}
{\cal H} \equiv g^2 ({\cal K}^u)^T h_{u v} {\cal K}^v \ , \quad
{\cal M} =\left(\begin{array}{cc}
 {\rm I} + 
 {\rm R}  {\rm I}^{-1}  {\rm R} & \,\,-  {\rm R} {\rm I} ^{-1} \\
- {\rm I} ^{-1}  {\rm R} &  {\rm I}^{-1} \\
\end{array}\right) \ , \quad \Omega = \left(\begin{array}{cc} 0 & -{\bf 1} \\ {\bf 1} & 0 \end{array}\right) \ .
\end{equation}
 
The above equations are further supplemented by the constraints
\begin{equation}\label{app:constr}
{\cal H} \Omega {\cal Q} = 0 \ , \qquad {\cal K}^u h_{uv} q'^v = 0 \ , \qquad {\cal Q}^x {\cal Q}^x = 1 \ .
\end{equation}

As already noted, in addition to the BPS equations, the Maxwell equations need to be imposed. The rest of the equations of motion then follow. The Maxwell equations sourced by the hypermultiplet scalars evaluated on the specified bosonic ansatz lead to \cite{Klemm:2016wng} a pair of coupled first order differential equations
\begin{equation}
	{\cal A}_t' = -\e^{2(U-\psi)} \Omega {\cal M} {\cal Q}\ , \qquad {\cal Q}' = -2 \e^{2 (\psi-2 U)} {\cal H} \Omega {\cal A}_t \ .
\end{equation}
They are immediately satisfied given the fifth row of \eqref{app:BHequations} together with the extra constraint
\begin{equation}
	2 \e^U {\cal H} \Omega \re (\e^{-i \alpha} {\cal V}) = {\cal H} \Omega {\cal A}_t\ .
\end{equation}

\subsection{Solution to the constraints}

Without making any further assumptions, we can already solve for some of the scalar fields using the constraints \eqref{app:constr}
that need to hold everywhere in spacetime. The first equation in \eqref{app:constr} gives
\begin{equation}
	g p^0 - m e_0 = 0 \ , \qquad (\zeta^2+\tilde{\zeta}^2) \sum_{I=1}^3 p^I = 0 \ ,
\end{equation}
while the last one further fixes
\begin{equation}
g\sum_{I=1}^3 p^I = \pm 1 \ .
\end{equation}
Hence,
\begin{equation}
\zeta = \tilde{\zeta} = 0 \ .
\end{equation}
The second equation in \eqref{app:constr} then yields 
\begin{equation}
	\sigma = \text{const.} \, .
\end{equation}
Following the above results, 
\bea\label{app:quantities-simplified}
{\cal Q}^x & = g{\tsumI} p^I \delta^{x, 3}  = \pm 1 \delta^{x, 3} \equiv \lambda \delta^{x, 3} \ , \\[.2cm]
{\cal W}^x & = g \e^{{\cal K}/2}\left[ \tsumI X^I -\tfrac{1}{2} \e^{2 \phi}(X^0 - c F_0) \right] \delta^{x, 3} \ , \\[.2cm]
{\cal Z} & = \e^{{\cal K}/2} \tsumI ( e_I X^I - p^I F_I) + \e^{{\cal K}/2} e_0 (X^0 - c F_0) \ ,  \\[.2cm]
{\cal L} &= \lambda g \e^{{\cal K}/2}\left[ \tsumI X^I -\tfrac{1}{2}\e^{2 \phi}(X^0 - c F_0) \right] ,
\eea
and the only components of the matrix ${\cal H}$ that remain non-vanishing are
\begin{align}
{\cal H}_{0 0} = \frac14 \e^{4 \phi}\ , \quad {\cal H}^{0 0} = \frac14 c^2 \e^{4 \phi}\ , \quad {\cal H}_{0}{}^{0} = {\cal H}^0{}_0 = \frac14 c \e^{4 \phi}\ .
\end{align}
We have already solved for three of the four hypermultiplet scalars, so it is worth writing explicitly the differential equation that determines the remaining scalar $\phi$, coming from the fourth equation of \eqref{app:BHequations}:
\begin{equation}\label{app:phi}
	\phi' = - g \kappa \lambda \e^{{\cal K}/2 - U} \im \left(\e^{-i \alpha} (X^0 - c F_0)\right)\ . 	
\end{equation}
The scalar $\phi$ is exactly the source that does not allow the charges $p^0 = c e_0$ to be conserved as it appears in the matrix ${\cal H}$ on the right-hand side of the Maxwell equations, 
\begin{equation}\label{app:charge}
	p'^0 = c e'_0 = -c \e^{2 \psi - 3 U} \e^{4 \phi} \re \left(\e^{-i \alpha} (X^0 - c F_0)\right)\ .
\end{equation}
Therefore, the charges $p^0$, $e_0$ cannot actually ``be felt'' by the field theory at the asymptotic AdS$_4$ boundary. 

The equations have been simplified, and are given by the scalar equation \eqref{app:phi}, the Maxwell equation \eqref{app:charge}, and the first three equations in \eqref{app:BHequations}. Note that in the absence of the hypermultiplet equations, \eqref{app:phi} and \eqref{app:charge} are solved trivially, and the remaining equations in \eqref{app:BHequations} can be solved analytically using standard special geometry identities. Here, we are unable to present an analytic solution for the full black hole geometry, exactly due to the complication of solving \eqref{app:phi} and \eqref{app:charge}.
We are however able to present an analytic solution for the two end-points of the black hole geometry, due to the extra condition of the scalars and charges being constant.

Before moving to the ``constant scalars and charges'' case, let us give the relevant components of the matrix ${\cal M}$ which allow us to write down the first equation in \eqref{app:BHequations}:
\begin{align}
{\cal M}^{00} &= -8 \e^{\cal K} |s|^2 |t|^2 |u|^2, \quad {\cal M}^0{}_0 =  -8 \e^{\cal K} \re (s)\ \re (t)\ \re (u) \ , \nn \\
{\cal M}^{01} &=-8 \e^{\cal K} \re (t)\ \re (u) |s|^2 \ , \nn \\
{\cal M}^{02} &=-8 \e^{\cal K} \re (s)\ \re (u) |t|^2 \ , \nn \\ 
{\cal M}^{03} &=-8 \e^{\cal K} \re (s)\ \re (t) |u|^2 \ , \nn \\
{\cal M}^{10} &= {\cal M}^{01} , \quad
{\cal M}^1{}_0 = -8 \e^{\cal K} \re (s) \ , \nn \\
{\cal M}^{11} &=-8 \e^{\cal K} |s|^2 \ , \quad 
{\cal M}^{12} =-8 \e^{\cal K} \re (s) \re(t) \ , \quad {\cal M}^{13} = -8 \e^{\cal K} \re (s) \re (u) 
\ , \nn \\
{\cal M}^{20} &= {\cal M}^{02} , \quad {\cal M}^2{}_0 = -8 \e^{\cal K} \re (t) \ , \nn \\
{\cal M}^{21} &= {\cal M}^{12} \ , \quad
{\cal M}^{22} = -8 \e^{\cal K} |t|^2 \ , \quad 
{\cal M}^{23} = -8 \e^{\cal K} \re (t) \re(u) \ , \nn \\
{\cal M}^{30} &= {\cal M}^{03} , \quad {\cal M}^3{}_0 = -8 \e^{\cal K} \re (u) \ , \nn \\
{\cal M}^{31} &= {\cal M}^{13} \ , \quad
{\cal M}^{32} = {\cal M}^{23} \ , \quad 
{\cal M}^{33} = -8 \e^{\cal K} |u|^2 \ , \nn \\
{\cal M}_{00} &= -8 \e^{\cal K} , \quad {\cal M}_0{}^0 = {\cal M}_1{}^1 = {\cal M}_2{}^2 = {\cal M}_3{}^3 =  -8 \e^{\cal K} \re (s)\ \re (t)\ \re (u)\ , \nn \\
{\cal M}_0{}^{1} &= {\cal M}^1{}_0 \ , \quad  {\cal M}_0{}^{2}= {\cal M}^2{}_0 \ , \quad  {\cal M}_0{}^{3} = {\cal M}^2{}_0 \ , \nn \\
{\cal M}_1{}^0 &= -8 \e^{\mathcal{K}} \re (s) |t|^2 |u|^2 \ , \quad 
{\cal M}_2{}^0 =  -8 \e^{\mathcal{K}} \re (t) |s|^2 |u|^2 \ , \quad  
{\cal M}_3{}^0 =  -8 \e^{\mathcal{K}} \re (u) |s|^2 |t|^2 \ , \nn \\
{\cal M}_1{}^2 &= -8 \e^{\cal K} |t|^2 \re (u) \ , \quad  {\cal M}_1{}^3 = -8 \e^{\cal K} |u|^2 \re (t)
\ , \nn \\ 
{\cal M}_2{}^1 &= -8 \e^{\cal K} |s|^2 \re (u) \ , \quad {\cal M}_2{}^3 = -8 \e^{\cal K} |u|^2 \re (s) \ ,  
 \nn \\
{\cal M}_3{}^1 &= -8 \e^{\cal K} |s|^2 \re (t) \ , \quad {\cal M}_3{}^2 = -8 \e^{\cal K} |t|^2 \re (s) \ .
\end{align}

\section{Constant scalars and charges}

The condition that all scalars and charges are constant, 
\begin{equation}
	s'= t'=u' = 0\ , \quad q'^u = 0\ , \quad {\cal Q}' = 0\ ,
\end{equation}
(based on the symmetries of AdS$_4$ and AdS$_2 \times \Sigma_\fg$), upon imposed on \eqref{app:phi}-\eqref{app:charge} yields
\begin{equation}
X^0 - c F_0 = 0 \ . 
\end{equation}
This is a strong constraint on the special K\"ahler manifold, leading to 
\begin{equation}
	s t u = -c \ ,
\end{equation}
and therefore  $(X^1,X^2,X^3,F_1, F_2,F_3) = (c/(t u),-t,-u, t u, -c/t, -c/u)$, which are consistent with the prepotential
\begin{equation}
	{\cal F}^\star = -\frac{3}{2} (-c)^{1/3} (X^1 X^2 X^3)^{2/3}\ .
\end{equation}

\subsection[Asymptotic AdS\texorpdfstring{$_4$}{(4)}]{Asymptotic AdS$_{\fakebold{4}}$}

The constant scalars and charges assumption holds for the AdS$_4$ vacuum, which satisfies the BPS equations asymptotically with
\begin{equation}
U = \log(r/L_{{\rm AdS}_4}), \qquad \psi = \log(r^2/L_{{\rm AdS}_4}),
\end{equation}
and
\begin{equation}
s = t = u = \e^{i \pi/3} c^{1/3}, \qquad \e^{2\phi} = 2 c^{-2/3} \ .
\end{equation}
If substitute the above field configuration in \eqref{app:BHequations}, as $r \rightarrow \infty$, we find 
\begin{equation}
	\alpha =- \frac{\pi}{6} \ , \qquad L_{{\rm AdS}_4} = \frac{1}{g}\ \frac{c^{1/6}}{3^{1/4}}\ , 
\end{equation}
which can be easily seen to solve the second and third equations in \eqref{app:BHequations}. The remaining one, ${\cal E} = 0$, is also asymptotically solved as can be verified by
\bea
 \frac{2}{L_{{\rm AdS}_4}} \im (\e^{-i \alpha} {\cal V}) &
 = - \kappa \Omega {\cal M} {\cal Q}^x {\cal P}^x \\
 & =3 g c \e^{\cal K}(c^{1/3}, -2 c^{-1/3},-2 c^{-1/3},-2 c^{-1/3}, c^{-2/3}, 1,1,1) \ .
\eea
Note that there is no way of fixing the asymptotic values of the massive vector charges $p^0 = c e_0$, but in the process we have fixed uniquely $\lambda$ to be aligned with $\kappa$ so that
\begin{equation}
	\kappa \lambda = -1 \quad \text{or} \qquad \lambda = -\kappa \ ,
\end{equation}
for a choice of positive electric coupling constant $g > 0$. 

\subsection{Near-horizon geometry}

The near-horizon equations are more involved than the asymptotic ones but we are again in the constant scalar case which guarantees that $s t u = -c$ solving the fourth and fifth equation in \eqref{app:BHequations}. The general near-horizon solution was analyzed in \cite{Guarino:2017pkw} but here we make an inspired ansatz for the scalars in a way that enforces $s t u = -c$:
\begin{equation}
	s = \frac{\e^{i \pi/3} c^{1/3} H^1}{(H^1 H^2 H^3)^{1/3}}\ , \quad t = \frac{\e^{i \pi/3} c^{1/3} H^2}{(H^1 H^2 H^3)^{1/3}}\ , \quad u = \frac{\e^{i \pi/3} c^{1/3} H^3}{(H^1 H^2 H^3)^{1/3}}\ ,
\end{equation}
under the condition that  $H^1 + H^2 + H^3 = 1$. With this ansatz we have imposed equal phases of the three scalars meaning we are killing some of the degrees of freedom, and practically restricting the solution to what we call ``purely magnetic'' solution (see the discussion in the main body of chapter \ref{ch:4}). The metric function ansatz is naturally given by
\begin{equation}
U = \log(r/L_{{\rm AdS}_2}), \qquad \psi = \log(L_{\Sigma_\fg} \cdot r / L_{{\rm AdS}_2}),
\end{equation}
where $L_{{\rm AdS}_2}$ is the radius of AdS$_2$ and $L_{\Sigma_\fg}$ that of the surface $\Sigma_\fg$. With this ansatz we solve the second and third equation in \eqref{app:BHequations} by setting
\begin{equation}
	\alpha = -\frac{\pi}{6}\ , \qquad L_{{\rm AdS}_2} = \frac{\e^{-{\cal K}/2} (H^1 H^2 H^3)^{1/3}}{2 g c^{1/3}}\ .
\end{equation}
The remaining symplectic vector of equations ${\cal E} = 0$ can be solved in several steps. The condition that $p^0 = c e_0$ imposes the constraint that ${\cal E}^0 = c {\cal E}_0$ which leads to
\begin{equation}
	\e^{2 \phi} = \frac{2 c^{-2/3}}{3 (H^1 H^2 H^3)^{1/3}}\ , \qquad p^0 = c e_0 = \frac{g c^{1/3}}{3 \sqrt{3} (H^1 H^2 H^3)^{1/3}} L^2_{\Sigma_\fg}\ , 
\end{equation}
while the electric charges are fixed by the components ${\cal E}_{1,2,3}$ to be
\begin{equation}
	e_1 = e_2 = e_3 = e = - \frac{g}{\sqrt{3}}  L^2_{\Sigma_\fg}\ .
\end{equation}
Note that the electric charges are equal and eventually fixed in terms of the magnetic charges, so they are not independent degrees of freedom. However, from the explicit expression it is clear that the value of $e$ is strictly not allowed to vanish, in accordance with the results in \cite{Guarino:2017eag,Guarino:2017pkw}. Finally, equations ${\cal E}^{1,2,3} = 0$ become
\begin{equation}
	\frac{2 g L^2_{\Sigma_\fg}}{3 \sqrt{3} c^{1/3} (H^1 H^2 H^3)^{2/3}}  = \frac{p^1}{H^1 (3 H^1-2)} = \frac{p^2}{H^2 (3 H^2-2)} = \frac{p^3}{H^3 (3 H^3-2)}\ ,
\end{equation}
which are solved by 
\begin{equation}
	L^2_{\Sigma_\fg} = - \frac{\sqrt{3}}{2 g} c^{1/3} (H^1 H^2 H^3)^{2/3} \sum_{I=1}^{3} \frac{p^I}{H^I} \, ,
\end{equation}
together with
\begin{equation}
3 H^I = 1 \pm \sum_{J, K} \frac{\left| \epsilon_{I J K} \big( p^J - p^K \big) \right|}{2 \sqrt{\big(\sqrt{\Theta} \pm p^I\big)^2 - p^J p^K}} \ ,
\end{equation}
where $\epsilon_{I J K}$ is the Levi--Civita symbol and $\Theta$ is defined in \eqref{theta}.

%


\chapter[Maxwell charges in \texorpdfstring{$\cN=2$}{N=2}, \texorpdfstring{${\rm D }=4$}{D=4} gauged supergravity]{Maxwell charges in $\fakebold{\cN=2}$, $\fakebold{{\rm D}=4}$ gauged supergravity}
\label{app:4d gauged sugra}

In this appendix we compute the Maxwell charges of a family of BPS black holes in ${\rm D }=4$, $\cN=2$ gauged supergravity
with prepotential \eqref{4d:prepotential:cubic}.
The ansatz for the metric and gauge fields is
\bea
 \label{ansatz:metric:gauge field:4d}
 \rd s^2 & = - \e^{2 U} \, \rd \tau^2 + \e^{-2 U} \, \rd r^2
 + \e^{2 \left( V - U \right)} \left( \rd \vartheta^2 + \sin^2\vartheta \, \rd \varphi^2 \right) \, , \\
 A^\Lambda & = \tilde{q}^\Lambda(r) \, \rd \tau - p^\Lambda(r) \cos\vartheta \, \rd \varphi \, .
\eea
The black hole magnetic and electric charges are then given by \cite{DallAgata:2010ejj}
\bea
 \label{em:charges:4d:period}
 p^\Lambda & \equiv \frac{1}{4 \pi} \int_{S^2} F^\Lambda \, , \\
 q_\Lambda & \equiv \frac{1}{4 \pi} \int_{S^2} G_\Lambda
 = \e^{2 \left( V - U \right)} \im \cN_{\Lambda \Sigma} \, \tilde{q}'^{\Sigma}
 + \re \cN_{\Lambda \Sigma} \,  p^{\Sigma} \, ,
\eea
where we defined the symplectic-dual gauge field strength,
\be
 \label{symplectic-dual:F}
 G_\Lambda \equiv 2 \frac{\delta \cL}{\delta F^{\Lambda}}
 =  \im \cN_{\Lambda \Sigma} \star_4 F^{\Sigma}
 + \re \cN_{\Lambda \Sigma} F^{\Sigma} \, ,
\ee
such that $\left( F^\Lambda , G_\Lambda \right)$ transforms as a $\left( 2 , n_{\rm V} + 2 \right)$
symplectic vector under electric-magnetic duality transformations.
Using \eqref{IIA:period:matrix} and \eqref{4d:prepotential:cubic} we find that
\bea
 \label{period:matrix:4d:scalars}
 \im \cN_{I J} & = \cG_{I J} \, , && \re \cN_{I J} = - C_{I J K} \re z^K \, ,\\
 \im \cN_{I 0} & = - \cG_{I J} \re z^J \, , && \re \cN_{I 0} = + \frac12 C_{I J K} \re z^J \re z^K \, ,\\
 \im \cN_{0 0} & = - \left( \e^{-6 \phi} - \cG_{I J} \re z^I \re z^J \right) \, , \qquad && \re \cN_{0 0} = - \frac13 C_{I J K} \re z^I \re z^J \re z^K  \, ,
\eea
where we defined
\bea
 \cG_{I J} \equiv C_{I J} - \frac{C_I C_J}{4 \e^{- 6 \phi}} \, .
\eea
Therefore, the electric charges read
\bea
 q_{0} & = - \e^{2 \left( V - U \right)} \tilde{q}'^{0}
 \left[ \e^{- 6 \phi} + \cG_{I J} \re z^J \left( \frac{\tilde{q}'^{I}}{\tilde{q}'^{0}} - \re z^I \right) \right]
+ \frac{p^{0}}{2} C_{I J K} \re z^J \re z^K \left( \frac{p^I}{p^0} - \frac23 \re z^I \right) \, , \\
 q_{I} & = \e^{2 \left( V - U \right)} \tilde{q}'^{0} \cG_{I J} \left( \frac{\tilde{q}'^{J}}{\tilde{q}'^{0}} - \re z^J \right)
 - p^{0} C_{I J K} \re z^K \left( \frac{p^J}{p^0} - \frac12 \re z^J \right) \, .
\eea

In the gauged STU model $(n_{\rm V} = 3)$ the only nonvanishing intersection numbers are $C_{1 2 3} = 1$ (and cyclic permutations).
Hence,
\bea
 \cG_{I J} =
 \begin{cases}
 - \frac{\im z^1 \im z^2 \im z^3}{ \left( \im z^I \right)^2} &\text{if } I = J \\
 0 &\text{otherwise}
 \end{cases} \, ,
\eea
and
\bea
 \label{em:charges:4d:scalars}
 q_{0} & = - \e^{2 \left( V - U \right)} \tilde{q}'^{0} \im z^{123}
 \left[ 1 - \sum_{I = 1}^{3} \frac{\re z^I}{\left( \im z^I \right)^2}
 \left( \frac{\tilde{q}'^{I}}{\tilde{q}'^{0}} - \re z^I \right) \right]
 - 2 p^0 \re z^{123} + \sum_{\substack{ I < J \\ ( \neq K)}} \re z^I \re z^J p^K \, , \\
 q_I & = - \e^{2 \left( V - U \right)} \tilde{q}'^{0} \frac{\im z^{123}}{\left( \im z^I \right)^2}
 \left( \frac{\tilde{q}'^{I}}{\tilde{q}'^{0}} - \re z^I \right)
 + \frac{\re z^{123}}{\re z^I} \bigg( p^0 - \sum_{J (\neq I)} \frac{p^J}{\re z^J} \bigg) \, ,
\eea
where we employed the following notation
\bea
 \im z^{123} \equiv \im z^1 \im z^2 \im z^3 \, , \qquad \qquad
 \re z^{123} \equiv \re z^1 \re z^2 \re z^3 \, .
\eea

%


\chapter{Supersymmetric Casimir energy}
\label{AppC}

The partition function of an $\cN = 1$ supersymmetric gauge theory with a non-anomalous $\U(1)_R$ symmetry
on a Hopf surface $\cH_{p, q} \simeq S^3 \times S^1$, with a complex structure characterized by two parameters $p$, $q$, can be written as
\bea\label{Zindex}
 Z \left[ \cH_{p,q} \right] = \e^{ - E_{\text{susy}}} I(p, q) \, .
\eea
Here, $I(p, q)$ is the superconformal index \cite{Kinney:2005ej,Romelsberger:2005eg} 
\bea I(p, q)= {\rm  Tr} (-1)^F  p^{  h_1 + r/2} q^{h_2+r/2} \, ,\eea
where $h_1$ and $h_2$ are the generators of rotation in orthogonal planes and $r$ is the superconformal R-symmetry.
$E_{\text{susy}}$ is the supersymmetric Casimir energy \cite{Assel:2014paa,Lorenzen:2014pna,Assel:2015nca},
\bea\label{Ecas}
 E_{\text{susy}} (b_1 , b_2) = \frac{4 \pi}{27} \frac{\left( |b_1| + |b_2| \right)^3}{|b_1| |b_2|} \left( 3 c - 2 a \right) - \frac{4 \pi}{3} \left( |b_1| + |b_2| \right) \left( c - a \right) \, ,
\eea 
where $p = \e^{- 2 \pi |b_1|}$, $q = \e^{- 2 \pi |b_2|}$, and $a$, $c$ are the central charges of the four-dimensional $\cN = 1$ theory.
We can extrapolate this result to include flavor symmetries by considering $a$ and $c$ as trial central charges,
depending on a set of chemical potentials $\hat\Delta_I$,
\bea\label{central charges}
 a(\hat \Delta_I) = \frac{9}{32}  {\rm  Tr} R(\hat \Delta_I)^3 - \frac{3}{32}{\rm  Tr} R(\hat \Delta_I) \, ,
 \qquad  c(\hat \Delta_I) = \frac{9}{32}  {\rm  Tr} R(\hat \Delta_I)^3 - \frac{5}{32}{\rm  Tr} R(\hat \Delta_I) \, ,
\eea
where $R$ is a choice of $\U(1)_R$ symmetry and the trace is over all fermions in the theory.
The supersymmetric Casimir energy can be also interpreted as the vacuum 
expectation value $\langle H_{\rm susy} \rangle$ of the Hamiltonian which generates time translations \cite{Assel:2015nca}.
It can also be obtained by integrating the anomaly polynomials in six dimensions \cite{Bobev:2015kza}.
In particular, the supersymmetric Casimir energy for $\cN = 4$ SYM with $\SU(N)$ gauge group,
where $a=c$, reads\footnote{The R-symmetry 't Hooft anomalies for $\cN = 4$ SYM are given by
$\Tr R(\hat \Delta_I)=(N^2-1) [\sum_{I=1}^3 (\hat\Delta_I-1)+1]=0$ and
$\Tr R(\hat \Delta_I)^3 = (N^2-1)  [\sum_{I=1}^3 (\hat\Delta_I-1)^3+1] =3 (N^2-1) \hat \Delta_1 \hat \Delta_2 \hat \Delta_3$.
The $\hat\Delta_I$ are the R-symmetries of the three adjoint chiral multiplets of $\cN=4$ SYM and
satisfy (\ref{constraint2:N=4}).}
\bea
 \label{Casimir:SUSY2:N=4}
 E_{\text{susy}} = \frac{\pi}{8} \left( N^2 - 1 \right) \frac{\left( |b_1| + |b_2| \right)^3}{|b_1| |b_2|} \hat \Delta_1 \hat \Delta_2 \hat \Delta_3 \, ,
\eea
where $\hat \Delta_{1, 2, 3}$ are the chemical potentials for the Cartan generators of the R-symmetry, fulfilling the constraint
\bea
 \label{constraint2:N=4}
 \hat \Delta_1 + \hat \Delta_2 + \hat \Delta_3 = 2 \, . 
\eea
We can rewrite  Eq.\,\eqref{Casimir:SUSY2:N=4} in the notation used in the main text as 
\bea
  \label{modified:Casimir:SUSY2:N=4}
 E_{\text{susy}} = - i \pi  (N^2-1)  \frac{\Delta_1 \Delta_2 \Delta_3}{\omega_1 \omega_2} 
 \, ,
\eea
where we defined $\Delta_I =i  \left( |b_1| + |b_2| \right) \hat \Delta_I / 2$ and $\omega_i=-i |b_i|$. 
They satisfy the constraint
\bea
 \label{modified:constraint2:N=4}
 \Delta_1 + \Delta_2 + \Delta_3 +\omega_1+\omega_2 = 0 \, .
\eea
When extended to the complex plane $\Delta_I$ and $\omega_i$ are defined modulo 1. 

The constraints \eqref{modified:constraint2:N=4} and \eqref{constraint:N=4} are closely related to the absence of pure and mixed gauge anomalies. Let us briefly see why.
Consider again a generic $\cN=1$ supersymmetric gauge theory. The constraint \eqref{modified:constraint2:N=4} is modified to
\bea
 \label{modified:constraint3:N=4}
\sum_{I\in W_a} \Delta_I +\sum_{i=1}^2 \omega_i = 0 \, ,
\eea
where $\Delta_I$ is the chemical potential for the $I$th field and the sum is restricted  to the fields entering the superpotential term $W_a$.
There is one constraint for each monomial $W_a$ in the superpotential of the theory.
The path integral on $S^3\times S^1$ localizes to a matrix model where one-loop determinants must be regularized.
$E_{\text{susy}}$  arises from the following regularization factors \cite{Assel:2015nca,Bobev:2015kza}
\bea \Psi(u)=  i \pi \frac{(\sum_{i=1}^2 \omega_i)^3}{24 \omega_1\omega_2}\left [ ( \hat u -1)^3 -\frac{\sum_{i=1}^2 \omega_i^2}{(\sum_{i=1}^2 \omega_i)^2} (\hat u-1)\right] \, ,\eea
where $\hat u =- 2 u/\sum_{i=1}^2\omega_i$ and $u$ is a chemical potential for gauge or flavor symmetries.
More precisely, denoting with $z$ the gauge variables, $E_{\text{susy}}$ gets an additive contribution $\Psi(z+\Delta_I)$ from each chiral multiplet and $-\Psi(z)$ for each vector multiplet. One can pull out regularization factors  from the matrix model only if they are independent of the gauge variables.
The constraint \eqref{modified:constraint3:N=4} implies that $\sum_{I\in W_a} \hat \Delta_I=2$, where again $\hat \Delta_I =- 2 \Delta_I/\sum_{i=1}^2\omega_i$.
Hence $\hat\Delta_I$ parameterizes a trial R-symmetry of the theory. One can see that, if all (pure and mixed) gauge anomalies cancel,  $E_{\text{susy}}$ is indeed independent of $z$ if the chemical potential satisfy \eqref{modified:constraint3:N=4} \cite{Bobev:2015kza}. The final result for 
$E_{\text{susy}}$ is then easily computed to be,
\bea\label{ES}  E_{\text{susy}} =  i \pi \frac{(\sum_{i=1}^2 \omega_i)^3}{24 \omega_1\omega_2}\left [ {\rm  Tr} R(\hat \Delta_I)^3 -\frac{\sum_{i=1}^2 \omega_i^2}{(\sum_{i=1}^2 \omega_i)^2} {\rm  Tr} R(\hat \Delta_I)\right] \, . \eea
Using \eqref{central charges} and $\omega_i=-i |b_i|$ we recover \eqref{Ecas}.

A similar quantity constructed from 
\bea \tilde\Psi(u)= i \pi \frac{(\sum_{i=1}^3 \omega_i)^3}{24 \omega_1\omega_2\omega_3}\left [ ( \hat u -1)^3 -\frac{\sum_{i=1}^3 \omega_i^2}{(\sum_{i=1}^3 \omega_i)^2} (\hat u-1)\right] \, ,\eea
appears in the modular transformation of the integrand of the matrix model \cite{Brunner:2016nyk}.\footnote{This can be
expressed in terms of Bernoulli polynomials as $\tilde\Psi(u)=  \frac{\pi i}{3} B_{3,3}(u; \omega_i)$  \cite{Brunner:2016nyk}.}
Here the angular momentum fugacities are written as $p=\e^{-2\pi i \omega_1/\omega_3},
q=\e^{-2\pi i \omega_2/\omega_3}$,  the gauge fugacity as $\e^{2\pi u/\omega_3}$, and
$\hat u =- 2 u/\sum_{i=1}^3\omega_i$. In a theory with no gauge anomalies, the sum of $\tilde\Psi(z+\Delta_I)$ from each chiral multiplet and $-\tilde\Psi(z)$ for each vector multiplet is  independent of the
gauge variables $z$  and can be pulled out of the integral if the constraint
\bea
 \label{modified:constraint4:N=4}
\sum_{I\in W_a} \Delta_I +\sum_{i=1}^3 \omega_i = 0 \, ,
\eea
is fulfilled \cite{Spiridonov:2012ww,Brunner:2016nyk}. The sum then becomes
\bea
 \label{FS}
 \varphi=  i \pi \frac{(\sum_{i=1}^3 \omega_i)^3}{24 \omega_1\omega_2\omega_3}
 \left [ {\rm  Tr} R(\hat \Delta_I)^3 -
 \frac{\sum_{i=1}^3 \omega_i^2}{(\sum_{i=1}^3 \omega_i)^2} {\rm  Tr} R(\hat \Delta_I)\right] \, .
\eea
This observation has been used in  \cite{Brunner:2016nyk} to write the index as
\bea
 \label{FS2}
 I(p,q,\Delta_I) = \e^{\varphi} I^{{\rm mod}} (\omega_i,\Delta_I) \, ,
\eea
where $I^{{\rm mod}}$ is a modified matrix model depending on modified elliptic gamma functions. 
As noticed in \cite{Brunner:2016nyk}, the supersymmetric Casimir energy can be extracted from $\varphi$ in the low-temperature limit. 
Indeed, in the limit $\omega_3 = 1/\beta \to 0$, $\varphi$  becomes exactly $\beta E_{\text{susy}}$.
It is interesting to notice that, for $\cN = 4$ SYM,  by setting one of the $\omega_i$
equal to $-1$, $\varphi$ reduces to our quantity \eqref{modified:Casimir:SUSY:N=4}
and \eqref{modified:constraint4:N=4}  to the constraint \eqref{constraint:N=4}. 
This is an observation to explore further. In particular,
it would be interesting to understand the physical meaning of $I^{{\rm mod}}$ and  the decomposition
\eqref{FS2}.

\end{appendices}

\begin{spacing}{0.9}


\bibliographystyle{ytphys}
\cleardoublepage
\bibliography{References/references} 



\end{spacing}

\printthesisindex 

\end{document}